	\documentclass[
		paper=letter,%
		abstract=true,
		11pt,american,pagesize,%
		DIV=12,headinclude=true,headlines=0,footinclude=true,footlines=-5,twoside=semi,%
	]{scrartcl}
	\def\docclass{koma}

	\def\version{arxiv}

	\def\draftmode{false} 

\usepackage[utf8]{inputenc}
\makeatletter
\usepackage{xifthen}

\newcommand\iflipics[2]{\ifthenelse{\equal{\docclass}{lipics}}{#1}{#2}}
\newcommand\ifkoma[2]{\ifthenelse{\equal{\docclass}{koma}}{#1}{#2}}
\newcommand\ifieee[2]{\ifthenelse{\equal{\docclass}{ieee}}{#1}{#2}}
\ifthenelse{ \equal{\docclass}{lipics} \OR \equal{\docclass}{koma} \OR \equal{\docclass}{ieee} }{
	\PackageInfo{paper}{Building paper with docclass = \docclass} 
}{
	\PackageWarning{paper}{docclass = "\docclass", but must be one of "lipics", "koma", "ieee"}
}

\newcommand\ifmanuscript[2]{\ifthenelse{\equal{\version}{manuscript}}{#1}{#2}}
\newcommand\ifarxiv[2]{\ifthenelse{\equal{\version}{arxiv}}{#1}{#2}}
\newcommand\ifsubmission[2]{\ifthenelse{\equal{\version}{submission}}{#1}{#2}}
\newcommand\ifproceedings[2]{\ifthenelse{\equal{\version}{proceedings}}{#1}{#2}}
\ifthenelse{ 
	\equal{\version}{manuscript} 
	\OR \equal{\version}{arxiv} 
	\OR \equal{\version}{submission} 
	\OR \equal{\version}{proceedings} 
}{
	\PackageInfo{paper}{Building paper version = \version} 
}{
	\PackageWarning{paper}{version = "\version", but must be one of "manuscript", "arxiv", "submission", "proceedings"}
}

\newcommand\ifdraft[2]{\ifthenelse{\equal{\draftmode}{true}}{#1}{#2}}
\ifthenelse{ \equal{\draftmode}{true} \OR \equal{\draftmode}{false} }{
	\PackageInfo{paper}{Building paper with draftmode = \draftmode} 
}{
	\PackageWarning{paper}{draftmode = "\draftmode", but must be "true" or "false"}
}

\usepackage[T1]{fontenc}
\ifieee{}{
	\usepackage{lmodern}
	\usepackage{slantsc}
}

\usepackage{babel}
\input{ushyphex.tex} 

\usepackage{array,multicol,makecell}
\usepackage[referable]{threeparttablex} 
\ifieee{
	\usepackage[cmex10]{amsmath,mathtools}
	\usepackage{amsfonts,amssymb}
}{
	\usepackage{amsmath,amsfonts,amssymb,mathtools}
}

\usepackage{mleftright}\mleftright 
\usepackage{relsize,xspace,booktabs,adjustbox,needspace,pbox,relsize}
\ifieee{
	
	\usepackage{enumitem}
}{
	\usepackage{enumitem}
}
\usepackage{graphicx}

\usepackage{colonequals}

\usepackage{wref}
\usepackage[bibtex]{url-doi-arxiv}

\newdimen\makeboxdimen
\newcommand\makeboxlike[3][l]{%
\setbox0=\hbox{#2}%
\global\makeboxdimen=\wd0%
\setbox1=\hbox{\makebox[\makeboxdimen][#1]{%
\makebox[0pt][#1]{#3}%
}}%
\ht1=\ht0%
\dp1=\dp0%
\box1%
}

\newcommand\like[3][c]{%
	\mathchoice{
		\makeboxlike[#1]{%
			\ensuremath{\displaystyle\relax#2}%
		}{%
			\ensuremath{\displaystyle\relax#3}%
		}%
	}{
		\makeboxlike[#1]{%
			\ensuremath{\textstyle\relax#2}%
		}{%
			\ensuremath{\textstyle\relax#3}%
		}%
	}{
		\makeboxlike[#1]{%
			\ensuremath{\scriptstyle\relax#2}%
		}{%
			\ensuremath{\scriptstyle\relax#3}%
		}%
	}{
		\makeboxlike[#1]{%
			\ensuremath{\scriptscriptstyle\relax#2}%
		}{%
			\ensuremath{\scriptscriptstyle\relax#3}%
		}%
	}
}

\newcommand\plaincenter[1]{%
	\mbox{}\hfill#1\hfill\mbox{}%
}

\ifkoma{
	\setlength\parindent{1.5em}
	\usepackage[headsepline]{scrlayer-scrpage}
	\pagestyle{scrheadings}
	\clearscrheadfoot
	\AtBeginDocument{%
		\automark[section]{}%
	}
	\ohead{\pagemark}
	\rehead{\mytitle}
	\lohead{\headmark}
	\addtokomafont{caption}{\sffamily\small}
	\addtokomafont{captionlabel}{\sffamily\textbf}
	\setcapmargin{2em}
}{}

\AtBeginDocument{%
	\let\mytitle\@title%
}

\let\oldthebibliography\thebibliography
\renewcommand\thebibliography[1]{%
	\oldthebibliography{#1}%
	\pdfbookmark[1]{References}{}%
}

\usepackage{lscape} 

\ifkoma{
	\usepackage{float}
	\floatstyle{plain}
	\usepackage{newfloat}
	\DeclareFloatingEnvironment[%
			name=Algorithm,%
			placement=thb,%
		]{algorithm}
}{}
\iflipics{
	\usepackage{newfloat}
	\DeclareFloatingEnvironment[%
			name=Algorithm,%
			placement=thb,%
		]{algorithm}
}{}

\usepackage{textcomp} 
\usepackage{tikz}

\usetikzlibrary{positioning,arrows.meta}
\usetikzlibrary{backgrounds,calc,trees,graphs}

\pgfdeclarelayer{background}
\pgfsetlayers{background,main}

\usepackage[framemethod=tikz,xcolor]{mdframed}

\mdfdefinestyle{examplebox}{
	backgroundcolor=black!10,
	roundcorner=5pt,
	linewidth=0pt,
}

\newmdenv[style=examplebox]{examplebox}

\iflipics{
	\newtheorem{fact}[theorem]{Fact}
	\newcommand\thmemph[1]{\emph{#1}}
	\newenvironment{proofof}[1]{%
		\begin{proof}[{{Proof of #1{}}}]%
	}{%
		\end{proof}%
	}
	
	\newmdtheoremenv[style=examplebox]{examplecorollary}[theorem]{Corollary}
	\surroundwithmdframed[style=examplebox]{example}
}{
	\usepackage[amsmath,hyperref,thmmarks]{ntheorem}
	
	\theorembodyfont{\slshape}
	\theoremseparator{:}
	\newtheoremstyle{proofstyle}%
	  {\item[\theorem@headerfont\hskip\labelsep ##1\theorem@separator]}%
	  {\item[\theorem@headerfont\hskip\labelsep ##3\theorem@separator]}
	
	\theorempreskip{\topsep} 
	
	\newtheorem{theorem}{Theorem}[section]
	
	\theoremstyle{plain}
	\theorempreskip{\topsep}

	\newtheorem{lemma}[theorem]{Lemma}
	
	\newtheorem{corollary}[theorem]{Corollary}
	\newtheorem{definition}[theorem]{Definition}
	
	\newmdtheoremenv[style=examplebox]{examplecorollary}[theorem]{Corollary}

	\theoremstyle{plain}
	\theorembodyfont{\upshape}

	\newtheorem{fact}[theorem]{Fact}
	\newtheorem{remark}[theorem]{Remark}
	\newtheorem{example}[theorem]{Example}
	
	\surroundwithmdframed[style=examplebox]{example}

	\theoremsymbol{\raisebox{-.25ex}{$\Box$}}
	\qedsymbol{\raisebox{-.25ex}{$\Box$}}

	\theoremstyle{proofstyle}
	\newtheorem{proof}{Proof}

	\newcommand\thmemph[1]{\textit{#1}}
}

\iflipics{
	\newenvironment{thmenumerate}[2][]{%
		\begin{enumerate}[
			label={\textsf{\textbf{\color{darkgray}{\makebox[\widthof{(iii)}][r]{\textup{(\roman*)}}}}}},
			ref={\ref{#2}\kern.1em--\kern.1em(\roman*)},
			itemsep=0pt,
			topsep=.5ex,
			leftmargin=2em,
			#1
		]%
	}{%
		\end{enumerate}%
	}
}{
	\newenvironment{thmenumerate}[2][]{%
		\begin{enumerate}[
			label={\makebox[\widthof{(iii)}][c]{\textup{(\roman*)}}},
			ref={\ref{#2}\kern.1em--\kern.1em(\roman*)},
			itemsep=0pt,
			#1
		]%
	}{%
		\end{enumerate}%
	}
}

\newcommand*\ie{\mbox{i.\hspace{.2ex}e.}}
\newcommand*\eg{\mbox{e.\hspace{.2ex}g.}}
\newcommand*\iid{\mbox{i.\hspace{.2ex}i.\hspace{.2ex}d.}\xspace}

\newcommand*\aka{\mbox{a.\hspace{.2ex}k.\hspace{.2ex}a.}\xspace}

\newcommand\N{\mathbb N}

\newcommand\Oh{O}

\usepackage{fixmath}

\newcommand{\ESymbol}{\mathbb{E}}

\newcommand{\ProbSymbol}{\ensuremath{\mathbb{P}}}

\DeclarePairedDelimiterXPP\Prob[1]{\ProbSymbol}[]{}{%
	#1%
}
\DeclarePairedDelimiterXPP\E[1]{\ESymbol}[]{}{%
	#1%
}
\DeclarePairedDelimiterXPP\Eover[2]{\ESymbol_{#1}}[]{}{%
	#2%
}
\DeclarePairedDelimiterXPP\ProbIn[2]{\ProbSymbol_{#1}}[]{}{%
	#2%
}
\providecommand{\Prob}{} 
\providecommand{\ProbIn}{} 
\providecommand{\E}{} 
\providecommand{\Eover}{} 

\newcommand{\surroundedmath}[3]{
	\mathchoice{
		#1{#2{#3}#2}%
	}{
		#1{#3}%
	}{
		#1{#3}%
	}{
		#1{#3}%
	}%
}
\newcommand\rel[1]{\surroundedmath{\mathrel}{\:}{#1}}
\newcommand\wrel[1]{\surroundedmath{\mathrel}{\;}{#1}}
\newcommand\wwrel[1]{\surroundedmath{\mathrel}{\;\;}{#1}}
\newcommand\bin[1]{\surroundedmath{\mathbin}{\:}{#1}}
\newcommand\wbin[1]{\surroundedmath{\mathbin}{\;}{#1}}

\newcommand{\relwithref}[3][c]{%
	\mathrel{\underset{\mathclap{\makebox[\widthof{$=$}][#1]{\scriptsize\wref{#2}}}}{#3}}%
}

\newcommand\ppe{\phantom{=}}

\iflipics{}{
	\makeatletter
	\let\oldalign\align
	\let\endoldalign\endalign
	\renewenvironment{align}{%
		\begingroup%
		\let\oldhalign\halign
		\def\halign{%
			\let\oldbreak\\%
			\def\nonnumberbreak{\nonumber\oldbreak*}%
			\def\\{%
				\@ifstar{\nonnumberbreak}{\oldbreak}%
			}%
			\oldhalign%
		}
		\oldalign%
	}{%
		\endoldalign%
		\endgroup%
	}
}
\newcommand*\numberthis[1][]{\stepcounter{equation}\tag{\theequation}}

\allowdisplaybreaks[3]

\newcommand\splitaftercomma[1]{%
  \begingroup
  \begingroup\lccode`~=`, \lowercase{\endgroup
    \edef~{\mathchar\the\mathcode`, \penalty0 \noexpand\hspace{0pt plus .25em}}%
  }\mathcode`,="8000 #1%
  \endgroup
}

\newcommand\separatedpar{%
	\medskip
	\plaincenter{%
		$*$\hspace{3em}$*$\hspace{3em}$*$%
	}\par
	\medskip
	\noindent
}

\def\mydots{\xleaders\hbox to.5em{\hfill.\hfill}\hfill}
\newlength\tmpLenNotations
\newenvironment{notations}[1][10em]{%
	\small%
	\newcommand\notationentry[1]{%
		\settowidth\tmpLenNotations{##1}%
		\ifthenelse{\lengthtest{\tmpLenNotations > \labelwidth}}{%
			\parbox[b]{\labelwidth}{%
				\makebox[0pt][l]{##1}\\%
			}%
		}{%
			\mbox{##1}%
		}%
		\mydots\relax%
	}%
	\begin{list}{}{%
		\setlength\labelsep{0em}%
		\setlength\labelwidth{#1}%
		\setlength\leftmargin{\labelwidth+\labelsep+1em}%
	}
	\ifdraft{%
		\newcommand\notation[1]{%
			\item[{##1}] \marginpar{%
				\adjustbox{scale={.55}{1},outer}{%
					\color{brown!80}%
					\scriptsize$\triangleright$\;\tiny\texttt{\detokenize{##1}}%
				}%
			}%
		}%
	}{%
		\newcommand\notation[1]{\item[{##1}]}%
	}%
	\raggedright%
}{%
	\end{list}%
}

\ifdraft{%
	\overfullrule=6mm
	
	\usepackage[color,notref,notcite]{showkeys}
	\definecolor{refkey}{gray}{.99}
	\colorlet{labelkey}{green!60!black!60}
	
	\usepackage[inline,nolabel]{showlabels}

	\showlabels{cite}
	\showlabels{citealt}
	\showlabels{citealp}
	\showlabels{citet}
	\showlabels{Citet}
	\showlabels{citep}
	\showlabels{citeauthor}
	\showlabels{Citeauthor}
	\showlabels{citefullauthor}
	\showlabels{citeyear}
	\showlabels{citeyearpar}
	\showlabels{wildref}
	\showlabels{wildpageref}
	\showlabels{wildtpageref}
}{}

\iflipics{
	\ifmanuscript{\hideLIPIcs}{}
	\ifarxiv{\hideLIPIcs}{}
	\ifsubmission{}{\nolinenumbers}
}{}

\ifdraft{}{%
	\usepackage{microtype}
}

\hypersetup{
	final,
	unicode=true, 
	bookmarks=true,
	bookmarksnumbered=true,
	bookmarksdepth=2,
	bookmarksopen=true,
	breaklinks=true,
	hidelinks,
}

\newsavebox\tmpbox

\iflipics{
	\let\oldparagraph\paragraph
	\renewcommand\paragraph[1]{%
		\subparagraph{#1.}
	}
}{
	\let\oldparagraph\paragraph
	\renewcommand\paragraph[1]{%
		\oldparagraph{#1.}
	}
}

\let\epsilon\varepsilon

\def\myacknowledgements{}
\ifkoma{
	\newcommand\acknowledgements[1]{\def\myacknowledgements{\paragraph{Acknowledgements}#1}}
}{}
\ifieee{
	\newcommand\acknowledgements[1]{\def\myacknowledgements{\section*{Acknowledgement}#1}}
}{}

\newcommand\idtt[1]{\texttt{\upshape #1}\xspace}

\def\DFUDS{\idtt{DFUDS}}

\def\BP{\idtt{BP}}

\def\TC{\idtt{TC}}
\def\pre{\idtt{PRE}}
\def\post{\idtt{POST}}
\def\inorder{\idtt{IN}}

\def\TrParent{\idtt{parent}}
\def\TrChild{\idtt{child}}
\def\TrLastChild{\idtt{last\_child}}
\def\TrFirstChild{\idtt{first\_child}}
\def\TrNextSibling{\idtt{next\_sibling}}
\def\TrPrevSibling{\idtt{prev\_sibling}}
\def\TrDepth{\idtt{depth}}
\def\TrLevAnc{\idtt{anc}}
\def\TrRank{\idtt{node\_rank}}
\def\TrSelect{\idtt{node\_select}}
\def\TrChildRank{\idtt{child\_rank}}
\def\TrDeg{\idtt{degree}}
\def\TrNbDesc{\idtt{nbdesc}}
\def\LCA{\idtt{LCA}}
\def\TrLCA{\idtt{LCA}}
\def\TrHeight{\idtt{height}}

\def\TrLeftLeaf{\idtt{leftmost\_leaf}}
\def\TrRightLeaf{\idtt{rightmost\_leaf}}

\def\TrLeafRank{\idtt{leaf\_rank}}
\def\TrLeafSel{\idtt{leaf\_select}}
\def\TrLevelLeft{\idtt{level\_leftmost}}
\def\TrLevelRight{\idtt{level\_rightmost}}
\def\TrLevelSucc{\idtt{level\_succ}}
\def\TrLevelPred{\idtt{level\_pred}}

\def\TrHeight{\idtt{height}}

\def\TrChildLeft{\idtt{left\_child}}
\def\TrChildRight{\idtt{right\_child}}

\newcommand\RMQ{\idtt{RMQ}}

\newcommand\FCNS{\operatorname{fcns}}
\newcommand\treenode{\mathnormal{\bullet}}

\newcommand\hypsuc[1]{|\mathsf{H}(#1)|}

\newcommand{\pos}{\operatorname{pos}}
\newcommand{\type}{\operatorname{type}}
\newcommand{\childtype}{\operatorname{childtype}}
\newcommand{\idd}{\operatorname{id}}

\tikzset{
	inner node/.style={circle, fill,minimum size=5pt,inner sep=0pt},
	leaf node/.style={rectangle,draw,minimum size=6pt,inner sep=0pt},
}

\wref@putshortname{sec}{Sec.\!}

\makeatother

\usepackage{tabularx}
\usepackage{rotating}

\usepackage{hyperref}

\title{Hypersuccinct Trees -- 
New universal tree source codes for optimal compressed tree data structures and range minima}
\AtBeginDocument{\def\mytitle{Hypersuccinct Trees}}
	\newcommand\email[1]{\texttt{#1}}
	\author{%
		J.\ Ian Munro%
			\footnote{University of Waterloo, Canada, 
			\email{imunro\,@\,uwaterloo.ca}} 
	\and 
		Patrick K. Nicholson%
			\footnote{
			\email{pat.nicholson\,@\,gmail.com}}
	\and 
		Louisa Seelbach Benkner%
			\footnote{Universität Siegen, Germany,
			\email{seelbach\,@\,eti.uni-siegen.de}, supported by the DFG project LO 748/10-2}
	\and
		Sebastian Wild%
			\footnote{University of Liverpool, UK, 
			\email{wild\,@\,liverpool.ac.uk}
			}
	}
	
	\date{\small\today}

\acknowledgements{We thank Conrado Martínez and Markus Lohrey for valuable discussions and feedback on earlier drafts of this paper.

This work has been supported in part by 
		the Canada Research Chairs Programme and an NSERC Discovery Grant,
		the DFG research project LO 748/10-2 (QUANT-KOMP), and 
		the NeST (Network Sciences and Technologies) EEECS School initiative of University of Liverpool.}

\begin{document}

\maketitle

	\vspace*{-6ex}%

\begin{abstract}
We present a new universal source code for distributions of unlabeled binary and ordinal trees
that achieves optimal compression to within lower order terms for all tree sources
covered by existing universal codes.
At the same time, it supports answering many navigational queries 
on the compressed representation in constant time on the word-RAM;
this is not known to be possible for any existing tree compression method.
The resulting data structures, ``hypersuccinct trees'', hence combine the 
compression achieved by the best known universal codes 
with the operation support of the best succinct tree data structures.

We apply hypersuccinct trees to obtain a universal compressed data structure 
for range-minimum queries. It has constant query time and 
the optimal worst-case space usage of $2n+o(n)$ bits,
but the space drops to $1.736n + o(n)$ bits on average
for random permutations of $n$ elements,
and $2\lg\binom nr + o(n)$ for arrays with $r$ increasing runs, respectively.
Both results are optimal; the former answers an open problem
of Davoodi et al.~(2014) and Golin et al.~(2016).

Compared to prior work on succinct data structures, 
we do not have to tailor our data structure to specific applications; 
hypersuccinct trees automatically adapt to the trees at hand.
We show that they simultaneously achieve the 
optimal space usage to within lower order terms
for a wide range of distributions over tree shapes, including:
binary search trees (BSTs) generated by insertions in random order / 
Cartesian trees of random arrays, 
random fringe-balanced BSTs,
binary trees with a given number of binary/unary/leaf nodes,
random binary tries generated from memoryless sources,
full binary trees, unary paths, 
as well as uniformly chosen weight-balanced BSTs, AVL trees, 
and left-leaning red-black trees.
\end{abstract}

	\bigskip
	\setcounter{tocdepth}{2}
	\noindent\plaincenter{\fbox{\resizebox{.9\linewidth}!{\begin{minipage}{1.8\linewidth}
	\sffamily
	\setlength\columnsep{5em}\begin{multicols}{3}
	\tableofcontents
	\end{multicols}
	\end{minipage}}}}

	\bigskip

	\noindent

%
%
%
%
%
%
%
%
%

%
%
%
%
%
%
%
%
%

%

%
%

\section{Introduction}\label{sec:introduction}

As space usage and memory access become the bottlenecks in computation,
working directly on a compressed representation 
(``computing over compressed data'')  has become 
a popular field.
For text data, substantial progress over the last two decades culminated in
compressed text indexing methods that had
wide-reaching impact on applications and satisfy strong analytical guarantees.
For structured data, the picture is much less developed and clear.
In this paper, we develop the analog of entropy-compressed string indices for trees:
a data structure that allows one to query a tree stored in compressed form,
with optimal query times and space matching the best universal tree codes.

Computing over compressed data became possible
by combining techniques from information theory, string compression, and data structures.
The central object of study in (classical) information theory is that of a \emph{source} 
of random strings, 
whose entropy rate is the fundamental limit for source coding.
The ultimate goal in compressing such strings is a \emph{universal code}, which achieves optimal compression (to within lower order terms) for \emph{distributions} of strings  from 
a large class of possible sources
\emph{without knowing the used source}.
\ifsubmission{}{\par}%
A classic result in this area is that 
Lempel-Ziv methods are universal codes for finite-state sources, \ie,
sources in which the next symbol's distribution depends 
on the previous $k$ emitted symbols (see, \eg, \cite[\S\,13]{CoverThomas2006}).
The same is true for methods based on the 
Burrows-Wheeler-transform~\cite{EffrosVisweswariahKulkarniVerdu2002}
and for grammar-based compression~\cite{KiefferYang2000}.
The latter two results were only shown around 2000, 
marking a renewed interest in compression methods.

The year 2000 also saw breakthroughs in compressed text indexing,
with the first compressed self-indices 
that can represent a string and support pattern matching queries using 
 $\Oh(n H_0)$ bits of space~\cite{GrossiVitter2000,GrossiVitter2005}
and $\Oh(n H_k)+o(n \log |\Sigma|)$ bits of space~\cite{FerraginaManzini2000}
for $H_k$ the $k$th order empirical entropy of the string (for $k \geq 0$);
many improvements have since been obtained on space and query time;
(see~\cite{NavarroMakinen2007,BelazzouguiMakinenValenzuela2014} for surveys and
\cite{Ganczorz19}~for lower bounds on redundancy;
\cite{Navarro2021a,Navarro2021b}~summarizes more recent trends).
For strings, computing over compressed data has mainly been achieved.

In this article, we consider structure instead of strings; 
focusing on one of the simplest forms of structured data:
unlabeled binary and ordinal trees.
Unlike for strings, 
the information theory of structured data is still in its infancy.
Random sources of binary trees have (to our knowledge) first been suggested and analyzed 
in 2009~\cite{KiefferYangSzpankowski2009};
a more complete formalization then appeared in~\cite{ZhangYangKieffer2014},
together with a first universal tree source code.

For trees, computational results predate information-theoretic developments.
\emph{Succinct data structures} date back to 1989~\cite{Jacobson1989} and have their
roots in storing trees space-efficiently while supporting fast queries.
A succinct data structure is allowed to use $\lg U_n(1+o(1))$ bits of space
to represent one out of $U_n$ possible objects of size $n$~-- 
corresponding to a uniform distribution over these objects.
This has become a flourishing field, and 
several succinct data structures for ordinal or cardinal (including binary) trees supporting 
many operations are known~\cite{Navarro2016}.
Apart from the exceptions discussed below (in particular~\cite{JanssonSadakaneSung2012,DavoodiNavarroRamanRao2014}),
these methods do not achieve any compression beyond $\lg U_n$
no matter what the input is.
\ifsubmission{}{\par}%
At the other end of the spectrum, more recent representations for highly repetitive trees~\cite{BilleGortzLandauWeimann2015, BilleLandauRamanSadakaneRaoWeimann2015, GanardiHuckeJezLohreyNoeth2017,GanardiHuckeLohreyNoeth2017,GasconLohreyManethRehSieber2020,GawrychowskiJez2016}
can realize exponential space savings over $\lg U_n$ in extreme cases, 
but recent lower bounds~\cite{Prezza2019} imply that these methods 
cannot simultaneously achieve constant time%
\footnote{%
	All running times assume the word-RAM model with word size $w = \Theta(\log n)$.
}
for queries;
they are also not known to be succinct when the tree is not highly compressible.

\begin{table}[thb]
	\caption{%
		Navigational operations on succinct binary trees. 
		($v$ denotes a node and $i$ an integer).
	}
	\plaincenter{\fbox{%
	\newcommand\opitem[2]{\footnotesize#1 &\footnotesize #2 \\[.2ex]}
		
		\begin{minipage}{35em}\footnotesize
		\begin{tabular}{p{9em}p{28em}}%
	
		\opitem{$\TrParent(v)$}{the parent of $v$, same as $\TrLevAnc(v,1)$}
		\opitem{$\TrDeg(v)$}{the number of children of $v$}
		\opitem{$\TrChildLeft(v)$}{the left child of node $v$}
		\opitem{$\TrChildRight(v)$}{the right child of node $v$}
		\opitem{$\TrDepth(v)$}{the depth of $v$, \ie, the number of edges between the root and $v$}
		\opitem{$\TrLevAnc(v,i)$}{the ancestor of node $v$ at depth $\TrDepth(v)-i$} 
		\opitem{$\TrNbDesc(v)$}{the number of descendants of $v$}
		\opitem{$\TrHeight(v)$}{the height of the subtree rooted at node $v$}
		\opitem{$\LCA(v, u)$}{the lowest common ancestor of nodes $u$ and $v$}
		\opitem{$\TrLeftLeaf(v)$}{the leftmost leaf descendant of $v$}
		\opitem{$\TrRightLeaf(v)$}{the rightmost leaf descendant of $v$}
		\opitem{$\TrLevelLeft(\ell)$}{the leftmost node on level $\ell$}
		\opitem{$\TrLevelRight(\ell)$}{the rightmost node on level $\ell$}
		\opitem{$\TrLevelPred(v)$}{the node immediately to the left of $v$ on the same level}
		\opitem{$\TrLevelSucc(v)$}{the node immediately to the right of $v$ on the same level}
		\opitem{$\TrRank_{X}(v)$}{the position of $v$ in the $X$-order, $X \in \{\pre,\post,\inorder\}$, \ie, in \newline a preorder, postorder, or inorder traversal of the tree}
		\opitem{$\TrSelect_{X}(i)$}{the $i$th node in the $X$-order, $X \in \{\pre,\post,\inorder\}$}
		\opitem{$\TrLeafRank(v)$}{the number of leaves before and including $v$ in preorder}
		\opitem{$\TrLeafSel(i)$}{the $i$th leaf in preorder}
	\end{tabular}
	\end{minipage}}}
	\label{tab:binary-operations}
\end{table}
In this paper, we fill this gap between succinct trees and dictionary-compressed trees 
by presenting the first data structure for unlabeled binary trees 
that answers all queries supported in previous succinct data structures (cf. \wref{tab:binary-operations})
in $O(1)$ time and 
simultaneously achieves optimal compression 
over the same tree sources as the best previously known universal tree codes.
We also extend the tree-source concepts and our data structure to unlabeled ordinal trees.
In contrast to previous succinct trees, 
we give a single, \emph{universal} data structure, 
the \emph{hypersuccinct trees}%
\footnote{%
	The name ``hypersuccinct trees'' is the escalation %
	of the ``ultrasuccinct trees'' of~\cite{JanssonSadakaneSung2012}.
},
that does not need to be adapted to specific classes or distributions of~trees.

Our hypersuccinct trees require only a minor modification of existing succinct tree 
data structures based on tree covering~\cite{GearyRamanRaman2006,HeMunroRao2012,FarzanMunro2014},
(namely Huffman coding micro-tree types);
the contribution of our work is the careful analysis of the
information-theoretic properties of the tree-compression method, 
the ``hypersuccinct code'', that underlies these data structures.

As a consequence of our results, 
we solve an open problem for succinct range-minimum queries (RMQ):
Here the task is to construct a data structure from an array $A[1..n]$ of comparable items
at preprocessing time that can answer subsequent queries without inspecting $A$ again.
The answer to the query $\RMQ(i,j)$, for $1\le i\le j\le n$, is the index (in $A$) of the
(leftmost) minimum in $A[i..j]$, \ie,
\(
		\RMQ(i,j)
	\wwrel=
		\mathop{\arg\min}_{i\le k\le j} A[k].
\)
We give a data structure that answers $\RMQ$ in constant time
using the optimal expected space of $1.736 n + o(n)$ 
bits when the array is a random permutation, (and $2n+o(n)$ in the worst case);
previous work either had suboptimal space~\cite{DavoodiNavarroRamanRao2014}
or $\Omega(n)$ query time~\cite{GolinIaconoKrizancRamanRaoShende2016}.
We obtain the same (optimal) space usage for storing a binary search tree (BST) 
built from insertions in random order (``random BSTs'' hereafter).
Finally, we show that the space usage of our RMQ data structure is 
also bounded by $2\lg \binom nr + o(n)$ whenever $A$ has $r$ increasing runs, 
and that this is again best possible.

\subparagraph{Outline}
The rest of our article is structured as follows: 
A comprehensive list of the contributions appears below in \wref{sec:contributions}.
\wref{sec:covering-to-trees} describes our compressed tree encoding.
In \wref{sec:examples},  we illustrate the techniques for proving universality of our
hypersuccinct code on two well-known types of binary-trees shape distributions~--
random BSTs and weight-balanced trees~-- and sketch the extensions necessary for the general results.
In \wref{sec:optimal-range-min}, we present our RMQ data structures.
Finally, \wref{sec:conclusion} concludes the paper with future directions.

	The appendix contains a comprehensive comparison to previous work (\wref{sec:related-work}), 
	full formal proofs of all results (\wref{part:binary}) and
	the extension to ordinal trees (\wref{part:ordinal}).
	(The proofs in 
	Sections~\ref{sec:memoryless-binary}, \ref{sec:fixed-size-height}, 
	\ref{sec:uniform-binary}, \ref{sec:omitted-proofs-rmq-runs},
	\ref{sec:degreeentropy}, \ref{sec:fixedsizeordinal}, 
	\ref{sec:label-shape}
	can all be read in isolation.)

\section{Results}
\label{sec:contributions}

In a binary tree, each node has a left and a right child, either of which can be empty (``null'').
For a binary tree $t$ we denote by $|t|$ the number of nodes in $t$.
Unless stated otherwise, $n = |t|$.
A binary tree source $\mathcal S$ emits a tree $t$ with a 
certain probability $\ProbIn{\mathcal S} t$ (potentially $\ProbIn{\mathcal S} t = 0$);
we write $\Prob t$ if $\mathcal S$ is clear from the context.
$\lg(1/0)$ is taken to mean~$+\infty$.

\begin{theorem}[Hypersuccinct binary trees]
\label{thm:main-binary}
	Let $t$ be a binary tree over $n$ nodes.
	The hypersuccinct representation of $t$ supports
	all queries from \wref{tab:binary-operations}
	in $\Oh(1)$ time and uses $\hypsuc t + o(n) $ bits of space,
	where 
	\[
			\hypsuc t
		\wwrel\le 
			\min\biggl\{
				2n+1,\; 
				\min_{\mathcal S} \lg\biggl(\frac1{\ProbIn{\mathcal S} t}\biggr) 
				+o(n)
			\biggr\},
	\]
	and $\ProbIn{\mathcal S} t$ is the probability that $t$ is emitted 
	by source $\mathcal S$.
	The minimum is taken over all binary-tree sources $\mathcal S$ 
	in the following families (which are explained in \wref{tab:sufficient-conditions}):
	\begin{thmenumerate}{thm:main-binary}
	\item memoryless node-type processes,
	\item $k$th-order node-type processes (for $k=o( \log n)$),
	\item monotonic fixed-size sources,
	\item worst-case fringe-dominated fixed-size sources,
	\item monotonic fixed-height sources,
	\item worst-case fringe-dominated fixed-height sources,
	\item tame uniform subclass sources.
	\end{thmenumerate}
\end{theorem}
\begin{corollary}[Hypersuccinct binary trees: Examples \& Empirical entropies]
\label{cor:main-binary-empirical-entropies}
	Hypersuccinct trees achieve optimal compression to within lower order terms
	for all example distributions listed in \wref{tab:examples}.
	Moreover, for every binary tree $t$, we have:
	\begin{thmenumerate}{cor:main-binary-empirical-entropies}
	\item 
		$\hypsuc t \le H_k^{\type}(t) + o(n)$ with
		$H_k^{\type}(t)$ the (unnormalized) $k$th-order empirical entropy of node types
		(leaf, left-unary, binary, or right-unary) for $k=o(\log n)$.
	\item 
		$\hypsuc t \le H_{\mathit{st}}(t) + o(n)$ with $H_{\mathit{st}}(t)$
		the ``subtree-size entropy'', \ie, 
		the sum of the logarithm of the subtree size of $v$ for all nodes $v$ in $t$,
		(\aka the splay-tree potential).
	\end{thmenumerate}
\end{corollary}

\begin{table}[tbp]
	\caption{Overview of random tree sources for binary and ordinal trees.}
	\adjustbox{max width=\textwidth}{%
	\setlength\extrarowheight{2pt}%
		\begin{tabular}{
			|>{\raggedright} m{5.5em}%
			|>{\centering}m{3.25em}%
			|m{17em} <{\rule[-6pt]{0pt}{1pt}}%
			|m{3.7em}%
			|m{12.7em}|%
		}
		\hline
		\textbf{Name} 
			&\scalebox{.8}[1]{\textbf{Notation}} 
			&\textbf{Intuition} 
			&\scalebox{.8}[1]{\textbf{Reference}} 
			&\textbf{Formal Definition} of $\mathbb{P}[t]$ 
		\\
		\hline
		\hline
		Memoryless Processes
			& $\tau$ 
			& A binary tree is constructed top-down, drawing each node's \textit{type} 
				($0 =$ leaf, $1 = $ left-unary, $2 =$ binary, $3 = $ right-unary) 
				\iid according to the distribution $(\tau_0,\tau_1,\tau_2,\tau_3)$.
			& \makecell[l]{\swref{sec:memoryless-binary}\\ \swref{eq:prob-higherordertypes}\\
				\cite{DavoodiNavarroRamanRao2014,GolinIaconoKrizancRamanRaoShende2016} }
			&\(\displaystyle \mathbb{P}[t] \wrel= \prod_{v \in t}\tau(\type(v))\)    
		\\
		\hline
		Higher-order \newline Processes
			& $(\tau_z)_{z}$
			& A binary tree is constructed top-down, drawing node $v$'s \textit{type}
				according to $\tau_{h_k(v)}:\{0,1,2,3\} \rightarrow [0,1]$,
				which depends on the types of the $k$ closest \textit{ancestors} of $v$.
			& \makecell[l]{\swref{sec:memoryless-binary}\\ \swref{eq:prob-higherordertypes}}
			&\(\displaystyle\mathbb{P}[t] \wrel= \prod_{v \in t} \tau_{h_k(v)}(\type(v))\)  
		\\
		\hline
		Fixed-size Binary Tree Sources
			& $\mathcal{S}_{\mathit{fs}}(p)$
			& A binary tree of size $n$ is constructed top-down, 
				asking source $p$ at each node for its left- and right \textit{subtree size}.
			& \makecell[l]{\swref{sec:fixed-size-height}\\ \swref{eq:prob-Tn=t}\\
				\scalebox{.825}[1]{\cite{ZhangYangKieffer2014,GanardiHuckeLohreySeelbachBenkner2019,SeelbachBenknerLohrey2018}}%
				}
			& \makecell[l]{\(\displaystyle\mathbb{P}[t] \wrel= \prod_{v \in t}p(|t_{\ell}(v)|, |t_r(v)|)\)\\
				\textsmaller{$t_{\ell/r}(v)$= left/right subtree of~$v$}}
		\\
		\hline
		Fixed-height Binary Tree Sources
			& $\mathcal{S}_{\mathit{fh}}(p)$
			& A binary tree of height $h$ is constructed top-down,
				asking source $p$ at each node for a left and right subtree \textit{height}.\strut
			& \makecell[l]{\swref{sec:fixed-size-height}\\ \swref{eq:prob-Tn=t2}\\
				\cite{ZhangYangKieffer2014,GanardiHuckeLohreySeelbachBenkner2019}
			}
			& \makecell[l]{\(\displaystyle\mathbb{P}[t] \wrel= \prod_{v \in t}p(h(t_{\ell}(v)), h(t_r(v)))\) \\
				\textsmaller{$h(t)$ = height of $t$}}
		\\
		\hline
		Uniform Subclass Sources
			& $\mathcal{U}_{\mathcal{P}}$
			& A binary tree is drawn uniformly at random from the set $\mathcal T_n(\mathcal P)$ 
				of all binary trees of size $n$ 
				that satisfy property $\mathcal{P}$.\strut
			& \makecell[l]{\swref{sec:uniform-binary}\\ \swref{eq:probuniform}}
			& \(\displaystyle\mathbb{P}[t] \wrel= \frac1{|\mathcal{T}_n(\mathcal{P})|}\)  
		\\
		\hline
		\hline
		Memoryless Ordinal Tree Sources
			&$d$
			& An ordinal tree is constructed top-down, drawing each node $v$'s 
				degree $\deg(v)$ according to distribution $d=(d_0,d_1,\ldots)$.\strut
			& \makecell[l]{\swref{sec:degreeentropy}\\ \swref{eq:degree-probability-ordinal}}
			&\(\displaystyle\mathbb{P}[t] \wrel= \prod_{v \in t}d_{\deg(v)}\) 
		\\
		\hline
		Fixed-size Ordinal Tree Sources
			& $\mathfrak{S}_{\mathit{fs}}(p)$
			& An ordinal tree of size $n$ is constructed top-down, asking 
				source $p$ at each node for the number and sizes of the subtrees.\strut
			& \makecell[l]{\swref{sec:fixedsizeordinal}}
			& \scalebox{.93}[1]{\(\displaystyle\mathbb{P}[t]=\prod_{v \in t}p(|t_1[v]|, \dots, |t_{\deg(v)}[v]|) \)} 
		\\
		\hline
		\end{tabular}
	}
		
		\label{tab:sources}
	\end{table}

\begin{table}[tbp]
	\caption{%
		An overview over the concrete examples of tree-shape distributions 
		that our hypersuccinct code compresses optimally (up to lower-order terms).%
	}
	\def\unknown{---\tnotex{tn:entropy-unknown}}%
	\adjustbox{max width=1\textwidth}{%
	\def\treesources{\unskip}%
	\begin{threeparttable}[b]
		\begin{tabular}{%
				|>{\raggedright}m{19em} %
				|m{5em} %
				|>{\raggedright}m{13em} %
				|>{\centering}m{3em} %
				|m{4em}|%
		}
			\hline
			\textbf{Tree-Shape Distribution} 
				& \textbf{Entropy}
				& \textbf{Corresponding Source} 
				& \textbf{Def.} 
				& \plaincenter{\textbf{Result}}  
			\\
			\hline
			\hline
			(Uniformly random) binary trees of size~$n$
				& $2n$
				& \makecell[l]{Memoryless binary \treesources, \\
					monotonic fixed-size binary \treesources }
				& \makecell[l]{\swref{exm:typebinary}\\%
					\swref{exm:uniform}}
				& \makecell[l]{\swref{cor:typeentropy} \\ \swref{cor:monotonic}}
			\\
			\hline
			(Uniformly random) full binary trees of size~$n$ 
				& $n$
				& Memoryless binary \treesources 
				& \makecell[l]{\swref{exm:typefullbinary} }
				& \makecell[l]{\swref{cor:typeentropy} }
			\\
			\hline
			(Uniformly random) unary paths of length~$n$ 
				& $n$
				& Memoryless binary \treesources 
				& \makecell[l]{\swref{exm:typeunarypath} }
				& \makecell[l]{\swref{cor:typeentropy} }
			\\
			\hline
			(Uniformly random) Motzkin trees of size~$n$ 
				& $1.585n$
				& Memoryless binary \treesources 
				& \makecell[l]{\swref{exm:typeMotzkin} }
				& \makecell[l]{\swref{cor:typeentropy}}
			\\
			\hline
			Binary search trees generated by insertions in random order (``random BSTs'') 
				& $1.736n$ 
				& Monotonic fixed-size binary \treesources 
				& \makecell[l]{\swref{exm:bst} }
				& \makecell[l]{\swref{cor:monotonic} \\ \swref{cor:fringe-dom}}
			\\
			\hline
			Binomial random trees 
				& $P(\lg n) n$\tnotex{tn:binomial-random-tree-entropy}
				& Average-case fringe-dominated fixed-size binary \treesources 
				& \makecell[l]{\swref{exm:dst}}
				& \makecell[l]{\swref{cor:fringe-dom} }
			\\
			\hline
			Almost paths 
				& \unknown
				& Monotonic fixed-size binary \treesources 
				& \makecell[l]{\swref{exm:unary-paths}}
				& \makecell[l]{\swref{cor:monotonic} }
			\\
			\hline
			Random fringe-balanced binary search trees 
				& \unknown
				& Average-case fringe-dominated fixed-size binary \treesources 
				& \makecell[l]{\swref{exm:fringe-balanced-bsts}}
				& \makecell[l]{\swref{cor:fringe-dom} }
			\\
			\hline
			(Uniformly random) AVL trees of height~$h$
				& \unknown
				& Worst-case fringe-dominated fixed-height binary \treesources 
				& \makecell[l]{\swref{exm:avl-uniform-height} }
				& \makecell[l]{\swref{cor:fringe-dom}}
			\\
			\hline
			(Uniformly random) weight-balanced binary trees of size~$n$ 
				& \unknown
				& Worst-case fringe-dominated fixed-size binary \treesources 
				& \makecell[l]{\swref{exm:weightbalanced} }
				& \makecell[l]{\swref{cor:fringe-dom} }
			\\
			\hline
			(Uniformly random)  AVL trees of size~$n$ 
				& $0.938n$
				& Uniform-subclass \treesources 
				& \makecell[l]{\swref{exm:AVLsize} }
				& \makecell[l]{\swref{cor:uniformsubclass} }
			\\
			\hline
			(Uniformly random) left-leaning red-black trees of size~$n$ 
				& $0.879n$
				& Uniform-subclass \treesources 
				& \makecell[l]{\swref{exm:redblack} }
				& \makecell[l]{\swref{cor:uniformsubclass} }
			\\
			\hline
			\hline
			(Uniformly random) full $m$-ary trees of size~$n$ 
				& $\lg(\frac m{m-1})n$
				& Memoryless ordinal \treesources 
				& \makecell[l]{\swref{exm:fullkary}}
				& \makecell[l]{\swref{cor:fullmary}}
			\\
			\hline
			Uniform composition trees 
				& \unknown
				& Monotonic fixed-size ordinal \treesources 
				& \makecell[l]{\swref{exm:uniformcomposition} }
				& \makecell[l]{\swref{cor:monotonicordinal} }
			\\
			\hline
			Random LRM-trees 
				& $1.736n$
				& Monotonic fixed-size ordinal \treesources 
				& \makecell[l]{\swref{exm:lrm}  }
				& \makecell[l]{\swref{cor:monotonicordinal} }
			\\
			\hline
		\end{tabular}%
	\begin{tablenotes}
		\item[a] \label{tn:binomial-random-tree-entropy}
			Here $P$ is a nonconstant, continuous, periodic function with period 1.
		\item[b] \label{tn:entropy-unknown}
			No (concise) asymptotic approximation known.
	\end{tablenotes}
	\end{threeparttable}}
	\label{tab:examples}
\end{table}

\begin{table}[tbp]
	\caption{%
		Sufficient conditions under which we show universality of our hypersuccinct code $\mathsf H$
		for binary trees.%
		\protect\ifproceedings{
			Proofs are given in the extended version of this article (\arxiv{2104.13457}).%
		}{}%
	}
	
	\adjustbox{max width=1\textwidth}{%
	\begin{threeparttable}
		\begin{tabular}{|>{\raggedright}m{8em}|c|c|c|c|}
			\hline
			\textbf{Family of sources} 
				& \textbf{Restriction}
				& \textbf{Redundancy}
				& \textbf{Def.} 
				& \scalebox{.8}[1]{\textbf{Reference}}
			\\
			\hline
			\hline
			Memoryless node-type
				& ---
				& $\Oh\left(n \log\log n / \log n\right)$
				& \swref{sec:memoryless-binary}
				& \swref{thm:kthorderdegree}
			\\
			\hline
			$k$th-order node-type
				& ---
				& $O((nk+ n \log \log n)/ \log n)$
				& \swref{sec:memoryless-binary}
				& \swref{thm:kthorderdegree}
			\\
			\hline
			Monotonic fixed-size 
				& \makecell{$p(\ell,r) \ge p(\ell+1,r)$ and $p(\ell,r) \ge p(\ell,r+1)$\\ 
					for all $\ell,r \in \N_0$}
				& $\Oh\left(n \log\log n / \log n\right)$
				& \swref{def:monotonic}
				& \swref{thm:binary-monotonic-fixed-size}
			\\
			\hline
			Worst-case fringe-dominated fixed-size
				& \makecell{$n_{\geq B}(t) = o(n/ \log \log n)$ \\
					for all $t$ with $\mathbb{P}[t]>0$; \\
					$n_{\geq B}(t) =$ \#nodes with subtree size in $\Omega(\log n)$
				}
				& \makecell{$\Oh\bigl(n_{\geq B}(t) \log \log n$\\${}+ n \log\log n / \log n\bigr)$}
				& \swref{def:wfringe-dominated}
				& \swref{thm:wfringe-dominated}
			\\
			\hline
			Weight-balanced fixed-size
				& \makecell{$\displaystyle\sum_{\frac{n}{c}\leq \ell \leq n-\frac{n}{c}}p(\ell-1,n-\ell-1) \wwrel= 1$\\[-2.5ex]
					\qquad\qquad\qquad for constant $c\ge3$}
				& $\Oh\left(n \log\log n / \log n\right)$
				& \swref{def:weight-balanced}
				& \swref{cor:weight-balanced-universal}
			\\
			\hline
			Average-case fringe-dominated fixed-size
				& \makecell{$\E{n_{\geq B}(T)} = o(n/ \log \log n)$ \\
								for random $T$ generated by source $\mathcal S$
							}
				& \makecell{$\Oh\bigl(n_{\geq B}(t) \log \log n$\\${}+ n \log\log n / \log n\bigr)$%
					\tnotex{tn:exp-redundancy}}
				& \swref{def:avfringe-dominated}
				& \swref{thm:fringe-dominated}
			\\
			\hline
			Monotonic fixed-height 
				& \makecell{$p(\ell,r) \ge p(\ell+1,r)$ and $p(\ell,r) \ge p(\ell,r+1)$\\ for all $\ell,r \in \N_0$}
				& $\Oh\left(n \log\log n / \log n\right)$
				& \swref{def:monotonic}
				& \swref{thm:binary-monotonic-fixed-size}
			\\
			\hline
			Worst-case fringe-dominated fixed-height
				& \makecell{$n_{\geq B}(t) = o(n/ \log \log n)$ \\
					for all $t$ with $\mathbb{P}[t]>0$
				}
				& \makecell{$\Oh\bigl(n_{\geq B}(t) \log \log n$\\${}+ n \log\log n / \log n\bigr)$}
				& \swref{def:wfringe-dominated}
				& \swref{thm:wfringe-dominated}
			\\
			\hline
			Tame uniform-subclass
				& \makecell{class of trees $\mathcal T_n(\mathcal P)$ is hereditary\\
					(\ie, closed under taking subtrees), \\
					$n_{\geq B}(t) = o(n/ \log \log n)$ for $t\in \mathcal T_n(\mathcal P)$,\\
					$\lg |\mathcal T_n(\mathcal P)| = c n + o(n)$ for constant $c>0$,\\
					heavy-twigged: if $v$ has subtree size $\Omega(\log n)$, \\
						$v$'s subtrees have size $\omega(1)$
				}
				&	$o(n)$
				& \swref{def:tame-uniform}
				& \swref{thm:uniformsubclass}
			\\
			\hline
		\end{tabular}
	\begin{tablenotes}
		\item[a] \label{tn:exp-redundancy}
			Stated redundancy is achieved in expectation for a random tree $t$ 
				generated by the source.
	\end{tablenotes}
	\end{threeparttable}}
	\label{tab:sufficient-conditions}
\end{table}

The hypersuccinct code is a \emph{universal code} for the 
families of binary-tree sources listed in \wref{thm:main-binary} 
with bounded maximal pointwise redundancy. 
We also present a more general class of sources, 
for which our code achieves $o(n)$ \emph{expected} redundancy
in the appendix; see also \wref{tab:sufficient-conditions}.

To our knowledge, the list in \wref{thm:main-binary}
is a comprehensive account of \emph{all} concrete binary-tree sources
for which any universal code is known.
Remarkably, in all cases the bounds on redundancies proven for the hypersuccinct code
are identical (up to constant factors) to those known for existing
universal binary-tree codes.
\emph{Our hypersuccinct code thus achieves the same compression as all previous universal codes,
but simultaneously supports constant-time queries on the compressed representation 
with $o(n)$ overhead.}

In terms of queries,
previous solutions either have suboptimal query times~\cite{BilleGortzLandauWeimann2015,BilleLandauRamanSadakaneRaoWeimann2015,GanardiHuckeLohreyNoeth2017}, 
higher space usage~\cite{Prezza2019},
or rely on tailoring the representation to a 
specific subclass of trees~\cite{JanssonSadakaneSung2012,FarzanMunro2014}
to achieve good space and time for precisely
these instances, but they fail to generalize to other use cases.
Some also do not support all queries.
We give a detailed comparison with the state of the art in 
\wref{sec:related-work}.

We focus here on our results for binary trees.
In \ifproceedings{the extended version of this article}{the appendix, \wref{part:ordinal}},
we extend the above notions of tree sources (except fixed-height sources)
to ordinal trees, which has not been done to our knowledge.
Moreover, we extend both our code and data structure to ordinal trees,
and show their universality for these sources.

\section{From Tree Covering to Hypersuccinct Trees}
\label{sec:covering-to-trees}
\label{sec:tree-covering-to-hypersuccinct}

Our universally compressed tree data structures are based on 
\emph{tree covering}~\cite{GearyRamanRaman2006,HeMunroRao2012,FarzanMunro2014}:
A~(binary or ordinal) tree $t$ is decomposed into \emph{mini trees}, 
each of which is further decomposed into \emph{micro trees}; 
the size of the latter, $B = B(n) = \lg n/8$, 
is chosen so that we can tabulate all possible shapes of micro trees 
and the answers to various micro-tree-local queries in one global lookup table
(the ``Four-Russian Table'' technique).
For each micro tree, its local shape is stored, \eg, using the
balanced-parenthesis (BP) encoding, using a total of exactly $2n$ bits 
(independent of the tree shape).
Using additional data structures occupying only $o(n)$ bits of space,
a long list of operations can be supported in constant time (\wref{tab:binary-operations}).
The space usage of this representation is optimal to within lower order terms for the worst case, since 
$\lg C_n \sim 2n$ bits are necessary to distinguish all 
$C_n = \binom{2n}n / (n+1)$ trees of $n$ nodes.
(This worst-case bound applies both to ordinal trees and binary trees).

\begin{figure}[tbh]
	\resizebox{\linewidth}!{\input{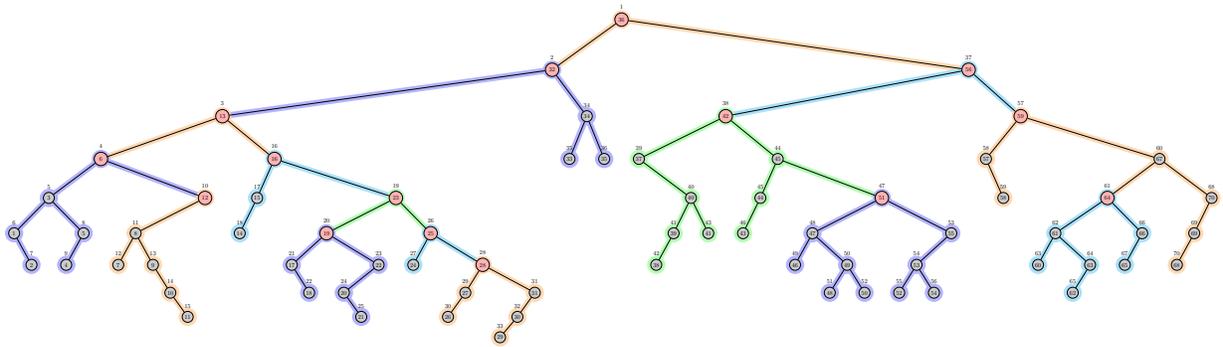}}
	\caption{%
		Example binary tree with $n=70$ nodes and micro trees computed by the 
		Farzan-Munro tree-covering algorithm \cite{FarzanMunro2014} with parameter $B=6$.
		For the reader's convenience, the algorithm is summarized in the appendix 
		(\wref{sec:farzan-munro}).
		The micro trees are indicated by colors.
		The algorithm guarantees that each node is part of exactly one micro tree
		and that each micro tree has at most three edges shared with other micro trees,
		namely to a parent, a left- and a right-child micro tree.
	}
	\label{fig:tree-partitioning-example}
\end{figure}

A core observation is that the dominant space in tree-covering data structures 
comes from storing the micro-tree types, and these can be further compressed
using a different code.
This has been used in an ad-hoc manner 
for specific tree classes~\cite{FarzanMunro2014,DavoodiNavarroRamanRao2014,Ganczorz2020}, 
but has not been investigated systematically.
A natural idea is to use a Huffman code for the micro tree types
to simultaneously beat the compression of all these special cases;
we dub this as the ``Four Russians and One American''%
\footnote{%
	It deems us only fair to do D. A. Huffman the same questionable honor 
	of reducing the person to a country of residence that
	V. L. Arlazarov, E. A. Dinic, M. A. Kronrod, and I. A. Faradžev
	have experienced ever since their table-lookup technique has 
	become known as the ``Four-Russians trick''.
}
trick.
Applying it to the data structures based on 
the Farzan-Munro tree-covering algorithm~\cite{FarzanMunro2014} 
yields our hypersuccinct trees.

The main contribution of our present work is the careful analysis 
of the potential of the Four Russians and One American trick 
for (binary and ordinal) tree source coding.
As an immediate corollary, we obtain 
a single data structure that achieves optimal compression 
for all special cases covered in previous work, plus
a much wider class of distributions over trees
for which no efficient data structure was previously known.

Our analysis builds on previous work on tree compression, 
specifically DAG compression and tree straight-line programs (TSLPs)~\cite{Lohrey2015}.
Our core idea is to interpret (parts of the) tree-covering data structures
as a \emph{code} for trees, the ``hypersuccinct code'': it stores
the type, \ie, the local shape, of all micro trees separately from 
how they interface to form the entire tree
(details are given in the appendix, 
\wref{sec:hypersuccinct-binary} for binary trees and 
\wref{sec:hypersuccinct-ordinal} for ordinal trees).
Intuitively, our hypersuccinct code is a restricted version of a grammar-based tree code,
where we enforce having nonterminals for certain subtrees;%
\footnote{%
	Differences in technical details make the direct comparison difficult, though:
	in TSLPs, holes in contexts must be stored (and encoded) alongside the local shapes
	as they are both part of the right-hand side of productions;
	in our hypersuccinct code, we separately encode the shapes of micro trees
	and the positions of portals, potentially gaining a small advantage.
	Our comment thus remains a motivational hint as to 
	why similar analysis techniques are useful in both cases,
	but falls short of providing a formal reduction.
}
we strengthen and extend existing universality proofs from general grammar-based tree codes 
to the restricted hypersuccinct code.

\section{Universality for Fixed-Size Sources}
\label{sec:examples}

In this section, we sketch the proof that our hypersuccinct trees 
achieve optimal compression for two exemplary tree-shape distributions:
random binary search trees and uniform weight-balanced trees (defined below).
These examples serve to illustrate the proof techniques and to showcase
the versatility of the approach.
The extension to the general sufficient conditions from \wref{tab:sufficient-conditions}
and full details of computations are spelled out in the 
appendix.

By \emph{random BSTs}, we mean
the distribution of tree shapes obtained by 
successively inserting $n$ keys in random order into an (initially empty)
unbalanced binary search tree (BST).
We obtain random BSTs from a fixed-size tree source 
$\mathcal{S}_{\mathit{fs}}(p_{\mathit{bst}})$ with
$p_{\mathit{bst}}(\ell,n-1-\ell) = \frac {1}{n}$ 
for all $\ell \in \{0, \dots, n-1\}$ and $n \in \mathbb{N}_{\ge1}$, \ie,
making every possible split equally likely.
(Any left subtree size $\ell$ is equally likely in a random BST of a given size~$n$.)
Hence, $\Prob t = \prod_{v\in t} 1/|t[v]|$
where $t[v]$ is the subtree rooted at $v$ and $|t[v]|$ its size (in number of nodes).

The second example are the shapes of uniformly random \emph{weight-balanced BSTs} 
($\mathrm{BB}[\alpha]$-trees, \cite{NievergeltReingold1973}): 
A binary tree $t$ is $\alpha$-weight-balanced if we have for every node $v$ in $t$
that 
$\min\{|t_{\ell}[v]|, |t_{r}[v]|\} + 1 \ge \alpha (|t[v]|+1)$.
Here $t_\ell[v]$ resp.\ $t_r[v]$ are the left resp.\ right subtrees of~$t[v]$.
We denote the set of $\alpha$-weight-balanced trees of size $n$ by $\mathcal{T}_n(\mathcal{W}_{\alpha})$.
We obtain random $\alpha$-weight-balanced trees from another fixed-size source 
$\mathcal S_{\mathit{fs}}(p_{\mathit{wb}})$ with
\begin{align*}
		p_{\mathit{wb}}(\ell,n-1-\ell)
	\wwrel=
		\begin{dcases}
			\frac{|\mathcal{T}_\ell(\mathcal{W}_{\alpha})| |\mathcal{T}_{n-1-\ell}(\mathcal{W}_{\alpha})|}
					{|\mathcal{T}_n(\mathcal{W}_{\alpha})|} \quad 
				&\text{if } \min\{\ell+1,n-\ell\} \geq \alpha (n+1),
		\\
			0 & \text{otherwise.}
	\end{dcases}
\end{align*}
It is easy to check that this yields the uniform probability distribution
on $\mathcal{T}_n(\mathcal{W}_{\alpha})$, \ie, with
$\Prob t = 1/|\mathcal{T}_n(\mathcal{W}_{\alpha})|$ for $t\in \mathcal{T}_n(\mathcal{W}_{\alpha})$ 
and $\Prob t = 0$ otherwise.
We note that computing $|\mathcal{T}_n(\mathcal{W}_{\alpha})|$ 
is a formidable challenge in combinatorics, but we never have to do so; 
we only require the \emph{existence} of the fixed-size source for weight-balanced BSTs.

The \emph{hypersuccinct code} $\mathsf H(t)$ is formed by partitioning the nodes of a given 
binary tree $t$ into $m=\Theta(n/\log n)$ micro trees $\mu_1,\ldots,\mu_m$,
each of which is a connected subtree of at most $\mu = \Oh(\log n)$ nodes;
an example is shown in \wref{fig:tree-partitioning-example}.
Previous work on tree covering shows how to compute these and how to encode everything
but the local shape of the micro trees in $o(n)$ bits of space~\cite{FarzanMunro2014}.
(For a mere encoding, $O(n\log \log n / \log n)$ bits suffice; see \wref{sec:hypersuccinct-binary}).

The dominant part of the hypersuccinct code is the list of \emph{types} of all micro trees, \ie, 
the (local) shapes of the induced subtrees formed by the set of nodes in the micro trees.
Let $C$ be a Huffman code for the string $\mu_1,\ldots,\mu_m$, where we identify micro trees 
with their types.
For a variety of different tree sources $\mathcal S$,
we can prove that $\sum_{i=1}^m |C(\mu_i)|$, the total length of codewords for the micro trees,
is upper bounded by $\lg(1/\Prob t) + \text{lower-order terms}$, 
where $\Prob t$ is the probability that $t$ is emitted by~$\mathcal S$; 
this is the best possible code length to within lower order terms achievable for that source.
We will now show this for our two example distributions.

\subsection{Random BSTs}
\label{sec:bsp-random-bst}

The proof consists of four steps that can be 
summarized as follows: 

\bigskip

{\centering%
\adjustbox{max width=\linewidth}{\sffamily\smaller%
\begin{tikzpicture}[
	Step/.style={draw,fill=white,rounded corners=4pt,inner sep=5pt,minimum height=18ex},
	node distance=.5em,
]
	\node[Step,align=center] (step1) at (0,0) {
		\textbf{Step 1}\\
		Construct a source-specific\\ 
		micro-tree encoding\\[.5ex]
		$D_{\mathcal S}\colon \{\mu_1,\ldots,\mu_m\}\to\{0,1\}^\star$\\[1ex]
		\textbf{Goal:} $|D_{\mathcal S}(\mu_i)| \approx \lg(1/\Prob{\mu_i})$
	} ;

	\node[Step,align=center,fill=black!10,right=of step1,anchor=west] (step2) {
		\textbf{Step 2}\\
		By optimality of\\ 
		Huffman codes:\\[1ex]
		\(\displaystyle
		\sum_{i=1}^m |C(\mu_i)| \le \sum_{i=1}^m |D_{\mathcal S}(\mu_i)|
		\)
	} ;
	
	\node[Step,align=center,right=of step2,anchor=west] (step3) {
		\textbf{Step 3}\\
		Use properties of $\mathcal S$\\ 
		to show that\\[1ex]
		\(\displaystyle
			\prod_{i=1}^m \Prob{\mu_i} \gtrsim \Prob t
		\)
	} ;
	
	\node[Step,align=center,fill=black!10,right=of step3,anchor=west] {
		\textbf{Step 4}\\[.5ex]
		Conclude\\
		\(\displaystyle
			\sum_{i=1}^m |C(\mu_i)| \approx \lg(1/\Prob t) 
		\)
	} ;

\end{tikzpicture}}%
}

\bigskip

\noindent
Steps 2 and~4 do not depend on the source and indeed follow immediately;
Steps 1 and 3 are the creative parts.
Ignoring proper tracing of error terms, the result then follows as
\begin{align*}
		\hypsuc t
	&\wwrel\sim
		\sum_{i=1}^n |C(\mu_i)|
	\wwrel\le
		\sum_{i=1}^n |D_{\mathcal S}(\mu_i)|
	\wwrel\le
		\sum_{i=1}^n \lg(1/\Prob{\mu_i})
	\wwrel\lesssim
		\lg(1/\Prob{t}).
\end{align*}

Let us consider $\mathcal S_{\mathit{fs}}(p_{\mathit{bst}})$,
the fixed-size source producing (shapes of) random BSTs,
and address these steps independently.

Our task in Step~1 is to find a code $D_\mathcal S$ for the micro-tree types 
that can occur in $t$, 
so that $|D_\mathcal S(\mu_i)| = \lg(1/|\ProbIn{\mathcal S}{\mu_i}) + \Oh(\log\log n)$.
This code may rely on the decoder to have knowledge of $\mathcal S$.
For random BSTs, $D_\mathcal S(t)$ can be constructed as follows: 
We initially store~$n$ using Elias gamma code%
\footnote{%
	Elias gamma code $\gamma: \mathbb{N} \to \{0,1\}^\star$ encodes an 
	integer $n \geq 1$ using $2\lfloor \lg n\rfloor+1$ bits by prefixing 
	the binary representation of $n$ with that representation's length encoded in unary.%
}
and then,
following a depth-first (preorder) traversal of the tree,
we encode the size of the left subtree using \emph{arithmetic coding}.
Inductively, the size of the currently encoded node is always known,
and the source-specific code is allowed to use the probability distributions 
hardwired into $\mathcal S$ without storing them; 
for random BSTs, we simply encode a number uniformly distributed in 
$[0..s-1]$ at a node with subtree size $s$, using exactly $\lg s$ bits.
Apart from storing the initial size and the small additive overhead 
from arithmetic coding, the code length of this ``depth-first arithmetic tree code'' 
is best possible:
$|D_\mathcal S(t)| \le \lg(1/\Prob t) + \Oh(\log|t|)$.
This concludes Step 1.

For Step~3, we have to show that the probability for the entire tree $t$
is at most the product of the probabilities for all micro-trees.
Recall that $\mu_1,\ldots,\mu_m$ are the micro trees in $t$.
We can write $\Prob t$ as a product over contributions of individual nodes,
and can collect factors in $\Prob t$ according to micro trees;
this works for any fixed-size source.
For random BSTs, we can use the ``monotonicity'' of node contributions
to show
\begin{align*}
		\Prob t
	&\wwrel=
		\prod_{v\in t} \frac1{|t[v]|}
	\wwrel=
		\prod_{i=1}^m \prod_{v\in\mu_i} \frac1{|t[v]|}
	\wwrel\le
		\prod_{i=1}^m \prod_{v\in\mu_i} \frac1{|\mu_i[v]|}
	\wwrel=
		\prod_{i=1}^m \Prob{\mu_i} .
\end{align*}
That completes Step~3, and hence the proof
that $\hypsuc t \le \ProbIn{\mathcal S_{\mathit{fs}}(p_{\mathit{bst}})}{t} + o(n)$.

\subsection{Weight-balanced trees}
\label{sec:bsp-weightbalanced-bst}

Let us now consider uniformly random weight-balanced trees, \ie, 
the source $\mathcal S = \mathcal S_{\mathit{fs}}(p_{\mathit{wb}})$.
We would like to follow the same template as above;
however, this is not possible: Step~3 from above is in general not true anymore.
The reason is that it is not clear whether the ``non-fringe'' micro trees,
\ie, those that do not contain all descendants of the micro-tree root,
have non-zero probability under $\mathcal S$. 
(A subtree of a tree is called \emph{fringe}, if it consists of a node and all its descendants).
Such micro trees will also make Step~1 impossible as they would require
a code length of 0.
While this issue is inevitable in general (\wref{rem:no-universal-code-fixed-size}),
we can under certain conditions circumvent Steps~1 and~3 altogether
by directly bounding $\sum_{i=1}^m |D_{\mathcal S}(\mu_i)| \le \Prob t + o(n)$.

As a first observation,
note that it suffices to have
$|D_\mathcal S(\mu_i)| = \lg(1/|\ProbIn{\mathcal S}{\mu_i}) + \Oh(\log\log n)$
for \emph{all but a vanishing fraction} of the micro trees in any tree $t$;
then we can still hope to show 
$\sum_{i=1}^m |D_{\mathcal S}(\mu_i)| \le \Prob t + o(n)$ overall.
Second, it is known~\cite{GanardiHuckeJezLohreyNoeth2017} that 
weight-balanced trees are \emph{``fringe dominated''} in the following sense:
Denoting by $n_{\ge B}(t)$ the number of ``heavy'' nodes, \ie,  
$v$ in $t$ with $|t[v]| \ge B = \lg n / 8$,
we have $n_{\ge B}(t) = \Oh(n/B) = o(n)$ for every weight-balanced tree 
$t\in\mathcal T_n(\mathcal W_\alpha)$.
Since only a vanishing fraction of nodes are heavy, 
one might hope that also only a vanishing fraction of micro trees are non-fringe,
making the above route succeed.
Unfortunately, that is not the case; the non-fringe micro trees can be a constant 
fraction of all micro trees.

Notwithstanding this issue, 
a more sophisticated micro-tree code $D_{\mathcal S}$ allows us to proceed.
$D_{\mathcal S}$ encodes any \emph{fringe} micro tree using a depth-first arithmetic code 
as for random BSTs.
Any non-fringe micro tree $\mu_i$, however, is broken up
into the subtree of heavy nodes, the ``boughs'' of $\mu_i$, 
and (fringe) subtrees $f_{i,j}$ hanging off the boughs.
It is a property of the Farzan-Munro algorithm that every micro-tree root is heavy,
hence all $f_{i,j}$ are indeed entirely contained within $\mu_i$.

$D_{\mathcal S}(\mu_i)$ then first encodes the bough nodes using 2 bits per node
(using a BP representation for the boughs subtree)
and then appends the depth-first arithmetic code for the $f_{i,j}$ (in left-to-right order).
While this does not actually achieve 
$|D_{\mathcal S}(\mu_i)| \approx \lg(1/\Prob{\mu_i})$ for entire micro trees $\mu_i$,
it does so for all the fringe subtrees $f_{i,j}$.
Any node not contained in a fringe subtree $f_{i,j}$ must be part of a bough and hence heavy;
by the fringe-dominance property, these nodes form a vanishing fraction of all nodes
and hence contribute $o(n)$ bits overall.
This shows that 
$\hypsuc t \le \ProbIn{\mathcal S_{\mathit{fs}}(p_{\mathit{wb}})}{t} + o(n)$.

\begin{remark}[A simple code whose analysis isn't]
It is worth pointing out that the source specific code $D_{\mathcal S}$ is only 
a vehicle for the \emph{analysis} of $|\mathsf H(t)|$;
the complicated encodings $D_{\mathcal S}$ do not ever need to be computed 
when using our codes or data structures.
\end{remark}

\subsection{Other Sources}

For memoryless sources, the analysis follows the four-step template, and
is indeed easier than the random BSTs since Step~3 becomes trivial.
For higher-order sources, in order to know the node types of the $k$ closest ancestors (in $t$) 
of all nodes of depth $\le k$ in $\mu_i$, 
we prefix the depth-first arithmetic code by the node types of the $k$ closest ancestors 
of the root of $\mu_i$.
Then the $k$ ancestor types are known inductively for all nodes 
in a preorder traversal of $\mu_i$.

The tame uniform-subclass sources require the most technical proof,
but it is conceptually similar to the weight-balanced trees from above.
The source-specific encoding for fringe subtrees is trivial here: 
we can simply use the rank in an enumeration of all trees of a given size, 
prefixed by the size of the subtree. 
Using the tameness conditions, one can show that a similar decomposition into boughs and
fringe subtrees yields an optimal code length for almost all nodes.
Details are deferred to the appendix (\wref{sec:uniform-binary}).

\separatedpar
Together with the observations from \wref{sec:tree-covering-to-hypersuccinct}
this yields \wref{thm:main-binary}.
We obtained similar results for ordinal trees; 
details are deferred to the appendix (\wref{part:ordinal}).

\begin{remark}[Restrictions are inevitable]
\label{rem:no-universal-code-fixed-size}
	We point out that some restrictions like the ones discussed above cannot possibly 
	be overcome in general.
	Zhang, Yang, and Kieffer~\cite{ZhangYangKieffer2014} prove that the unrestricted 
	class of fixed-size sources 
	(leaf-centric binary tree sources in their terminology) does not
	allow a universal code, even when only considering expected redundancy.
	The same is true for unrestricted fixed-height and uniform-subclass sources.
	While each is a natural formalism to describe possible binary-tree sources, 
	additional conditions are strictly necessary for any interesting compression statements
	to be made.
	Our sufficient conditions are the weakest such restrictions for which 
	any universal source code is known to exist (\cite{ZhangYangKieffer2014, GanardiHuckeLohreySeelbachBenkner2019, SeelbachBenknerLohrey2018}),
	even without the requirement of efficient queries.
\end{remark}

\section{Hypersuccinct Range-Minimum Queries}
\label{sec:optimal-range-min}

We now show how hypersuccinct trees imply an optimal-space
solution for the range-minimum query (RMQ) problem.%
\footnote{%
	A technical report containing preliminary results for random RMQ,
	but including more details on the data structure aspects of our solution, 
	can be found on arXiv~\cite{MunroWild2019}.
}
Let $A[1..n]$ store the numbers $x_1,\ldots,x_n$,
\ie, $x_j$ is stored at index $j$ for $1\le j\le n$.
	While duplicates naturally arise in some applications, \eg, in the longest-common extension (LCE) problem,
	we assume here that $x_1,\ldots,x_n$ are $n$ distinct numbers to simplify the presentation.
	However, our RMQ solution works regardless of which minimum-value index is to be returned
	so long as the tie breaking rule is deterministic and fixed at construction time.

\subsection{Cartesian Trees}
\label{sec:rmq-cartesian-tree-lca}

The \emph{Cartesian tree} $T$ for $x_1,\ldots,x_n$ (resp.\ for $A[1..n]$) 
is a binary tree defined recursively as follows:
If $n=0$, it is the empty tree (``null''). 
Otherwise it consists of a root whose left child is the Cartesian tree for $x_1,\ldots,x_{j-1}$
and its right child is the Cartesian tree for $x_{j+1},\ldots,x_n$ where
$j$ is the position of the minimum, $j = \mathop{\arg\min}_k A[k]$.
A classic observation of Gabow et al.~\cite{GabowBentleyTarjan1984} is that 
range-minimum queries on $A$ are equivalent to lowest-common-ancestor (LCA) queries on $T$
when identifying nodes with their inorder~rank:
\[
		\RMQ_A(i,j) 
	\wwrel= 
		\TrRank_{\inorder}\Bigl( 
			\TrLCA\bigl(
				\TrSelect_{\inorder}(i),
				\TrSelect_{\inorder}(j)
			\bigr)
		\Bigr).
\]
We can thus reduce an RMQ instance (on an arbitrary input) to an LCA instance
on binary trees of the same size; (the number of nodes in $T$ equals the length of the array).

\subsection{Random RMQ}
\label{sec:random-rmq-random-bst}

We first consider the random permutation model for RMQ: Every (relative) ordering of
the elements in $A[1..n]$ is equally likely.
Without loss of generality, we identify the $n$ elements with their ranks,
\ie, $A[1..n]$ contains a random permutation of $[1..n]$.
We refer to this as a random RMQ instance.

We can characterize the distribution of the Cartesian tree 
associated with such a random RMQ instance:
Since the minimum in a random permutation is located at every position $i\in[n]$ with 
probability $\frac 1n$, the inorder index of the root is uniformly distributed in $[n]$.
Apart from renaming, the subarrays $A[1..i-1]$ (resp.\ $A[i+1..n]$) contain a random permutation
of $i-1$ (resp.\ $n-i$) elements, and these two permutations are independent of each other
conditional on their sizes.
Cartesian trees of random RMQ instances thus have the
same distribution as random BSTs, and in particular shape $t$ arises with probability
$\Prob{t} = \prod_{v\in t} \frac1{|t[v]|}$.
The former are also known as
random increasing binary trees~\cite[Ex.\,II.17\,\&\,Ex.\,III.33]{FlajoletSedgewick2009}).

Since the sets of answers to range-minimum queries is in bijection with Cartesian trees,
the entropy $H_n$ of the distribution of the shape of the Cartesian tree (and hence random BSTs)
gives an information-theoretic lower bound 
for the space required by any RMQ data structure (in the encoding model studied here).
Kieffer, Yang and Szpankowski~\cite{KiefferYangSzpankowski2009} show%
\footnote{%
	Hwang and Neininger~\cite{HwangNeininger2002} showed already in 2002 that the quicksort
	recurrence can be solved explicitly for arbitrary toll functions. 
	$H_n$ satisfies this recurrence with toll function $\lg n$, 
	hence they implicitly proved \weqref{eq:entropy-random-BSTs}.%
}
that the entropy of random BSTs 
$H_n = \Eover[\big] T{\lg(1/\Prob{T})} = \Eover[\big] T{\sum_{v\in T} \lg(|T[v]|)}$ is 
\begin{align}
\label{eq:entropy-random-BSTs}
		H_n 
	&\wwrel= 
		\lg(n) + 2(n+1) \sum_{i=2}^{n-1} \frac{\lg i}{(i+2)(i+1)}
	\wwrel\sim
		2n\sum_{i=2}^{\infty} \frac{\lg i}{(i+2)(i+1)}
	\wwrel\approx 1.7363771 n.
\end{align}

With these preparations, we are ready to prove our first result on range-minimum queries.

\begin{corollary}[Average-case optimal succinct RMQ]
\label{cor:average-case-RmQ}
	There is a data structure that supports (static) range-minimum queries on an array $A$ of 
	$n$ (distinct) numbers in $\Oh(1)$ worst-case time and which
	occupies $H_n + o(n) \approx 1.736 n + o(n)$
	bits of space on average over all possible permutations of the elements in $A$.
	The worst case space usage is $2n + o(n)$ bits.
\end{corollary}

\begin{proof}
We construct a hypersuccinct tree on the Cartesian tree for $A$.
It supports $\TrRank_{\inorder}$, $\TrSelect_{\inorder}$, and $\TrLCA$ in $O(1)$ time
and thus $\RMQ$ in constant time without access to $A$.
By \wref{cor:main-binary-empirical-entropies}, the space usage of hypersuccinct trees 
is at most $\min\{2n,\lg(1/\Prob{t})\} + o(n)$ for $\Prob{t} = \prod_{v\in t} \frac1{|t[v]|}$.
By the above observations, this is the probability to obtain $t$ 
as the Cartesian trees of a random permutation, so we store $t$ with
maximal pointwise redundancy of $o(n)$, hence also $o(n)$ expected redundancy over 
the entropy $H_n\sim 1.736 n$.
\end{proof}

\subsection{RMQ with Runs}
\label{sec:main-rmq-runs}

A second example of compressible RMQ instances results from partially sorted arrays.
Suppose that $A[1..n]$ can be split into $r$ runs, \ie, maximal contiguous ranges $[j_i,j_{i+1}-1]$, 
($i=1,\ldots,r$ with $j_1=1$ and $j_{r+1} = n+1$),
so that $A[j_i] \le A[j_i+1] \le \cdots\le A[j_{i+1}-1]$.

\begin{theorem}[Lower bound for RMQ with runs]
\label{thm:rmq-runs-lower-bound}
	Any range-minimum data structure in the encoding model for an array of length $n$
	that contains $r$ runs must occupy at least $\lg N_{n,r} \ge 2 \lg \binom nr - O(\log n)$ 
	bits of space where
	\(
		N_{n,r} 
			\wwrel= 
				\frac1n \binom nr \binom n{r-1}
	\)
	are the \thmemph{Narayana numbers}.
\end{theorem}
The proof follows from a bijection between Cartesian trees on sequences
of length $n$ with exactly $r$ runs and
mountain-valley diagrams (\aka Dyck paths) of length $2n$ with exactly $r$ ``peaks'';
the latter is known to be counted by the \emph{Narayana numbers}~\cite{OEIS-Narayana-numbers}.
Details are given in 
\ifproceedings{\wref{app:rmq-runs}.}{\wref{sec:omitted-proofs-rmq-runs} in the appendix.}

\begin{corollary}[Optimal succinct RMQ with runs]
\label{cor:rmq-runs-ds}
	There is a data structure that supports (static) range-minimum queries on an array $A$ of 
	$n$ numbers that consists of $r$ runs in $\Oh(1)$ worst-case time and which
	occupies $2\lg\binom nr + o(n) \le 2n + o(n)$ bits of space.
\end{corollary}
This follows from the observation that a node's type in the Cartesian tree,
\ie, whether or not its left resp.\ right child is empty,
closely reflects the runs in $A$:
A binary node marks the beginning of a non-singleton run, a leaf node marks the last position
in a non-singleton run, a right-unary node (\ie, left child empty, right child nonempty) 
is a middle node of a run, and a left-unary node corresponds to a singleton run.
With $s\in[0..r]$ the number of singleton runs, 
we can bound the space for a hypersuccinct tree in terms of 
its empirical node-type entropy by
\(
		H_0^{\mathrm{type}}(T) + o(n)
\wwrel=
		n H\left(\frac{r-s}n,\frac{s}n,\frac{n-2r+s}n,\frac{r-s}n\right)
		+ o(n)
\),
which can be shown to be no more than $2\lg \binom nr + o(n)$ for any value of $s$;
again, details are deferred to \wref{sec:omitted-proofs-rmq-runs}.

\separatedpar
We close by pointing out that hypersuccinct trees simultaneously achieve the optimal bounds for 
RMQ on random permutations and arrays with $r$ runs without taking explicit precautions for either.
The same is true for any other shape distributions of Cartesian trees that 
can be written as one of the sources from \wref{tab:sufficient-conditions}.

\section{Conclusion}
\label{sec:conclusion}

We presented the first succinct tree data structures with optimally adaptive space usage 
for a large variety of random tree sources, both for binary trees and for ordinal trees.
This is an important step towards the goal of efficient computation over
compressed structures,
and has immediate applications, \eg, as illustrated above for the range-minimum problem.

A goal for future work is to reduce the redundancy of $o(n)$,
which becomes dominant for sources with sublinear entropy.
While this has been considered for tree covering in principle~\cite{Tsur2018},
many details remain to be thoroughly investigated.

For very compressible trees, the space savings in hypersuccinct trees 
are no longer competitive.
On the other hand, with current methods for random access on dictionary-compressed sequences,
constant-time queries are not possible in the regime of mildly compressible strings; 
the same applies to known approaches to represent trees.
An interesting question is whether these opposing approaches can be combined in a way
to complement each other's strengths.
We leave this direction for future work.

\myacknowledgements

\clearpage
\appendix
\addpart{Appendix}

	In the appendix, we give full formal proof for all claims presented in the previous sections (in particular \wref{sec:contributions}) of the paper.
	Furthermore, we present a comprehensive discussion of related work and applications of hypersuccinct trees.
	
	The appendix is structured as follows:
	\wref{sec:related-work} puts the work in broader context and surveys relevant results
	from information theory, tree compression, and succinct data structures.
	In \wref{sec:preliminaries}, we introduce common notations, give basic definitions and recall important properties with respect to trees and succinct data structures. Additionally, we briefly recapitulate the Farzan-Munro algorithm from \cite{FarzanMunro2014}.  
	
	\wref{part:binary} gives full details for our results and proofs on binary trees: \wref{sec:hypersuccinct-binary} formally defines our compressed tree encoding, respectively, data structure (the \emph{hypersuccinct trees}). In \wref{sec:memoryless-binary} to \wref{sec:uniform-binary} we show that our hypersuccinct tree encoding is universal with respect to the various types of tree sources:
	In \wref{sec:memoryless-binary}, we formally define memoryless and higher order tree sources and prove our results with respect to these sources.
	In \wref{sec:fixed-size-height}, we consider fixed-size and fixed-height binary tree sources: In particular, the results and proof sketches presented in \wref{sec:bsp-random-bst} and \wref{sec:bsp-weightbalanced-bst} with respect to random BSTs and weight-balanced BSTs follow as special cases from more general results (\wref{thm:binary-monotonic-fixed-size} and \wref{thm:wfringe-dominated}) proven in \wref{sec:fixed-size-height}.
	In \wref{sec:uniform-binary}, we introduce and prove our results with respect to uniform subclass sources.

	
	\wref{part:ordinal} presents our results for ordinal trees: We describe our hypersuccinct tree encoding, respectively, data structure in \wref{sec:hypersuccinct-ordinal}.
	Furthermore, in \wref{sec:degreeentropy} and \wref{sec:fixedsizeordinal} we generalize the concepts and results with respect to memoryless/higher order and fixed-size tree sources from binary to ordinal trees. Additionally, we show that our hypersuccinct encoding achieves the so-called \emph{Label-Shape-Entropy}, a concept introduced in \cite{HuckeLohreySeelbachBenkner2019} as a measure of empirical entropy for both labeled and unlabeled trees, in \wref{sec:label-shape}.
	For the reader's convenience, \wref{sec:notation} has a comprehensive 
	list of used notation.

\section{Related Work}
\label{sec:related-work}
We discuss related work here, focusing on methods that are (also) meaningful
for unlabeled structures.

\subsection{Information Theory of Structure}

Compared with the situation for sequences (see, \eg, \cite{CoverThomas2006}),
the information theory of structured data is much less developed.
The last decade has seen increasing efforts to change that.
Sources and their entropies have been studied for 
binary trees~\cite{KiefferYangSzpankowski2009,ZhangYangKieffer2014,MagnerTurowskiSzpankowski2018,GolebiewskiMagnerSzpankowski2019}
and families of 
graphs~\cite{ChoiSzpankowski2012,LuczakMagnerSzpankowski2019}.
We are not aware of similar works specifically focusing on ordinal trees.

Some natural notions of structure sources contain more information 
(more degrees of freedom)
that can possibly be extracted from a given object.
In particular the \emph{leaf-} and 
\emph{depth-centric binary trees sources}
of~\cite{ZhangYangKieffer2014} as general classes of sources do \emph{not} 
admit a universal code~\cite[Ex.\,6\,\&\,Ex.\,8]{ZhangYangKieffer2014}
for that reason, making suitable restrictions necessary.

Other work has focused on notions of empirical entropies.
Jansson et al.~\cite{JanssonSadakaneSung2012} study the degree entropy of ordinal trees,
\ie, the zeroth-order entropy of sequence of node degrees $H^{\deg}(t)$, 
and show that $H^{\deg}(t) n$ bits are asymptotically necessary and sufficient to
represent a tree of size $n$ with given node degree frequencies.
In~\cite{HuckeLohreySeelbachBenkner2019}, a notion of $k$th order empirical entropy 
is introduced for full binary trees, where the type of a node $v$ (binary / leaf) depends
on the direction (left/right) of the last $k$ edges on the path from the root to $v$.

\subsection{Tree Compression}

The most widely studied methods for compressing trees are
\emph{DAG compression}, 
\emph{top-tree compression}, and
\emph{grammar-based compression}.
\textbf{DAG compression} is the oldest method. 
It stores identical shared fringe subtrees only once
and hence transforms a tree $t$ into a DAG. The smallest such DAG is
unique and can be computed in linear time~\cite{DowneySethiTarjan1980}.
While good enough to yield universal binary-tree codes for
fringe-dominated trees 
(cf.\ the ``Representation Ratio Negligibility Property'' in~\cite{ZhangYangKieffer2014} 
and similar sufficient conditions~\cite{SeelbachBenknerLohrey2018}), 
it is easy to construct examples where DAG compression
is exponentially worse than the other methods~\cite{Lohrey2015} because
repeated patterns ``inside'' the tree are not exploited.

\textbf{Top-tree compression}~\cite{BilleGortzLandauWeimann2015} avoids this shortcoming
by DAG compressing a \emph{top tree}~\cite{AlstrupHolmDeLichtenbergThorup2005}
of $t$ instead of $t$ itself. 
A top tree represents a hierarchy of clusters of the tree edges: 
leaves are individual edges, internal (binary) nodes are merging operations of
child clusters.
Top tree compression is presented for node-labeled ordinal trees, 
but can be applied to unlabeled trees, as well, and we formulate its properties 
for these here.
Top trees of best possible worst-case size $\Oh(n/\log n)$
and of height $\Oh(\log n)$
can be computed in linear time~\cite{LohreyRehSieber2017,DudekG18}
from an ordinal tree $t$ on $n$ nodes. Furthermore, any top DAG (of arbitrary height) for an ordinal tree $t$ of size $n$ can be transformed with a constant multiplicative blow-up in linear time into a top DAG of height $\Oh(\log n)$  for $t$ \cite{GanardiJL19}.

We can write a tree $t$ as a term (see also \wref{sec:preliminaries}), 
thus transforming it into a string.%
\footnote{%
	For terms, it is natural to have node labels (functions in the term)
	imply a given degree (function arity); such trees are called \emph{ranked}.
	When this is not the case, trees are called unranked.
	Working with ranked node labels does not preclude to study
	unlabeled trees; we can imagine nodes to be labeled with their 
	degree for this purpose.
	Any tree code must necessarily store each node's degree, so this 
	does not add additional information.
}
Any DAG for $t$ corresponds to a straight-line program (SLP)~\cite{KiefferYang2000}
for this string, but with the restriction that every nonterminal produces
(the term of) a fringe subtree of $t$.
To allow better compression through exploiting repeated patterns inside the tree, 
one can either give up the correspondence of nonterminals to subtrees/tree patters
or move to a more expressive grammar formalism.

The latter approach leads to (linear) \textbf{tree straight-line programs (TSLPs)}~\cite{Lohrey2015}, 
which can be seen as a \emph{multiple context-free grammar}~\cite[\S2.8]{Wild2010}: 
here, a rank-$k$ nonterminal derives $k+1$ substrings separated by $k$ gaps 
(instead of a single substring in context-free grammars).
That gives us the flexibility to let nonterminals produce (the term of) 
a \emph{context} $c$,
a fringe subtree $t[v]$ with $k$ holes, \ie, 
$k$ nodes are removed together with their subtree from $t[v]$ to obtain $c$.
Let $r$ denote the maximal degree in $t$, then we 
can transform any TSLP into one with only rank-1 and rank-0 nonterminals
with a blow-up of $O(r)$ in grammar size (the total size of all right-hand sides)~\cite{LohreyManethSchmidtSchauss2012}.
Like for top-tree compression, a TSLP of size $\Oh(n/\log n)$ and height $\Oh(\log n)$ can be computed from
an unlabeled ranked (constant maximal degree) tree of $n$ nodes in linear time~\cite{GanardiHuckeJezLohreyNoeth2017};
unlike for top-trees this result does not directly generalize to ordered trees
with arbitrary degrees.

Unsurprisingly, TSLPs yield universal codes for all the 
classes of binary-tree sources for which the DAG-based
code of~\cite{ZhangYangKieffer2014} is universal~\cite{GanardiHuckeLohreySeelbachBenkner2019}
(the worst-case or average-case fringe-dominated sources);
but they are also shown to be universal for the class of monotonic sources~\cite{GanardiHuckeLohreySeelbachBenkner2019}, which are not in 
general compressed optimally using DAGs,
and achieve compression to the above mentioned 
$k$th-order empirical entropy for binary trees~\cite{HuckeLohreySeelbachBenkner2019}.

Unlike top DAGs, TSLPs cannot decompose trees ``horizontally'' 
(splitting the children of one node), which makes them less effective for 
trees of large degree.
\textbf{Forest straight-line programs (FSLPs)}~\cite{GasconLohreyManethRehSieber2020}
add such an operation; they are shown to achieve 
the same compression up to constant factors
as TSLPs for the first-child-next-sibling encoding of a tree
and top DAGs (for unlabeled trees)~\cite{GasconLohreyManethRehSieber2020}.
(For labeled trees over an alphabet of size~$\sigma$, it is shown in \cite{GasconLohreyManethRehSieber2020} that a top DAG can be transformed in $\Oh(n)$ time into an equivalent FSLP with a constant multiplicative blow-up, whereas the transformation from an FSLP to a top DAG needs time $\Oh(\sigma n)$ and a multiplicative blow-up of size $\Oh(\sigma)$ is unavoidable.)

The other approach mentioned above~-- 
using unrestricted (string) \textbf{SLPs on a linearization} of a tree $t$~-- 
is investigated in~\cite{BilleLandauRamanSadakaneRaoWeimann2015}.
They consider compressing the balanced-parenthesis (BP) encoding of an ordinal tree $t$ on $n$ nodes, and show that an SLP proportional in size to the smallest 
DAG can be computed from the 
DAG~\cite[Lem.\,8.1]{BilleLandauRamanSadakaneRaoWeimann2015}.

A similar approach is taken in~\cite{GanardiHuckeLohreyNoeth2017},
focusing on ranked trees.
It is shown there that an SLP for the depth-first degree sequence (DFDS)
can be exponentially smaller than the smallest TSLP (but a TSLP
with factor $O(h\cdot d)$, for $h$ the height and $d$ the maximal degree of $t$, 
can always be computed from an DFDS-SLP), and also exponentially smaller 
than the minimal SLP for the BP sequence of an ordinal tree.
On the other hand, any TSLP (and hence DAG) can be transformed into an SLP for the DFDS 
with a factor $\Oh(d)$ blowup, where $d$ is the maximal degree in $t$.
The latter can still be more desirable as many algorithmic problems are 
efficiently solvable for TSLP-compressed trees~\cite{Lohrey2015}.

Other approaches include an LZ77-inspired methods for 
ranked trees~\cite{GawrychowskiJez2016}; it is not known
to support operations on the compressed representation.

\subsection{Succinct Trees}

The survey of Raman and Rao~\cite{RamanRao2013} and Navarro's book~\cite{Navarro2016}
give an overview of the various known succinct ordinal-tree data structures;
cardinal trees and binary trees are covered also in~\cite{FarzanMunro2014,DavoodiRamanRao2017}.
From a theoretical perspective, the tree-covering technique~-- initially suggested
by Geary, Raman and Raman~\cite{GearyRamanRaman2006}; extended and 
simplified in~\cite{HeMunroRao2012,FarzanMunro2014,DavoodiNavarroRamanRao2014}~--
might be seen as the most versatile representation~\cite{FarzanRamanRao2009}.

A typical property of succinct data structures is that their space usage is determined only by the \emph{size} of the input. 
For example, all of the standard tree representations use $2n+o(n)$ bits of space
for \emph{any} tree with $n$ nodes.
Notable exceptions are ultrasuccinct trees~\cite{JanssonSadakaneSung2012}
that compresses ordinal trees (indeed, their DFDS) to the (zeroth-order) empirical node-degree entropy and otherwise employs the data structures designed for 
the depth-first \emph{unary} degree sequeunce (DFUDS) representation.
Gańczorz~\cite{Ganczorz2020} recently extended this shape compression to labeled trees,
in which the labels are also stored in compressed form, and
Davoodi et al.~\cite{DavoodiNavarroRamanRao2014} achieved space bounded by the 
empirical node-type entropy for binary trees.
The latter two works are closest to ours in terms of their data structures; 
both are based on (variants) of tree covering.

\paragraph{``Four Russians and an American''}

Using a Huffman code for the lookup-level 
in a data structure is an arguably obvious idea,
but to the last author's surprise, 
this trick does not seem to be part of
the standard toolbox in the field.
We refer to it as the ``Four-Russians-One-American'' trick.
While explicitly mentioned in~\cite[\S4.1.2]{Navarro2016} for
higher-order-entropy-compressed bitvectors,
a recent work on run-length compressed bitvectors~\cite{ArroyueloRaman2019} 
does not discuss four Russians and one American as an option, 
although it is competitive (asymptotically) with
some of their results, \eg, \cite[Thm.\,4]{ArroyueloRaman2019}.
The survey~\cite{Grossi2013} on compressed storage schemes for strings 
does not mention four Russians and one American as an option, 
although it yields the same time-space bounds as
the (conceptually more complicated) methods discussed there 
(\S3.2 and \S3.3, based on~\cite{GonzalezNavarro2006} 
resp.~\cite{FerraginaVenturini2007}).
Finally~-- closest to our work~-- compressing micro tree types in tree-covering data structures is used in several works~\cite{FarzanMunro2014,DavoodiNavarroRamanRao2014,Tsur2018,Ganczorz2020}~--
only Gańczorz~\cite{Ganczorz2020} makes use of Four Russians and one American.
Moreover, it does not seem to have been used before to compare against measures
of compressibility other than (empirical) entropy.

\subsection{Compressed Tree Data Structures}

Some of the tree compression methods discussed above have also been turned into
compressed data structures.
Compressed tree data structures typically achieve $\Oh(\log n)$ query times, 
which is in general close to optimal as discussed below.
The exact set of supported operations for all discussed data structures is reported 
in \wref{tab:comparison-operations},
which also lists the main approaches for succinct data structures for comparison.

\begin{table}[tbhp]
	\caption{%
		Supported operations and their running time for different static-tree representations: 
		balanced parentheses (\BP), 
		depth-first unary degree sequence (\DFUDS), 
		tree covering (\TC),
		compression using top-DAGs (top directed acyclic graphs),
		forest straight-line programs (FSLP), and
		compression using straight-line programs for the \BP sequence (SLP(BP)) 
		resp.\ depth-first degree sequence (SLP(DFDS)).
		\BP includes the range-min-max-tree based data structure of~\cite{NavarroSadakane2014}; 
		ultrasuccinct trees~\cite{JanssonSadakaneSung2012} are based on \DFUDS;
		\TC is used in~\cite{FarzanMunro2014,DavoodiNavarroRamanRao2014,Ganczorz2020,Tsur2018} and in the present work.
	}
	\centering
	\smaller[2]
	\newcommand\Const{$\Oh(1)$}
	\newcommand\Log{$\Oh(\log n)$}
	\setlength\extrarowheight{3pt}
	\adjustbox{max width=\linewidth}{\begin{tabular}{|l|c|c|c|c|c|c|}
	\hline
		\textbf{Operations}                          &  \BP   & \DFUDS &  \TC   & top DAG/FSLP & SLP(BP) & SLP(DFDS) \\
	\hline\hline
		$\TrParent$                                  & \Const & \Const & \Const &     \Log     &  \Log   &   \Log    \\
	\hline
		$\TrDeg$                                     & \Const & \Const & \Const &              &         &   \Log    \\
	\hline
		$\TrFirstChild$, $\TrNextSibling$            & \Const & \Const & \Const &     \Log     &  \Log   &   \Log    \\
	\hline
		$\TrLastChild$                               & \Const & \Const & \Const &              &  \Log   &   \Log    \\
	\hline
		$\TrPrevSibling$                             & \Const & \Const & \Const &              &  \Log   &   \Log    \\
	\hline
		$\TrChild$                                   & \Const & \Const & \Const &              &         &   \Log    \\
	\hline
		$\TrChildRank$                               & \Const & \Const & \Const &              &         &   \Log    \\
	\hline
		$\TrDepth$,                                  & \Const & \Const & \Const &     \Log     &  \Log   &           \\
	\hline
		$\LCA$                                       & \Const & \Const & \Const &     \Log     &  \Log   &   \Log    \\
	\hline
		$\TrLevAnc$                                  & \Const & \Const & \Const &     \Log     &  \Log   &           \\
	\hline
		$\TrNbDesc$                                  & \Const & \Const & \Const &     \Log     &  \Log   &   \Log    \\
	\hline
		$\TrHeight$                                  & \Const &        & \Const &     \Log     &  \Log   &           \\
	\hline
		$\TrLeftLeaf$, $\TrRightLeaf$                & \Const & \Const & \Const &              &  \Log   &           \\
	\hline
		$\TrLeafRank$, $\TrLeafSel$                  & \Const & \Const & \Const &              &         &           \\
	\hline
		$\TrLevelLeft$, $\TrLevelRight$              & \Const &        & \Const &              &  \Log   &           \\
	\hline
		$\TrLevelPred$, $\TrLevelSucc$               & \Const &        & \Const &              &  \Log   &           \\
	\hline
		$\TrRank_{\pre}$, $\TrSelect_{\pre}$         & \Const & \Const & \Const &              &  \Log   &           \\
	\hline
		$\TrRank_{\inorder}$, $\TrSelect_{\inorder}$ & \Const &        & \Const &              &         &           \\
	\hline
		$\TrRank_{\post}$, $\TrSelect_{\post}$       & \Const &        & \Const &              &  \Log   &           \\
	\hline
		$\TrRank_{\DFUDS}$, $\TrSelect_{\DFUDS}$     &        & \Const & \Const &              &         &           \\
	\hline
	\end{tabular}}
	\label{tab:comparison-operations}
\end{table}

A DAG-compressed top tree of with $d$ nodes can be augmented to a $\Oh(d \log n)$ bit 
data structure~\cite{BilleGortzLandauWeimann2015,HubschleSchneiderRaman15} for ordinal trees.
Many more operations are supported by the data structure 
of~\cite{BilleLandauRamanSadakaneRaoWeimann2015},
which uses the machinery developed in the same paper 
for providing random access to SLP-compressed strings
to store an SLP for the BP string of an ordinal tree
and simulate access to the excess sequence used in~\cite{NavarroSadakane2014}.
The data structure of~\cite{GanardiHuckeLohreyNoeth2017}
also uses a string SLP, but for the depth-first degree sequence instead of the BP,
thus building on further indices for DFUDS-based succinct trees.
In both cases, the size of the data structure becomes $\Oh(g \log n)$ bits
when $g$ is the size of the SLP.

\paragraph{Lower bounds}
Since all of the above methods are dictionary-based (in the sense of~\cite{KempaPrezza2018}),
a recent lower bound~\cite{Prezza2019} applies to them. 
It builds on earlier work for SLPs~\cite{VerbinYu2013}, which proved that
if $g$ is the size of an SLP $G$ for a string $T$ with $n=|T|=\Theta(g^{1+\epsilon})$ for an $\epsilon>0$,
random access to $T$ requires $\Omega(\log n / \log \log n)$ time
for any data structure using $\Oh(g \operatorname{polylog} (n))$ space; 
(\cite{VerbinYu2013} has other tradeoffs for more compressible strings, too).
Prezza~\cite{Prezza2019} showed that also all operations required by tree data structures
based on LOUDS, DFUDS or BP sequences require $\Omega(\log n / \log \log n)$ time
on $\Oh(\alpha \operatorname{polylog} (n))$-space data structures, where $\alpha$
is the size of any dictionary compressor (and $n=\Theta(\alpha^{1+\epsilon})$).

\paragraph{Average-case behavior}
While dictionary-based compression has the ability to dramatically compress some specific trees,
simple information-theoretic arguments show that the vast majority are only slightly
compressible.
Clearly, this is true for uniformly chosen trees, but also for a vast variety of less balanced
sources as those considered in this article.
For such ``average-case'' trees, the compressed object (top dag, SLP) is of size 
$\alpha = \Oh(n/\log n)$. While the above data structures then still use $\Oh(n)$ bits of space,
none is known to be succinct (have a constant of $2$ in front of $n$).

Also, queries take $\Oh(\log n)$ time, while the random-access lower bound 
no longer applies with $g=\Omega(n/\log n)$.
Indeed, constant-time random access to SLPs is generally possible using 
$\Oh(n^\epsilon g^{1-\epsilon} |\Sigma| \log n)$ bits of space~\cite{Prezza2019}
(setting $\tau = (n/g)^\epsilon$), and that seems to be the best known bound.
With $g=\Omega(n/\log n)$, that bound is $\Omega(|\Sigma| n  \log^\epsilon(n)) = \omega(n)$.
It therefore seems not currently possible to build universally compressed data structures
on top of any dictionary-based compressor that answers queries in constant time
and has optimal space for the tree sources.

\subsection{Range-Minimum Queries}

Via the connection to lowest-common-ancestor (LCA) queries in Cartesian trees 
(see, \eg,~\cite{DavoodiNavarroRamanRao2014}),
we can formulate the RMQ problem as a task on trees:
Any (succinct) data structure for binary trees that supports 
finding nodes by inorder index ($\TrSelect_{\inorder}$),
\TrLCA, and
finding the inorder index of a node ($\TrRank_{\inorder}$)
immediately implies a (succinct) solution for RMQ.

Worst-case optimal succinct data structures for the RMQ problem have been presented by 
Fischer and Heun~\cite{FischerHeun2011},
with subsequent simplifications by Ferrada and Navarro~\cite{FerradaNavarro2017}
and Baumstark et al.~\cite{BaumstarkGogHeuerLabeit2017}.
Implementations of (slight variants) of these solutions are part of
widely-used programming libraries for succinct data structures, 
such as Succinct~\cite{succinct} and SDSL~\cite{GogBellerMoffatPetri2014}.

The above approaches use the same $2n+o(n)$ space on any input, but there are 
few attempts to exploit compressible instances.
Fischer and Heun~\cite{FischerHeun2011} show that range-minimum queries can still be
answered efficiently when the array is compressed to $k$th order empirical entropy.
For random permutations, the model we considered here, 
this does not result in significant savings.
Barbay, Fischer and Navarro~\cite{BarbayFischerNavarro2012} used LRM-trees to
obtain an RMQ data structure that adapts to presortedness in $A$, 
\eg, the number of (strict) runs by storing the tree as an ultrasuccinct tree.
Again, for the random permutations considered here, this would not result in space reductions.

Recently, Gawrychowski et al.~\cite{GawrychowskiJoMozesWeimann2020} designed RMQ solutions 
for grammar-compressed input arrays resp.\ DAG-compressed Cartesian trees.
The amount of compression for random permutation is negligible for the former; 
for the latter it is less clear, but in both cases, they have to give up constant-time queries.
The node-type entropy-compressed data structure for binary trees~\cite{DavoodiNavarroRamanRao2014} 
is the first constant-time RMQ data structure that compresses
random RMQ instances. They show that a node in the Cartesian tree has probability $\frac13$
to be binary resp.\ a leaf, and probability $\frac16$ to have a single left resp.\ right child.
The resulting entropy is $\mathcal H(\frac13,\frac13,\frac16,\frac16) \approx 1.91$ bit per node
instead of the $2$ bit for a trivial encoding.

Golin et al.~\cite{GolinIaconoKrizancRamanRaoShende2016} showed that
$1.736n$ bits are (asymptotically) necessary and sufficient to encode a random RMQ instance,
but they do not present a data structure that is able to make use of their encoding.
The constant in the lower bound also appears in the entropy of BSTs build from random 
insertions~\cite{KiefferYangSzpankowski2009}, and indeed the shape distributions are
the same~\cite[\S3]{MunroWild2019}.
The encoding of Golin et al.\ has independently been described by 
Magner et al.~\cite{MagnerTurowskiSzpankowski2018}
to compress trees (without attempts to combine it with efficient access to the stored object).
Our result closes this gap between the lower bound and 
the best data structure with efficient queries, both for RMQ and for representing binary trees.

\subsubsection{Applications}

The RMQ problem is an elementary building block in many data structures.
We discuss two exemplary applications here, in which a non-uniform distribution
over the set of RMQ answers is to be expected.

\paragraph{Range searching}
A direct application of RMQ data structures lies in 3-sided orthogonal 2D range searching. Given a set of points in the plane with coordinates $(x, y)$, the goal is to report all points in $x$-range $[x_1,x_2]$ and $y$-range $(-\infty, y_1]$ for some $x_1, x_2, y_1 \in \mathbb{R}$. 
Given such a set of points in the plane, we maintain an array of the points sorted by $x$-coordinates
and build a range-minimum data structure for the array of $y$-coordinates
and a predecessor data structure for the set of $x$-coordinates.
To report all points in $x$-range $[x_1,x_2]$ and $y$-range $(-\infty, y_1]$,
we find the indices $i$ and $j$ of the outermost points enclosed in $x$-range, \ie, 
the ranks of (the successor of) $x_1$ resp.\ (the predecessor of) $x_2$.
Then, the range-minimum in $[i,j]$ is the first candidate, 
and we compare its $y$-coordinate to $y_1$.
If it is smaller than $y_1$, we report the point and recurse in both subranges;
otherwise, we stop.

A natural testbed is to consider random point sets.
When $x$- and $y$-coordinates are independent of each other,
the ranking of the $y$-coordinates of points sorted by $x$ form a random permutation, 
and we obtain the exact setting studied in this paper.

\paragraph{Longest-common extensions}
A second application of RMQ data structures is the longest-common extension (LCE) problem on strings:
Given a string $T$, the goal is to create a data structure that allows to answer LCE queries,
\ie, given indices $i$ and $j$, what is the largest length $\ell$, so that 
$T_{i,i+\ell-1} = T_{j,j+\ell-1}$.
LCE data structures are a building block, \eg, for finding tandem repeats in genomes;
(see Gusfield's book~\cite{Gusfield1997} for many more applications).

A possible solution is to compute
the suffix array $\mathit{SA}[1..n]$, its inverse $\mathit{SA}^{-1}$, and 
the longest common prefix array $\mathit{LCP}[1..n]$ for the string $T$, where 
$\mathit{LCP}[i]$ stores the length of the longest common prefix of the $i$th and $(i-1)$st 
suffixes of $T$ in lexicographic order.
Using an RMQ data structure on $\mathit{LCP}$, $\mathit{lce}(i,j)$ is found as 
$\mathit{LCP}\bigl[\mathit{rmq}_\mathit{LCP}\bigl(\mathit{SA}^{-1}(i)+1,\mathit{SA}^{-1}(j)\bigr) \bigr]$.

Since LCE effectively asks for lowest common ancestors of leaves in suffix trees, %
the tree shapes arising from this application are related to the shape of the suffix tree of $T$.
This shape heavily depends on the considered input strings, but
for strings generated by a Markov source, it is known that 
random suffix trees behave asymptotically similar to random tries constructed
from independent strings of the same source~\cite[Chap.\,8]{JacquetSzpankowski2015}.
Those in turn have logarithmic height.
This gives some hope that the RMQ instances arising from LCE
are compressible;
we could confirm this on example strings, but further study is needed here.

\section{Preliminaries}
\label{sec:preliminaries}

In this section we introduce some basic definitions and notations; 
a comprehensive list of our notation is given in \wref{sec:notation}.
We write $[n..m] = \{n,\ldots,m\}$ and $[n] = [1..n]$ for integers $n$, $m$.
We use the standard Landau notation (\ie, $O$-notation etc.) and write $\lg$ for $\log_2$. We leave the basis of $\log$ undefined (but constant);
(any occurrence of $\log$ outside a Landau-term should thus be considered a mistake). 
We make the convention that $0 \lg (0) = 0$ and $0 \lg (x/0) = 0$ for $x \geq 0$.

\subsection{Trees}

Let $\mathcal{T}$  denote the set of all \emph{binary trees}, that is, of ordered rooted trees, such that each node has either 
(i) exactly two children, or
(ii) a single left child, or
(iii) a single right child, or
(iv) is a leaf.
For technical reasons, we also include the \emph{empty tree} $\Lambda$
(also called ``null'' in analogy of representing trees via left/right-child pointers), 
which consists of zero nodes, in the set of binary trees. 
A \emph{fringe subtree} of a binary tree $t$ is a subtree that consists of a node of $t$ and all its descendants. 
With $t[v]$ we denote the fringe subtree rooted at node $v$ and with $t_\ell[v]$ (resp. $t_r[v]$) we
denote the fringe subtree rooted in $v$'s left (resp.\ right) child: 
If $v$ does not have a left (resp., right) child, then $t_{\ell}[v]$ (resp., $t_r[v]$) is the empty binary tree. 
If $v$ is the root node of $t$, we shortly write $t_\ell$ and $t_r$ instead of $t_\ell[v]$ and $t_r[v]$.
With $|t|$ we denote the \emph{size} (\ie, number of nodes) of $t$. 
Moreover, let $h(t)$ denote the \emph{height} of $t$, which is inductively defined by $h(\Lambda) = 0$ 
and $h(t) = 1+\max(h(t_\ell), h(t_r))$, for $t \neq \Lambda$. 
Let $\mathcal{T}_n$ denote the set of binary trees with $n$ nodes and 
let $\mathcal{T}^h$ denote the set of binary trees of height $h$. 
We write trees inline as (unranked) terms with an anonymous function $\treenode$ 
representing a vertex; for example
$t = \treenode\bigl(\treenode(\Lambda,\Lambda),\treenode(\treenode(\Lambda,\Lambda),\Lambda)\bigr) \in \mathcal T_4$
represents the binary tree
\[
t \wwrel=
\begin{tikzpicture}[scale=.4,baseline=(2.south)]
	\node[inner node] (1) at (0,2) {};
	\node[inner node] (2) at (-1,1) {};
	\node[inner node] (3) at (1.25,1) {};
	\node[inner node] (4) at (.5,0) {};
	\node[leaf node] (l1) at (-1.5,0) {}; 
	\node[leaf node] (l2) at (-.5,0) {}; 
	\node[leaf node] (l3) at (0,-1) {}; 
	\node[leaf node] (l4) at (1,-1) {}; 
	\node[leaf node] (l5) at (2,0) {}; 
	\draw (1) -- (2) (1) -- (3) -- (4);
	\draw (2) -- (l1) (2) -- (l2) 
		(l3) -- (4) -- (l4)
		(3) -- (l5)
		;
\end{tikzpicture}.
\]
(We followed the convention to draw empty subtrees as squares).
A binary tree is called a \emph{full binary tree}, if every node has either exactly two children or is a leaf, \ie, there are no unary nodes. Note that there is a natural one-to-one correspondence between the set $\mathcal{T}_n$ of binary trees of size $n$ and the set of full binary trees with $n+1$ leaves.
Every binary tree $t$ of size $n$ uniquely corresponds to a full binary tree $t'$ with $n+1$ leaves by identifying the nodes of $t$ with the internal nodes of $t'$. 
Thus, results from \cite{GanardiHuckeLohreySeelbachBenkner2019, KiefferYangSzpankowski2009, SeelbachBenknerLohrey2018, ZhangYangKieffer2014} stated in the setting of full binary trees naturally transfer to our setting. 

With $\mathfrak T$ we denote the set of \emph{ordinal trees} 
(\aka Catalan trees, planted plane trees); 
every node has a potentially empty sequence of children, 
each of which is a (nonempty) ordinal tree.
Again, $\mathfrak T_n$ are ordinal trees with $n$ nodes, $|t|$ denotes the size (number of nodes) of an ordinal tree $t \in \mathfrak{T}$, and $t[v]$ denotes the fringe subtree rooted in node $v$ of $t \in \mathfrak{T}$.
We use square brackets for writing ordinal trees (to distinguish from binary trees);
for example 
$t = \treenode\bigl[\treenode[],\treenode[\treenode[]],\treenode[\treenode[],\treenode[],\treenode[]],\treenode[]\bigr] \in \mathfrak T_9$
stands for the ordinal tree
\[
t \wwrel=
\begin{tikzpicture}[scale=.4,baseline=(21.south)]
	\node[inner node] (1) at (0,2) {};
	\node[inner node] (21) at (-2,1) {};
	\node[inner node] (22) at (-1,1) {};
	\node[inner node] (23) at (.5,1) {};
	\node[inner node] (24) at (2,1) {};
	\node[inner node] (31) at (-1,0) {};
	\node[inner node] (32) at (-.1,0) {};
	\node[inner node] (33) at (0.5,0) {};
	\node[inner node] (34) at (1.1,0) {};
	\draw (21) -- (1) -- (22) -- (31) (1) -- (23) -- (32) (23) -- (33) 
		(34) -- (23) (1) -- (24)
	;
\end{tikzpicture}.
\]

\begin{definition}[BP encoding]
\label{def:bp-sequence}
	We define the \textit{balanced-parenthesis encoding} of binary trees
	$\mathit{BP}:\mathcal T \to \{\texttt{\textbf(},\texttt{\textbf)}\}^\star$, 
	recursively as follows:
	\begin{align*}
			\mathit{BP}(t)
		&\wwrel=
			\begin{dcases*}
				\varepsilon & if $t = \Lambda$\\
				\texttt{\textbf(}\cdot \mathit{BP}(t_\ell) \cdot \texttt{\textbf)} \cdot \mathit{BP}(t_r)
					& if $t = \treenode(t_l,t_r)$.
			\end{dcases*}
	\end{align*}
	Similarly, we define for ordinal trees
	$\mathit{BP_o}:\mathfrak T \to \{\texttt{\textbf(},\texttt{\textbf)}\}^\star$
	recursively:
	\begin{align*}
			\mathit{BP_o}(t)
		&\wwrel=
			\begin{dcases*}
				\varepsilon & if $t = \Lambda$\\
				\texttt{\textbf(}
				\cdot \mathit{BP_o}(t_1) 
				\cdots
				\mathit{BP_o}(t_k)
				\cdot \texttt{\textbf)} 
					& if $t = \treenode[t_1,\ldots,t_k]$, $k\in \N_0$.
			\end{dcases*}
	\end{align*}
\end{definition}
Here $\varepsilon$ denotes the empty sequence. For technical reasons, we also define \emph{forests}, which are (possibly empty) sequences of trees from $\mathfrak{T}$: With $\mathfrak{F}$, we denote the set of all forests. We have $\mathfrak F = \mathfrak T^\star$. The balanced parenthesis mapping $\mathit{BP_o}$ for ordinal trees naturally extends to a mapping $\mathit{BP_o}: \mathfrak{F} \to \{\texttt{\textbf(},\texttt{\textbf)}\}^\star$ by setting 
$\mathit{BP_o}(t_1 \cdots t_k) =\mathit{BP_o}(t_1)  \cdots  \mathit{BP_o}(t_k) $.
\begin{definition}[FCNS]\label{def:fcns}
	We define the \textit{first-child-next-sibling} mapping 
	$\FCNS:\mathfrak F \to \mathcal T$ from ordinal forests 
	 to binary trees
	recursively as follows:
	\begin{align*}
			\FCNS(\varepsilon)
		&\wwrel=
			\Lambda,
	\\
			\FCNS(t_1 = \treenode[c_1,\ldots,c_k], t_2,\ldots,t_j)
		&\wwrel=
			\treenode(\FCNS(c_1,\ldots,c_k),\FCNS(t_2,\ldots,t_j)).
	\end{align*}
\end{definition}

\begin{example}
Let $t = \treenode\bigl[\treenode[],\treenode[\treenode[]],\treenode[\treenode[],\treenode[],\treenode[]],\treenode[]\bigr]$.
Then 
\begin{align*}
\FCNS(t) 
&= 
\treenode(\FCNS(\treenode[],\treenode[\treenode[]],\treenode[\treenode[],\treenode[],\treenode[]],\treenode[]),\Lambda)
\\&=
\treenode(\treenode(\Lambda,\FCNS(\treenode[\treenode[]],\treenode[\treenode[],\treenode[],\treenode[]],\treenode[])),\Lambda)
\\&=
\treenode(\treenode(\Lambda,\treenode(\treenode(\Lambda,\Lambda),\FCNS(\treenode[\treenode[],\treenode[],\treenode[]],\treenode[]))),\Lambda)
\\&=
\treenode(\treenode(\Lambda,\treenode(\treenode(\Lambda,\Lambda),
\treenode(\FCNS(\treenode[],\treenode[],\treenode[]),\FCNS(\treenode[])))),\Lambda)
\\&=
\treenode(\treenode(\Lambda,\treenode(\treenode(\Lambda,\Lambda),
\treenode(\treenode(\Lambda,\treenode(\Lambda,\treenode(\Lambda,\Lambda))),\treenode(\Lambda,\Lambda)))),\Lambda).
\end{align*}
\end{example}
It is a folklore result that $\FCNS$ is a bijection between 
ordinal forests and binary trees, which is easily seen by noting that:
\begin{align*}
	\forall f\in\mathfrak F \wrel:
		\FCNS(f)
		\wrel=
		\mathit{BP}^{-1}(\mathit{BP_o}(f))
&&\text{and}&&
	\forall t\in\mathcal T \wrel:
		\FCNS^{-1}(t)
		\wrel=
		\mathit{BP_o}^{-1}(\mathit{BP}(t))
		.
\end{align*}

An easy, uniquely decodable binary-tree code is obtained by storing the size plus one, $|t|+1$, of the binary tree in \emph{Elias-gamma-code}, $\gamma(|t|+1)$, using $|\gamma(|t|+1)|=2\lfloor \lg(|t|+1) \rfloor+1$ many bits, followed by the balanced parenthesis encoding $ \mathit{BP}(t)$ of the binary tree, using $2|t|$ many bits. (We store the size plus one of the binary tree, instead of its size, in order to take the case into account that $t$ might be the empty binary tree). 
We can use this encoding to obtain a simple length-restricted version of any binary-tree code $C$ as follows:
\begin{definition}[Worst-case bounding trick]\label{def:worst-case-bounding}
Let $C: \mathcal{T} \to \{0,1\}^\star$ denote a uniquely decodable encoding of binary trees. We define a simple \emph{length-restricted version}
$\bar C: \mathcal{T} \to \{0,1\}^\star$ of the binary-tree code $C$ as follows:
\begin{align*}
		\bar C(t)
	&\wwrel=
		\begin{dcases*}
		\texttt{0} \cdot \gamma(|t|+1) \cdot \mathit{BP}(t),
			& if $|C(t)| > 2|t| + 2\lfloor \lg(|t|+1)\rfloor$; \\
		\texttt{1} \cdot C(t), & otherwise.
		\end{dcases*}
\end{align*}
\end{definition}
The length-restricted code $\bar C: \mathcal{T} \to \{0,1\}^\star$ then uses
\begin{align}\label{eq:worst-case-bounding}
|\bar{C}(t)| \leq \min\{|C(t)|, 2|t|+2\lfloor \lg(|t|+1) \rfloor+1\} +1
\end{align}
many bits in order to encode a binary tree $t$ of size $|t|$,
that is, by spending one extra bit to indicate the used encoding,
we can get the best of both worlds.
In a similar way, using the balanced parenthesis mapping $\mathit{BP_o}:\mathfrak T \to \{\texttt{\textbf(},\texttt{\textbf)}\}^\star$
for ordinal trees, we can obtain a \emph{lenght-restricted version} of any ordinal-tree encoding.

\subsection{Succinct Data Structures}

We use the data structure of 
Raman, Raman, and Rao~\cite{RamanRamanRao2007}
for compressed bitvectors. 
They show the following result;
we use it for more specialized data structures below.

\begin{lemma}[Compressed bit vector]
\label{lem:compressed-bit-vectors}
	Let $\mathcal{B}$ be a
	bit vector of length $n$, containing $m$ $1$-bits.  In the
	word-RAM model with word size $w=\Theta(\lg n)$ bits, there is a data
	structure of size 
	\begin{align*}
			\lg \binom{n}{m} \wbin+ O\biggl(\frac{n\log \log n}{\log n}\biggr)
		&\wwrel\leq 
			m \lg \Bigl(\frac nm\Bigr) \wbin+ O\biggl(\frac{n \log \log n}{\log n}+m\biggr)
	\end{align*}
	bits that
	supports the following operations in $O(1)$ time, 
	for any $i \in [1,n]$:
	\begin{itemize}
		\item $\mathit{access}(\mathcal{B}, i)$: return the bit at index $i$ in $\mathcal{V}$.
		\item $\mathit{rank}_\alpha(\mathcal{B}, i)$: return the number of bits with
		value $\alpha \in \{0,1\}$ in $\mathcal{B}[1..i]$.
		\item $\mathit{select}_\alpha(\mathcal{B}, i)$: return the index of the $i$-th
		bit with value $\alpha \in \{0,1\}$.
	\end{itemize}
\end{lemma}

\paragraph{Variable-cell arrays}
A standard trick (``two-level index'') allows us to store variable cell arrays:
Let $o_1,\ldots, o_m$ be $m$ objects
where $o_i$ needs $s_i$ bits of space.
The goal is to store
an ``array'' $O$ of the objects contiguously in memory, so that
we can access the $i$th element in constant time as $O[i]$;
in case $s_i > w$ (where $w$ denotes the word size in the word-RAM model), we mean by ``access'' to find its starting position.
We call such a data structure a variable-cell array.

\begin{lemma}[Variable-cell arrays]
\label{lem:variable-cell-arrays}
	There is a \textit{variable-cell array} data structure for 
	objects $o_1,\ldots, o_m$ of sizes $s_1,\ldots,s_m$
	that occupies
	\[n \bin+ m \lg(\max s_i) \bin+ 2 m \lg \lg n  \bin+ O(m)\]
	bits of space, where $n =\sum_{i=1}^m s_i$ is the total size of all objects.
\end{lemma}

\begin{proof}
Denote by $s = \min s_i$, $S = \max s_i$ and $\bar s = n/m$
the minimal, maximal and average size of the objects, respectively.
We store the concatenated bit representation in a bitvector $B[1..n]$
and use a two-level index to find where the $i$th object begins.
More in detail, we store the starting index of every $b$th object
in an array $\mathit{blockStart}[1..\lceil m/b\rceil]$.
The space usage is $\frac mb \lg n$ (ignoring ceilings around the logarithms).
In a second array $\mathit{blockLocalStart}[1..m]$,
we store for every object its starting index within its block.
The space for this is $m \lg(b S)$ (again, ignoring ceilings around the logarithms):
we have to prepare for the worst case of a block full of maximal objects.

It remains to choose the block size;
$b = \lg^2 n$ yields the claimed bounds.
Note that $\mathit{blockStart}$ is $o(n)$ 
(for $b = \omega(\lg n / \bar s)$), but
$\mathit{blockLocalStart}$ has, in general, non-negligible space overhead.
The error term only comes from ignoring ceilings around the logarithms;
its constant can be bounded explicitly.
\end{proof}

\subsection{The Farzan-Munro Algorithm}
\label{sec:farzan-munro}

We briefly recapitulate the Farzan-Munro algorithm~\cite[\S3]{FarzanMunro2014}.
Recall that we have a parameter $B$ governing the sizes of micro trees.

\subsubsection{Ordinal Trees}

The Farzan-Munro algorithm builds components bottom-up, through a recursive procedure which returns a component
containing the root of the subtree it is called on, collecting nodes
until a component contains at least $B$ nodes:
Let $v$ be a node of the tree $t$ and
suppose that components for all children $u_1,\ldots,u_k$ 
of $v$ have been computed recursively; 
the returned components will be called the \emph{active} components $C_1,\ldots,C_k$
of the children, whereas some components might be already declared \emph{permanent} and remain invariant.
The normal mode of operations~-- ``greedy packing''~-- 
is to start a new component $C$ containing 
just $v$ and to keep including the active components of $v$'s children, left to right.
If we reach $|C|\ge B$, $C$ is declared permanent, 
and we start a new component $C\gets\{v\}$.
When all children are processed, we declare $C$ permanent~-- 
except for the case when $|C|<B$ \emph{and} it contains all children 
$u_1,\ldots,u_k$ of $v$. Finally, we return $C$.

This mode in isolation is not sufficient for our goal.
An external edge of a component connects a \emph{non-root node} of the component
with the root of another component.
Greedily packing leads to potentially many external edges per component.
To achieve at most one external edge, the Farzan-Munro algorithm
distinguishes heavy and light nodes; a node $u$ is \emph{heavy} if $|t[u]|\ge B$.
The entire subtree of a light node fits into one component, 
so these do not have external edges and can be combined safely.
For heavy children of $v$, there will be further connections, 
so we must avoid grouping several heavy children into one component to 
have at most one external edge per component.
This leads to a problem since the active components of these nodes can be too small
to remain ungrouped, and in general, there can be $\Theta(n)$ heavy nodes, 
so we cannot afford to keep that many components around.

\begin{figure}[htb]
	\resizebox{\linewidth}!{\input{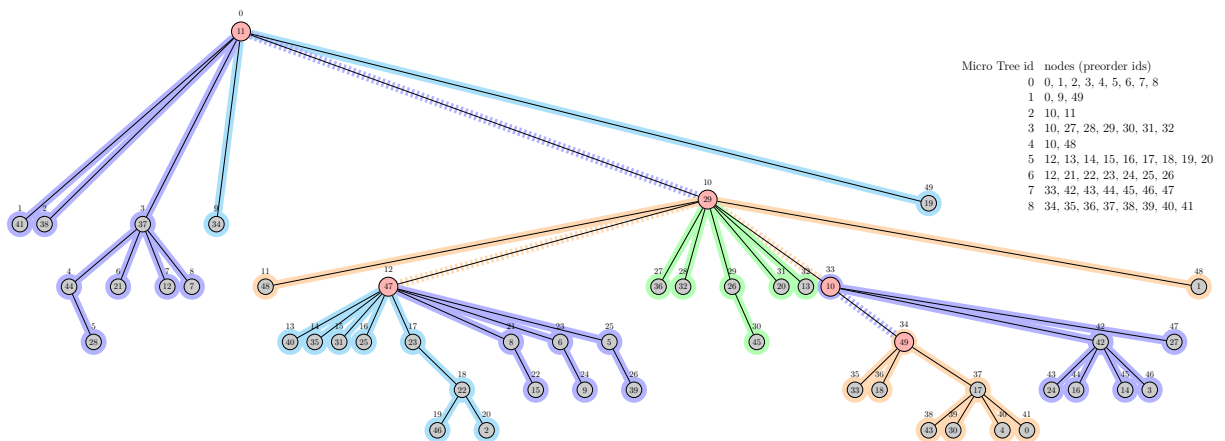}}
	\caption{%
		Example tree with $n=50$ nodes, partitioned using $B=8$.
		The root is a path node and the micro tree $\{0,9,49\}$ (preorder ids) 
		shows a split of the children, omitting the heavy child $10$.
		$10$ itself is a branching node; note that here children in components
		are contiguous. (The leftmost and rightmost children here are \emph{not}
		in the same micro tree despite the color). 
	}
	\label{fig:farzan-munro-special-cases}
\end{figure}

However, the number of \emph{branching nodes}~-- 
nodes with at least two heavy children~-- 
is always $\Oh(n/B)$. 
If $v$ is a branching node, we can declare the active components of 
heavy children permanent and use greedy packing separately 
in the gaps between/outside heavy children.
(This leads to some undersized components, but they can be charged to branching nodes, 
so remain bounded by $\Oh(n/B)$ in number.)

The remaining, and only truly ``abnormal'' case, is that of a \emph{path node},
happening when $v$ has exactly one heavy child $u_j$.
This makes two special treatments necessary.
First, we cannot bound the number of path nodes, 
so we cannot afford to declare the active component of $u_j$ permanent right away.
But that is also not necessary, for there is only one heavy child anyways.
So we just greedily pack as if all children were light.
If, however, $C_j$ was already declared permanent, 
we cannot add $v$ to it without risking an \emph{oversized} component~-- 
so $C_j$ must stay untouched~-- but we also cannot pack the children left resp.
right of $u_j$ separately since that could lead to two external edges from $v$
for the component that we pass up the tree.
Therefore, we here~-- and only in this case~-- pack \emph{across} the gap left by $u_j$,
allowing a component that contains a range of $v$'s children with one gap.

An example illustrating the special cases is shown in 
\wref{fig:farzan-munro-special-cases}.

\subsubsection{Properties}

From the procedure above, we immediately observe the following properties.

\begin{fact}
\label{fact:farzan-munro}
	Suppose we apply the Farzan-Munro algorithm with parameter $B$
	to a tree $t$ with $n$ nodes.
	For the resulting micro trees $\mu_1,\ldots,\mu_m$, we find:
	\begin{thmenumerate}{fact:farzan-munro}
	\item \label{fact:farzan-munro-micro-root-heavy}
		Every micro-tree root is heavy.
	\item \label{fact:farzan-munro-fringe-micro-ge-B}
		Every fringe micro tree has $\ge B$ nodes.
	\item \label{fact:farzan-munro-first-heavy-is-micro-root}
		If $v$ is a \thmemph{heavy leaf} ($v$ is heavy, but none of its children is),
		$v$ is a micro-tree root (potentially shared among several components).
		All components with children of $v$ contain \thmemph{one interval} of children.
	\item \label{fact:farzan-munro-branching-node}
		If $v$ is a \thmemph{branching node} (at least 2 heavy children),
		all components with children of $v$ contain \thmemph{one interval} of children.
	\item \label{fact:farzan-munro-path-node}
		If $v$ is a \thmemph{path node} (exactly 1 heavy child),
		the components containing $v$ also each contain 
		\thmemph{one interval} of children,
		except for the heavy child, which may be missing from the components of the 
		surrounding interval.
	\item \label{fact:farzan-munro-node-covering}
		Every node $v$ appears in at least one micro tree;
		if $v$ appears in several micro trees, then as the shared root of all of them.
	\end{thmenumerate}
\end{fact}

\subsubsection{Binary Trees}

When applying the Farzan-Munro algorithm to binary trees, 
simplifications arise from the bounded degree of nodes;
in particular, we never obtain components that share nodes.
\wpref{fig:tree-partitioning-example} shows an example.

It is illustrative to consider the possible cases that can arise.
Let $v$ be a node with children $u_1$ and $u_2$ (potentially null)
whose active components are $C_1$ and $C_2$, respectively.
If $u_1$ and $u_2$ are light $|C_1|,|C_2| < B$, greedy packing
yields a single component $\{v\}\cup C_1\cup C_2$.
If both children are heavy~-- $v$ is a branching node~-- 
we keep $C=\{v\}$ and declare $C$, $C_1$ and $C_2$ permanent.

If only one child, say $u_1$, is heavy, there are two cases
depending on whether $C_1$ is permanent.
If it is, we keep it and pass $\{v\}\cup C_2$ up the tree.
If $C_1$ is not permanent, it must be small, $|C_1|<B$,
and greedy packing yields a single component $\{v\}\cup C_1\cup C_2$.

	\part{Binary Trees}
\label{part:binary}

We now present our results on binary trees.
We begin by describing our code and data structure (\wref{sec:hypersuccinct-binary}),
then define the various classes of sources, state properties, list concrete examples and state and prove universality of
our hypersuccinct code for the classes of sources introduced (\wref{sec:memoryless-binary}{}\,--\,\ref{sec:uniform-binary}).
For an overview over the classes of sources sources and the concrete examples considered in our paper, see \wref{tab:sources} and \wref{tab:examples}.

\section{Hypersuccinct Binary Trees}
\label{sec:hypersuccinct-binary}

Here, we describe our compressed tree code resp.\ data structure.
Both are based on the Farzan-Munro algorithm~\cite{FarzanMunro2014} to decompose a tree 
into connected subtrees (so-called \emph{micro trees}).
It was originally designed for ordinal trees; 
we state its properties here when applied on binary trees. 
The results follow directly from the result proven in~\cite{FarzanMunro2014}
and the fact that node degrees are at most two.
For the reader's convenience,
we describe the relevant details of the method in \wref{sec:farzan-munro}.

\begin{lemma}[{{Binary tree decomposition, \cite[Theorem~1]{FarzanMunro2014}}}]
\label{lem:tree-decomposition-binary}
	For any parameter $B\ge 1$, a binary tree with $n$ nodes can be decomposed,
	in linear time, into $\Theta(n/B)$ pairwise disjoint subtrees (so-called micro trees) of $\le 2B$ nodes each.
	Moreover, each of these micro trees has at most three connections to other micro trees:
	\begin{thmenumerate}{lem:tree-decomposition-binary}
		\item an edge from a parent micro tree to the root of the micro tree,
		\item an edge to another micro tree in the left subtree of the micro tree root,
		\item an edge to another micro tree in the right subtree of the micro tree root.
		\item At least one of the edges to a child micro tree (if both of them exist) emanates from the root itself.
	\end{thmenumerate}
	In particular, contracting micro trees into single nodes yields again a binary tree.
\qed\end{lemma}
If a node $v$'s parent $u$ belongs to a different micro tree,
$u$ will have a ``null pointer'' within its micro tree, \ie, it loses its child there.
To recover these connections between micro trees, we do not only need the information 
which micro tree is a child of which other micro tree, but also which null pointer inside a micro tree
leads to the lost child.
We refer to this null pointer as the \emph{portal} of the (parent) micro tree 
(to the child micro tree).

An additional property that we need is stated in the following lemma;
it follows directly from the construction of micro trees.
\begin{lemma}[Micro-tree roots are heavy]
\label{lem:bough-implies-heavy}
	Let $v$ be the root of a micro tree constructed using the tree parameter $B$, 
	respectively, any ancestor of a micro tree root.
	Then $|t[v]| \ge B$.
\qed\end{lemma}

\needspace{4\baselineskip}
\subsection{Hypersuccinct Code}\label{sec:hypersuccinct-code}

Based on the above properties of this tree partitioning algorithm,
we design a universal code $\mathsf{H}: \mathcal{T} \rightarrow \{0,1\}^\star$ for binary trees: Given a binary tree $t$ of size $n$, we apply the Farzan-Munro algorithm with parameter $B = \lceil \frac{1}{8} \lg(n)\rceil$ to decompose the tree into micro trees $\mu_1, \dots, \mu_m$, where $m = \Theta(n/\log n)$. The size of the micro trees $\mu_1, \dots, \mu_m$ is thus upper-bounded by $\mu = \lceil \frac{1}{4} \lg(n)\rceil$. With $\Upsilon$ we denote the \emph{top tier} of the tree $t$, which is obtained from $t$ by contracting each micro tree $\mu_i$ into a single node (it forms a graph minor of $t$ in the graph-theoretic sense). In particular, as each micro tree $\mu_i$ has at most 3 connections to other micro trees (a parent micro tree and (up to) two child micro trees, see \wref{lem:tree-decomposition-binary}), $\Upsilon$ is again a binary tree, and the size of $\Upsilon$ equals the number $m$ of micro trees. With $\Sigma_{\mu} \subseteq \bigcup_{s \leq \mu} \mathcal{T}_s$ we denote the set of shapes of micro trees that occur in the tree $t$: We observe that because of the limited size of micro trees, 
there are fewer different possible shapes of binary trees
than we have micro trees. The crucial idea of our hypersuccinct encoding is to treat each shape of a micro tree as a letter in the alphabet $\Sigma_{\mu}$ and to compute a Huffman code $C: \Sigma_{\mu} \to \{0,1\}^{\star}$ based on the frequency of occurrences of micro tree shapes in the sequence $\mu_1, \dots, \mu_m \in \Sigma_{\mu}^m$: For our hypersuccinct code, we then use a \emph{length-restricted} version $\bar{C}: \Sigma_{\mu} \to \{0,1\}^{\star}$ obtained from $C$ using the simple cutoff technique 
from \wref{def:worst-case-bounding}.
Finally, for each micro tree, we have to encode which null pointers (external leaves) 
are portals to left and right child components
(if they exist). 
For that, we store the portals' rank 
in the micro-tree-local in left-to-right order of the null pointers
using $\lceil \lg(\mu+1)\rceil $ bits each.
We can thus encode $t$ as follows: 
\begin{enumerate}
\item Store $n$ and $m$ in Elias gamma code,
\item followed by the balanced-parenthesis (BP) bitstring for $\Upsilon$ (see \wref{def:bp-sequence}).
\item Next comes an encoding for $\bar{C}$; for simplicity, we simply list all possible codewords 
and their corresponding binary trees by storing the size (in Elias code) followed by their BP sequence.
\item Then, we list the length-restricted Huffman codes $\bar{C}(\mu_i)$ of all micro trees in DFS order (of $\Upsilon$).
\item Finally, we store 2 $\lceil \lg (\mu+1) \rceil$-bit integers to encode the portal nulls for each 
	micro tree, again in DFS order (of $\Upsilon$).
\end{enumerate}
Altogether, this yields our \emph{hypersuccinct encoding} $\mathsf{H}: \mathcal{T} \rightarrow \{0,1\}^\star$ for binary trees.
Decoding is obviously possible by first recovering $n$, $m$, and $\Upsilon$ from the BP,
then reading the Huffman code and finally replacing each node in $\Upsilon$
by its micro tree in a depth-first traversal, using the information about portals
to identify nodes from components that are adjacent in $\Upsilon$.
With respect to the length of the hypersuccinct code, we find the following:

\begin{lemma}[Hypersuccinct binary tree code]
\label{lem:hypersuccinct-code}
	Let $t\in\mathcal T_n$ be a binary tree of $n$ nodes, decomposed into
	micro trees $\mu_1,\ldots,\mu_m$ by the Farzan-Munro algorithm.
	Let $C$ be an ordinary Huffman code for the string $\mu_1\ldots\mu_m$.
	Then, the hypersuccinct code encodes $t$ with a binary codeword of length
	\begin{align*}
			\hypsuc t
		&\wwrel\le
			\sum_{i=1}^m |C(\mu_i)| + \Oh\biggl(n \frac{\log \log n}{\log n}\biggr).
	\end{align*}
\end{lemma}

\begin{proof}
We first show that, among the five parts of the hypersuccinct binary-tree code for $t\in\mathcal T_n$, 
all but the second to last one contribute $O(n\log\log n/\log n)$ bits.
Part 1 clearly needs $\Oh(\log n)$ bits and
Part 2 requires $2m = \Theta(n/\log n)$ bits.
For Part 3, observe that 
$$|\Sigma_\mu| \le \sum_{s \le \lceil \lg n /4 \rceil} 4^{s} < \frac43 \cdot 4^{\lg n / 4+1} = \frac{16}{3} \sqrt n.$$
With the worst-case cutoff technique from \wref{def:worst-case-bounding}, $\bar{C}(\mu_i) \le 2 +2\lg(\mu+1)+ 2\mu \leq O(\mu)$, so
we need asymptotically $\Oh(\sqrt n)$ entries / codewords in the table, 
each of size $\Oh(\mu) = \Oh(\log n)$,
for an overall table size of $\Oh(\sqrt n \log n)$.
Part 5 uses $m \cdot 2 \lceil \lg (\mu+1)\rceil  = \Theta(n \cdot \frac{\log \log n}{\log n})$
bits of space.
It remains to analyze Part 4.
We note that 
by applying the worst-case pruning scheme of \wref{def:worst-case-bounding},
we waste $1$ bit per micro tree compared to a pure, non-restricted Huffman code.
But these wasted bits amount to $m = \Oh(n/ \log n)$ bits in total,
and so are again a lower-order term:
\begin{align*}
\sum_{i=1}^m\bar{C}(\mu_i) &= \sum_{i=1}^m\min\{|C(\mu_i)|+1,2|\mu_i|+2\lfloor\lg|\mu_i|+1\rfloor+2 \} \\
&\leq \sum_{i=1}^m (|C(\mu_i)|+1) = \sum_{i=1}^m |C(\mu_i)| + O(n/\log n),
\end{align*}
where the first equality comes from \wref{eq:worst-case-bounding}.
This finishes the proof.
\end{proof}

\subsection{Tree Covering Data Structures}

What sets hypersuccinct code apart from other known codes
is that it can be turned into a universally compressed tree data structure
with constant-time queries.
For that, we use a well-known property of tree covering
that can be formalized as follows.

\begin{theorem}[Tree-covering index~\cite{FarzanMunro2014}]
\label{thm:binary-tree-covering}
	Given a binary tree $t\in\mathcal T_n$, decomposed into micro trees
	$\mu_1,\ldots,\mu_m$ with tree covering.
	Assuming access to a data structure that maps $i$
	to $\mathit{BP}(\mu_i)$ in constant-time (for any $i\in[m]$),
	there is a data structure occupying $o(n)$ additional bits
	of space that supports all operations from \wref{tab:binary-operations}
	in constant-time.
\end{theorem}

We will use this to turn our hypersuccinct code into a full-blown tree data structure;
any results proven about the space of the former via \wref{lem:hypersuccinct-code}
can then be transferred to this data structure.
To realize the mapping of micro tree ids to shapes,
we will keep a variable-cell bitvector (\wref{lem:variable-cell-arrays}) 
storing $i \mapsto \bar C(\mu_i)$
where $\bar C(\mu_i)$ is the (length-restricted) Huffman code of $\mu_1,\ldots,\mu_m$.
To get the balanced-parenthesis strings, we additionally store a 
lookup-table for $\bar C(\mu_i) \mapsto \mathit{BP}(\mu_i)$.
The space for the former is $\sum_{i=1}^m |\bar C(\mu_i)| + \Oh(n \log \log n / \log n)$ by \wref{lem:variable-cell-arrays}
and the latter is $O(\sqrt n \operatorname{polylog} n)$ because of the length restriction (see the proof of \wref{lem:hypersuccinct-code}).

\section{Memoryless and Higher-Order Binary-Tree Sources}

\label{sec:memoryless-binary}

Let $t \in \mathcal{T}$ be a binary tree. We define the \emph{type} of a node $v$ as 
\begin{align*}
\type(v) = \begin{cases} 0 \quad &\text{if } v \text{ is a leaf,}\\
1 \quad &\text{if } v \text{ has a single left child (and no right child),}\\
2 \quad &\text{if } v \text{ is a binary node,}\\
3 \quad &\text{if } v \text{ has a single right child (and no left child).}
\end{cases}
\end{align*}
For a node $v$ of a binary tree $t$, we inductively define the \emph{history} of $v$, 
$h(v) \in \{1,2,3\}^\star$, as follows: 
If $v$ is the root node, we set $h(v) = \varepsilon$, (\ie, the empty string). 
If $v$ is the child node of node $w$ of $t$, we set $h(v) = h(w)\type(w)$, \ie, 
in order to obtain $h(v)$, we concatenate the types of $v$'s ancestors. 
Note that $\type(v)$ is not part of the history of $v$. 
Moreover, we define the $k$-history of $v$, $h_k(v)$, as the length-$k$-suffix of $1^k h(v)$, 
\ie, if $|h(v)|\geq k$, $h_k(v)$ equals the last $k$ characters of $h(v)$, 
and if $|h(v)|< k$, we pad this too short history with $1$'s in order to obtain 
a string $h_k(v)$ of length $k$.%
\footnote{%
	This is an ad-hoc decision: Alternatively, we could allow histories of length smaller than $k$.%
}

Let $k \geq 0$, let $z \in \{1,2,3\}^k$ and let $i \in \{0,1,2,3\}$. 
With $n_{z}^t$ we denote the number of nodes of $t$ with $k$-history $z$ 
and with $n_{z,i}^t$ we denote the number of nodes of type $i$ of $t$ and $k$-history $z$ .
A \emph{$k$th-order type process} $\tau =(\tau_z)_{z \in \{1,2,3\}^k}$ 
is a tuple of probability distributions $\tau_z: \{0,1,2,3\} \rightarrow [0,1]$. 
A $k$th order type process assigns a probability $\mathbb{P}[t]$ to a binary tree $t$ by
\begin{align}\label{eq:prob-higherordertypes}
		\mathbb{P}[t] 
	\wwrel= 
		\prod_{v \in t} \tau_{h_k(v)}(\type(v)) 
	\wwrel= 
		\mkern-15mu\prod_{z \in \{1,2,3\}^k} \;\prod_{i=0}^3
			\left( \tau_z(i)\right)^{n_{z,i}^t}.
\end{align}
If $k=0$, we call such a $k$th-order type process a 
\emph{memoryless binary-tree source}: 
in this case, the probability distribution on the node types is 
independent of the node's ancestors' node types.
If $k>0$, we call the $k$th-order type process a 
\emph{higher-order binary-tree source.}

A $k$th-order type process randomly constructs a binary tree $t$ as follows: 
In a top-down way, starting at the root node, 
we determine for each node $v$ its type, 
where this decision depends on the $k$-history $h_k(v)$ of the node: 
The probability that a node $v$ is of type $i$ is given by $\tau_{h_k(v)}(i)$. 
If $i=0$, then this node becomes a leaf and the process stops at this node. 
If $i=1$, we attach a single left child to the node, 
if $i=2$, we attach a left and a right child to the node, 
and if $i=3$, we attach a single right child to the node. 
The process then continues at these child nodes. Note that this process might produce infinite trees with non-zero probability.

We define the following higher-order empirical entropy for binary trees:
\begin{definition}[Empirical type entropy]\label{def:empiricaltype}
	Let $k\geq 0$ be an integer, and let $t \in \mathcal{T}$ be a binary tree. 
	The (unnormalized) \emph{$k$th-order type entropy} $H_k^{\type}(t)$ of $t$ is defined as
	\begin{align*}
			H_k^{\type}(t) 
		\wwrel= 
			\sum_{z \in \{1,2,3\}^k}\sum_{i=0}^3 
				n_{z,i}^t \lg\left( \frac{n_z^t}{n_{z,i}^t} \right).
	\end{align*}
\end{definition}
The corresponding normalized tree entropy is obtained by dividing by the tree size.
The zeroth order empirical type entropy is a slight variant of the degree entropy defined for ordinal trees by Jansson, Sadakane, and Sung~\cite{JanssonSadakaneSung2012} and occurs implicitly in \cite{DavoodiNavarroRamanRao2014}.

We say that the $k$th-order type process $(\tau_z)_{z}$ is the 
\emph{empirical $k$th-order type process of a tree $t$}, 
if $\tau_z(i) = \frac{n_{z,i}^t}{n_z^t}$ for all $z \in \{1,2,3\}^k$ and $i \in \{0,1,2,3\}$. 
In particular, if $(\tau_z)_z$ is the empirical $k$th-order type process 
of a binary tree $t \in \mathcal{T}$, we have
\begin{align*}
		\lg\left(\frac{1}{\mathbb{P}[t]}\right) 
	&\wwrel= \mkern-10mu
		\sum_{z \in \{1,2,3\}^k}\sum_{i=0}^3n_{z,i}^t\lg \left(\frac{1}{\tau_z(i)}\right) 
	\wwrel= \mkern-10mu
		\sum_{z \in \{1,2,3\}^k}\sum_{i=0}^3n_{z,i}^t\lg \left(\frac{n_z^t}{n_{z,i}^t}\right) 
	\wwrel= 
		H_k^{\type}(t).
\end{align*}
This shows that the empirical entropy is precisely the number of bits an optimal code
can achieve for this source.

\begin{example}[Uniform binary trees]\label{exm:typebinary}
In order to encode a (uniformly random) binary tree of size $n$, $2n$ bits are necessary \cite{FarzanMunro2014}. Let $\tau$ denote the memoryless type process defined by $\tau(0)=\tau(1)=\tau(2)=\tau(3) = \frac14$, then for every binary tree $t$ of size $n$ we have $\mathbb{P}[t] = 4^{-n}$ and in particular, $\lg (1/\mathbb{P}[t]) =2n$.
\end{example}

\begin{example}[Full binary trees]\label{exm:typefullbinary}
Probability distributions over full binary trees are obtained from type processes $(\tau_z)_{z \in \{1,2,3\}^k}$ with $\tau_z (1) = \tau_z(3) = 0$ for all $z \in \{1,2,3\}^k$. 
Recall that every full binary tree consists of an odd number $n=2\nu+1$ of nodes: $\nu$ binary nodes and $\nu+1$ leaves for some integer $\nu$. 
If $\tau$ is a memoryless type process, we thus have 
$\mathbb{P}[t] = \tau(0)^{\nu+1}\tau(2)^\nu$ for every $t \in \mathcal{T}_{n}$. 
Setting $\tau(0)=\tau(2) = \frac{1}{2}$ yields 
\begin{align*}
		\lg\left(\frac{1}{\mathbb{P}[t]}\right)
	\wwrel=
		(\nu+1)\log\left(2\right) + \nu\log\left(2\right) 
	\wwrel= n,
\end{align*} 
and $n=2\nu+1$ is the minimum number of bits needed to represent a 
(uniformly chosen) full binary tree $t \in \mathcal{T}_{n}$ \cite{HuckeLohreySeelbachBenkner2019}.
\end{example}

\begin{example}[Unary paths]\label{exm:typeunarypath}
Type processes $(\tau_z)_{z \in \{1,2,3\}^k}$ with $\tau_z(2) = 0$ yield probability distributions over unary-path trees, \ie, trees only consisting of unary nodes and one leaf. In order to encode a unary-path tree of size $n+1$, we need $n$ bits (to encode the $n$ ``directions'' left/right). For a fixed integer $n$, let $\tau^{n}$ denote the memoryless type process with $\tau^{n}(1) =\tau^{n}(3)=1/(2+\varepsilon_n)$ and $\tau^{n}(0) =\varepsilon_n/(2+\varepsilon_n)$, for $\varepsilon_n = 2/n$. We have 
\begin{align*}
		\lg\left(\frac{1}{\mathbb{P}[t]}\right) 
	&\wwrel= 
		n\lg(2+\varepsilon_n) + \lg\left(1+\frac{2}{\varepsilon_n}\right)
	\wwrel=
		n\lg\left(2+\frac{2}{n}\right) + \lg (n+1)
\\	&\wwrel\leq 
		n+\lg(n+1)+\frac{1}{\ln(2)},
\end{align*}
for every unary-path tree $t \in \mathcal{T}_{n+1}$.
\end{example}

\begin{example}[Motzkin trees]\label{exm:typeMotzkin}
Motzkin trees are binary trees with only one type of unary nodes: Probability distributions over Motzkin trees can be modeled by type processes $(\tau_z)_{z \in \{1,2,3\}^k}$ with $\tau_z(3) = 0$ for every $z \in \{1,2,3\}^k$. For encoding (uniformly random) Motzkin trees of size $n$, asymptotically $\lg(3)n\approx1.58496n$ bits are necessary \cite[Theorem 6.16]{SedgewickFlajolet1996}: Let $\tau$ denote the memoryless type process with $\tau(0) = \tau(1) = \tau(2) = \frac{1}{3}$, then $\mathbb{P}(t) = 3^{-n}$ for every Motzkin tree of size $n$. In particular, we have
$
\lg (1/\mathbb{P}[t]) = \lg (3) n.
$
\end{example}

\begin{example}
Let $(\tau_z)_{z \in \{1,2,3\}^k}$ denote a higher-order type process with $\tau_{z'1}(1)=\tau_{z'3}(1)=\tau_{z'1}(3)=\tau_{z'3}(3)=0$ for every $z' \in \{1,2,3\}^{k-1}$, then $\tau$ only generates binary trees with non-zero probability, in which children of unary nodes are either binary nodes or leaves, \ie, on each path from the root node to a leaf of the tree, we do not pass two consecutive unary nodes. For example, a first-order type process $(\tau_z)_{z \in \{1,2,3\}}$ which satisfies this property is obtained by setting $\tau_{1}(2)=\tau_{1}(0)= 1/2$, $\tau_{3}(2)=\tau_{3}(0)= 1/2$, $\tau_{2}(0)=\tau_{2}(1)=\tau_{2}(2)=\tau_{2}(3)= 1/4$.
\end{example}

\begin{example}
The \emph{random binary search tree model} assigns a probability to a binary tree of size $n$ by setting
\begin{align*}
\mathbb{P}_{\mathit{bst}}(t) = \prod_{v \in t}\frac{1}{|t[v]|},
\end{align*}
where the product ranges over all nodes $v$ of $t$, see \wref{exm:bst} and \wref{sec:random-rmq-random-bst} for more information. 
This distribution over binary trees arises for binary search trees (BST)s, when they are built by successive insertions from a uniformly random permutation. 
In \cite{GolinIaconoKrizancRamanRaoShende2016}, 
it was shown that the average numbers of node types in a random binary search tree $t$ of size $n$ satisfy 
$
		\sum_{t \in \mathcal{T}_n}\mathbb{P}_{\mathit{bst}}[t]n_0^t
	\sim
		\sum_{t \in \mathcal{T}_n}\mathbb{P}_{\mathit{bst}}[t]n_2^t 
	\sim 
		\frac13n
$ and 
$\sum_{t \in \mathcal{T}_n}\mathbb{P}_{\mathit{bst}}[t]n_1^t=\sum_{t \in \mathcal{T}_n}\mathbb{P}_{\mathit{bst}}[t]n_3^t \sim \frac16n$.
Thus, it seems natural to consider the memoryless type process given by $\tau(0)=\tau(2)=\frac13$ and $\tau(1)=\tau(3)=\frac16$: In
\cite{DavoodiNavarroRamanRao2014}, 
a data structure supporting RMQ in constant time using 
\begin{align*}
	&
		\sum_{t \in \mathcal{T}_n}
			\mathbb{P}_{\mathit{bst}}(t)\lg\left(\frac{1}{\mathbb{P}_{\tau}(t)}\right) 
			+ o(n)
\\	&\wwrel=
		\sum_{t \in \mathcal{T}_n}
			\mathbb{P}_{\mathit{bst}}(t)\left(
				n_0^t\lg(3)+n_2^t\lg(3)+n_1^t\lg(6)+n_3^t\lg(6)
			\right) 
			+ o(n)
\\	&\wwrel=
		 \frac{1}{3}\lg(3)n
		+\frac{1}{3}\lg(3)n
		+\frac{1}{6}\lg(6)n
		+\frac{1}{6}\lg(6)n 
\\	&\wwrel\approx 
		1.919n + o(n)
\end{align*}
many bits in expectation is introduced. 
However, to achieve the asymptotically optimal 
$
		\sum_{t \in \mathcal{T}_n}
		\mathbb{P}_{\mathit{bst}}(t)
			\lg\bigl(\frac{1}{\mathbb{P}_{\mathit{bst}}[t]}\bigr)
		+o(n) 
	\approx 1.736n+o(n)
$ 
bits on average (see \wref{sec:random-rmq-random-bst}), 
it is necessary to consider a different kind of binary-tree sources.
\end{example}

\subsection{Universality of Memoryless and Higher-Order Sources}
\label{sec:universal-memoryless}
In order to show universality of the hypersuccinct code from \wref{sec:hypersuccinct-code} with respect to memoryless and higher-order binary tree sources, we first derive a source-specific encoding (a so-called depth-first arithmetic code) with respect to the memoryless/higher-order source, against which we will then compare our hypersuccinct code. An overview of the strategy is given in \wref{sec:bsp-random-bst}. 

The formula for $\mathbb{P}[t]$, Equation \eqref{eq:prob-higherordertypes}, suggests a 
route for an (essentially) optimal \emph{source-specific} encoding of any binary tree $t$ 
with $\mathbb{P}[t]>0$ that, given a $k$th-order type process $(\tau_z)_z$, spends $\lg(1/\mathbb{P}[t])$ (plus lower-order terms) many bits in order to encode a binary tree $t \in \mathcal{T}$ with $\mathbb{P}[t]>0$: Such an encoding may spend $\lg (1/\tau_z(i))$ 
many bits per node $v$ of type $i$ and of $k$-history $z$ of $t$. (Note that as $\mathbb{P}[t]>0$ by assumption, we have $\tau_{h_k(v)}(\type(v))>0$ for every node $v$ of $t$).
Assuming that we ``know'' the $k$th-order type process $(\tau_z)_z$ --  \ie, that it need not be stored as 
part of the 
encoding~-- we can use \emph{arithmetic coding}~\cite{WittenNealCleary1987}  in order to encode the type of node $v$ 
in that many bits. A simple (source-dependent) encoding $D_{\tau}$, dependent on a given $k$th-order type process 
$(\tau_z)_z$, thus stores a tree $t$ as follows: While traversing the tree in depth-first 
order, we always know the $k$-history of each node $v$ we pass, and encode $\type(v)$ 
of each node $v$, using arithmetic coding: To encode $\type(v)$, we feed the arithmetic 
coder with the model that the next symbol is a number $i \in \{0,1,2,3\}$ with probability $\tau_z(i)$, where $z$ is the $k$-history of $v$. 
We refer to this (source-dependent) code $D_{\tau}$ as the \emph{depth-first arithmetic code} for the type process $\tau$.
We can reconstruct the tree $t$ 
recursively from its code $D_{\tau}(t)$, as we always know the node types of nodes we 
have already visited in the depth-first order traversal of the tree, and the $k$-history of 
the node which we will visit next. As arithmetic 
coding needs $\lg \left(1/\tau_{h_k(v)}(\type(v))\right)$ many bits per node $v$, plus at most $2$ bits of overhead, the total number of bits needed to store a binary tree $t \in \mathcal{T}$ is thus
\begin{align}\label{eq:depth-first-memoryless}
		|D_{\tau}(t)| 
	\wwrel\leq 
		\sum_{v \in t}\lg \left(\frac{1}{\tau_{h_k(v)}(\type(v))}\right)+2 
	\wwrel= 
		\lg \left(\frac{1}{\mathbb{P}[t]}\right)+2.
\end{align}
Note that $D_{\tau}$ is a single prefix-free code for the set of all binary trees which satisfy $\mathbb{P}[t]>0$ with respect to the $k$th-order type process $(\tau_z)_z$.
We now start with the following lemma:

\begin{lemma}\label{lem:kthorderdegree}
Let $(\tau_z)_z$ be a $k$th-order type process and let $t \in \mathcal{T}$ be a binary tree of size $n$ with $\mathbb{P}[t]>0$. Then
\begin{align*}
		\sum_{i=1}^m |C(\mu_i)|  
	\wwrel\leq 
		\lg\left(\frac{1}{\mathbb{P}[t]}\right) + O\left(\frac{nk}{\log n}+\frac{n \log\log n}{\log n}\right),
\end{align*}
where $C$ is a Huffman code for the sequence of the micro trees $\mu_1,\ldots,\mu_m$ obtained from the tree covering scheme (see \wref{sec:hypersuccinct-code}).
\end{lemma}
\begin{proof}

Let $v$ be a node of $t$ and let $\mu_i$ denote the micro tree of $t$ that contains $v$. For the sake of clarity, let $\type^t(v)$ denote the type of $v$ viewed as a node of $t$, and let $\type^{\mu_i}(v)$ denote the type of $v$ in $\mu_i$.
We find that $\type^{\mu_i}(v)=\type^t(v)$, unless $v$ is a parent of a portal null: In this case, the degree of $v$ in $\mu_i$ is strictly smaller than the degree of $v$ in $t$. By definition of the tree covering scheme 
(\wref{lem:tree-decomposition-binary}), there are at most two parents 
of portal nulls per micro tree $\mu_i$. If a tree $\mu_i$ 
contains two parents of portal nulls, one of those two nodes is the root node by \wref{lem:tree-decomposition-binary}. 
Let $\pi_{i,1}$ denote the root node of $\mu_i$ and let $\pi_{i,2}$ denote the parent 
node of the portal null in $\mu_i$ which is not the root node, if it exists. Moreover, let $\pos(\pi_{i,2})$ denote the preorder index of node $\pi_{i,2}$ in $\mu_i$. 

Again for the sake of clarity, let $h_k^t(v)$ denote the $k$-history of $v$ in $t$, and let $h_k^{\mu_i}(v)$ denote the $k$-history of $v$ in micro tree $\mu_i$.
If $v$ is of depth smaller than $k$ (within $\mu_i$), then its $k$-history $h_k^{\mu_i}(v)$ in $\mu_i$
might not coincide with its $k$-history $h_k^t(v)$ in $t$, and if $v$ is a descendant of order smaller than $k$ of node $\pi_{i,2}$ (\ie, $v$ is of depth smaller than $k$ in the subtree of $\mu_i$ rooted in $\pi_{i,2}$), then its $k$-history in $\mu_i$ does not coincide with its $k$-history in $t$, as $\pi_{i,2}$ changes its node type. 

However, if we know the $k$-history $h_k^t(\pi_{i,1})$ of the root node $\pi_{i,1}$ of $\mu_i$, the type $\type^t(\pi_{i,1})$, and the preorder position (in $\mu_i$) and type (in $t$) of the node $\pi_{i,2}$, we are able to recover the $k$-history $h_k^t(v)$ of every node $v \in \mu_i$. 
We define the following modification of $D_{\tau}$ (\ie, the depth-first arithmetic code defined at the beginning of \wref{sec:universal-memoryless}), under the assumption that we know $h_k^t(\pi_{i,1})$, $\type^t(\pi_{i,1})$, $\type^t(\pi_{i,2})$ and $\pos(\pi_{i,2})$:
While traversing the micro-tree $\mu_i$ in depth-first order, we 
 encode $\type^{\mu_i}(v)$ (\ie, $\type^t(v)$) for every node $v$ of $\mu_i$ except for nodes $\pi_{i,1}$ and $\pi_{i,2}$ (if it exists), for which we encode $\type^t(\pi_{i,1})$ and $\type^t(\pi_{i,2})$ (which we know, by assumption, as well as the preorder position of $\pi_{i,2}$);
as we know $h_k^t(\pi_{i,1})$ by assumption, as well as the node types of $\pi_{i,1}$ and $\pi_{i,2}$, we know $h_k^t(v)$ at every node $v$ we pass: we therefore encode $\type^t(v)$ using arithmetic coding by feeding the arithmetic coder with the model that the next symbol is a number $i \in \{0,1,2,3\}$ with probability $\tau_{h_k^t(v)}(i)$. We denote this modification of $D_{\tau}(\mu_i)$ with $D_{\tau}^{h_k^t(\pi_{i,1})}(\mu_i)$ and find that it spends at most
\begin{align}\label{eq:dtau}
D_{\tau}^{h_k^t(\pi_{i,1})}(\mu_i)\leq \sum_{v \in \mu_i}\lg\left(\frac{1}{\tau_{h_k^t(v)}(\type^t(v))}\right)+2
\end{align}
many bits in order to encode a micro tree $\mu_i$.

Furthermore, let $S: \{0,1,2,3\}^\star \to \{0,1\}^\star$ denote any uniquely decodable binary encoding which spends $2|z|$ bits in order to encode $z \in \{0,1,2,3\}^\star$. 
Let $\mathcal{I}_0$ denote the set of indices $i \in [m]$ for which $\mu_i$ is fringe, 
let $\mathcal{I}_1$ denote the set of indices $i \in [m]$ for which the root node of $\mu_i$ is a parent of a portal null, but no other portal null exists, 
let $\mathcal{I}_3$ denote the set of indices $i \in [m]$, for which the root node of $\mu_i$ is not a parent of a portal null, but node $\pi_{i,2}$ is a parent of a portal null, 
and let $\mathcal{I}_3 = [m]\setminus (\mathcal{I}_0 \cup \mathcal{I}_1 \cup \mathcal{I}_2)$. 
We define a modified encoding of $\mu_i$ as follows:

\begin{align*}
		\tilde{D}_{\tau}(\mu_i) 
	\wwrel= \begin{dcases} 
		00\cdot S(h_k^t(\pi_{i,1})) \cdot D_{\tau}^{h_k^t(\pi_{i,1})}(\mu_i), 
			&\text{if } i \in \mathcal{I}_0;\\[1ex]
		\begin{aligned}
			&01\cdot S(h_k^t(\pi_{i,1})) \cdot \gamma(\type^t(\pi_{i,1})+1) \cdot D_{\tau}^{h_k^t(\mu_i)}(\mu_i), 
		\end{aligned}
			&\text{if } i \in \mathcal{I}_1;\\[1ex]
			\begin{aligned}
			&10\cdot S(h_k^t(\pi_{i,1}))\cdot \gamma(\pos(\pi_{i,2})) \cdot \gamma(\type^t(\pi_{i,2}+1))
		\\	& \phantom{11}
			\cdot D_{\tau}^{h_k^t(\mu_i)}(\mu_i), 
		\end{aligned}
			&\text{if } i \in \mathcal{I}_1;\\[1ex]
		\begin{aligned}
			&11\cdot S(h_k^t(\pi_{i,1}))
				\cdot \gamma(\type^t(\pi_{i,1}+1))\cdot \gamma(\pos(\pi_{i,2}))
		\\ &\phantom{11}
				\cdot \gamma(\type^t(\pi_{i,2})+1)
				\cdot D_{\tau}^{h_k^t(\pi_{i,1})}(\mu_i)
				\mkern-20mu
		\end{aligned}
			\quad &\text{if } i \in \mathcal{I}_3.
	\end{dcases}
\end{align*}

Note that formally, $\tilde{D}_{\tau}$ is \emph{not} a prefix-free code over $\Sigma_{\mu}$, as there can be micro tree shapes that are assigned several codewords by $\tilde{D}_{\tau}$, depending on which and how many nodes are portals to other micro trees. But $\tilde{D}_{\tau}$ is uniquely decodable to local shapes of micro trees, and can thus be seen as a \emph{generalized prefix-free code}, where more than one codeword per symbol is allowed. In terms of encoding length, assigning more than one codeword is not helpful~-- removing all but the shortest one never makes the code worse~-- so a Huffman code minimizes the encoding length over the larger class of \emph{generalized} prefix-free codes.
Thus, as a Huffman code minimizes the encoding length over the class of \emph{generalized} prefix-free codes, we find
\begin{align*}
		\sum_{i=1}^m |C(\mu_i)| 
	&\wwrel\leq 
		\sum_{i=1}^m|\tilde{D}_{\tau}(\mu_i)| 
	\wwrel= 
		\sum_{j=0}^3\sum_{i \in \mathcal{I}_j}|\tilde{D}_{\tau}(\mu_i)|
\\	&\wwrel\leq
		\sum_{i=1}^m|S(h_k^t(\pi_{i,1}))|
		\bin+ \sum_{j=0}^3\sum_{i \in \mathcal{I}_j}|D_{\tau}^{h_k^t(\pi_{i,1})}(\mu_i)|  \bin+ \Oh(m\log \mu),
\intertext{%
as $\pos(\pi_{i,2}) \leq \mu$. With the estimate \eqref{eq:dtau}, this is upper-bounded by
}
	&\wwrel\le 
		\sum_{i=1}^m|S(h_k^t(\pi_{i,1}))|
		\bin+\sum_{j=0}^3\sum_{i \in \mathcal{I}_j} 
			\sum_{v \in \mu_i} \lg \left(\frac{1}{\tau_{h_k^t(v)}(\type^t(v))}\right) \bin+ O(m \log \mu).
\end{align*}
Finally, as $|S(h_k^{t}(\pi_{i,1}))|=2k$, we have
\begin{align*}
		\sum_{i=1}^m |C(\mu_i)| 
	&\wwrel\leq 
		\sum_{v \in t}\lg \left(\frac{1}{\tau_{h_k^t(v)}(\type^t(v))}\right) 
		\bin+ O(m \log \mu) + O(km)
\\	&\wwrel= 
		\lg \left(\frac{1}{\mathbb{P}[t]}\right) 
		\bin+ O\left(\frac{nk}{\log n }+\frac{ n \log \log n}{\log n}\right),
\end{align*}
as $m = \Theta(n/\log n)$ and $\mu = \Theta(\log n)$ (see \wref{sec:hypersuccinct-code}).
This finishes the proof of the lemma.

\end{proof}
From \wref{lem:kthorderdegree} and \wref{lem:hypersuccinct-code}, we find that our hypersuccinct code is universal with respect to memoryless/higher-order type processes of order $k$, if $k =o(\log n)$:
\begin{theorem}
\label{thm:kthorderdegree}
	Let $(\tau_z)_z$ be a $k$th-order type process. The hypersuccinct code $\mathsf{H}: \mathcal{T} \rightarrow \{0,1\}^\star$ satisfies
	\begin{align*}
			|\mathsf{H}(t)| 
		\wwrel\leq 
			\lg\left(\frac{1}{\mathbb{P}[t]}\right) 
			\bin+ O\left(\frac{nk+n \log \log n}{\log n}\right)
	\end{align*}
	for every $t \in \mathcal{T}_n$ with $\mathbb{P}[t]>0$.
	In particular, if $(\tau_z)_z$ is the empirical $k$th-order type process of the binary tree $t$, we have
	\begin{align*}
			|\mathsf{H}(t)| 
		\wwrel\leq 
			H_k^{\type}(t) + O\left(\frac{nk+ n \log \log n}{\log n}\right).
	\end{align*}
\end{theorem}

From \wref{thm:kthorderdegree} and  \wref{exm:typebinary}, \wref{exm:typefullbinary}, \wref{exm:typeunarypath} and 
\wref{exm:typeMotzkin} we obtain the following corollary:

\begin{examplecorollary}\label{cor:typeentropy}
The hypersuccinct code $\mathsf{H}: \mathcal{T} \rightarrow \{0,1\}^\star$ optimally compresses
\begin{thmenumerate}{cor:typeentropy}
\item  \textbf{binary trees} $t$ of size $n$, drawn uniformly at random from the set of all binary trees of size $n$, using $|\mathsf{H}(t)| \leq 2n + O(n \log \log n/\log n)$ many bits,
\item \textbf{full binary trees} $t$ of size $n$, drawn uniformly at random from the set of all full binary trees of size $n$, using 
$|\mathsf{H}(t)| \leq n + O(n \log \log n/\log n)$ many bits,
\item \textbf{unary-path trees} $t$ of size $n+1$, drawn uniformly at random from the set of all unary-path trees of size $n+1$, using
$|\mathsf{H}(t)| \leq n + O(n \log \log n/\log n)$ many bits, and
\item \textbf{Motzkin trees} $t$ of size $n$, drawn uniformly at random from the set of all Motzkin trees of size $n$, using
$|\mathsf{H}(t)| \leq \lg(3)n + O(n \log \log n/\log n)$ many bits.
\end{thmenumerate}

\end{examplecorollary}

\begin{remark}[Shape entropy]\label{rem:shape-entropy}
Another notion of empirical entropy for unlabeled full binary trees was defined 
in~\cite{HuckeLohreySeelbachBenkner2019}: The authors define the $k$-history of a node $v$ of
a full binary tree $t$ as the string consisting of the last $k$ directions (left/right) on
the path from the root node of the tree to node $v$, and define the (normalized) $k$th
order empirical entropy $\mathcal{H}_k^s(t)/|t|$ of the full binary tree as the expected
uncertainty of the node types conditioned on the $k$-history of the node. In particular,
it is then shown in \cite{HuckeLohreySeelbachBenkner2019}, that
 the length of the binary encoding of full binary trees based on TSLPs 
 from~\cite{GanardiHuckeLohreySeelbachBenkner2019} can be upper-bounded in terms of this
empirical entropy plus lower-order terms.
 As this notion of empirical entropy $\mathcal{H}_k(t)$ for full binary trees is
conceptually quite similar to the empirical entropy of the node types $H_k^{\type}(t)$, the
main ideas of our proof that $|\mathsf{H}(t)| \leq H_k^{\type}(t) + o(n)$ can be
transferred to the setting from \cite{HuckeLohreySeelbachBenkner2019} in order to show
that $|\mathsf{H}(t)| \leq \mathcal{H}_k(t) + o(n)$ holds for full binary trees of
size $n$, as~well, if $k = o(\log n)$. For a formal definition and further details on shape entropy, see \wref{sec:label-shape}.
\end{remark}

\section{Fixed-Size and Fixed-Height Binary Tree Sources}
\label{sec:fixed-size-height}
A general concept to model probability distributions on various sets of binary trees was
introduced by Zhang, Yang, and Kieffer in \cite{ZhangYangKieffer2014} 
(see also~\cite{KiefferYangSzpankowski2009}), where the authors extend the classical notion of an
information source on finite sequences to so-called structured binary-tree sources, or
binary-tree sources for short: So-called leaf-centric binary-tree sources induce
probability distributions on the set of full binary trees with $n$ leaves and correspond
to fixed-size binary-tree sources which we will introduce below,
while so-called depth-centric binary tree souces induce probability distributions on the
set of full binary trees of height $h$ and correspond to fixed-height binary-tree sources,
also to be introduced below in this section.
For a formal introduction of structure sources and underlying concepts,
see~\cite{ZhangYangKieffer2014}.

\subsection{Fixed-Size Binary Tree Sources}\label{sec:fixedsize}
A fixed-size binary tree source $\mathcal{S}_{\mathit{fs}}(p)$ is defined by a function $p: \N_0^2 \to [0,1]$, such that
\begin{align*}
 \sum_{\ell=0}^{n} p(\ell,n-\ell) \rel= 1 \qquad \text{for all } n \in \N_0.
\end{align*}
A fixed-size tree source $\mathcal{S}_{\mathit{fs}}(p)$ induces a probability distribution over the set of all binary trees of size $n$ by
\begin{align}
		\Prob{t}
	\wwrel=
		\prod_{v\in t} p(|t_\ell[v]|,|t_r[v]|),
\label{eq:prob-Tn=t}
\end{align}
where the product ranges over all nodes $v$ of the binary tree $t$. If $t$ is the empty tree, we set $\mathbb{P}[t]=1$.
Intuitively, this corresponds to generating a binary tree by a (recursive) depth-first traversal as follows:
Given a target tree size $n$, ask the source for a left subtree size $\ell \in\{0,...,n-1\}$: The probability of a left subtree size $\ell$ is $p(\ell,n-1-\ell)$.
Create a node and recursively generate its left subtree of size $\ell$ and its
right subtree of size $n-1-\ell$. The random choices in the left and right subtree are
independent conditional on their sizes.
An inductive proof over $n$ verifies that $\sum_{t\in\mathcal T_n} \Prob{t} = 1$ for every $n \in \N_0$.

Note that the concept of fixed-size binary-tree sources is equivalent to the concept of leaf-centric binary-tree sources considered in \cite{GanardiHuckeLohreySeelbachBenkner2019, KiefferYangSzpankowski2009, SeelbachBenknerLohrey2018, ZhangYangKieffer2014} in the setting of full binary trees.

\begin{example}[Random binary search tree model]
\label{exm:bst}
	The (arguably) simplest example of a fixed-size tree source is 
	the \emph{random binary search tree (BST) model} $\mathcal{S}_{\mathit{fs}}(p_{\mathit{bst}})$.
	This corresponds to setting $p_{\mathit{bst}}(\ell,n-\ell) = \frac {1}{n+1}$ for all $\ell \in \{0, \dots, n\}$ and $n \in \mathbb{N}_0$.
	The very same distribution over binary trees arises for (unbalanced) 
	binary search trees (BSTs), when they are build by successive insertions 
	from a uniformly random permutation (``random BSTs''), and 
	also for the shape of Cartesian trees build from a uniformly random permutation
	(\aka random increasing binary trees~\cite[Ex.\,II.17\,\&\,Ex.\,III.33]{FlajoletSedgewick2009});
	see \wref{sec:random-rmq-random-bst}.
\end{example}

\needspace{8\baselineskip}
\begin{example}[Uniform model]
\label{exm:uniform}
	Perhaps the most elementary distribution on the set $\mathcal{T}_n$ 
	is the \emph{uniform probability distribution}, 
	\ie, $\mathbb{P}(t) = \frac{1}{|\mathcal{T}_n|}$ for every $t \in \mathcal{T}_n$. 
	This distribution corresponds to the fixed-size tree source 
	$\mathcal{S}_{\mathit{fs}}(p_{\mathit{uni}})$ defined by
	\begin{align*}
			p_{\mathit{uni}}(\ell,n-\ell) 
		\wwrel= 
			\frac{|\mathcal{T}_{\ell}||\mathcal{T}_{n-\ell}|}{|\mathcal{T}_{n+1}|} 
			\qquad \text{  for every } \ell \in \{0, \dots, n-1\} \text{ and } n \in \mathbb{N}_0.
	\end{align*}
\end{example}

\begin{example}[Binomial random tree model]\label{exm:dst}
	Fix a constant $0 < \alpha < 1$.
	The \emph{binomial random tree model} $\mathcal{S}_{\mathit{fs}}(p_{bin})$ is defined by
	\begin{align*}
			p_{bin}(\ell,n-\ell)
		\wwrel= 
			\alpha^{\ell}(1-\alpha)^{n-\ell}\binom{n}{\ell}
	\end{align*}
	for every $\ell \in \{0,\dots, n \}$ and $n \in \N_0$. 
	It is a slight variant of the digital search tree model, studied in~\cite{Martinez1992} 
	(see also~\cite{KiefferYangSzpankowski2009,ZhangYangKieffer2014,SeelbachBenknerLohrey2018}),
	and corresponds to (simple) \emph{tries} built from $n$ bitstrings 
	generated by a Bernoulli$(\alpha)$ (memoryless) source. 
\end{example}

\begin{example}[{Almost paths}]
\label{exm:unary-paths}
	Setting $p(0,n) = p(n,0) = \frac12$ for $n\ge 2$ yields a 
	fixed-size source which produces unary paths;
	(this is a special case of \cite[Ex.\,6]{ZhangYangKieffer2014}).
	One can generalize the example so that $p(\ell,r) > 0$ implies $\min\{\ell,r\} \le K$
	for some constant $K$ by setting
	\begin{align*}
			p_{path}(\ell, r)
		\wwrel=
			\begin{dcases} 
				\min\left\{\frac{1}{\ell+r+1}, \frac{1}{2(K+1)}\right\} \quad &\text{ if } \ell \leq K \text{ or } r \leq K,\\
				0 &\text{ otherwise.}
			\end{dcases}
	\end{align*}
A fixed-size source $\mathcal{S}_{\mathit{fs}}(p_{path})$ only generates binary trees for which at each node, the left or right subtree has at most $K$ nodes. Unary paths correspond to $K=0$.
	
\end{example}

\begin{example}[Random fringe-balanced BSTs]
\label{exm:fringe-balanced-bsts}	
	Let $t\in\N_0$ be a parameter, and define
\begin{align*}
p_{bal}(k,n-k-1) \wwrel= \begin{dcases} \binom{k}{t}\binom{n-k-1}{t}\Bigg/\binom{n}{2t+1}
\quad &\text{if } n \geq 2t+1, \\
\frac{1}{n} &\text{otherwise.}
\end{dcases}
\end{align*}
This is the shape of a random $(2t+1)$-fringe-balanced BST;
	(see \cite[\S4.3]{Wild2018} and the references therein for background on these trees).
\end{example}

\subsection{Fixed-Height Binary-Tree Sources}\label{sec:fixedheight}
A fixed-height binary tree source $\mathcal{S}_{\mathit{fh}}(p)$ is defined by a function $p: \N_0^2 \to [0,1]$, such that
\begin{align*}
 \sum_{\substack{i,j \in \mathbb{N}_0 \\\max(i,j)=h}} \mkern-20mu p(i,j) \rel= 1 
 	\qquad \text{for all } h \in \N_0.
\end{align*}
A fixed-height tree source $\mathcal{S}_{\mathit{fh}}(p)$ induces a probability distribution over the set of all binary trees of height $h$ by
\begin{align}
		\Prob{t}
	\wwrel=
		\prod_{v\in t} p(h(t_\ell[v]),h(t_r[v])),
\label{eq:prob-Tn=t2}
\end{align}
where the product ranges over all nodes of the binary tree $t$. If $t$ is the empty tree, we set $\mathbb{P}[t]=1$.
Intuitively, this corresponds to generating a binary tree by a (recursive) depth-first traversal as follows:
Given a target height  $h$ of the tree, ask the source for the height  $\ell$ of the left subtree and the height $r$ of the right subtree conditional on $\max(\ell,r) = h-1$.
The probability of a pair of heights $(\ell,r)$ with $\max(\ell,r)=h-1$ is $p(\ell,r)$.
Create a node and recursively generate its left subtree of height $\ell$ and its
right subtree of height $r$. The random choices in the left and right subtree are
independent conditional on their heights.
An inductive proof over $h$ verifies that $\sum_{t\in\mathcal T^h} \Prob{t} = 1$ for every $h \in \mathbb{N}_0$. 
Note that the concept of fixed-height binary-tree sources is equivalent to the concept of depth-centric binary-tree sources considered in \cite{GanardiHuckeLohreySeelbachBenkner2019, KiefferYangSzpankowski2009} in the setting of full binary trees.

\needspace{8\baselineskip}
\begin{example}[AVL trees by height]\label{exm:avl-uniform-height}
An AVL tree is a binary tree $t$, such that for every node $v$ of $t$, we have $|h(t_{\ell}[v])-h(t_r[v])| \leq 1$. Let $\mathcal{T}^h(\mathcal{A})$ denote the set of AVL trees of height $h$. The number of AVL trees of height $h$ satisfies the following recurrence relation:
\begin{align*}
		|\mathcal{T}^h(\mathcal{A})|
	\wwrel=
		2|\mathcal{T}^{h-1}(\mathcal{A})||\mathcal{T}^{h-2}(\mathcal{A})|
		\bin+|\mathcal{T}^{h-1}(\mathcal{A})||\mathcal{T}^{h-1}(\mathcal{A})|.
\end{align*}
Set
\begin{align*}
		p(j,k) 
	\wrel= 
		\begin{dcases} 
			\frac{|\mathcal{T}^j(\mathcal{A})||\mathcal{T}^{k}(\mathcal{A})|}{|\mathcal{T}^h(\mathcal{A})|} \quad &\text{for } (j,k) \in \{(h{-}2,h{-}1),(h{-}1,h{-}1),(h{-}1,h{-}2)\}\\
			0 & \text{otherwise,} 
		\end{dcases}
\end{align*}
for every $h \geq 2$. Then $\mathcal{S}_{\mathit{fh}}(p)$ corresponds to a uniform probability distribution on the set $\mathcal{T}^h(\mathcal{A})$ of AVL trees of \emph{height} $h$ for every $h \in \mathbb{N}$. 
\end{example}

\subsection{Entropy of Fixed-Size and Fixed-Height Sources}
\label{sec:lower-bound}

Given a fixed-size tree source $\mathcal{S}_{\mathit{fs}}(p)$ or fixed-height tree source $\mathcal{S}_{\mathit{fh}}(p)$, we write $H_n(\mathcal{S}_{\mathit{fs}}(p))$, respectively $H_h(\mathcal{S}_{\mathit{fh}}(p))$ for the entropy of the distribution it induces over the set of binary trees $\mathcal{T}_n$, respectively, $\mathcal{T}^h$:
If $\mathcal{S}_{\mathit{fs}}(p)$ is a fixed-size tree source, we have
\begin{align*}
		H_n(\mathcal{S}_{\mathit{fs}}(p))
	&\wwrel=
		\sum_{t\in \mathcal T_n} \Prob{t} \lg\left(\frac 1{\Prob{t}}\right)
\\	&\wwrel=
		\sum_{t\in \mathcal T_n} \biggl(\prod_{v\in t} p(|t_\ell[v]|, |t_r[v]|) \biggr)
			\cdot \sum_{v\in t} \lg\left(\frac1{p(|t_\ell[v]|, |t_r[v]|)}\right).
\end{align*}
Similarly, if $\mathcal{S}_{\mathit{fh}}(p)$ is a fixed-height tree source, we have
\begin{align*}
		H_h(\mathcal{S}_{\mathit{fh}}(p))
	&\wwrel=
		\sum_{t\in \mathcal T^h} \Prob{t} \lg\left(\frac 1{\Prob{t}}\right)
\\	&\wwrel=
		\sum_{t\in \mathcal T^h} \biggl(\prod_{v\in t} p(h(t_\ell[v]), h(t_r[v])) \biggr)
			\cdot \sum_{v\in t} \lg\left(\frac1{p(h(t_\ell[v]), h(t_r[v]))}\right).
\end{align*}
(Recall our convention $0 \lg (1/0) = 0$).
\begin{examplebox}
In \cite{KiefferYangSzpankowski2009}, the growth of $H_n$ was examined with respect to several types of fixed-size binary-tree sources, like the uniform model $\mathcal{S}_{\mathit{fs}}(p_{\mathit{uni}})$  from \wref{exm:uniform} and the binomial random tree model $\mathcal{S}_{\mathit{fs}}(p_{bin})$ from \wref{exm:dst}. In particular, for the random BST model $\mathcal{S}_{\mathit{fs}}(p_{\mathit{bst}})$ from \wref{exm:bst}, it was shown in \cite{KiefferYangSzpankowski2009} that 
\begin{align*}
		H_n(\mathcal{S}_{\mathit{fs}}(p_{\mathit{bst}})) 
	\wwrel\sim
		2n\sum_{i=2}^{\infty} \frac{\lg i}{(i+2)(i+1)}
	\wwrel\approx 
		1.7363771 n;
\end{align*}
see \wpref{sec:random-rmq-random-bst} for more discussion of this example.
\end{examplebox}

In the following, we present several properties of fixed-size and fixed-height binary-tree sources, for which we will be able to derive universal codes.

\subsection{Monotonic Tree Sources}

The first property was introduced in \cite{GanardiHuckeLohreySeelbachBenkner2019}, where it was shown that a certain binary encoding of binary trees based on tree straight-line programs yields universal codes with respect to fixed-size  and fixed-height sources satisfying this property:
\begin{definition}[Monotonic source]\label{def:monotonic}
	A fixed-size or fixed-height binary tree source is \textit{monotonic} if
	$p(\ell,r) \ge p(\ell+1,r)$ and $p(\ell,r) \ge p(\ell,r+1)$ for all $\ell,r \in \N_0$.
\end{definition}
Clearly, the binary search tree model $\mathcal{S}_{\mathit{fs}}(p_{\mathit{bst}})$ from \wref{exm:bst} is a monotonic fixed-size tree source, and one can easily show that the uniform model $\mathcal{S}_{\mathit{fs}}(p_{\mathit{uni}})$ from \wref{exm:uniform} is another one. 
Furthermore, the fixed-size source $\mathcal{S}_{\mathit{fs}}(p_{\mathit{path}})$ from \wref{exm:unary-paths} is monotonic.
In contrast, the binomial random tree model $\mathcal{S}_{\mathit{fs}}(p_{\mathit{bin}})$ from \wref{exm:dst} 
and the fringe-balanced BSTs (\wref{exm:fringe-balanced-bsts})
are not monotonic.

For monotonic tree sources, we find the following:
\begin{lemma}[Monotonicity implies submultiplicativity]
	\label{lem:bi-monotonic}
	Let $t \in \mathcal{T}$, and let $\mu_1, \dots, \mu_m$ be a partition of $t$ into disjoint subtrees, in the sense that every node of $t$ belongs to exactly one subtree $\mu_i$. If $p$ corresponds to a monotonic fixed-size or monotonic fixed-height tree source, then
	\begin{align*}
		\mathbb{P}[t] 
	\wwrel\leq 
		\prod_{i=1}^m \mathbb{P}[\mu_i]
	\end{align*}
\end{lemma}
\begin{proof}
Let $v$ be a node of $t$ and let $\mu_i$ denote the subtree that $v$ belongs to.
As $\mu_i$ is a subtree of $t$, we find $|{\mu_i}_{\ell}[v]| \leq |t_{\ell}[v]|$, $|{\mu_i}_{r}[v]| \leq |t_{r}[v]|$, $h({\mu_i}_{\ell}[v]) \leq h(t_{\ell}[v])$ and $h({\mu_i}_{r}[v]) \leq h(t_{r}[v])$. From the definition of monotonicity, we thus have
$p(|{\mu_i}_{\ell}[v]|, |{\mu_i}_{r}[v]|) \geq p(|t_{\ell}[v]|,|t_{r}[v]|)$, if $p$ corresponds to a fixed-size source, respectively, $p(h({\mu_i}_{\ell}[v]), h({\mu_i}_{r}[v])) \geq p(h(t_{\ell}[v]),h(t_{r}[v]))$, if $p$ corresponds to a fixed-height source. As every node of $t$ belongs to exactly one subtree $\mu_i$, we find for monotonic fixed-size sources $p$:
\begin{align*}
	\mathbb{P}[t] 
	\wwrel= 
	\prod_{v \in t} p\bigl(|t_{\ell}[v]|,|t_{r}[v]|\bigr) 
	\wwrel\leq 
	\prod_{i=1}^m \prod_{v \in \mu_i}p\bigl(|{\mu_i}_{\ell}[v]|, |{\mu_i}_{r}[v]|\bigr) 
	\wwrel= 
	\prod_{i=1}^m \mathbb{P}[\mu_i].
\end{align*}
For monotonic fixed-height sources, we similarly find
\begin{align*}
	\mathbb{P}[t] 
	\wwrel= 
	\prod_{v \in t} p\bigl(h(t_{\ell}[v]),h(t_{r}[v])\bigr) 
	\wwrel\leq 
	\prod_{i=1}^m \prod_{v \in \mu_i} p\bigl(h({\mu_i}_{\ell}[v]), h({\mu_i}_{r}[v])\bigr) 
	\wwrel= 
	\prod_{i=1}^m \mathbb{P}[\mu_i].
\end{align*}
\end{proof}
\wref{lem:bi-monotonic} depicts the crucial property of monotonic 
sources, based on which we will be able prove universality of our hypersuccinct encoding from \wref{sec:hypersuccinct-code}.

\subsection{Fringe-Dominated Tree Sources}
\label{sec:fringe-dom}

A second class of tree sources, for which we will be able to show universality of our encoding, is the following: Let $n_b(t)$ be the number of nodes $v$ in $t$ with $|t[v]| = b$
and let $n_{\ge b}(t)$ likewise be the number of nodes $v$ in $t$ with $|t[v]| \ge b$. 
\begin{definition}[Average-case fringe-dominated]
\label{def:avfringe-dominated}
	We call a fixed-size binary tree source \thmemph{average-case $B$-fringe dominated} for a function $B$ with $B(n)=\Theta(\log (n))$, if 
	\begin{align*}
		\sum_{t \in \mathcal{T}_n}\mathbb{P}[t]n_{\geq B(n)}(t)
	\wwrel=
		o\left(\frac{n}{\log (B(n))}\right).
	\end{align*}
\end{definition}

\begin{definition}[Worst-case fringe-dominated]
\label{def:wfringe-dominated}	
	We call a fixed-size or fixed-height binary tree source \thmemph{worst-case $B$-fringe dominated} for a function $B$  with $B(n)=\Theta(\log (n))$, if 
	\[
			n_{\geq B(n)}(t) 
		\wwrel= 
			o(n/ (\log B(n)))
	\]
	for every tree $t \in \mathcal{T}_n$ with $\mathbb{P}[t]>0$.
\end{definition}
Note that \wref{def:wfringe-dominated} treats fixed-size and fixed-height binary tree sources,
but \wref{def:avfringe-dominated} only covers fixed-size binary tree sources 
(to avoid averaging over trees of different sizes).
Moreover, a fixed-size tree source that is worst-case $B$-fringe-dominated is clearly average-case $B$-fringe-dominated as well.

Sufficient conditions for fixed-size sources to be average-case fringe-dominated are
given in~\cite{SeelbachBenknerLohrey2018} in the context
of DAG-compression of trees. 
The classes for which our hypersuccinct code from \wref{sec:hypersuccinct-code} is universal
happen to be exactly the classes for which the DAG-based
compression provably yields best possible compression:
\begin{definition}[$\psi$-nondegenerate \cite{SeelbachBenknerLohrey2018}]\label{def:psi-nondegenerate}
Let $\psi: \mathbb{R} \to (0,1]$ denote a monotonically decreasing function. A fixed-size tree source $\mathcal{S}_{\mathit{fs}}(p)$ is called \emph{$\psi$-nondegenerate}, if $p(\ell,n-\ell) \leq \psi(n)$ for every $\ell \in \{0, \dots, n\}$ and sufficiently large $n$.
\end{definition}
\begin{definition}[$\varphi$-weakly-weight-balanced \cite{SeelbachBenknerLohrey2018}]\label{def:phi-weakly-weight}
Let $\varphi: \mathbb{R} \to (0,1]$ denote a monotonically decreasing function and let $c\geq 3$ denote a constant. A fixed-size tree source $\mathcal{S}_{\mathit{fs}}(p)$ is called \emph{$\varphi$-weakly-weight-balanced}, if
\begin{align*}
		\sum_{\frac{n}{c}\leq \ell \leq n-\frac{n}{c}}p(\ell-1,n-\ell-1) \wwrel\geq \varphi(n)
		\end{align*}
		for every $n \in \mathbb{N}$.
\end{definition}
The following two lemmas follow from results shown in \cite{SeelbachBenknerLohrey2018} (note that in \cite{SeelbachBenknerLohrey2018}, the  authors consider full binary trees with $n$ leaves, whereas we consider (not necessarily full) binary trees with $n$ nodes, so there is an off-by-one in the definition of the tree size $n$):

\begin{lemma}[{{$\psi$-nondegeneracy implies fringe dominance, \cite[Lemma 4]{SeelbachBenknerLohrey2018}}}]
\label{lem:leaf-centric-average-nondegenerate}
	Let $\mathcal{S}_{\mathit{fs}}(p)$ be a $\psi$-nondegenerate fixed-size tree source, then
		\begin{align*}
				\sum_{t \in \mathcal{T}_n}\mathbb{P}[t]\cdot n_{\geq B(n)}(t) 
			\wwrel\leq 
				O(n \psi(B(n))),
		\end{align*}
		for every function $B$ with $B(n) = \Theta(\log n)$.
\end{lemma}

\begin{lemma}[{{$\varphi$-balance implies fringe dominance, \cite[Lemma 14]{SeelbachBenknerLohrey2018}}}]
	\label{lem:leaf-centric-average-balanced}
	Let $\mathcal{S}_{\mathit{fs}}(p)$ be a $\varphi$-weakly-weight-balanced fixed-size tree source, then
		\[
			\sum_{t \in \mathcal{T}_n}\mathbb{P}[t]\cdot n_{\geq B(n)}(t) 
		\wwrel\leq 
			O\left(\frac{cn}{\varphi(n) B(n)} \right),
		\]
	for every function $B$ with $B(n) = \Theta(\log n)$.	
\end{lemma} 
Thus, if a fixed-size tree source $\mathcal{S}_{\mathit{fs}}(p)$ is $\psi$-nondegenerate for a function $\psi$ with $\psi(n) \in o(1/\log(n))$, or $\varphi$-weakly-weight-balanced for a function $\varphi$ with $\varphi(n) \in \omega(\log\log n/\log n)$ (under the assumption that $B = \Theta(\log n)$),
then it is average-case fringe dominated. 
For the binary search tree model $\mathcal{S}_{\mathit{fs}}(p_{\mathit{bst}})$ (\wref{exm:bst}), \wref{lem:leaf-centric-average-nondegenerate} and \wref{lem:leaf-centric-average-balanced} both yield $\sum_{t \in \mathcal{T}_n}\mathbb{P}[t]n_{\geq B(n)}(t) \in O(n/B(n))$, by choosing $\psi(n) \in \Theta(1/n)$ and $\varphi(n)\in \Theta(1)$. 
Moreover, for the binomial random tree model $\mathcal{S}_{\mathit{fs}}(p_{bin})$ from \wref{exm:dst}, we find $\sum_{t \in \mathcal{T}_n}\mathbb{P}[t]n_{\geq B(n)}(t) \in O(n/B(n))$ from \wref{lem:leaf-centric-average-balanced} 
(see also \cite[Ex.\,16]{SeelbachBenknerLohrey2018}). Additionally, for random fringe-balanced BSTs from \wref{exm:fringe-balanced-bsts}, it is easy to show that $\sum_{t \in \mathcal{T}_n}\mathbb{P}[t]n_{\geq B(n)}(t)\in O(n/B(n))$ by choosing $\psi(n) = \Theta(1/n)$ in \wref{lem:leaf-centric-average-nondegenerate} (see also \cite[{\href{https://www.wild-inter.net/publications/html/wild-2016.pdf.html\#pf66}{Lemma~2.38}}]{Wild2016}).

Intuitively, $\varphi$-weakly-weight-balanced fixed-size tree sources lower-bound the probability of balanced binary trees in terms of the function $\varphi$. 
They generalize a class of tree sources considered in \cite[Lemma 4 and Theorem 2]{GanardiHuckeLohreySeelbachBenkner2019}, as well as so-called leaf-balanced (called weight-balanced below) tree sources introduced in \cite{ZhangYangKieffer2014} and further analyzed in \cite{GanardiHuckeLohreySeelbachBenkner2019}: 
\begin{definition}[Weight-balanced]
\label{def:weight-balanced}
	A weight-balanced tree source is a $\varphi$-weakly-weight-balanced tree source with $\varphi = 1$, that is, there is a constant $c\geq 3$, such that
	\begin{align*}
		\sum_{\frac{n}{c}\leq \ell \leq n-\frac{n}{c}}p(\ell-1,n-\ell-1) \wwrel= 1
	\end{align*}
	for every $n \in \mathbb{N}$.
\end{definition}
Weight-balanced tree sources constitute an example of fixed-size tree sources which are worst-case fringe-dominated:
\begin{lemma}[Weight-balance implies fringe dominance]
\label{lem:weight-balanced}
	Let $\mathcal{S}_{\mathit{fs}}(p)$ be a weight-balanced fixed-size tree source. Then 
	\[n_{\geq B(n)}(t) \wwrel= O\left(\frac{n}{B(n)}\right)\] 
	for every tree $t \in \mathcal{T}_n$ with $\mathbb{P}[t]>0$ and function $B$, \ie, $\mathcal{S}_{\mathit{fs}}(p)$ is worst-case $B$-fringe dominated.
\end{lemma}
\begin{proof} \wref{lem:weight-balanced} follows from results shown in \cite{GanardiHuckeJezLohreyNoeth2017} (see also \cite[Lemma 3]{GanardiHuckeLohreySeelbachBenkner2019}): Let $0 < \beta \leq 1$. In \cite{GanardiHuckeJezLohreyNoeth2017}, the authors introduce so-called $\beta$-balanced binary trees: A node $v$ of a binary tree $t$ is called $\beta$-balanced, if $|t_{\ell}[v]|+1 \geq \beta (|t_{r}[v]|+1)$ and $|t_{r}[v]|+1\geq \beta (|t_{\ell}[v]|+1)$ (note that in \cite{GanardiHuckeJezLohreyNoeth2017}, the authors count leaves of full binary trees, such that there is an off-by-one in the definition of $\beta$-balanced nodes).
A binary tree is called $\beta$-balanced, if for all internal nodes $u,v$ of $t$ such that $u$ is the parent node of $v$, we have that $u$ is $\beta$-balanced or $v$ is $\beta$-balanced. In the proof of \cite[Lemma 10]{GanardiHuckeJezLohreyNoeth2017}, it is shown in the context of DAG-compression of trees that for every $\beta$-balanced tree $t \in \mathcal{T}_n$, we have $n_{\geq b}(t)\leq 4\alpha n/b$ for every constant $b \in \mathbb{N}$, where $\alpha=1+\log_{1+\beta}(\beta^{-1})$. 
Now let $\mathcal{S}_{\mathit{fs}}(p)$ be a weight-balanced fixed-size tree source and let $t$ be a binary tree with $\mathbb{P}[t]>0$. It remains to show that $t$ is $\beta$-balanced for some constant $\beta$:
Let $v$ be a node of $t$. As $\mathbb{P}[t]>0$, we find that $p(|t_{\ell}[v]|, |t_r[v]|)>0$, and thus, there is a constant $c$, such that $n/c \leq |t_{\ell}[v]|+1, |t_{r}[v]|+1 \leq n-n/c$: In particular, we find that $|t_{\ell}[v]|+1 \geq (|t_r[v]|+1)/c$ and $|t_{r}[v]|+1 \geq (|t_{\ell}[v]|+1)/c$. Thus, $t$ is $\beta$-balanced with $\beta = 1/c$. 
\end{proof}

Finally, we will present a class of fixed-height binary tree sources 
that generalizes AVL-trees and is worst-case $B$-fringe dominated
(and thus amenable to compression using our techniques).

\begin{definition}[$\delta$-height-balanced]\label{def:height-balanced}
A fixed-height tree source $\mathcal{S}_{\mathit{fh}}(p)$ is called \emph{$\delta$-height-balanced}, if there is a monotonically increasing function $\delta: \N \to \N_0$, such that for all $(i,j) \in \N_0 \times \N_0 $ with $p(i,j) > 0$ and $\max(i,j) = k-1$ we have $|i-j| \leq \delta(k)$.
\end{definition}
For $\delta$-height-balanced tree sources, we find the following:
\begin{lemma}[Height balance implies fringe dominance]
\label{lem:height-balanced}
Let $\mathcal{S}_{\mathit{fh}}(p)$ be a $\delta$-height-balanced fixed-height tree source, then
	\begin{align*}
			n_{\geq B(n)}(t) 
		\wwrel\leq 
			O\left(\frac{\delta(n)n\log B(n)}{B(n)}\right)
	\end{align*}
	for every tree $t \in \mathcal{T}_n$ with $\mathbb{P}[t]>0$ and function $B$. 
\end{lemma}
In particular, under the assumption that $B(n) = \Theta(\log n)$,  $\mathcal{S}_{\mathit{fh}}(p)$ is worst-case fringe-dominated if $\delta(k) \in o(\log k /(\log\log k)^2)$.
The class of $\delta$-height-balanced fixed-height tree sources generalizes so-called depth-balanced tree sources introduced in \cite{GanardiHuckeLohreySeelbachBenkner2019}. The fixed-height binary tree source from \wref{exm:avl-uniform-height} is an example of a $1$-height-balanced fixed-height tree source.
\wref{lem:height-balanced} follows from combining, respectively, generalizing known results from \cite[Lemma 7]{GanardiHuckeLohreySeelbachBenkner2019} and \cite[Lemma 2]{HubschleSchneiderRaman15}, the latter presented in the context of top-tree compression; in the following, we give a self-contained proof in our notation: We start with showing the following lemma based upon \cite[Lemma 2]{HubschleSchneiderRaman15}, 
which is wider interest for establishing fringe dominance.
\begin{lemma}[Log-height implies fringe dominance]
\label{lem:hubschleschneider}
	Let $t$ be a binary tree and let $b\in \N$. 
	If there is a constant $c>1$, such that 
	$h(t[v]) \le \log_c (|t[v]|+1) = \frac1{\lg(c)}\cdot \lg(|t[v]|+1)$
	for every node $v$ of $t$, 
	then the number $n_{\geq b}(t)$ of nodes $v$ with $|t[v]|\geq b$ in $t$ satisfies 
		\begin{align*}
		n_{\geq b}(t) \wwrel\leq \frac{4|t|(\lg b+2)}{b\lg c}+\frac{2|t|}{b}.
		\end{align*}
\end{lemma}
\begin{proof}
We call a node $v$ of $t$ \emph{heavy}, if $|t[v]| \geq b$, otherwise, we call the node $v$  \emph{light}. Furthermore, we call the empty binary tree \emph{light}.
Thus, our goal is to upper-bound the number of heavy nodes in $t$. The total number of heavy nodes consists of
\begin{itemize}
\item[(i)] the number of heavy nodes with only light children plus
\item[(ii)] the number of heavy nodes with one heavy child and one light child (which might be the empty tree), plus
\item[(iii)] the number of heavy nodes with two heavy children.
\end{itemize}
We start with upper-bounding the number (i) of heavy nodes with only light children: These nodes are not in an ancestor-descendant relationship with each other, and as they are heavy, the subtrees rooted in those nodes are of size at least $b$: Thus, there are at most  $|t|/b$ many of those nodes. 

In order to upper-bound number (ii) of heavy nodes with one heavy child and one light child, we adapt the following definition from \cite{HubschleSchneiderRaman15}: We say that a node $v$ is in \emph{class} $i$ for an integer $i \in \mathbb{N}_0$, if $|t[v]| \in [2^{i}, 2^{i+1}-1]$.
Moreover, we call a node a \emph{top-class} $i$ node, if its parent belongs to class $j > i$ and we say that a node is a \emph{bottom-class} $i$ node, if its children both belong to classes $i_1, i_2 < i$. 

We find that if a node is heavy, then it is in class $i$ for an integer $\lfloor \lg b\rfloor \leq i \leq \lfloor \lg |t| \rfloor$. Moreover, if a node $v$ is in class $i$, then at most one of its children $u,w$ is in class $i$ as well: If both nodes $u,w$ belonged to class $i$, then $|t[v]| = 1+|t[u]|+|t[w]| \geq 1+2^{i}+2^{i}>2^{i+1}$, 
a contradiction to the fact that $v$ belongs to class $i$.

Let $v$ be a top-class $i$ node. 
By the above considerations, there is exactly one path of class $i$ nodes in $t[v]$, which leads from $v$ to a bottom-class $i$ node $w$, and there are no other class $i$ nodes in $t[v]$. We upper-bound the length of this path from node $v$ to node $w$ as follows: By assumption, we find that $h(t[v])\leq \lg(|t[v]|+1)(\lg c)^{-1}\leq \lg(2^{i+1})(\lg c)^{-1}=(i+1)(\lg c)^{-1} $. Thus, $h(t[v]) - h(t[w]) \leq (i+1)(\lg c)^{-1}$. Hence, $t[v]$ contains at most $(i+1)(\lg c)^{-1}$ many class $i$ nodes and in particular, $t[v]$ contains at most $(i+1)(\lg c)^{-1}$ many class $i$ heavy nodes with one heavy child and one light child.

As top-class $i$ nodes are not in an ancestor-descendant relationship with each other, there are at most $|t|/2^{i}$ many top-class $i$ nodes in $t$. Thus, there are at most $|t|/2^{i} \cdot (i+1)(\lg c)^{-1}$ class $i$ heavy nodes with one heavy child and one light child, respectively, only one heavy child, in $t$.  
Altogether, there are at most
\begin{align*}
		\sum_{i=\lfloor \lg b\rfloor}^{\lfloor \lg |t| \rfloor}\frac{|t|(i+1)}{2^{i}(\lg c)}
	\wwrel\leq
		\frac{4|t|(\lg b+2)}{b\lg c}
\end{align*}
many heavy nodes with one heavy child and one light child in $t$.

It remains to upper-bound number (iii) of heavy nodes with two heavy children: For this, note that all heavy nodes of $t$ form a (non-fringe) subtree $t'$ of $t$ rooted in the root of $t$. Heavy nodes of type (i), \ie, heavy nodes with only light children, are the leaves of this subtree $t'$, while nodes of type (ii) are unary nodes in $t'$ and heavy nodes of type (iii) are binary nodes in $t'$. Thus, the number (iii) of heavy nodes with two heavy children is upper-bounded by the number (i), which is upper-bounded by $|t|/b$. This finishes the proof.
\end{proof}
 
 With \wref{lem:hubschleschneider}, we are able to prove \wref{lem:height-balanced}:
 
 \begin{proof}[ \wref{lem:height-balanced}]
 Let $\beta \in \N$. We call a binary tree $t$ \emph{$\beta$-height-balanced}, if for every node $v$ of $t$, we have $|h(t_\ell[v])- h(t_r[v])| \leq \beta$. This property of trees was called $\beta$-depth-balanced trees in~\cite{GanardiHuckeLohreySeelbachBenkner2019}.
Note that every subtree of a $\beta$-height-balanced tree is $\beta$-height-balanced as well. In \cite[Lemma 7]{GanardiHuckeLohreySeelbachBenkner2019}, it is shown that for every $\beta$-height-balanced tree $t$, we have $|t|+1\geq c^{h(t)}$ with $c = 1+1/(1+\beta)$ (note that in \cite{GanardiHuckeLohreySeelbachBenkner2019}, the authors consider full binary trees and measure size as the number of leaves, such that there is an off-by-one in the meaning of $|t|$). Thus, \wref{lem:hubschleschneider} applies to $\beta$-height-balanced trees.

Now let $\mathcal{S}_{\mathit{fh}}(p)$ be a fixed-height tree source, and let $\delta: \mathbb{N} \to \mathbb{N}_0$ be a monotonically increasing function, such that for all $(i,j) \in \N_0 \times \N_0$ with $p(i,j) >0$ and $\max(i,j) = k-1$, we have $|i-j| \leq \delta(k)$. Moreover, let $t \in \mathcal{T}_n$ be a binary tree of size $n$ with $\mathbb{P}[t]>0$. Then $|h(t_{\ell}[v])-h(t_{r}[v])|\leq \delta(h(t[v]))$ for every node $v$ of $t$. In particular, as $\delta$ is monotonically increasing, we find that $t$ is $\beta$-height balanced with $\beta = \delta(h(t))$ and as $h(t) \leq |t|=n$, $t$ is $\delta(n)$-height-balanced. By \wref{lem:hubschleschneider}, we thus find that 
\begin{align*}
	n_{\geq B(n)}(t) \wwrel\leq \frac{4n(\lg B(n)+2)}{B(n) \lg c} + \frac{2n}{B(n)},
\end{align*}
with $c = 1+1/(1+\delta(n))$. By the mean-value theorem, we find
\begin{align*}
		\lg\left(1+\frac{1}{1+\delta(n)}\right) 
	&\wwrel= 
		\lg\left(\frac{2+\delta(n)}{1+\delta(n)}\right) 
	\wwrel= 
		\lg (2+\delta(n)) - \lg (1+\delta(n))
\\	&\wwrel\geq 
		\frac{1}{(2+\delta(n))\ln(2)}.
\end{align*}
Thus
\begin{align*}
	n_{\geq B(n)}(t) 
	\wwrel\leq 
	\frac{4\ln(2)(2+\delta(n))n(\lg B(n)+2)}{B(n)} + \frac{2n}{B(n)} 
	\wwrel= 
	O\left(\frac{\delta(n)n\log B(n)}{B(n)}\right).
\end{align*}
This proves the lemma.
\end{proof}
 
\subsection{Universality of Fixed-Size and Fixed-Height Sources}\label{sec:universality-fixed-size-height}
In order to show universality of the hypersuccinct code from \wref{sec:hypersuccinct-code} with respect to fixed-size and fixed-height sources, we proceed in a similar way as in the case of memoryless and higher-order sources: An overview of the strategy is given in \wref{sec:bsp-random-bst}. 
First, we derive a source-specific encoding (a so-called depth-first order arithmetic code) with respect to the fixed-size or fixed-height source, against which we will then compare our hypersuccinct code:

The formulas for $\Prob{t}$, \weqref{eq:prob-Tn=t} and \weqref{eq:prob-Tn=t2},
immediately suggest a route for an (essentially) optimal \emph{source-specific} 
encoding of any binary tree $t$ with $\Prob{t}>0$ that, given a fixed-size or fixed-height source $p$, spends $\lg(1/\mathbb{P}[t])$ (plus lower-order terms) many bits in order to encode a binary tree $t \in \mathcal{T}$ with $\mathbb{P}[t]>0$: For a given fixed-size source, such an encoding may spend $-\lg \bigl( p(|t_\ell[v]|, |t_r[v]|) \bigr)$ many bits per node $v$, while for a fixed-height source, it may spend 
$-\lg \bigl( p(h(t_\ell[v]),h(t_r[v])) \bigr)$ many bits per node $v$. (Note that as $ \Prob{t}>0$ by assumption, we have $p(|t_\ell[v]|, |t_r[v]|)>0$, respectively, $p(h(t_\ell[v]),h(t_r[v]))>0$ for every node $v$ of $t$.) 
Assuming that we ``know'' $p$~-- 
\ie, assuming it is ``hard-wired'' into the code and need not be stored as part of the encoding~--
and assuming that we have already stored $|t[v]|$, if $p$ corresponds to a fixed-size source, respectively, $h(t[v])$, if $p$ corresponds to a fixed-height source, we can use \emph{arithmetic coding}~\cite{WittenNealCleary1987} 
to store $|t_\ell[v]|$ (from which we will then be able to determine $|t_r[v]|$), if $p$ corresponds to a fixed-size source, respectively, $h(t_{\ell}[v])$ and $h(t_r[v])$, if $p$ corresponds to a fixed-height source.

First, let us assume that $p$ corresponds to a fixed-size binary tree source. 
A simple (source-dependent) encoding $D_p$ 
thus stores a tree $t \in \mathcal T_n$ as follows: We initially encode the size of the tree in Elias gamma code: If the tree consists of $n$ nodes, we store the Elias gamma code of $n+1$, $\gamma(n+1)$, in order to take the case into account that $t$ is the empty binary tree.
Additionally, while traversing the tree in depth-first order, we encode $|t_\ell[v]|$ for each node $v$, using arithmetic coding: To encode $|t_\ell[v]|$, we feed the arithmetic coder with the model that the next symbol
is a number $\ell\in \{0, \dots, |t[v]|-1\}$ with respective probabilities
$p(\ell, |t[v]|-1-\ell)$.

If $p$ corresponds to a fixed-height binary tree source, we proceed similarly: A (source-dependent) encoding $D_p$ 
with respect to a fixed-height source $\mathcal{S}_{\mathit{fh}}(p)$ stores a tree $t \in \mathcal{T}^h$ by initially encoding $h+1$, \ie, the height of the tree plus one, in Elias gamma code, $\gamma(h+1)$, followed by an encoding of $(h(t_\ell[v]), h(t_r[v]))$ for every node $v$ in depth-first order, stored using arithmetic encoding: Note that there are $2h(t[v])-1$ many different possibilities for $(h(t_\ell[v]), h(t_r[v]))$, thus, we can represent a pair $(h(t_\ell[v]), h(t_r[v]))$ by a number $i \in \{0, 2h(t[v])-2\}$, (\eg, by letting $i$ represent the pair $(i,h(t[v])-1)$ if $i \leq h(t[v])-1$ and $(h(t[v])-1,2h(t[v])-2-i)$, otherwise). 
To encode $(h(t_\ell[v]), h(t_r[v]))$, we feed the arithmetic coder with the model that the next symbol is a number $i \in \{0, 2h(t[v])-2\}$ with respective probabilities $p(i,h(t[v])-1)$, if $i \leq h(t[v])-1$, and $p(h(t[v])-1,2h(t[v])-2-i)$, otherwise.

We refer to this (source-dependent) code $D_p$ as the \emph{depth-first arithmetic code} for the binary tree source with probabilities $p$.
We can reconstruct the tree $t$ recursively from its code $D_p(t)$:
Since we always know the subtree size, respectively, subtree height, we know how many and what size the bins 
for the next left subtree size, respectively, pair of subtree heights, uses in the arithmetic code. Finally, if a subtree size or height is $1$ or $0$, we know the subtree itself.
Recalling that arithmetic coding compresses to the entropy of the given input
plus at most 2 bits of overhead, we need at most $\lg(1/\mathbb{P}[t]) + 2$
bits to store $t$ when we know $|t|$, respectively $h(t)$ (depending on the type of tree source). With $h(t)\leq |t|$, and as the Elias-gamma code satisfies $|\gamma(n)| \leq 2\lfloor \lg(n)\rfloor +1$,
we find that the total encoding length is upper-bounded by
\begin{align*}
		|D_p(t)|
	&\wwrel\le
		\lg(1/\mathbb{P}[t])
		\bin+
		2\lfloor \lg(|t|+1)\rfloor \bin+ 3.
\numberthis\label{eq:depth-first-arithmetic-code-length}
\end{align*}
If $p$ corresponds to a fixed-size tree source, taking expectations over the tree $t $ to encode, depth-first arithmetic coding 
thus stores a binary tree with $n$ nodes using 
$H_n(\mathcal{S}_{\mathit{fs}}(p)) + O(\log n)$ bits on average.

\begin{figure}[th]

	\plaincenter{\resizebox{.75\textwidth}!{
	\begin{tikzpicture}[
			scale=.7,
			tree node/.style         = {circle, draw, fill=black!20, font=\small, minimum size=18pt, inner sep=0pt},
			preorder label/.style    = {font=\scriptsize},
			subtreesize label/.style = {font=\small, text=blue},
			tree edge/.style         = {thin},
			dfs block/.style         = {draw=blue!30,line width=10pt},
			dfs block span/.style    = {dfs block,line width=3pt,|-|,shorten <=-1pt,shorten >=-1pt},
	]
	
		\def\hobbysize{20pt}
		\node[tree node] (n1) at (1,3.409091) {$1$} ;
		\node[preorder label,above=2pt of n1] {$7$} ;
		\node[subtreesize label,below=2pt of n1] {$0/1$} ;
		\node[tree node] (n2) at (2,4.909091) {$2$} ;
		\node[preorder label,above=2pt of n2] {$6$} ;
		\node[subtreesize label,below=2pt of n2] {$1/3$} ;
		\node[tree node] (n3) at (3,3.409091) {$3$} ;
		\node[preorder label,above=2pt of n3] {$8$} ;
		\node[subtreesize label,below=2pt of n3] {$0/1$} ;
		\node[tree node] (n4) at (4,6.681818) {$4$} ;
		\node[preorder label,above=2pt of n4] {$5$} ;
		\node[subtreesize label,below=2pt of n4] {$3/4$} ;
		\node[tree node] (n5) at (5,8.727273) {$5$} ;
		\node[preorder label,above=2pt of n5] {$4$} ;
		\node[subtreesize label,below=2pt of n5] {$4/8$} ;
		\node[tree node] (n6) at (6,4.909091) {$6$} ;
		\node[preorder label,above=2pt of n6] {$10$} ;
		\node[subtreesize label,below=2pt of n6] {$0/1$} ;
		\node[tree node] (n7) at (7,6.681818) {$7$} ;
		\node[preorder label,above=2pt of n7] {$9$} ;
		\node[subtreesize label,below=2pt of n7] {$1/3$} ;
		\node[tree node] (n8) at (8,4.909091) {$8$} ;
		\node[preorder label,above=2pt of n8] {$11$} ;
		\node[subtreesize label,below=2pt of n8] {$0/1$} ;
		\node[tree node] (n9) at (9,11.045455) {$9$} ;
		\node[preorder label,above=2pt of n9] {$3$} ;
		\node[subtreesize label,below=2pt of n9] {$8/9$} ;
		\node[tree node] (n10) at (10,13.636364) {$10$} ;
		\node[preorder label,above=2pt of n10] {$2$} ;
		\node[subtreesize label,below=2pt of n10] {$9/18$} ;
		\node[tree node] (n11) at (11,3.409091) {$11$} ;
		\node[preorder label,above=2pt of n11] {$16$} ;
		\node[subtreesize label,below=2pt of n11] {$0/1$} ;
		\node[tree node] (n12) at (12,4.909091) {$12$} ;
		\node[preorder label,above=2pt of n12] {$15$} ;
		\node[subtreesize label,below=2pt of n12] {$1/2$} ;
		\node[tree node] (n13) at (13,6.681818) {$13$} ;
		\node[preorder label,above=2pt of n13] {$14$} ;
		\node[subtreesize label,below=2pt of n13] {$2/3$} ;
		\node[tree node] (n14) at (14,8.727273) {$14$} ;
		\node[preorder label,above=2pt of n14] {$13$} ;
		\node[subtreesize label,below=2pt of n14] {$3/5$} ;
		\node[tree node] (n15) at (15,6.681818) {$15$} ;
		\node[preorder label,above=2pt of n15] {$17$} ;
		\node[subtreesize label,below=2pt of n15] {$0/1$} ;
		\node[tree node] (n16) at (16,11.045455) {$16$} ;
		\node[preorder label,above=2pt of n16] {$12$} ;
		\node[subtreesize label,below=2pt of n16] {$5/8$} ;
		\node[tree node] (n17) at (17,6.681818) {$17$} ;
		\node[preorder label,above=2pt of n17] {$19$} ;
		\node[subtreesize label,below=2pt of n17] {$0/1$} ;
		\node[tree node] (n18) at (18,8.727273) {$18$} ;
		\node[preorder label,above=2pt of n18] {$18$} ;
		\node[subtreesize label,below=2pt of n18] {$1/2$} ;
		\node[tree node] (n19) at (19,16.500000) {$19$} ;
		\node[preorder label,above=2pt of n19] {$1$} ;
		\node[subtreesize label,below=2pt of n19] {$18/20$} ;
		\node[tree node] (n20) at (20,13.636364) {$20$} ;
		\node[preorder label,above=2pt of n20] {$20$} ;
		\node[subtreesize label,below=2pt of n20] {$0/1$} ;
	
		\draw[tree edge] (n19) to (n20) ;
		\draw[tree edge] (n19) to (n10) ;
		\draw[tree edge] (n10) to (n16) ;
		\draw[tree edge] (n10) to (n9) ;
		\draw[tree edge] (n9) to (n5) ;
		\draw[tree edge] (n5) to (n7) ;
		\draw[tree edge] (n5) to (n4) ;
		\draw[tree edge] (n4) to (n2) ;
		\draw[tree edge] (n2) to (n3) ;
		\draw[tree edge] (n2) to (n1) ;
		\draw[tree edge] (n7) to (n8) ;
		\draw[tree edge] (n7) to (n6) ;
		\draw[tree edge] (n16) to (n18) ;
		\draw[tree edge] (n16) to (n14) ;
		\draw[tree edge] (n14) to (n15) ;
		\draw[tree edge] (n14) to (n13) ;
		\draw[tree edge] (n13) to (n12) ;
		\draw[tree edge] (n12) to (n11) ;
		\draw[tree edge] (n18) to (n17) ;

	\end{tikzpicture}
	}}
	\caption{%
		Example of a binary tree $t$ with 20 nodes.
		Each node shows the inorder number (in the node), 
		its preorder index (above the node) and the sizes of left subtree and its total subtree 
		(blue, below the node).
		Assuming the random BST model,
		we have $\lg(1/\mathbb{P}[t]) \approx 28.74$, slightly below the expectation
		$H_{20}(\mathcal{S}_{\mathit{fs}}(p_{\mathit{bst}})) \approx 29.2209$.
		The arithmetic DFS code for the left tree sizes
		is \texttt{111011010111101011110101011111}, \ie, 30 bits.
		This compares very favorably to a balanced-parenthesis representation
		\texttt{((((((()()))(()())))((((()))())(())))())}
		which would use $40$ bits.
	}
	\label{fig:example-bst}
\end{figure}
 
\subsubsection{Universality for Monotonic Fixed-Size and Fixed-Height Sources}
In this subsection, we show universality of our hypersuccinct code from \wref{sec:hypersuccinct-code} with respect to \emph{monotonic} fixed-size and fixed-height sources, as defined in \wref{def:monotonic}. We start with the following lemma:

\begin{lemma}[Monotonic bounds micro-tree code]
\label{lem:estimate-bi-monotonic}
	Let $\mathcal{S}_{\mathit{fs}}(p)$, respectively, $\mathcal{S}_{\mathit{fh}}(p)$, be a fixed-size or fixed-height tree source and let $t \in \mathcal{T}_n$ with $\mathbb{P}[t]>0$. If $\mathcal{S}_{\mathit{fs}}(p)$, respectively, $\mathcal{S}_{\mathit{fh}}(p)$ is monotonic, then 
	\begin{align*}
		\sum_{i=1}^m  |C(\mu_i)| 
		\wwrel\leq 
		\lg\left(\frac{1}{\mathbb{P}[t]}\right) + O\left(\frac{n \log \log n}{\log n}\right),
	\end{align*}
	where $C$ is a Huffman code for the sequence of micro trees $\mu_1,\ldots,\mu_m$ obtained from our tree covering scheme (see \wref{sec:hypersuccinct-code}).
\end{lemma}
\begin{proof}
Let us denote by $D_p: \mathcal T \to \{0,1\}^\star$ the
depth-first arithmetic code as introduced in the beginning of \wref{sec:universality-fixed-size-height}. In particular, by \wref{lem:bi-monotonic}, we find that $\mathbb{P}[\mu_i] \geq \mathbb{P}[t]$ for all micro trees $\mu_i$ of $t$, and thus, $D_p(\mu_i)$ is well-defined for every micro tree $\mu_i$.
Restricting $D_p$ to $\Sigma_\mu$ yields a prefix-free code for $\Sigma_\mu$, so we know by the optimality of Huffman codes that
\begin{align*}
		\sum_{i=1}^m  |C(\mu_i)|
	\wwrel\leq 
		\sum_{i=1}^m  |D_p(\mu_i)|.
\end{align*}
By our estimate \eqref{eq:depth-first-arithmetic-code-length} for $|D_p|$, we find that
\begin{align*}
		\sum_{i=1}^m  |D_p(\mu_i)| 
	&\wwrel\leq 
		\sum_{i=1}^m \left(\lg\left(\frac{1}{\mathbb{P}[\mu_i]}\right)+3+2\lfloor \lg (|\mu_i|+1) \rfloor\right)
\\	&\wwrel\leq 
		\sum_{i=1}^m \lg\left(\frac{1}{\mathbb{P}[\mu_i]}\right) + O(m\log \mu).
\end{align*}
Note that the subtrees $\mu_1, \dots, \mu_m$ form a partition of $t$ in the sense that every node of $t$ belongs to exactly one subtree $\mu_i$: Thus, and as $p$ corresponds to a monotonic fixed-size or fixed-height source, we find by \wref{lem:bi-monotonic}:
\begin{align*}
		\sum_{i=1}^m \lg\left(\frac{1}{\mathbb{P}[\mu_i]}\right) + O(m\log \mu) 
	&\wwrel\leq 
		\lg\left(\frac{1}{\mathbb{P}[t]}\right) + O(m \log \mu).
\end{align*}
Altogether, with $m = \Theta(n/\log n) $ and $\mu = \Theta(\log n)$ (see \wref{sec:hypersuccinct-code}), we thus obtain
\begin{align*}
		\sum_{i=1}^m |C(\mu_i)| 
	&\wwrel\leq 
		\lg\left(\frac{1}{\mathbb{P}[t]}\right) 
		\bin+ O\left(\frac{n \log \log n}{\log n}\right).
\end{align*}
\end{proof}

From \wref{lem:estimate-bi-monotonic} and \wref{lem:hypersuccinct-code}, we obtain the following result for monotonic tree sources (defined in \wref{def:monotonic}):
\begin{theorem}[Universality for monotonic sources]
\label{thm:binary-monotonic-fixed-size}
Let $\mathcal{S}_{\mathit{fs}}(p)$, respectively, $\mathcal{S}_{\mathit{fh}}(p)$, be a monotonic fixed-size or fixed-height tree source.  Then the hypersuccinct code $\mathsf{H}: \mathcal{T} \rightarrow \{0,1\}^\star$ satisfies
\begin{align*}
|\mathsf{H}(t)|\wwrel\leq  \lg\left(\frac{1}{\mathbb{P}[t]}\right) + O\left(\frac{n\log\log n}{\log n}\right)
\end{align*}
for every $t \in \mathcal{T}_n$ with $\mathbb{P}[t]>0$.
\end{theorem}
The binary tree sources from \wref{exm:bst}, \wref{exm:uniform}, and \wref{exm:unary-paths} 
are monotonic fixed-size binary tree sources. 
Thus, together with \wref{thm:binary-monotonic-fixed-size}, we obtain the following corollary:

\needspace{15\baselineskip}
\begin{examplebox}
\begin{corollary}\label{cor:monotonic}
The hypersuccinct code $\mathsf{H}: \mathcal{T} \rightarrow \{0,1\}^\star$ satisfies the following:
\begin{itemize}
\item[(i)] A \textbf{(random) binary search tree} (BST) (see  \wref{exm:bst}) $t$ of size $n$ is encoded using $$|\mathsf{H}(t)| \wwrel\leq \lg(1/\mathbb{P}[t]) + O(n \log \log n/\log n)$$ many bits. In particular, we need on average
\begin{align*}
		\sum_{t \in \mathcal{T}_n}\mathbb{P}[t]|\mathsf{H}(t)| 
	&\wwrel\leq 
		H_n(\mathcal{S}_{\mathit{fs}}(p_{bst})) +  O(n \log \log n/\log n)
\\	&\wwrel\approx 
		1.736n + O(n \log \log n/\log n)
\end{align*}
many bits (see \cite{KiefferYangSzpankowski2009}) in order to encode a random BST of size $n$.
\item[(ii)] \textbf{Almost-path binary trees} (for arbitrary $K \geq 0$) from \wref{exm:unary-paths} are encoded using
$|\mathsf{H}(t)| \wwrel\leq \lg\left(\frac{1}{\mathbb{P}[t]}\right) + O(n \log \log n/\log n)$ many bits.
\end{itemize}
\end{corollary}
As the uniform probability distribution on the set $\mathcal{T}_n$ of binary trees of size $n$ can be modeled as a monotonic fixed-size binary tree source (see \wref{exm:uniform}), we find moreover that \wref{cor:typeentropy}, part (i) follows from \wref{thm:binary-monotonic-fixed-size}.
\end{examplebox}

\subsubsection{Universality for Fringe-Dominated Fixed-Size and Fixed-Height Sources}

Recall that our hypersuccinct code from Section \wref{sec:hypersuccinct-code} decomposes $t$ into micro trees
$\mu_1,\ldots,\mu_m$ using \wref{lem:tree-decomposition-binary}
and uses a Huffman code $C$ for $\mu_1,\ldots,\mu_m$.
Some of these micro trees might be \emph{``fringe''},
\ie, correspond to fringe subtrees of $t$ and leaves in the top tier tree $\Upsilon$,
but many will be \emph{internal} micro trees, 
\ie, have child micro trees in the top tier tree $\Upsilon$.
That means, micro-tree-local subtree sizes, resp. heights, and global subtree sizes, resp., heights, differ
for nodes that are ancestors of the portal to the child micro tree~-- 
and only for those nodes do they differ: This will be the crucial observation in order to show that our hypersuccinct code is universal with respect to fringe-dominated sources.

Formally, let $v$ be a node of $t$. If $v$ is contained in a fringe micro tree $\mu_i$, respectively, in a non-fringe micro tree $\mu_i$ but not an ancestor of a portal node, then ${\mu_i}[v] = t[v]$, and thus $p(|{\mu_i}_{\ell}[v]|, |{\mu_i}_{r}[v]|) = p(|t_\ell[v]|,|t_r[v]|)$, respectively, $p(h({\mu_i}_{\ell}[v]), h({\mu_i}_{r}[v])) = p(h(t_\ell[v]),h(t_r[v]))$. On the other hand, if $v$ is an ancestor of a portal node in a non-fringe subtree $\mu_i$, then $\mu_i[v] \neq t[v]$. In order to take this observation into consideration, we make the following definitions:
Let $\mu_i$ be an internal (non-fringe) micro tree.
By $\mathit{bough}(\mu_i)$, we denote the subtree of $\mu_i$ induced
by the set of nodes that are ancestors of $\mu_i$'s child micro trees 
(ancestors of the portals);
the boughs of a micro tree are the paths from the portals 
to the micro tree root. In particular, if $v$ denotes a node of $t$ contained in a subtree $\mu_i$, then $t[v] \neq \mu_i[v]$ if and only if $\mu_i$ is not fringe and $v$ is contained in $\mathit{bough}(\mu_i)$. Hanging off the boughs of $\mu_i$ are (fringe) subtrees 
$f_{i,1},\ldots,f_{i,|\mathit{bough}(\mu_i)|+1}$, 
listed in depth-first order of the bough nulls these subtrees are attached to. In particular, some of these subtrees might be the empty tree. Recall that the portal nodes themselves are not part of $\mu_i$ and hence not part of $\mathit{bough}(\mu_i)$. We now find the following:

\begin{lemma}[bough decomposition]
\label{lem:fringe-dominated-submultiplikativ}
Let $\mathcal{S}_{\mathit{fs}}(p)$, respectively, $\mathcal{S}_{\mathit{fh}}(p)$, be a fixed-size, respectively, fixed-height binary tree source. Furthermore, let 
$\mathcal{I}_0 =\{i \in [m] : \mu_i $ is a fringe micro tree in $ t\}$ and let $\mathcal{I}_1 = [m] \setminus \mathcal{I}_0$. Then
	\begin{align*}
			\sum_{i \in \mathcal{I}_0} \lg\left(\frac{1}{\mathbb{P}[\mu_i]}\right)
			\bin+\sum_{i \in \mathcal{I}_1}
				\sum_{j=1}^{|\mathit{bough}(\mu_i)|+1} 
					\lg\left(\frac{1}{\mathbb{P}[f_{i,j}]}\right)
		\wwrel\leq 
			\lg \left(\frac{1}{\mathbb{P}[t]}\right).
	\end{align*}
\end{lemma}
\begin{proof}
The statement follows immediately from the facts that (i) all the subtrees $\mu_i$ for $i \in \mathcal{I}_0$ and $f_{i,j}$ for $i \in \mathcal{I}_1$ and $j \in \{1, \dots, |\mathit{bough}(\mu_i)|+1\}$ are fringe subtrees of $t$, and (ii) every node $v$ of $t$ occurs in at most one of these fringe subtrees.
Assume that $p$ corresponds to a fixed-size tree source, then we find:
\begin{align*}
	&
		\sum_{i \in \mathcal{I}_0} \lg\left(\frac{1}{\mathbb{P}[\mu_i]}\right)
		+\sum_{i \in \mathcal{I}_1}\sum_{j=1}^{|\mathit{bough}(\mu_i)|+1} 		\lg\left(\frac{1}{\mathbb{P}[f_{i,j}]}\right) 
\\	&\wwrel= 
		-\sum_{i \in \mathcal{I}_0} \sum_{v \in \mu_i} \lg(p(|{\mu_i}_\ell[v]|, |{\mu_i}_r[v]|))
		-\sum_{i \in \mathcal{I}_1}\sum_{j=1}^{|\mathit{bough}(\mu_i)|+1}
				\sum_{v \in f_{i,j}}\lg(p(|{f_{i,j}}_\ell[v]|, |{f_{i,j}}_r[v]|))
\\	&\wwrel{\overset{(i)}=} 
		-\sum_{i \in \mathcal{I}_0} \sum_{v \in \mu_i} \lg(p(|{t}_\ell[v]|, |{t}_r[v]|))
		-\sum_{i \in \mathcal{I}_1}\sum_{j=1}^{|\mathit{bough}(\mu_i)|+1}
			\sum_{v \in f_{i,j}}\lg(p(|t_\ell[v]|, |t_r[v]|))
\\	&\wwrel{\overset{(ii)}\leq} 
		-\sum_{v \in t}\lg(p(|t_\ell[v]|, |t_r[v]|))
\\	&\wwrel=
		\lg\left(\frac{1}{\mathbb{P}[t]}\right).
\end{align*}
The proof for fixed-height sources is similar.
\end{proof}

We now find the following:
\begin{lemma}[Great-branching lemma]
\label{lem:great-branching}
Let $\mathcal{S}_{\mathit{fs}}(p)$, respectively, $\mathcal{S}_{\mathit{fh}}(p)$, be a fixed-size, respectively, fixed-height tree source and let $t \in \mathcal{T}_n$ with $\mathbb{P}[t]>0$. Then
\begin{align*}
\sum_{i=1}^m  |C(\mu_i)| \wwrel\leq \lg \left(\frac{1}{\mathbb{P}[t]}\right) + O\left(n_{\geq B}(t) \log B\right),
\end{align*}
where $C$ is a Huffman code for the sequence of micro trees $\mu_1,\ldots,\mu_m$  from our tree covering scheme and $B=B(n) \in \Theta(\log n)$ is the parameter of the tree covering scheme (see \wref{sec:hypersuccinct-code}).
\end{lemma}
\begin{proof}
We construct a new encoding for micro trees against which we can compare the hypersuccinct code, the ``great-branching'' code, $G_B$, as follows:
\begin{align*}
		G_B(\mu_i)
	&\wwrel=
		\begin{dcases*}
		\texttt{0}\cdot D_p(\mu_i), & if $\mu_i$ is a fringe micro tree; \\[1ex]
		\begin{aligned}[c]
			\texttt{1}\cdot 
			\gamma(|\mathit{bough}(\mu_i)|) \cdot 
			\mathit{BP}(\mathit{bough}(\mu_i)) \cdot {}\\
				D_p(f_{i,1})\cdots D_p(f_{i,|\mathit{bough}(\mu_i)|+1}),\;\;
		\end{aligned}
			& otherwise,
		\end{dcases*}
\end{align*}
where $D_p: \mathcal T \to \{0,1\}^\star$ is the
depth-first order arithmetic code as introduced in the beginning of \wref{sec:universality-fixed-size-height}.
Note that $G_B$ is well-defined, as the encoding $D_p$ is only applied to fringe subtrees $\mu_i$ and $f_{i,j}$ of $t$, for which $\mathbb{P}[\mu_i], \mathbb{P}[f_{i,j}]>0$ follows from $\mathbb{P}[t] >0$.
Moreover, note that formally, $G_B$ is \emph{not} a prefix-free code over $\Sigma_\mu$: 
there can be micro tree shapes that are assigned \emph{several} codewords by $G_B$, 
depending on which nodes are portals to other micro trees (if any). 
But $G_B$ is uniquely decodable to local shapes of micro trees, and can thus 
be seen as a \emph{generalized prefix-free code}, 
where more than one codeword per symbol is allowed.
In terms of the encoding length, assigning more than one codeword is 
not helpful~-- removing all but the shortest one never makes the code worse~--
so a Huffman code minimizes the encoding length over the larger class of
\emph{generalized} prefix-free codes.
In particular, the Huffman code $C$ for micro trees used in the hypersuccinct code 
achieves no worse encoding length than the great-branching code:
\begin{align*}
\sum_{i=1}^m  |C(\mu_i)|
	&\wwrel\le
		\sum_{i=1}^m  |G_B(\mu_i)|.
\end{align*}
With $\mathcal{I}_0 = \{i \in [m] : \mu_i \text{ is a fringe micro tree in } t\}$, and   
$\mathcal{I}_1 = [m] \setminus \mathcal{I}_0$, we have
\begin{align*}
	&
		\sum_{i=1}^m  |G_B(\mu_i)| 
	\wwrel= \sum_{i \in \mathcal{I}_0}|G_B(\mu_i)| +\sum_{i \in \mathcal{I}_1}|G_B(\mu_i)|
\\	&\wwrel\leq 
		\sum_{i \in \mathcal{I}_0}\left(1+|D_p(\mu_i)|\right) 
		\bin+ \sum_{i \in \mathcal{I}_1}\left(2 + 2\lg(|\mathit{bough}(\mu_i)|) + 2|\mathit{bough}(\mu_i)| 
					+\mkern-20mu\sum_{j=1}^{|\mathit{bough}(\mu_i)|+1}\mkern-20mu|D_p(f_{i,j})|\right).
\end{align*}
With the estimate \eqref{eq:depth-first-arithmetic-code-length}, this is upper-bounded by
\begin{align*}
&
		\sum_{i \in \mathcal{I}_0}\left(4+\lg\frac{1}{\mathbb{P}[\mu_i]}+2\lg(|\mu_i|+1)\right) 
		\bin+ \sum_{i \in \mathcal{I}_1}\sum_{j=1}^{|\mathit{bough}(\mu_i)|+1}\!\!\left(3+\lg\frac{1}{\mathbb{P}[f_{i,j}]}+2\lg(|f_{i,j}|+1)\right)
\\*	&\wwrel\ppe{}
		\bin+ \sum_{i \in \mathcal{I}_1}\left(2 + 2\lg(|\mathit{bough}(\mu_i)|) 
			+ 2|\mathit{bough}(\mu_i)|\right)
\\	&\wwrel\leq 
		\sum_{i \in \mathcal{I}_0}\lg\frac{1}{\mathbb{P}[\mu_i]}
		\bin+ \sum_{i \in \mathcal{I}_1}\sum_{j=1}^{|\mathit{bough}(\mu_i)|+1}
			\mkern-20mu\lg\frac{1}{\mathbb{P}[f_{i,j}]} 
		\bin+O\left(m\log\mu\right)+O\left(\sum_{i \in \mathcal{I}_1}|\mathit{bough}(\mu_i)|\log\mu\right).
\end{align*}
By \wref{lem:fringe-dominated-submultiplikativ}, we have
\begin{align*}
		\sum_{i \in \mathcal{I}_0}\lg\frac{1}{\mathbb{P}[\mu_i]}
		\bin+ \sum_{i \in \mathcal{I}_1}
			\sum_{j=1}^{|\mathit{bough}(\mu_i)|+1}\lg\frac{1}{\mathbb{P}[f_{i,j}]} 
	\wwrel\leq 
		\lg \left(\frac{1}{\mathbb{P}[t]}\right).
\end{align*}
It remains to upper-bound the error terms: \wref{lem:bough-implies-heavy} implies that any node $v$ in the bough of a micro tree satisfies $|t[v]|\geq B$. Thus, the total number of nodes of $t$ which belong to a bough of $t$ is therefore upper-bounded by $n_{\geq B}(t)$. Altogether, we thus obtain
\begin{align*}
		\sum_{i=1}^m  |C(\mu_i)| 
	\wwrel\leq 
		\lg \left(\frac{1}{\mathbb{P}[t]}\right) 
		\bin+ O\left(m\log\mu\right) 
		\bin+ O\bigl(n_{\geq B}(t) \log \mu\bigr).
\end{align*} 
By a pigeon-hole argument, we find $n_{\ge B}(t) = \Omega(n/B)$. As $\mu = \Theta(B(n))=\Theta(\log n)$ and $m = \Theta(n/B(n))=\Theta(n/\log n)$ (see \wref{sec:hypersuccinct-code}), we have
\begin{align*}
		\sum_{i=1}^m  |C(\mu_i)| 
	\wwrel\leq 
		\lg \left(\frac{1}{\mathbb{P}[t]}\right)  
		\bin+ O\left(n_{\geq B(n)}(t) \log \log n\right).
\end{align*} 
\end{proof}
For average-case fringe-dominated fixed-size binary tree sources (defined in \wref{def:avfringe-dominated}), we obtain the following result from \wref{lem:great-branching} and \wref{lem:hypersuccinct-code}:

\begin{theorem}[Universality from average-case fringe dominance]
\label{thm:fringe-dominated}
	Let $\mathcal{S}_{\mathit{fs}}(p)$ be an average-case fringe-dominated fixed-size binary tree source. 
	Then the hypersuccinct code $\mathsf{H}: \mathcal{T} \rightarrow \{0,1\}^\star$
	satisfies
	\begin{align*}
		\sum_{t \in \mathcal{T}_n}\mathbb{P}[t]|\mathsf{H}(t)| 
		\wwrel\leq 
		H_n(\mathcal{S}_{\mathit{fs}}(p))+ o(n).
	\end{align*}
\end{theorem}
For worst-case fringe-dominated fixed-size, respectively, fixed-height binary tree sources (defined in \wref{def:wfringe-dominated}), we get the following result from \wref{lem:great-branching} and \wref{lem:hypersuccinct-code}:
\begin{theorem}[Universality from worst-case fringe dominance]
\label{thm:wfringe-dominated}
	Let $\mathcal{S}_{\mathit{fs}}(p)$, respectively, $\mathcal{S}_{\mathit{fh}}(p)$ be a worst-case fringe-dominated fixed-size or fixed-height binary tree source. 
	Then the hypersuccinct code $\mathsf{H}: \mathcal{T} \rightarrow \{0,1\}^\star$
	satisfies
	\begin{align*}
		|\mathsf{H}(t)| 
		\wwrel\leq 
		\lg\left(\frac{1}{\mathbb{P}[t]}\right) + o(n)
	\end{align*}	
	for every binary tree $t \in \mathcal{T}_n$ with $\mathbb{P}[t]>0$.
\end{theorem}
In \wref{sec:fringe-dom}, we have presented several general classes of fixed-size and fixed-height tree sources, which are average-case or worst-case fringe-dominated. For these classes, we now obtain the following universality results of our hypersuccinct encoding from \wref{lem:great-branching} and \wref{lem:hypersuccinct-code}.
With \wref{lem:leaf-centric-average-nondegenerate} we find for $\psi$-nondegenerate fixed-size binary tree sources (defined in \wref{def:psi-nondegenerate}): 
\begin{corollary}[Universality from $\psi$-nondegeneracy]
\label{cor:psicorollary}
	Let $\mathcal{S}_{\mathit{fs}}(p)$ be a $\psi$-nondegenerate fixed-size binary tree source. Then the hypersuccinct code $\mathsf{H}: \mathcal{T} \rightarrow \{0,1\}^\star$
	satisfies
	\begin{align*}
	\sum_{t \in \mathcal{T}_n}\mathbb{P}[t]|\mathsf{H}(t)| \wwrel\leq H_n(\mathcal{S}_{\mathit{fs}}(p)) + O(n\psi(\log n)\log\log n).
	\end{align*}	
\end{corollary}
With \wref{lem:leaf-centric-average-balanced}, we obtain for $\varphi$-weakly-weight-balanced fixed-size binary tree sources (defined in \wref{def:phi-weakly-weight}):
\begin{corollary}[Universality from $\varphi$-balance]
\label{cor:varphicorollary}
	Let $\mathcal{S}_{\mathit{fs}}(p)$ be a $\varphi$-weakly-weight-balanced fixed-size binary tree source. Then the hypersuccinct code $\mathsf{H}: \mathcal{T} \rightarrow \{0,1\}^\star$
	satisfies
	\begin{align*}
	\sum_{t \in \mathcal{T}_n}\mathbb{P}[t]|\mathsf{H}(t)| \wwrel\leq H_n(\mathcal{S}_{\mathit{fs}}(p)) + O\left(\frac{n\log\log n}{\varphi(n)\log n}\right).
	\end{align*}	
\end{corollary}
Moreover, with \wref{lem:weight-balanced}, we find for weight-balanced fixed-size binary tree sources (defined in \wref{def:weight-balanced}):
\begin{corollary}[Universality from weight-balance]
\label{cor:weight-balanced-universal}
	Let $\mathcal{S}_{\mathit{fs}}(p)$ be a weight-balanced fixed-size binary tree source. Then the hypersuccinct code $\mathsf{H}: \mathcal{T} \rightarrow \{0,1\}^\star$
	satisfies
	\begin{align*}
	|\mathsf{H}(t)| \wwrel\leq \lg\left(\frac{1}{\mathbb{P}[t]}\right) + O\left(\frac{n \log\log n}{\log n}\right)
	\end{align*}	
	for every binary tree $t \in \mathcal{T}_n$ with $\mathbb{P}[t]>0$.
\end{corollary}
Finally, with \wref{lem:height-balanced}, we obtain for $\delta$-height-balanced fixed-height binary tree sources (defined in \wref{def:height-balanced}):
\begin{corollary}[Universality from height-balance]
\label{cor:deltacorollary}
	Let $\mathcal{S}_{\mathit{fh}}(p)$ be a $\delta$-height-balanced fixed-height binary tree source. Then the hypersuccinct code $\mathsf{H}: \mathcal{T} \rightarrow \{0,1\}^\star$
	satisfies
	\begin{align*}
	|\mathsf{H}(t)|
	\wwrel\leq 
	\lg\left(\frac{1}{\mathbb{P}[t]}\right) + O\left( \frac{\delta(n)n\log\log n}{\log n}\right)
	\end{align*}	
	for every binary tree $t \in \mathcal{T}_n$ with $\mathbb{P}[t]>0$.
\end{corollary}

As the fixed-size and fixed-height tree sources from \wref{exm:dst}, \wref{exm:fringe-balanced-bsts}, \wref{exm:avl-uniform-height} and \wref{exm:weightbalanced} are (average-case or worst-case) fringe dominated, we obtain the following corollary from \wref{thm:fringe-dominated} and \wref{thm:wfringe-dominated}:

\begin{examplecorollary}\label{cor:fringe-dom}
	The hypersuccinct code $\mathsf{H}: \mathcal{T} \to \{0,1\}^\star$ satisfies the following:
	\begin{thmenumerate}{cor:fringe-dom}
	\item[(i)] 
		A binary tree of size $n$ randomly generated by the \textbf{binomial random tree model} $\mathcal{S}_{\mathit{fs}}(p_{bin})$ from \wref{exm:dst} is average-case optimally encoded:
		\begin{align*}
		\sum_{t \in \mathcal{T}_n}\mathbb{P}[t]|\mathsf{H}(t)| 
		\wwrel\leq 
		H_n(\mathcal{S}_{\mathit{fs}}(p_{bin})) + o(n).
		\end{align*}
	\item[(ii)] 
		A binary tree of size $n$ randomly generated by the \textbf{random fringe-balanced BST model} $\mathcal{S}_{\mathit{fs}}(p_{bal})$ from \wref{exm:fringe-balanced-bsts} is average-case optimally encoded:
		\begin{align*}
		\sum_{t \in \mathcal{T}_n}\mathbb{P}[t]|\mathsf{H}(t)| 
		\wwrel\leq 
		H_n(\mathcal{S}_{\mathit{fs}}(p_{bal})) + o(n).
		\end{align*}
	\item[(iii)] 
		An \textbf{AVL tree} $t$ of size $n$ and height $h$, drawn uniformly at random from the set $\mathcal{T}^h(\mathcal{A})$ of all AVL trees of height $h$, is optimally compressed using $|\mathsf{H}(t)| \leq \lg(|\mathcal{T}^h(\mathcal{A})|)+o(n)$ many bits (see \wref{exm:avl-uniform-height}).  
	\item[(iv)] 
		An $\alpha$-\textbf{weight-balanced BST} of size $n$, drawn uniformly at random from the set $\mathcal{T}_n(\mathcal{W}_{\alpha})$ of all $\alpha$-weight-balanced binary trees of size $n$, is optimally compressed using $|\mathsf{H}(t)| \leq\lg(|\mathcal{T}_n(\mathcal{W}_{\alpha})|) + o(n)$ many bits (see \wref{exm:weightbalanced}).
	\end{thmenumerate}
\end{examplecorollary}
We remark that using \wref{lem:leaf-centric-average-balanced} and \wref{lem:leaf-centric-average-nondegenerate}, it is possible to determine a more precise redundancy term for the results from \wref{cor:fringe-dom}, part (i) and part (ii). Moreover, we remark that the average-case result from \wref{cor:monotonic}, part (i), also follows from \wref{thm:fringe-dominated}.

\section{Uniform-Subclass Sources}
\label{sec:uniform-binary}

Finally, another class for which we will be able to prove universality of our code are 
so-called \emph{uniform-subclass sources}. Let $\mathcal{T}(\mathcal{P})$ (resp. $\mathcal{T}_n(\mathcal{P})$) denote the 
subset of binary trees $t \in \mathcal{T}$ (resp. $t \in \mathcal{T}_n$), which satisfy a certain \emph{property} $
\mathcal{P}$ (examples will be given below).
A uniform subclass source $\mathcal{U}_{\mathcal{P}}$ with respect to a property 
$\mathcal{P}$ assigns a probability to a binary tree $t \in \mathcal{T}_n$ by
\begin{align}\label{eq:probuniform}
		\mathbb{P}[t]
	\wwrel=
	\begin{cases} 
		|\mathcal{T}_n(\mathcal{P})|^{-1} \quad 
			&\text{if } t \in \mathcal{T}_n(\mathcal{P});
	\\
		0 & \text{otherwise}.
	\end{cases}
\end{align}
That is, a uniform subclass source $\mathcal{U}_{\mathcal{P}}$ induces a uniform 
probability distribution on the sets $(\mathcal{T}_n(\mathcal{P}))_n$ of all binary trees of 
size $n$ which satisfy property $\mathcal{P}$. For technical reasons, we include the empty binary tree $\Lambda$ in the set $\mathcal{T}(\mathcal{P})$ and set $\mathbb{P}[\Lambda]=1$.
We cannot hope to obtain universal codes for uniform-subclass sources in full generality.
We therefore restrict our attention to 
\emph{tame uniform subclass sources} $\mathcal{U}_{\mathcal{P}}$,
which we define to mean the following four conditions:
\begin{itemize}
\item[(i)] \emph{Fringe-hereditary:} We call a property $\mathcal{P}$ 
\emph{fringe-hereditary}, if every fringe subtree of a binary tree $t\in \mathcal{T}(\mathcal{P})$ belongs to $\mathcal{T}(\mathcal{P})$ as well. Furthermore, we call a uniform subclass source $\mathcal{U}_{\mathcal{P}}$ \emph{fringe-hereditary}, if the property $\mathcal{P}$ is fringe-hereditary. 
\item[(ii)]\emph{Worst-case fringe dominated:} Recall that $n_{\geq b}(t)$ denotes the number of nodes $v$ of a binary tree $t$, for which $|t[v]| \geq b$, where $b$ is a parameter. We call a uniform subclass source $\mathcal{U}_{\mathcal{P}}$ \emph{worst-case $B$-fringe-dominated} for a function $B=B(n)$ with $B(n) = \Theta(\log n)$, if $n_{\geq B(n)}(t)\in o(n/\log B(n))$ for every binary tree $t$ in $\mathcal{T}_n(\mathcal{P})$.
\item[(iii)] \emph{Log-linear}: A uniform subclass source $\mathcal{U}_{\mathcal{P}}$ is called \emph{log-linear}, if there is a constant $c>0$ and a function $\vartheta$ with $\vartheta(n) \in o(n)$, such that 
\begin{align*}
\lg\left(|\mathcal{T}_n(\mathcal{P})|\right) \wrel= c \cdot n + \vartheta(n).
\end{align*}
\item[(iv)] \emph{Heavy twigged}: A property $\mathcal{P}$ is called \emph{$B$-heavy twigged} for a function $B=B(n)$ with $B(n) \in \Theta(\log(n))$,
if every $t$ in $\mathcal{T}_n(\mathcal{P})$ satisfies the following condition: If $v$ is a node of $t$ with $|t[v]| \geq B=B(n)$, then both its subtrees satisfy $|t_{\ell}[v]|, |t_r[v]| \in \omega(1)$. A uniform subclass source $\mathcal{U}_{\mathcal{P}}$ is called \emph{$B$-heavy twigged}, if $\mathcal{P}$ is $B$-heavy-twigged.
\end{itemize}

\begin{definition}[Tame uniform-subclass sources]
\label{def:tame-uniform}
	A uniform-subclass source $\mathcal{U}_{\mathcal{P}}$ is called \thmemph{tame},
	if it is fringe-hereditary, worst-case fringe dominated, log-linear, and heavy twigged.
\end{definition}

\begin{example}[AVL trees]\label{exm:AVLsize}
	An example of a property which satisfies all of these four conditions is being an AVL tree:
	An AVL tree is a binary tree $t$ which is $1$-height-balanced, that is, for every node $v$ of $t$, we have $|h(t_{\ell}[v])-h(t_r[v])| \leq 1$.
	Let $\mathcal{A}$ denote this property of being an AVL tree, then $\mathcal{U}_{\mathcal{A}}$ yields the uniform probability distribution on the set of AVL trees of a given \emph{size}. By definition, we find that $\mathcal A$ is fringe-hereditary. 
	Moreover, from \wref{lem:hubschleschneider} and  \cite[Lemma 7]{GanardiHuckeLohreySeelbachBenkner2019}, we find that $\mathcal{U}_{\mathcal{A}}$ is worst-case fringe-dominated for any function $B$ with $B(n)=\Theta(\log n)$.
	
	A precise asymptotic for the number $a_n$ of AVL trees of size $n$
	is reported by Odlyzko~\cite{Odlyzko1984}: $a_n \sim \alpha^{-n} n^{-1} u(\ln n)$
	as $n\to\infty$, where $\alpha = 0.5219024\dots$ is a numerically known constant
	and $u(x)$ is a fixed, continuous periodic function.
	(Curiously, a detailed proof does not seem to have been published.)
	We obtain $\lg a_n \sim c n$ with $c \approx 0.938148$, that is, $\mathcal{U}_{\mathcal{A}}$ is log-linear. 
	
	Finally, $\mathcal{A}$ is heavy-twigged: Let $v$ be a node of $t \in \mathcal{T}(\mathcal{A})$ with $|t[v]|\geq B$. As $t[v]$ is a binary tree, we have $h(t[v])\geq \lg(B)$. Moreover, as $t$ is an AVL tree, we have $h(t_{\ell}[v]), h(t_r[v]) \geq \lg(B)-2$ and thus $|t_{\ell}(v)|, |t_r[v]|\geq \lg(B)-2$, which is in $\omega(1)$ for $B = \Theta(\log(n))$. 
\end{example}

\needspace{15\baselineskip}
\begin{example}[Red-black trees]\label{exm:redblack}
	Another example is the property $\mathcal R$, which holds if $t$ is the shape of a red-black tree: 
	A (left-leaning) red-black tree is a binary tree in which the edges are (implicitly) colored red and black, 
	so that the following conditions hold:
	\begin{itemize}[leftmargin=2em]
	\item[(a)] The number of black edges on any root-to-leaf path is the same.
	\item[(b)] No root-to-leaf path contains two consecutive red edges.
	\item[(c)] If a node has only one red child edge, it must be the left child edge.
	\end{itemize}
	It is easy to check that $\mathcal R$ is fringe-hereditary.
	One can show inductively that the height of a red-black tree is at most $2\lg n + O(1)$,
	which together with fringe-hereditary and \wref{lem:hubschleschneider} 
	implies that $\mathcal U_{\mathcal R}$ is worst-case fringe-dominated.

	For the log-linearity, we have to determine $\lg r_n$, for $r_n$ the number of 
	left-leaning red-black trees of size $n$.
	Since left-leaning red-black trees are in bijection with 2-3-4-trees~\cite{Sedgewick2008},
	we can also count the latter. 
	The similar 2-3 trees are enumerated (where the size is the number of external leaves) in~\cite{MillerPippengerRosenbergSnyder1979,Odlyzko1984}
	and the same technique allows to determine the exponential growth rate.
	We obtain $\lg r_n \sim c n$ with $c \approx 0.879146$.
	
	For the heavy-twigged property, let $t[v]$ be a fringe subtree in a red-black tree
	with $|t[v]| \ge B$. We have $h(t[v]) \ge \lg(B)$ (as for any binary tree). 
	Moreover, since black-heights must be equal and at most every other edge can be red, 
	we have $h(t_\ell[v]), h(t_r[v]) \ge \frac12 h(t[v]) -1 \ge \frac12\lg(B)-1$,
	which also lower bounds the size of $t_\ell[v]$ and $t_r[v]$.
	So $t_\ell[v]|, |t_r[v]| = \omega(1)$ as $B\to\infty$.
\end{example}

\begin{example}[Weight-balanced BSTs]\label{exm:weightbalanced}
	Let $\mathcal{T}(\mathcal{W}_{\alpha})$ denote the set of \emph{$\alpha$-weight-balanced binary trees} (in the sense of $\mathrm{BB}[\alpha]$, \cite{NievergeltReingold1973}): A binary tree is $\alpha$-weight-balanced, if for every node $v$ of $t$, we have $|t_{\ell}[v]|+1 \geq \alpha (|t[v]|+1)$ and $|t_r[v]|+1 \geq \alpha (|t[v]|+1)$ (note that this is a special case of $\beta$-balanced binary trees considered in the proof of \wref{lem:weight-balanced}).
	The property $\mathcal{W}_{\alpha}$ is fringe-hereditary by definition and it is easy to see that $\mathcal{W}_{\alpha}$ is heavy-twigged. 
	
	From the proof of \wref{lem:weight-balanced}, we furthermore find that $\alpha$-weight-balanced binary trees are worst-case fringe dominated. Unfortunately, we are not aware of a counting result for these trees, and so it remains a conjecture that $\alpha$-weight-balanced binary trees are log-linear and thus amenable to the same treatment.
	
	However, the uniform subclass source $\mathcal{U}_{\mathcal{W}_{\alpha}}$ can be modeled as a worst-case fringe dominated fixed-size source: If we set
	\begin{align*}
			p(\ell,n-\ell)
		\wwrel=
			\begin{dcases}
				\frac{|\mathcal{T}_\ell(\mathcal{W}_{\alpha})||\mathcal{T}_{n-\ell}(\mathcal{W}_{\alpha})|}
						{|\mathcal{T}_{n+1}(\mathcal{W}_{\alpha})|} \quad 
					&\text{if } \ell+1,n-\ell+1 \geq \alpha(n+2),
			\\
				0 & \text{otherwise}
		\end{dcases}
	\end{align*}
	for every $n \in \mathbb{N}$, then the corresponding fixed-size tree source $\mathcal{S}_{\mathit{fs}}(p)$ corresponds to a uniform probability distribution on $\mathcal{T}_n(\mathcal{W}_{\alpha})$ for every $n \in \mathbb{N}$.
\end{example}

\subsection{Universality for Uniform-Subclass Sources}

In order to show universality of the hypersuccinct code from \wref{sec:hypersuccinct-code} with respect to uniform subclass sources, we first derive a source-specific encoding with respect to the uniform subclass source, against which we will then compare our hypersuccinct code:

An encoding $E_{\mathcal{P}}(t)$ that stores a given binary tree $t \in \mathcal{T}_n(\mathcal{P})$ in $\lg (|\mathcal{T}_n(\mathcal{P}|)+O(\log n)$ many bits is obtained as follows: Let $t_1, \dots, t_{|\mathcal{T}_n(\mathcal{P})|}$ denote an enumeration of all elements in $\mathcal{T}_n(\mathcal{P})$. In order to encode a binary tree $t \in \mathcal{T}_n(\mathcal{P})$, we first encode its size (plus one, in order to incorporate the case that $t$ is the empty binary tree), in gamma code, $\gamma(n+1)$, followed by its number $i \in [|\mathcal{T}_n(\mathcal{P})|]$ in the enumeration of all binary trees in $\mathcal{T}_n(\mathcal{P})$, using $\lfloor\lg (|\mathcal{T}_n(\mathcal{P})|)\rfloor +1$ many bits.
Thus, such an encoding $E_{\mathcal{P}}: \mathcal{T}(\mathcal{P}) \rightarrow \{0,1\}^\star$ spends at most
\begin{align}\label{eq:estimateuniformsublclass}
	|E_{\mathcal{P}}(t)| 
	\wwrel\leq 
	\lg (|\mathcal{T}_n(\mathcal{P})|)+2\lg(n+1)+2 
	\wwrel= 
	\lg \left(\frac{1}{\mathbb{P}[t]}\right)+2\lg(n+1)+2
\end{align}
many bits in order to encode $t \in \mathcal{T}_n(\mathcal{P})$. We remark that $E_{\mathcal{P}}: \mathcal{T}(\mathcal{P}) \rightarrow \{0,1\}^\star$ is a single prefix-free code on $\mathcal{T}(\mathcal{P})$.
We find the following:
\begin{lemma}[Great-branching lemma for $\mathcal U_{\mathcal P}$]
\label{lem:great-branching2}
	Let $\mathcal{U}_{\mathcal{P}}$ be a fringe-hereditary, worst-case fringe-dominated, log-linear, heavy-twigged uniform subclass source and let $t \in \mathcal{T}_n$ with $\mathbb{P}[t]>0$. Then
	\begin{align*}
	\sum_{i=1}^m  |C(\mu_i)| \leq \lg \left(\frac{1}{\mathbb{P}[t]}\right) + o(n),
	\end{align*}
	where $C$ is a Huffman code for the sequence of micro trees $\mu_1, \dots, \mu_m$ from our tree covering scheme (see \wref{sec:hypersuccinct-code}).
\end{lemma}
\begin{proof}
The proof works in a similar way as the proof of \wref{lem:great-branching}: Let $\mu_i$ be an internal (non-fringe) micro tree. By $\mathit{bough}(\mu_i)$, we again denote the subtree of $\mu_i$ induced by the set of nodes that are ancestors of $\mu_i$'s child micro trees (ancestors of the portals); the boughs of a micro tree are the paths from the portals to the micro tree root. Hanging off the boughs of $\mu_i$ are (fringe) subtrees $f_{i,1}, \dots, f_{i,|\mathit{bough}(\mu_i)|+1}$, listed in depth-first order of the bough nulls these subtrees are attached to. In general, some of these subtrees might be the empty tree~-- however, as the uniform subclass source $\mathcal{U}_{\mathcal{P}}$ we consider is \emph{heavy-twigged}, and as every node $v$ that belongs to $\mathit{bough}(\mu_i)$ satisfies $|t[v]| \geq B$ by \wref{lem:bough-implies-heavy} (where $B=B(n)$ is the parameter from the tree covering algorithm), we find that $|f_{i,j}| \in \omega(1)$, except for possibly two exceptions, as the portals are replaced by null pointers in $\mu_i$.
Recall that the portal nodes themselves are not part of $\mu_i$ and hence not part of $\mathit{bough}(\mu_i)$.
As in the proof of \wref{lem:great-branching}, we construct a new encoding for micro trees against which we can compare the hypersuccinct code, another `` great-branching'' code, $\tilde{G}_B$, as follows:
\begin{align*}
		\tilde{G}_B(\mu_i)
	&\wwrel=
		\begin{dcases*}
		\texttt{0}\cdot E_{\mathcal{P}}(\mu_i), & if $\mu_i$ is a fringe micro tree; \\[1ex]
		\begin{aligned}[c]
			\texttt{1}\cdot 
			\gamma(|\mathit{bough}(\mu_i)|) \cdot 
			\mathit{BP}(\mathit{bough}(\mu_i)) \cdot {}\\
				E_{\mathcal{P}}(f_{i,1})\cdots E_{\mathcal{P}}(f_{i,|\mathit{bough}(\mu_i)|+1}),\;\;
		\end{aligned}
			& otherwise.
		\end{dcases*}
\end{align*}
Note that $\tilde{G}_B$ is well-defined: As the encoding $E_{\mathcal{P}}$ is only applied to fringe subtrees $\mu_i$ and $f_{i,j}$ of $t$, which satisfy property $\mathcal{P}$ as $\mathcal{U}(\mathcal{P})$ is fringe-hereditary, we find that $\mathbb{P}[\mu_i], \mathbb{P}[f_{i,j}]>0$.
Moreover, note that formally, $\tilde{G}_B$ is \emph{not} a prefix-free code over $\Sigma_\mu$: 
there can be micro tree shapes that are assigned \emph{several} codewords by $\tilde{G}_B$, 
depending on which nodes are portals to other micro trees (if any). 
But $\tilde{G}_B$ is uniquely decodable to local shapes of micro trees, and can thus be seen as a \emph{generalized} prefix-free code.
In terms of the encoding length, assigning more than one codeword never makes the code worse,
thus a Huffman code minimizes the encoding length over the larger class of generalized prefix-free codes.
In particular, the Huffman code $C$ for micro trees used in the hypersuccinct code 
achieves no worse encoding length than the great-branching code:
\begin{align*}
\sum_{i=1}^m  |C(\mu_i)|
	&\wwrel\le
		\sum_{i=1}^m  |\tilde{G}_B(\mu_i)|.
\end{align*}
With $\mathcal{I}_0 = \{i \in [m] : \mu_i \text{ is a fringe micro tree in } t\}$, and   
$\mathcal{I}_1 = [m] \setminus \mathcal{I}_0$, we have
\begin{align*}
&
		\sum_{i=1}^m  |\tilde{G}_B(\mu_i)| 
	\wwrel= 
		\sum_{i \in \mathcal{I}_0}|\tilde{G}_B(\mu_i)| +\sum_{i \in \mathcal{I}_1}|\tilde{G}_B(\mu_i)|
\\	&\wwrel\leq 
		\sum_{i \in \mathcal{I}_0}\left(1+|E_{\mathcal{P}}(\mu_i)|\right) 
\\*	&\wwrel\ppe{} 
		\bin+ \sum_{i \in \mathcal{I}_1}\Biggl(2 + 2\lg(|\mathit{bough}(\mu_i)|) + 2|\mathit{bough}(\mu_i)| 
					+\!\!\sum_{j=1}^{|\mathit{bough}(\mu_i)|+1}\!\!|E_{\mathcal{P}}(f_{i,j})|\Biggr).
\end{align*}
With estimate \eqref{eq:estimateuniformsublclass} this is upper-bounded by
\begin{align*}
&
		\sum_{i \in \mathcal{I}_0}\left(3+\lg\frac{1}{\mathbb{P}[\mu_i]}
		\bin+2\lg(|\mu_i|+1)\right) 
		\bin+ \sum_{i \in \mathcal{I}_1}\sum_{j=1}^{|\mathit{bough}(\mu_i)|+1}\!\!\left(
			2+\lg\frac{1}{\mathbb{P}[f_{i,j}]}+2\lg(|f_{i,j}|+1)\right)
\\*	&\wwrel\ppe{}
		\bin+ \sum_{i \in \mathcal{I}_1}\left(2 + 2\lg(|\mathit{bough}(\mu_i)|) 
		+ 2|\mathit{bough}(\mu_i)|\right)
\\	&\wwrel\leq 
		\sum_{i \in \mathcal{I}_0}  \mkern-20mu\lg\frac{1}{\mathbb{P}[\mu_i]}
		\bin+ \sum_{i \in \mathcal{I}_1}\sum_{j=1}^{|\mathit{bough}(\mu_i)|+1}\lg\frac{1}{\mathbb{P}[f_{i,j}]} \bin+O\left(m\log\mu\right)+O\left(\sum_{i \in \mathcal{I}_1}|\mathit{bough}(\mu_i)|\log\mu\right).
\end{align*}
By the log-linearity of the uniform subclass source $\mathcal{U}_{\mathcal{P}}$, we  find $\lg(1/\mathbb{P}[\mu_i]) = \lg\left(|\mathcal{T}_{|\mu_i|}(\mathcal{P})|\right) = c |\mu_i| + \vartheta(|\mu_i|)$ and $\lg(1/\mathbb{P}[f_{i,j}]) = \lg\left(|\mathcal{T}_{|f_{i,j}|}(\mathcal{P})|\right) = c |f_{i,j}| + \vartheta(|f_{i,j}|)$, with $\vartheta(n) \in o(n)$ and $c >0$ constant, for the fringe subtrees $\mu_i$ and $f_{i,j}$ (if $f_{i,j}$ is the empty binary tree, we simply have $\lg(1/\mathbb{P}[f_{i,j}])=0$ by assumption). As $\mathcal{U}_{\mathcal{P}}$ is heavy-twigged, we have $|f_{i,j}| \in \omega(1)$ for all subtrees $f_{i,j}$ which are not the empty tree. Furthermore, we find $|\mu_i| \in \omega(1)$ for all fringe micro trees $\mu_i$ of $t$ by \wref{lem:bough-implies-heavy}.
Hence, as $\vartheta(n) \in o(n)$, and as the trees $\mu_i$ for $i \in \mathcal{I}_0$ and $f_{i,j}$ are disjoint subtrees of $t$, we have
\[
		\sum_{i \in \mathcal{I}_0}\vartheta(|\mu_i|)
		\bin+ \sum_{i \in \mathcal{I}_1	} \sum_{j = 1}^{|\mathit{bough}(\mu_i)|+1}\vartheta(|f_{i,j}|)
	\wwrel= 
	o(n).
\]
Thus, we find
\begin{align*}
&
		\sum_{i \in \mathcal{I}_0}\lg\frac{1}{\mathbb{P}[\mu_i]}
		\bin+ \sum_{i \in \mathcal{I}_1}\sum_{j=1}^{|\mathit{bough}(\mu_i)|+1}\lg\frac{1}{\mathbb{P}[f_{i,j}]} 
\\	&\wwrel= 
	 	c\sum_{i \in \mathcal{I}_0}\left(|\mu_i|+\vartheta(|\mu_i|)\right)
		\bin+ c\sum_{i \in \mathcal{I}_1}\sum_{j=1}^{|\mathit{bough}(\mu_i)|+1}
			\left(|f_{i,j}|+\vartheta(|f_{i,j}|)\right) \\
\\	&\wwrel\leq  
		c|t|+ o(n)
	\wwrel= 
		\lg\left(|\mathcal{T}_{n}(\mathcal{P})|\right)+o(n) 
	\wwrel= 
		\lg \left(\frac{1}{\mathbb{P}[t]}\right)+o(n).
\end{align*}
It remains to upper-bound the error terms: \wref{lem:bough-implies-heavy} implies that any node $v$ in the bough of a micro tree satisfies $|t[v]|\geq B$. Thus, the total number of nodes of $t$ which belong to a bough of $t$ is therefore upper-bounded by $n_{\geq B}(t)$. Altogether, we thus obtain
\begin{align*}
		\sum_{i=1}^m  |C(\mu_i)| 
	\wwrel\leq 
		\lg \left(\frac{1}{\mathbb{P}[t]}\right) 
		\bin+ O\left(m\log\mu\right) 
		\bin+ O\left(n_{\geq B}(t) \log \mu\right)
		\bin+o(n).
\end{align*} 
By a pigeon-hole argument, we find $n_{\ge B}(t) = \Omega(n/B)$ and as $\mathcal{P}$ is worst-case fringe-dominated, we have $n_{\ge B(n)}(t)\in o(n/\log(B(n)))$.
Furthermore, as $\mu = \Theta(\log n)$ and $m = \Theta(n/\log n)$ (see \wref{sec:hypersuccinct-code}), we have
\begin{align*}
		\sum_{i=1}^m  |C(\mu_i)| 
	\wwrel\leq 
		\lg \left(\frac{1}{\mathbb{P}[t]}\right)  + o(n).
\end{align*} 
\end{proof}

\begin{theorem}[Universality for tame uniform sources]
\label{thm:uniformsubclass}
Let $\mathcal{U}_{\mathcal{P}}$ be a fringe-hereditary, worst-case fringe-dominated, log-linear, heavy-twigged uniform subclass source. The hypersuccinct code $\mathsf{H}: \mathcal{T} \rightarrow \{0,1\}^\star$ satisfies
\begin{align*}
		\hypsuc t 
	\wwrel\leq 
		\lg \left(\frac{1}{\mathbb{P}[t]}\right) + o(n)
\end{align*} 
for every binary tree $t$ of size $n$ with $\mathbb{P}[t]>0$.
\end{theorem} 
\wref{thm:uniformsubclass} follows from \wref{lem:great-branching2} and \wref{lem:hypersuccinct-code}. 
In particular, we obtain the following corollary from \wref{thm:uniformsubclass} (see \wref{exm:AVLsize} and \wref{exm:redblack}):

\enlargethispage{1\baselineskip}
\begin{examplecorollary}\label{cor:uniformsubclass}
	The hypersuccinct code $\mathsf{H}: \mathcal{T} \rightarrow \{0,1\}^\star$ optimally compresses
	\begin{itemize}
	\item[(i)] 
		\textbf{AVL trees} of size $n$, drawn uniformly at random from the set $\mathcal{T}_n(\mathcal{A})$ of all AVL trees of size $n$, using 
		\[
			|\mathsf{H}(t)|
			\wwrel\leq 
			\lg\left(|\mathcal{T}_n(\mathcal{A})|\right)+o(n)
			\wwrel\approx 
			0.938148n+o(n)
		\] 
		many bits and
	\item[(ii)] 
		\textbf{red-black trees} of size $n$, drawn uniformly at random from the set $\mathcal{T}_n(\mathcal{R})$ of all red-black trees of size $n$, using 
		\[
			|\mathsf{H}(t)|
			\wwrel\leq \lg\left(|\mathcal{T}_n(\mathcal{R})|\right)+o(n)
			\wwrel\approx 
			0.879146n+o(n)
		\] 
		many bits.%
	\end{itemize}%
\end{examplecorollary}

\clearpage
\section{Range-Minimum Queries With Runs}
\label{sec:omitted-proofs-rmq-runs}
\label{app:rmq-runs}

In this appendix, we give the proofs of the results from \wref{sec:main-rmq-runs}.

\subsection{Lower Bound}
\label{sec:rmq-runs-lower-bound}

In this section, we proof \wref{thm:rmq-runs-lower-bound}.

We refer to a run of length one as a singleton run.
The types of nodes in the Cartesian tree 
(whether or not their left and right children exist\ifproceedings{}{, see \wref{sec:memoryless-binary}})
directly reflect their role in runs:
A binary node is a run head of a non-singleton run, a leaf is the last node
of non-singleton run, a right-unary node (\ie, unary node with a right child) 
is a middle node of run and a left-unary node is a singleton run.
The leftmost node, \ie, the node with smallest inorder rank, is the only exception to this rule:
if the leftmost run is a singleton run, the leftmost node is a leaf; otherwise it is right-unary.

In any case, a Cartesian tree for an array with $r$ runs that has $b$ binary nodes and $u_\ell$ 
left-unary nodes thus satisfies $r = b + u_\ell + 1$: every binary node represents the non-singleton run
that begins with it, every left-unary node represents the singleton run at that position,
and the leftmost run is counted separately. (Note that we do not double count the latter
because the leftmost node is by definition neither binary nor left-unary.)
We therefore obtain a lower bound for the number of equivalence classes 
among length-$n$ arrays with $r$ runs under range-minimum queries by counting binary trees with a
given number of nodes $n$ and a given number of nodes of certain types.

That $r$ is the sum of two quantities is inconvenient, hence we instead consider 
the following sequence of bijections (see \wref{fig:bijections}).
First, we map Cartesian trees $t$ of $n$ nodes bijectively to balanced-parenthesis (BP) strings of $n$ 
pairs of parentheses as follows:
The empty tree corresponds to the empty string. For a nonempty tree, we recursively compute
the BP strings of the (potentially empty) left resp.\ right subtrees of the root;
let these be denoted by $L$ and $R$. Then the BP string for the entire tree is obtained as
$L \texttt ( R \texttt)$.
(This is a variation of the canonical BP representation\ifproceedings{}{ used in \wref{part:binary}}.)

\begin{figure}[tbh]
	\plaincenter{\includegraphics[width=.45\textwidth]{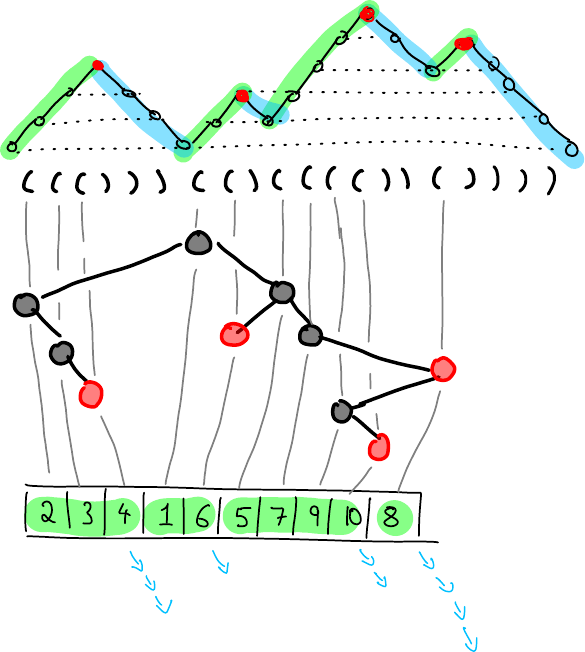}}
	\caption{%
		An example illustrating the bijections: The input array is $A=(2,3,4,1,6,5,7,9,10,8)$, 
		the min-oriented Cartesian tree is shown above with run ends highlighted in red.
		The BP string for the Cartesian tree is shown above the tree, with tree nodes connected
		to the corresponding opening parenthesis (note that nodes appear in inorder in the BP string).
		The maintain-valley (excess, Dyck path) representation of the BP string is on top;
		run ends correspond to peaks there.
	}
	\label{fig:bijections}
\end{figure}

It is easy to check that the resulting sequence is indeed the push/pop sequence of a max-stack~\cite{FischerHeun2011,GawrychowskiNicholson2015}
where `\texttt (' means push and `\texttt)' means pop.
We map this sequence to a lattice path by replacing 
`\texttt (' by step vector $(1,+1)$ and `\texttt )' by $(1,-1)$;
the resulting lattice path is a mountain-valley diagram (Dyck paths).

The important property of the above bijections is that they \emph{preserve runs}:
A \emph{run end} is an index where the next number is smaller (or nonexistent).
In the Cartesian tree, these are the leaves \emph{and} left-unary nodes, in the BP string, 
these are the occurrences of `\texttt{()}' and in the mountain-valley representations,
these are the \emph{peaks}.
The latter is known to be counted by the \emph{Narayana numbers}:
There are 
\begin{align}
\label{eq:narayana}
		N_{n,r} 
	\wwrel= 
		\frac1n \binom nr \binom n{r-1}
	\wwrel=
		\frac{r}{n(n+1-r)}\binom{n}{r}^2
\end{align}
mountain-valley diagrams of length $2n$ with exactly $r$ peaks~\cite{OEIS-Narayana-numbers}.
This concludes the proof of \wref{thm:rmq-runs-lower-bound};
the asymptotic approximation for $\lg N_{n,r}$ immediately follows from the above closed form.

\subsection{Hypersuccinct RMQ with Runs}
\label{sec:rmq-ds-linear-r}

In this section, we prove \wref{cor:rmq-runs-ds}.
To this end, we show that using a hypersuccinct tree to represent the Cartesian tree
of an array $A[1..n]$ with $r$ increasing runs has a space usage
that is bounded by $\lg N_{n,r} + o(n)$ bits.
By \wref{thm:rmq-runs-lower-bound}, this space usage is optimal up to the $o(n)$ term.

As noted in \wref{sec:rmq-runs-lower-bound},
the correspondence between runs and node types in the Cartesian tree can be
made more specific by also specifying the number $s\in[r]$ of singleton runs:
Singleton runs correspond to the left-unary nodes in the (min-oriented) Cartesian tree $t$,
except possibly for a leftmost singleton run (which corresponds to a leaf).
In either case, we will have $u_\ell = s\pm1$ left-unary nodes and $\ell = r-s\pm1$ leaves. 
That implies a number of binary nodes of $b=\ell-1\pm1$;
the remaining $u_r = n-b-\ell-u_\ell = n-2r+s \pm 1$ nodes are right-unary nodes.

By \wref{cor:main-binary-empirical-entropies}, 
the hypersuccinct representation of the Cartesian tree $t$ for $A[1..n]$ 
supports LCA-queries on $t$ in $O(1)$ time and uses 
\begin{align*}
	|\mathsf{H}(t)|+o(n) \wwrel\leq H_0^{\type}(t) \bin+ o(n)
\end{align*}  
bits of space. We show that $H_0^{\type}(t) \leq \lg N_{n,r} + o(n)$. 
Let 
\[
		p
	\wwrel=
		\left(\frac{b}{n}, \frac{u_l}{n}, \frac{u_r}{n}, \frac{\ell}{n}\right)
\] 
denote the empirical distribution of node types in $t$, 
and let $H(p)$ denote the entropy of this distribution. 
(For probability distribution $d=(d_1, d_2,  \dots, d_k)$, its entropy is defined by
\begin{align*}
H(d) \wwrel= \sum_{i=1}^kd_i\lg\left(\frac{1}{d_i}\right),
\end{align*}
as usual.)

By definition of the type-entropy $H_0^{\type}$\ifproceedings{}{ (cf.\ \wref{def:empiricaltype})}, 
we find $H_0^{\type}(t)=nH(p)$.
By our previous observations, $p$ differs from 
\[
		p' 
	\wwrel= 
		\left(\frac{r-s}n,\frac{s}n,\frac{n-2r+s}n,\frac{r-s}n\right)
\]
only by $\|p-p'\|_{\infty} \le \frac 2n$. 
Using \cite[{\href{https://www.wild-inter.net/publications/html/wild-2016.pdf.html\#pf6a}{Prop.\,2.42}}]{Wild2016},
we thus find $H(p) \le H(p') + \Oh(n^{-0.9})$ (this follows from Hölder-continuity of $x\mapsto x \ln x$).
It thus remains to show that $n\,H(p') \leq \lg N_{n,r}+o(n)$. 
By the grouping property of $H$, we have
\begin{align*}
	n \, H(p') 
&\wwrel=
		n \left(H\left(\frac rn, \frac{n-r}n\right) + \frac rn H\left(\frac sr, \frac{r-s}r\right) 
			+ \frac{n-r}{n} H\left(\frac{r-s}{n-r},\frac{(n-r)-(r-s)}{n-r}\right)\right).
\end{align*}
In order to estimate the right-hand side, observe that it follows 
from \cite[Eq.\,(5.22)]{GrahamKnuthPatashnik1994} that 
$\sum_{s=0}^r \binom rs\binom{n-r}{r-s} = \binom n{r}$. 
Since all summands are positive, we have $\binom rs\binom{n-r}{r-s} \le \binom nr$ 
and hence 
\begin{align}
\label{eq:lg-vandermonde}
		\lg\binom rs+ \lg\binom{n-r}{r-s} 
	&\wwrel\le 
		\lg\binom nr,
	\qquad \text{for all } s\in[r].
\end{align}
For a number $q \in [0,1]$, we set $h(q)=q\lg(1/q)+(1-q)\lg(1/(1-q))$. 
Using the standard inequality
\begin{align}
\label{eq:binom-entropy}
		\frac{2^{n h(q)}}{n+1}
	&\wwrel\le
		\binom n{qn}
	\wwrel\le
		2^{n h(q)},
	\qquad nq\in[0..n],
\end{align}
we find
\begin{align*}
		r h\left(\frac sr\right)
		\bin + (n-r) h\left(\frac{r-s}{n-r}\right)
	&\wwrel{\relwithref{eq:binom-entropy}\le}
		\lg\binom rs+ \lg\binom{n-r}{r-s}
		\bin+ \lg (r+1) + \lg(n-r+1)
\\	&\wwrel{\relwithref{eq:lg-vandermonde}\le}
		\lg\binom nr
		\bin+ \lg (r+1) + \lg(n-r+1)
\\	&\wwrel{\relwithref{eq:binom-entropy}\le}
		nh\left(\frac{r}n\right)
		\bin+ \lg (r+1) + \lg(n-r+1).
\end{align*}
We thus have
\begin{align*}
		n H(p') 
	&\wwrel=
		n \left(H\left(\frac rn, \frac{n-r}n\right) + \frac rn H\left(\frac sr, \frac{r-s}r\right) 
			+ \frac{n-r}{n} H\left(\frac{r-s}{n-r},\frac{(n-r)-(r-s)}{n-r}\right)\right)
\\	&\wwrel= 
		n\left(h\left(\frac{r}{n}\right)+\frac{r}{n}h\left(\frac{s}{r}\right)+\frac{n-r}{n}h\left(\frac{r-s}{n-r}\right)\right)
\\	&\wwrel\leq 	 
		2n h\left(\frac rn\right)  \bin+ \Oh(\log n)
\\	&\wwrel{\relwithref{eq:binom-entropy}\le}
		2\lg\binom nr  \bin+ \Oh(\log n)
\\	&\wwrel{\relwithref{eq:narayana}\le}
		\lg N_{n,r} \bin+ \Oh(\log n).
\end{align*}
So in total, we have shown that 
\begin{align*}
		|\mathsf{H}(t)|
	&\wwrel\le
		\lg N_{n,r} \wbin+ o(n),
\end{align*}
which implies \wref{cor:rmq-runs-ds}.

	\needspace{5\baselineskip}
\part{Ordinal Trees}
\label{part:ordinal}

Most results for binary trees can be extended to ordinal trees, but some additional 
arguments resp.\ restrictions are necessary because of large-degree nodes.
Our results with respect to ordinal trees are presented in this part.

\section{Hypersuccinct Ordinal Trees}\label{sec:hypersuccinct-ordinal}

The Farzan-Munro tree decomposition algorithm~\cite{FarzanMunro2014} is used to decompose an ordinal tree into subtrees, so-called \emph{micro trees}.
In the following, we recall the properties of this tree covering method (for more details, see  \wref{sec:farzan-munro}):

\begin{lemma}[{{Tree covering, \cite[Thm.\,1]{FarzanMunro2014}}}]
\label{lem:tree-decomposition-ordinal}
	For any parameter $B\ge 1$, an ordinal tree with $n$ nodes can be decomposed,
	in linear time, into connected subtrees (so-called \emph{micro trees}) with the following properties:
	\begin{thmenumerate}{lem:ordinal-tree-decomposition}
		\item Micro trees are pairwise disjoint except for (potentially) 
			sharing a common micro tree root.
		\item Each micro tree contains at most $2B$ nodes.
		\item The overall number of micro trees is $\Theta(n/B)$.
		\item Apart from edges leaving the micro tree root, at most one other 
			edge leads to a node outside of this micro tree.
			This edge is called the ``external edge'' of the micro tree.
	\end{thmenumerate}
\end{lemma}

By inspection of the proof in \cite{FarzanMunro2014}, we can say a bit more:
If $v$ is a node in the tree and is also the root of 
several micro trees of the decomposition, 
then the way that $v$'s children (in the entire tree) are
divided among the micro trees is into \emph{consecutive} blocks.
Each micro tree contains at most two of these blocks. 
(This case arises when the micro tree root has 
exactly one heavy child in the decomposition algorithm.)
In binary trees, a micro tree is always an entire fringe subtree except for
at most two entire subtrees, which are removed from it.
In ordinal trees, the possibility of large node degrees makes such a 
decomposition impossible: here an arbitrary number of children (and their subtrees)
can be missing in a micro tree root, and a single node in the original tree
can be the (shared) root of many micro trees.

\subsection{Hypersuccinct Code}\label{sec:hypersuccinct-code-ordinal}

In this section, we describe a universal code for ordinal trees
based on the Farzan-Munro algorithm using just one level
of micro trees.
The purpose is to give a self-contained description of the mere representation 
of an ordinal tree (as opposed to a succinct data structure) 
that admits compression as a universal code.
The exposition in~\cite{FarzanMunro2014} mixes this description with
the details of the data structures needed for navigation.

We fix the parameter $B$, so that the maximal micro tree size is $\mu = \lceil \frac14 \lg n \rceil$ 
\ie, we set $B = \lceil \frac18\lg n\rceil$.
The code of the ordinal tree $t\in\mathfrak T_n$ is then obtained as follows:
Decompose the tree into micro trees $\mu_1,\ldots,\mu_m$ where $m=\Theta(n/B) = \Theta(n/\log n)$.
Recall that each micro tree $\mu_i$ can have the following connections to other micro trees:
\begin{itemize}
\item an edge to one parent micro tree,
\item an external edge to one child micro tree, leaving from some 
	node of the micro tree (and inserted at some child rank),
\item an arbitrary number of other subtrees of the shared root;
	these micro trees can contain the shared root or not.
\end{itemize}
The \emph{top-tier} $\Upsilon$ of the tree is obtained by 
contracting each micro tree into a single node;
shared roots are copied to each micro tree.
Two micro trees are connected by an edge in $\Upsilon$ if there is an edge between some nodes 
in these micro trees in $t$.
Since several micro trees can contain the root of the tree, we add a dummy root
to $\Upsilon$ to turn it into a single tree.
\wref{fig:farzan-munro-top-tier-tree} shows an example.

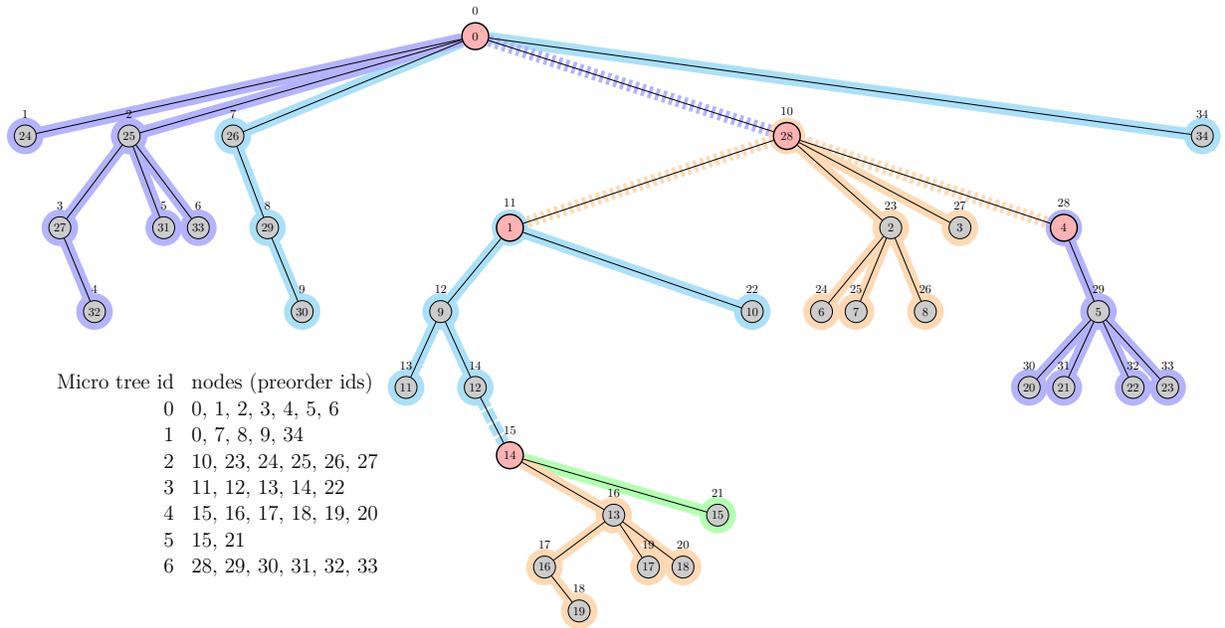
\begin{figure}
	\resizebox{\linewidth}!{\begin{tikzpicture}[
		tree node/.style         = {circle, draw, fill=black!20, font=\small, minimum size=18pt, inner sep=0pt},
		null/.style              = {rectangle,draw=black!50,thin,fill=white,minimum size=6pt},		marked null/.style       = {null,fill=red!30,},		marked node/.style       = {very thick, fill=red!30, minimum size=22pt},		preorder label/.style    = {font=\small},
		balance label/.style     = {font=\scriptsize, text=black!50},
		subtreesize label/.style = {font=\scriptsize, text=black!50},
		tree edge/.style         = {thin},
		index tree node/.style   = {tree node,marked node, opacity=.5,minimum size=18pt},
		index tree edge/.style   = {very thick,red!50, opacity=.5},
		block edge 11/.style     = {draw=blue!30,line width=10pt},
		block edge 12/.style     = {draw=cyan!30,line width=10pt},
		block edge 21/.style     = {draw=orange!30,line width=10pt},
		block edge 22/.style     = {draw=green!30,line width=10pt},
		crossing edge/.style     = {dashed},
		external edge/.style     = {dotted},
		block node 11/.style     = {circle,fill=blue!30,   minimum size=30pt},
		block node 12/.style     = {circle,fill=cyan!30,   minimum size=30pt},
		block node 21/.style     = {circle,fill=orange!30, minimum size=30pt},
		block node 22/.style     = {circle,fill=green!30,  minimum size=30pt},
		block node null/.style   = {red},
		block edge null/.style   = {red},
		connector lines/.style   = {thin,loosely dotted},
]

	\node[tree node,marked node] (n-0) at (13.000000,19.500000) {$0$} ;
	\node[preorder label,above=2pt of n-0] {$0$} ;
	\node[tree node] (n-1) at (0.000000,16.615385) {$24$} ;
	\node[preorder label,above=2pt of n-1] {$1$} ;
	\node[tree node] (n-2) at (3.000000,16.615385) {$25$} ;
	\node[preorder label,above=2pt of n-2] {$2$} ;
	\node[tree node] (n-3) at (1.000000,13.961538) {$27$} ;
	\node[preorder label,above=2pt of n-3] {$3$} ;
	\node[tree node] (n-4) at (2.000000,11.538462) {$32$} ;
	\node[preorder label,above=2pt of n-4] {$4$} ;
	\node[tree node] (n-5) at (4.000000,13.961538) {$31$} ;
	\node[preorder label,above=2pt of n-5] {$5$} ;
	\node[tree node] (n-6) at (5.000000,13.961538) {$33$} ;
	\node[preorder label,above=2pt of n-6] {$6$} ;
	\node[tree node] (n-7) at (6.000000,16.615385) {$26$} ;
	\node[preorder label,above=2pt of n-7] {$7$} ;
	\node[tree node] (n-8) at (7.000000,13.961538) {$29$} ;
	\node[preorder label,above=2pt of n-8] {$8$} ;
	\node[tree node] (n-9) at (8.000000,11.538462) {$30$} ;
	\node[preorder label,above=2pt of n-9] {$9$} ;
	\node[tree node,marked node] (n-10) at (22.000000,16.615385) {$28$} ;
	\node[preorder label,above=2pt of n-10] {$10$} ;
	\node[tree node,marked node] (n-11) at (14.000000,13.961538) {$1$} ;
	\node[preorder label,above=2pt of n-11] {$11$} ;
	\node[tree node] (n-12) at (12.000000,11.538462) {$9$} ;
	\node[preorder label,above=2pt of n-12] {$12$} ;
	\node[tree node] (n-13) at (11.000000,9.346154) {$11$} ;
	\node[preorder label,above=2pt of n-13] {$13$} ;
	\node[tree node] (n-14) at (13.000000,9.346154) {$12$} ;
	\node[preorder label,above=2pt of n-14] {$14$} ;
	\node[tree node,marked node] (n-15) at (14.000000,7.384615) {$14$} ;
	\node[preorder label,above=2pt of n-15] {$15$} ;
	\node[tree node] (n-16) at (17.000000,5.653846) {$13$} ;
	\node[preorder label,above=2pt of n-16] {$16$} ;
	\node[tree node] (n-17) at (15.000000,4.153846) {$16$} ;
	\node[preorder label,above=2pt of n-17] {$17$} ;
	\node[tree node] (n-18) at (16.000000,2.884615) {$19$} ;
	\node[preorder label,above=2pt of n-18] {$18$} ;
	\node[tree node] (n-19) at (18.000000,4.153846) {$17$} ;
	\node[preorder label,above=2pt of n-19] {$19$} ;
	\node[tree node] (n-20) at (19.000000,4.153846) {$18$} ;
	\node[preorder label,above=2pt of n-20] {$20$} ;
	\node[tree node] (n-21) at (20.000000,5.653846) {$15$} ;
	\node[preorder label,above=2pt of n-21] {$21$} ;
	\node[tree node] (n-22) at (21.000000,11.538462) {$10$} ;
	\node[preorder label,above=2pt of n-22] {$22$} ;
	\node[tree node] (n-23) at (25.000000,13.961538) {$2$} ;
	\node[preorder label,above=2pt of n-23] {$23$} ;
	\node[tree node] (n-24) at (23.000000,11.538462) {$6$} ;
	\node[preorder label,above=2pt of n-24] {$24$} ;
	\node[tree node] (n-25) at (24.000000,11.538462) {$7$} ;
	\node[preorder label,above=2pt of n-25] {$25$} ;
	\node[tree node] (n-26) at (26.000000,11.538462) {$8$} ;
	\node[preorder label,above=2pt of n-26] {$26$} ;
	\node[tree node] (n-27) at (27.000000,13.961538) {$3$} ;
	\node[preorder label,above=2pt of n-27] {$27$} ;
	\node[tree node,marked node] (n-28) at (30.000000,13.961538) {$4$} ;
	\node[preorder label,above=2pt of n-28] {$28$} ;
	\node[tree node] (n-29) at (31.000000,11.538462) {$5$} ;
	\node[preorder label,above=2pt of n-29] {$29$} ;
	\node[tree node] (n-30) at (29.000000,9.346154) {$20$} ;
	\node[preorder label,above=2pt of n-30] {$30$} ;
	\node[tree node] (n-31) at (30.000000,9.346154) {$21$} ;
	\node[preorder label,above=2pt of n-31] {$31$} ;
	\node[tree node] (n-32) at (32.000000,9.346154) {$22$} ;
	\node[preorder label,above=2pt of n-32] {$32$} ;
	\node[tree node] (n-33) at (33.000000,9.346154) {$23$} ;
	\node[preorder label,above=2pt of n-33] {$33$} ;
	\node[tree node] (n-34) at (34.000000,16.615385) {$34$} ;
	\node[preorder label,above=2pt of n-34] {$34$} ;

	\draw[tree edge] (n-0) to (n-34) ;
	\draw[tree edge] (n-0) to (n-10) ;
	\draw[tree edge] (n-0) to (n-7) ;
	\draw[tree edge] (n-0) to (n-2) ;
	\draw[tree edge] (n-0) to (n-1) ;
	\draw[tree edge] (n-2) to (n-6) ;
	\draw[tree edge] (n-2) to (n-5) ;
	\draw[tree edge] (n-2) to (n-3) ;
	\draw[tree edge] (n-3) to (n-4) ;
	\draw[tree edge] (n-7) to (n-8) ;
	\draw[tree edge] (n-8) to (n-9) ;
	\draw[tree edge] (n-10) to (n-28) ;
	\draw[tree edge] (n-10) to (n-27) ;
	\draw[tree edge] (n-10) to (n-23) ;
	\draw[tree edge] (n-10) to (n-11) ;
	\draw[tree edge] (n-11) to (n-22) ;
	\draw[tree edge] (n-11) to (n-12) ;
	\draw[tree edge] (n-12) to (n-14) ;
	\draw[tree edge] (n-12) to (n-13) ;
	\draw[tree edge] (n-14) to (n-15) ;
	\draw[tree edge] (n-15) to (n-21) ;
	\draw[tree edge] (n-15) to (n-16) ;
	\draw[tree edge] (n-16) to (n-20) ;
	\draw[tree edge] (n-16) to (n-19) ;
	\draw[tree edge] (n-16) to (n-17) ;
	\draw[tree edge] (n-17) to (n-18) ;
	\draw[tree edge] (n-23) to (n-26) ;
	\draw[tree edge] (n-23) to (n-25) ;
	\draw[tree edge] (n-23) to (n-24) ;
	\draw[tree edge] (n-28) to (n-29) ;
	\draw[tree edge] (n-29) to (n-33) ;
	\draw[tree edge] (n-29) to (n-32) ;
	\draw[tree edge] (n-29) to (n-31) ;
	\draw[tree edge] (n-29) to (n-30) ;

	\begin{pgfonlayer}{background}
		\draw[block edge 12] (n-0) to (n-34) ;
		\draw[block edge 11,crossing edge] (n-0) to (n-10) ;
		\draw[block edge 12] (n-0) to (n-7) ;
		\draw[block edge 11] (n-0) to (n-2) ;
		\draw[block edge 11] (n-0) to (n-1) ;
		\draw[block edge 11] (n-2) to (n-6) ;
		\draw[block edge 11] (n-2) to (n-5) ;
		\draw[block edge 11] (n-2) to (n-3) ;
		\draw[block edge 11] (n-3) to (n-4) ;
		\draw[block edge 12] (n-7) to (n-8) ;
		\draw[block edge 12] (n-8) to (n-9) ;
		\draw[block edge 21,crossing edge] (n-10) to (n-28) ;
		\draw[block edge 21] (n-10) to (n-27) ;
		\draw[block edge 21] (n-10) to (n-23) ;
		\draw[block edge 21,crossing edge] (n-10) to (n-11) ;
		\draw[block edge 12] (n-11) to (n-22) ;
		\draw[block edge 12] (n-11) to (n-12) ;
		\draw[block edge 12] (n-12) to (n-14) ;
		\draw[block edge 12] (n-12) to (n-13) ;
		\draw[block edge 12,external edge] (n-14) to (n-15) ;
		\draw[block edge 22] (n-15) to (n-21) ;
		\draw[block edge 21] (n-15) to (n-16) ;
		\draw[block edge 21] (n-16) to (n-20) ;
		\draw[block edge 21] (n-16) to (n-19) ;
		\draw[block edge 21] (n-16) to (n-17) ;
		\draw[block edge 21] (n-17) to (n-18) ;
		\draw[block edge 21] (n-23) to (n-26) ;
		\draw[block edge 21] (n-23) to (n-25) ;
		\draw[block edge 21] (n-23) to (n-24) ;
		\draw[block edge 11] (n-28) to (n-29) ;
		\draw[block edge 11] (n-29) to (n-33) ;
		\draw[block edge 11] (n-29) to (n-32) ;
		\draw[block edge 11] (n-29) to (n-31) ;
		\draw[block edge 11] (n-29) to (n-30) ;
		\node[block node 11] at (n-1) {} ;
		\node[block node 11] at (n-2) {} ;
		\node[block node 11] at (n-3) {} ;
		\node[block node 11] at (n-4) {} ;
		\node[block node 11] at (n-5) {} ;
		\node[block node 11] at (n-6) {} ;
		\node[block node 12] at (n-7) {} ;
		\node[block node 12] at (n-8) {} ;
		\node[block node 12] at (n-9) {} ;
		\node[block node 21] at (n-10) {} ;
		\node[block node 12] at (n-11) {} ;
		\node[block node 12] at (n-12) {} ;
		\node[block node 12] at (n-13) {} ;
		\node[block node 12] at (n-14) {} ;
		\node[block node 21] at (n-16) {} ;
		\node[block node 21] at (n-17) {} ;
		\node[block node 21] at (n-18) {} ;
		\node[block node 21] at (n-19) {} ;
		\node[block node 21] at (n-20) {} ;
		\node[block node 22] at (n-21) {} ;
		\node[block node 12] at (n-22) {} ;
		\node[block node 21] at (n-23) {} ;
		\node[block node 21] at (n-24) {} ;
		\node[block node 21] at (n-25) {} ;
		\node[block node 21] at (n-26) {} ;
		\node[block node 21] at (n-27) {} ;
		\node[block node 11] at (n-28) {} ;
		\node[block node 11] at (n-29) {} ;
		\node[block node 11] at (n-30) {} ;
		\node[block node 11] at (n-31) {} ;
		\node[block node 11] at (n-32) {} ;
		\node[block node 11] at (n-33) {} ;
		\node[block node 12] at (n-34) {} ;
	\end{pgfonlayer}

	\node[anchor=north west,scale=1.2] at (0.5,10) {
			\Large
			\begin{tabular}{rl}
				 Micro tree id & nodes (preorder ids) \\
				0	 & 0, 1, 2, 3, 4, 5, 6 \\
				1	 & 0, 7, 8, 9, 34 \\
				2	 & 10, 23, 24, 25, 26, 27 \\
				3	 & 11, 12, 13, 14, 22 \\
				4	 & 15, 16, 17, 18, 19, 20 \\
				5	 & 15, 21 \\
				6	 & 28, 29, 30, 31, 32, 33 \\
			\end{tabular}
		} ;

\end{tikzpicture}}
	\caption{%
		Example tree with $n=34$ nodes, partitioned using $B=6$.
	}
	\label{fig:farzan-munro-top-tier-example}
\end{figure}

\begin{figure}
	\resizebox{\linewidth}!{{%
		\def\svgwidth{2.4\linewidth}\sffamily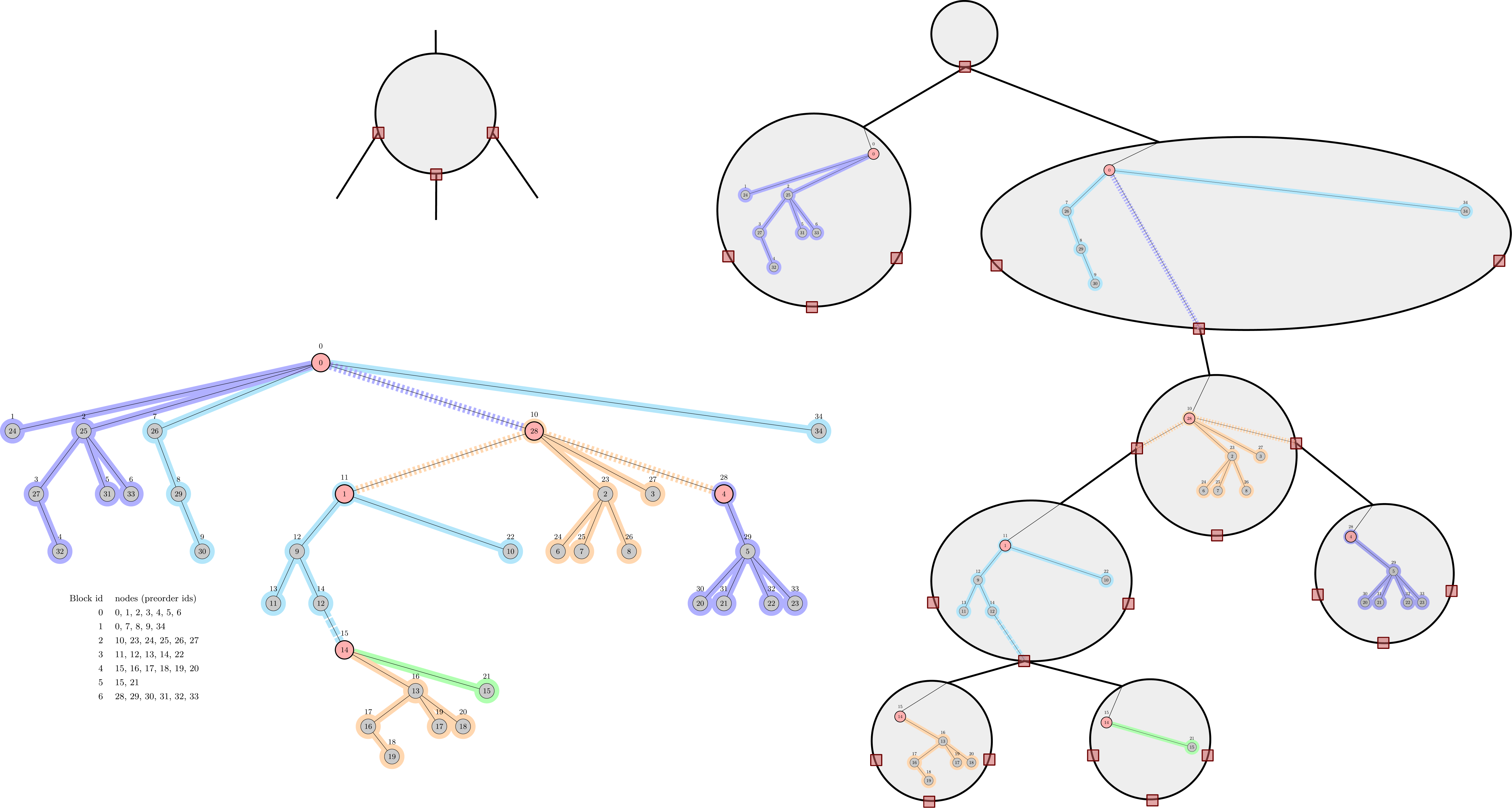%
	}}
	\caption{%
		The tree from \wref{fig:farzan-munro-top-tier-example} (left) and the top-tier 
		tree $\Upsilon$ (right) corresponding to the covering with the edge types.
		Edge types are also visualized through different exit points for 
		leftmost, rightmost, and external edges for illustration purposes.
	}
	\label{fig:farzan-munro-top-tier-tree}

\end{figure}
To be able to distinguish the different forms of interactions listed above,
additional information for parent-child edges in $\Upsilon$ is stored.
By construction, edges between micro trees always lead to the root 
of the child micro tree, but the other endpoint will have to be encoded.
We observe that there are the following types of edges between a parent 
micro tree $\mathsf{P}$ and its child~$\mathsf{C}$:
\begin{enumerate}[
	label={(\roman*)},
	leftmargin=2em,
]
\item \emph{new leftmost root child}\\
	The root of $\mathsf{C}$ is a child of the root of $\mathsf{P}$ and comes before all children
	of $\mathsf{P}$'s root that lie inside $\mathsf{P}$ in the left-to-right order of the children.
	Moreover, there is no other child component $\mathsf{C}'$ of $\mathsf{P}$ that shares the root with $\mathsf{C}$
	and comes before $\mathsf{C}$ in the child order.
\item \emph{continued leftmost root child}\\
	The root of $\mathsf{C}$ is a child of the root of $\mathsf{P}$ and comes before all children
	of $\mathsf{P}$'s root that lie inside $\mathsf{P}$ in the left-to-right order of the children, but it shares its root with the child component
	immediately before $\mathsf{C}$ in the child order.
\item \emph{new rightmost root child}\\
	The root of $\mathsf{C}$ is a child of the root of $\mathsf{P}$ and  
	$\mathsf{C}$'s root comes after all root children included in $\mathsf{P}$. Moreover, there is no other child component $\mathsf{C}'$ of $\mathsf{P}$ that shares the root with $\mathsf{C}$.
\item \emph{continued rightmost root child}\\
	The root of $\mathsf{C}$ is a child of the root of $\mathsf{P}$ and  
	$\mathsf{C}$'s root comes after all root children included in $\mathsf{P}$, but it shares its root with the child component
	immediately before $\mathsf{C}$ in the child order.
\item \emph{external-edge child}\\
	Any other edge. By construction, all external-edge child components of $\mathsf{P}$
	share a common root, 
	so there is no need to distinguish new and continued external edges.
	
	We note that path nodes can give rise to an external-edge child $\mathsf{C}$ whose
	root is a child of $\mathsf{P}$'s root. This happens only when we greedily pack across the
	gap left by the permanent component of a single heavy child.
	$\mathsf{P}$ cannot have another external edge then, so we are free to use $P$'s 
	external-edge ``slot'' to link to $\mathsf{C}$.
\end{enumerate}
The top tier is again an ordinal tree, 
$\Upsilon\in\mathfrak T_{m+1}$.
For the micro trees, we observe that because of their limited size, 
there are fewer different possible shapes of ordinal trees
than we have micro trees. The crucial idea of our hypersuccinct encoding is again to treat each shape of a micro tree as a letter in the alphabet $\Sigma_{\mu} \subseteq \bigcup_{s \le\mu} \mathfrak T_s$ of micro tree shapes and to compute a Huffman code $C: \Sigma_\mu \to \{0,1\}^\star$ based on the frequency of occurrences of micro tree shapes in the sequence $\mu_1,\ldots,\mu_m \in \Sigma_\mu^m$.
For our hypersuccinct code, we then use a \emph{length-restricted} version $\bar C: \Sigma_\mu \to \{0,1\}^\star$ obtained from $C$ using a  variant of the simple cutoff technique from \wref{def:worst-case-bounding} for ordinal trees (using the balanced parenthesis encoding for ordinal trees).
Furthermore, for each micro tree, we have to encode the portal 
for the external edges (if they exist) and the type of its parent edge (i)--(v).
For that, we store the micro-tree-local preorder rank of the node
and the child rank at which the external edges have to be inserted
using $\lceil \lg(\mu+1)\rceil $ bits each.

We can thus encode an ordinal tree $t\in\mathfrak T_n$ as follows: 
\begin{enumerate}
\item Store $n$ and $m$ in Elias gamma code,
\item followed by the balanced-parenthesis (BP) bitstring for $\Upsilon$.
\item Next comes an encoding for $\bar C$; for simplicity, we simply list all possible codewords 
and their corresponding ordinal trees by storing the size (in Elias-gamma code) followed by their BP sequence.
\item Then, we list the Huffman codes $\bar C(\mu_i)$ of all micro trees in DFS order (of $\Upsilon$).
\item Then, we store 2 $\lceil \lg (\mu+1) \rceil$-bit integers 
	to encode the portal of each 
	micro tree in DFS order (of $\Upsilon$).
\item Finally, we encode the type of the parent edge using $3$ bits 
	of each micro tree, again in DFS order.
\end{enumerate}
Altogether, this yields our \emph{hypersuccinct code} $\mathsf{H}: \mathfrak{T} \rightarrow \{0,1\}^\star$ for ordinal trees.
Decoding is possible by first recovering $n$, $m$, and $\Upsilon$ from the BP,
then reading the Huffman code. 
We then replace each node in $\Upsilon$ by its micro tree 
in a depth-first traversal. 
Herein, we use the information about edge types in $\Upsilon$
to correctly connect the micro trees: 
partitioning children into leftmost and rightmost root children places them 
in the appropriate order into the list of children of the parent component's root.
For type (ii) and (iv) children, we delete the component root and instead add its children
to the next type (i) resp.\ (iii) siblings component's root.
Finally, for type (v) children, we use the information about portals to find their place
in a node's child list, and for all but the leftmost of them, also merge their roots
with the left sibling component. With respect to the length of the hypersuccinct code, we find the following: 

\begin{lemma}[Hypersuccinct ordinal tree code]
\label{lem:hypersuccinct-code-ordinal}
	Let $t\in\mathfrak T_n$ be an ordinal tree of $n$ nodes, decomposed into
	micro trees $\mu_1,\ldots,\mu_m$ by the Farzan-Munro algorithm.
	Let $C$ be an ordinary  Huffman code for the string $\mu_1\ldots\mu_m$,
	the local shapes of the micro trees.
	Then, the hypersuccinct code encodes $t$ with a binary codeword of length
	\begin{align*}
			\hypsuc t
		&\wwrel\le
			\sum_{i=1}^m |C(\mu_i)| + \Oh\biggl(n \frac{\log \log n}{\log n}\biggr).
	\end{align*}
\end{lemma}

\begin{proof}

It is easy to check that all parts of the hypersuccinct ordinal-tree code
except Part 4 require $O(n \log \log n/\log n)$ bits of space.
Let $t\in\mathfrak T_n$.
The analysis of the number of bits needed to store parts 1--5 is identical to the binary-tree case:
Part 1 needs $\Oh(\log n)$ bits and
Part 2 requires $2m+2 = \Theta(n/\log n)$ bits.
For Part 3, observe that 
\[
	|\Sigma_\mu| 
	\wrel\le 
	\sum_{s \le \lceil\lg n /4\rceil} 4^{s} 
	\wrel< 
	\frac43 \cdot 4^{\lg n / 4+1} 
	\wrel= 
	\frac{16}{3} \sqrt n.
\]
With the worst-case cutoff technique (adapted to ordinal trees) from \wref{def:worst-case-bounding}, $\bar C(\mu_i) \le 1 + 2\mu \sim \frac12 \lg n$, so
we need asymptotically $\Oh(\sqrt n)$ entries / codewords in the table, 
each of size $\Oh(\mu) = \Oh(\log n)$,
for an overall table size of $\Oh(\sqrt n \log n)$.
Part 5 uses $m \cdot 2 \lceil \lg (\mu+1)\rceil = \Theta(\frac nB \log B) = \Theta(n \cdot \frac{\log \log n}{\log n}) = o(n)$
bits of space.
Part 6 uses $3m = \Theta(n/\log n)$ bits.
It remains to analyze Part 4, which is again similar to the binary-tree case: We note that by applying the worst-case pruning scheme of \wref{def:worst-case-bounding}, we waste $1$ bit per micro tree compared to a pure, non-restricted Huffman code. But the wasted bits amount to $m = O(n/\log n)$ bits in total:
\begin{align*}
		\sum_{i=1}^m\bar{C}(\mu_i) 
	&\wwrel= 
		\sum_{i=1}^m\min\{|C(\mu_i)|+1, 2|\mu_i|+2\lceil \lg |\mu_i|+1\rceil+2\}
\\	&\wwrel\leq 
		\sum_{i=1}^m\left(|C(\mu_i)|+1\right) 
\\	&\wwrel= 
		\sum_{i=1}^m|C(\mu_i)| + O(n/\log n).
\end{align*}
This finishes the proof.
\end{proof}

\begin{table}[tbh]
	\plaincenter{\fbox{%
	\begin{minipage}{.97\textwidth}\footnotesize
		\newcommand\opitem[2]{\footnotesize#1 &\footnotesize #2 \\[.2ex]}
		\begin{tabular}{p{9em}p{33em}}
		\opitem{$\TrParent(v)$}{the parent of $v$, same as $\TrLevAnc(v,1)$}
		\opitem{$\TrDeg(v)$}{the number of children of $v$}
		\opitem{$\TrChild(v,i)$}{the $i$th child of node $v$ ($i\in\{1,\ldots,\TrDeg(v)\}$)}
		\opitem{$\TrChildRank(v)$}{the number of siblings to the left of node $v$ plus $1$}
		\opitem{$\TrDepth(v)$}{the depth of $v$, \ie, the number of edges between the root and $v$}
		\opitem{$\TrLevAnc(v,i)$}{the ancestor of node $v$ at depth $\TrDepth(v)-i$} 
		\opitem{$\TrNbDesc(v)$}{the number of descendants of $v$}
		\opitem{$\TrHeight(v)$}{the height of the subtree rooted at node $v$}
		\opitem{$\LCA(v, u)$}{the lowest common ancestor of nodes $u$ and $v$}
		\opitem{$\TrLeftLeaf(v)$}{the leftmost leaf descendant of $v$}
		\opitem{$\TrRightLeaf(v)$}{the rightmost leaf descendant of $v$}
		\opitem{$\TrLevelLeft(\ell)$}{the leftmost node on level $\ell$}
		\opitem{$\TrLevelRight(\ell)$}{the rightmost node on level $\ell$}
		\opitem{$\TrLevelPred(v)$}{the node immediately to the left of $v$ on the same level}
		\opitem{$\TrLevelSucc(v)$}{the node immediately to the right of $v$ on the same level}
		\opitem{$\TrRank_{X}(v)$}{the position of $v$ in the $X$-order, $X \in \{\pre,\post,\inorder,\DFUDS\}$, \ie, in a preorder, postorder, inorder, {\DFUDS} order, or level-order traversal of the tree}
		\opitem{$\TrSelect_{X}(i)$}{the $i$th node in the $X$-order, $X \in \{\pre,\post,\inorder,\DFUDS\}$}
		\opitem{$\TrLeafRank(v)$}{the number of leaves before and including $v$ in preorder}
		\opitem{$\TrLeafSel(i)$}{the $i$th leaf in preorder}
	\end{tabular}
	\end{minipage}}}
	\smallskip
	\caption{%
		Navigational operations on succinct ordinal trees. 
		($v$ denotes a node and $i$ an integer).
	}
	\label{tab:operations}
\end{table}

As for binary trees, the representation of ordinal trees based on the hypersuccinct code can be turned into a data structure:

\begin{theorem}[Tree covering index for ordinal trees~\cite{FarzanMunro2014}]
	Let $t\in\mathfrak T_n$ denote an ordinal tree, decomposed into micro trees
	$\mu_1,\ldots,\mu_m$ with the tree covering algorithm.
	Assuming access to a data structure that maps $i$
	to $\mathit{BP}(\mu_i)$ in constant-time,
	there is a data structure occupying $o(n)$ additional bits
	of space that supports all operations from \wref{tab:operations}
	in constant time.
\end{theorem}

\section{Memoryless Ordinal Tree Sources}\label{sec:degreeentropy}

For an ordinal tree $t \in \mathfrak{T}$ and a node $v$ of $t$, let $\deg_t(v)$ denote the (out-)degree of $v$. We leave out the subscript $t$, if the tree $t$ is clear from the context. With $\nu_i^t$ we denote the number of nodes of degree $i$ of $t$. 
A \emph{degree distribution} $d = (d_i)_{i \in \N_0}$ is a sequence of non-negative real numbers, such that $\sum_{i =0}^{\infty}d_i = 1$. A degree distribution assigns a probability $\mathbb{P}[t]$ to an ordinal tree by
\begin{align}\label{eq:degree-probability-ordinal}
	\mathbb{P}[t] 
	\wwrel= 
	\prod_{v \in t}d_{\deg(v)} 
	\wwrel=
	\prod_{i=0}^{|t|}(d_{i})^{\nu_i^t}.
\end{align}
That is, a degree distribution $d$ can be used to randomly construct an ordinal tree as follows: In a top-down way, starting at the root node, we determine for each node its degree $i$: The probability that a node is of degree $i$ is given by $d_i$. If $i=0$, then this node becomes a leaf, otherwise we attach $i$ many children to the node and continue the process at these children. Note that this process might produce infinite trees with non-zero probability.
In order to obtain finite trees with non-zero probability, we assume that $d_0 >0$. 
In \cite{JanssonSadakaneSung2012}, the following notion of empirical entropy for trees was introduced:
\begin{definition}[Degree-entropy]\label{def:degree-entropy}
Let $t \in \mathfrak{T}$. The (unnormalized) \emph{degree-entropy} $H^{\deg}(t)$ of $t$ is the zeroth order entropy of the node degrees:
\begin{align*}
	H^{\deg}(t) \wwrel= \sum_{i=0}^{|t|}\nu_i^t\lg\left(\frac{|t|}{\nu_i^t}\right).
\end{align*}
\end{definition}
We say that a degree distribution is the \emph{empirical degree distribution} of an ordinal tree $t$, if $d_i = \nu_i^t/|t|$ for every index $0 \leq i \leq |t|$. In particular, if $d$ is the empirical degree distribution of an ordinal tree $t \in \mathfrak{T}$, we have
\begin{align*}
	\lg\left(\frac{1}{\mathbb{P}[t]}\right) 
	\wwrel= 
	\sum_{i=0}^{|t|}\nu_i^t\lg \left(\frac{1}{d_i}\right)
	\wwrel= 
	\sum_{i=0}^{|t|}\nu_i^t\lg \left(\frac{|t|}{\nu_i^t}\right) 
	\wwrel= 
	H^{\deg}(t).
P\end{align*}

\begin{example}[Full $m$-ary trees]\label{exm:fullkary}
Probability distributions over full $m$-ary trees, \ie, trees where each node has either exactly $m$ or $0$ children, are obtained from degree distributions $(d_i)_{i \in \mathbb{N}_0}$ with $d_0, d_m>0$ and $d_i =0$ for $i \neq m,0$. It is easy to see that a full $m$-ary tree $t$ with $\nu_m^t$ many inner nodes (of degree $m$) always consists of $\nu_0^t = (m-1)\nu_m^t+1$ many leaves, and is thus always of size $m\nu_m^t+1$. The number of full $m$-ary trees of size $n=m\nu+1$, for $\nu \in \mathbb{N}$, is given by \cite{FlajoletSedgewick2009}:
\begin{align}\label{eq:numberfullmary}
\frac{1}{m\nu+1}\binom{m\nu+1}{\nu}.
\end{align}
Let $d$ be the degree distribution with $d_0=1/m$ and $d_m=(m-1)/m$. We have
\begin{align*}
	\lg\left(\frac{1}{\mathbb{P}[t]}\right)
\wwrel=
	\nu\lg\left(m\right)
	+\bigl((m-1)\nu+1\bigr)\lg\left(\frac{m}{m-1}\right)
\end{align*}
for every full $m$-ary tree $t$ of size $m\nu+1$, which is asymptotically, by \eqref{eq:numberfullmary}, the minimum number of bits needed to represent a full $m$-ary tree of size $m\nu+1$.
\end{example}

Given a degree distribution $d$, Equation \eqref{eq:degree-probability-ordinal} suggests a route for an encoding that encodes an ordinal tree $t \in \mathfrak{T}$ with $\mathbb{P}[t]>0$ in $\lg(1/\mathbb{P}[t])$ (plus lower-order terms) many bits: Such an encoding may spend $\lg(1/d_i)$ many bits per node $v$ of $t$ of degree $\deg(v) = i$. Assuming that the degree distribution is known (and need not be stored as part of the encoding), we can use \emph{arithmetic coding} to encode the degree of node $v$ in that many bits: However, $d$ can possibly consist of countably many positive coefficients, thus, we have to adapt the process of arithmetic coding slightly: In order to encode the degree $\deg(v) \in \mathbb{N}_0$ of a node $v$, we consider $\deg(v)$ as a unary string $s=0^{\deg(v)}1$, which we encode using arithmetic coding as follows: In order to encode the $k$th symbol of $s$, we feed the arithmetic coder with the model that the next symbol is a number $s[k] \in \{0,1\}$, the probability for $s[k]=1$ being $d_{k-1}/(d_{k-1}+d_{k}+d_{k+1}+\dots)$. Thus, arithmetic coding uses
\begin{align*}
&\sum_{k=0}^{\deg(v)-1}\lg \left(\left(1-\frac{d_k}{\sum_{i \geq k}d_i}\right)^{-1}\right)+\lg \left(\frac{\sum_{i \geq \deg(v)}d_i}{d_{\deg(v)}}\right)\\
& \wwrel= \sum_{k=0}^{\deg(v)-1}\left(\lg \left(\sum_{i \geq k} d_i\right)-\lg \left(\sum_{i \geq k+1} d_i\right)\right)+ \lg \left(\sum_{i \geq \deg(v)}d_i\right)+\lg \left(\frac{1}{d_{\deg(v)}}\right) 
\\&\wwrel=
\lg \left(\frac{1}{d_{\deg(v)}}\right)
\end{align*}
many bits to encode $s=0^{\deg(v)}1$.
An encoding $D_d$, dependent of a given degree-distribution $d$, stores a tree $t$ as follows: While traversing the tree in depth-first order, we encode the degree $\deg(v)$ of each node $v$, using arithmetic encoding as described above. We can reconstruct the tree $t$ recursively from its code $D_d(t)$, as we always know the degrees of the nodes we have already visited in the depth-first order traversal of the tree. As arithmetic encoding needs $\lg(1/d_{\deg(v)})$ bits per node $v$, plus at most $2$ bits of overhead, the total number of bits needed in order to store an ordinal tree $t \in \mathfrak{T}$ with $\mathbb{P}[t]>0$ is thus
\begin{align*}
|D_d(t)| \wwrel\leq \sum_{v \in t} \lg \left(\frac{1}{d_{\deg(v)}}\right) +2. 
\end{align*}
If a degree distribution $d$ is the empirical degree distribution of an ordinal tree $t$, \ie, $d_{i}=\nu_i^t/|t|$ for every $i \in [t]$, we find in particular:
\begin{align*}
|D_d(t)| \wwrel\leq \sum_{i=0}^{|t|}\nu_i^t\lg\left(\frac{|t|}{\nu_i^t}\right)+2 \wwrel= H^{\deg}(t) +2.
\end{align*}

The encoding $D^d$ yields a prefix-free code for the set of ordinal trees which satisfy $\mathbb{P}[t]>0$ with respect to the degree distribution $d$. In order to show that our hypersuccinct code is universal with respect to degree-distribution sources, we start with the following lemma:
\begin{lemma}[Micro tree code bound]
\label{lem:degreeordinal}
Let $d$ be a degree distribution and let $t \in \mathfrak{T}_n$ be an ordinal tree of size $n$ with $\mathbb{P}[t]>0$. Then
\begin{align*}
\sum_{i=1}^m |C(\mu_i)| \leq \lg\left(\frac{1}{\mathbb{P}[t]}\right) + O\left(\frac{n \log \log n}{\log n}\right)
\end{align*}
where $C$ is a Huffman code for the sequence of micro trees $\mu_1,\ldots,\mu_m$ from our tree covering scheme (see \wref{sec:hypersuccinct-code-ordinal}).
\end{lemma}
\begin{proof}
Recall that the micro trees $\mu_1, \dots, \mu_m$ from our tree partitioning scheme for ordinal trees are pairwise disjoint except for (potentially) sharing a common subtree root and that apart from edges leaving the subtree root, at most one other edge leads to a node outside of the subtree (\wref{lem:tree-decomposition-ordinal}). Thus, there are at most two nodes in each micro tree $\mu_i$, whose degree in $\mu_i$ might not coincide with their degree in $t$: The root of $\mu_i$, which we denote with $\rho_i$, and a node $\pi_i \neq \rho_i$. In particular, for every node $v \neq \pi_i, \rho_i$ of $\mu_i$, we have
$\deg_{\mu_i}(v) = \deg_t(v)$.
Let $\pos(\pi_i)$ denote the depth-first order position of $\pi_i$ in $\mu_i$. With $D_d(\mu_i \setminus \rho_i)$ (respectively, $D_d(\mu_i \setminus \rho_i, \pi_i)$), we denote the following modification of $D_d$: While traversing the tree $\mu_i$ in depth-first order, we encode the degree $\deg_{\mu_i}(v)$ of each node $v$ of $\mu_i$, using arithmetic coding as in the encoding $D_d$, except that we skip the root $\rho_i$ of $\mu_i$ (respectively, we skip the root $\rho_i$ of $\mu_i$ and the node $\pi_i \neq \rho_i$ in $\mu_i$ from which an edge to a node outside of $\mu_i$ emerges). This is well-defined: We have $d_{\deg(v)}>0$ for every node $v \neq \rho_i, \pi_i$ of $\mu_i$ whose degree we encode, as its degree in $\mu_i$ coincides with its degree in $t$ and as $\mathbb{P}[t]>0$. If we know  $\deg_{\mu_i}(\rho_i)$, respectively, $\deg_{\mu_i}(\rho_i)$, $\deg_{\mu_i}(\pi_i)$ and $\pos(\pi_i)$, we are able to recover $\mu_i$ from $D_d(\mu_i \setminus \rho_i)$, respectively, $D_d(\mu_i \setminus \rho_i, \pi_i)$. Let $\mathcal{I}_0$ denote the set of indexes $i \in [m]$ for which $\mu_i$ does not contain a node other than (possibly) the root node from which an edge to a node outside of $\mu_i$ emerges, and let $\mathcal{I}_1 = [m] \setminus \mathcal{I}_0$. We define the following modified encoding:
\begin{align*}
\tilde{D}_d(\mu_i) = \begin{cases}
0\cdot \gamma(\deg_{\mu_i}(\rho_i))\cdot D_d(\mu_i \setminus \rho_i) \quad &\text{if } i \in \mathcal{I}_0,\\
1 \cdot \gamma(\deg_{\mu_i}(\rho_i))\cdot \gamma(\deg_{\mu_i}(\pi_i)+1) \cdot \gamma(\pos(\pi_i))\cdot D_d(\mu_i \setminus \rho_i, \pi_i) & \text{otherwise}.
\end{cases}
\end{align*}
Note that formally, $\tilde{D}_d$ is \emph{not} a prefix-free code over $\Sigma_{\mu}$, as there can be micro tree shapes that are assigned \emph{several} codewords by $\tilde{D}_d$. But $\tilde{D}_d$ can again be seen as a \emph{generalized prefix-free code}, where more than one codeword per symbol is allowed, as $\tilde{D}_d$ is uniquely decodable to local shapes of micro trees. Thus, as a Huffman code minimizes the encoding length over the class of \emph{generalized} prefix-free codes, we find:
\begin{align*}
\sum_{i=1}^m|C(\mu_i)| &\leq \sum_{i=1}^m|\tilde{D}_d(\mu_i)| =\sum_{i \in \mathcal{I}_0}|\tilde{D}_d(\mu_i)|+\sum_{i \in \mathcal{I}_1}|\tilde{D}_d(\mu_i)|\\
&\leq \sum_{i \in \mathcal{I}_0} \left(|D_d(\mu_i \setminus \rho_i)| + 2\lg \mu + 2\right) + \sum_{i \in \mathcal{I}_1}\left(|D_d(\mu_i \setminus \rho_i, \pi_i)| + 6\lg \mu + 4\right),
\end{align*}
as $\deg_{\mu_i}(\rho_i), \deg_{\mu_i}(\pi_i)+1, \pos(\pi_i) \leq \mu$. By definition of $|D_d(\mu_i \setminus \rho_i)|$ and $|D_d(\mu_i \setminus \rho_i, \pi_i)|$, and as $|\mathcal{I}_0|+ |\mathcal{I}_1|= m$, this is upper-bounded by
\begin{align*}
\sum_{i \in \mathcal{I}_0} \sum_{v \in \mu_i \atop v \neq \rho_i}\lg\left(\frac{1}{d_{\deg_{\mu_i}(v)}}\right) +\sum_{i \in \mathcal{I}_1} \sum_{v \in \mu_i \atop v \neq \rho_i, \pi_i}\lg\left(\frac{1}{d_{\deg_{\mu_i}(v)}}\right)+6m\lg \mu + 6 m.
\end{align*}
As every node $v$ of $t$ which is not the root node of a micro tree $\mu_i$ is contained in at most one subtree $\mu_i$ and as $\deg_{\mu_i}(v) = \deg_{t}(v)$ for every node $v \neq \pi_i, \rho_i$, we have
\begin{align*}
\sum_{i=1}^m|C(\mu_i)| \leq \sum_{v \in t}\lg \left(\frac{1}{d_{\deg_t(v)}}\right)+6m \lg \mu + 6m = \lg \left(\frac{1}{\mathbb{P}[t]}\right) + O\left(\frac{n \log \log n}{\log n}\right),
\end{align*}
as $m = \Theta(n/\log n)$ and $\mu = \Theta(\log n)$ (see \wref{sec:hypersuccinct-code-ordinal}).
This finishes the proof.
\end{proof}

\begin{theorem}[Universality for degree distribution]
\label{thm:degreeordinal}
Let $d$ be a degree distribution. The hypersuccinct code $\mathsf{H}: \mathfrak{T} \rightarrow \{0,1\}^\star$ satisfies
\begin{align*}
|\mathsf{H}(t)| \leq \lg\left(\frac{1}{\mathbb{P}[t]}\right) + O\left(\frac{n \log \log n}{\log n}\right)
\end{align*}
for every $t \in \mathfrak{T}_n$ with $\mathbb{P}[t]>0$.
In particular, if $d$ coincides with the empirical degree distribution of $t$, we have
\begin{align*}
|\mathsf{H}(t)| \leq H^{\deg}(t)+ O\left(\frac{n \log \log n}{\log n}\right).
\end{align*}
\end{theorem}
 follows from \wref{lem:degreeordinal} and \wref{lem:hypersuccinct-code-ordinal}.

In particular, for full $m$-ary trees from \wref{exm:fullkary}, we obtain the following corollary from \wref{thm:degreeordinal}:

\begin{examplebox}
\begin{corollary}\label{cor:fullmary}
The hypersuccinct code $\mathsf{H}: \mathfrak{T} \to \{0,1\}^\star$ optimally compresses encodes \textbf{\boldmath full $m$-ary trees} $t$ of size $n=m\nu+1$, drawn uniformly at random from the set of all full $m$-ary trees of size $n$, using $|\mathsf{H}(t)|\leq \nu\lg(m) +(m-1)\nu\lg(m/(m-1)) + O(n \log\log n/\log n)$ many bits.
\end{corollary}
\end{examplebox}

\section{Fixed-Size Ordinal Tree Sources}\label{sec:fixedsizeordinal}

For ordinal trees, we can define fixed-size sources in a similar way as for binary trees;
such a source is characterized by a function $p:\N^+ \to [0,1]$ with
\[
	\quad\sum_{\mathclap{\substack{k\in \N\\n_1,\ldots,n_k \in\N\\[.2ex]n_1+\cdots+n_k = n-1}}}		\;
		p(n_1,\ldots,n_k)
	\wwrel=
	1
\]
for all $n\in\N$. 
The function $p$ assigns a probability to each possible grouping of the $n-1$ descendants
of an $n$-node ordinal tree into subtrees of the root.
Note that the choice of subtree sizes of the root is equivalent to choosing a
\emph{composition} of $n-1$ into strictly positive summands;
there are $2^{n-2}$ of these compositions
(between each consecutive pair of $n-1$ dots, we can either place a barrier or not)~--
a lot more than the $n$ choices for binary trees.

\subsection{Monotonic Fixed-Size Sources}\label{sec:monotonic-ordinal}

\begin{definition}[Monotonic source]
\label{def:monotonic-ordinal-tree-source}
	A fixed-size ordinal-tree source $\mathfrak S_{\mathit{fs}}(p)$ 
	is called \thmemph{monotonic} if $p$ is 
	\begin{thmenumerate}{def:monotonic-ordinal-tree-source}
	\item weakly decreasing in every component,
	\[
			p(n_1,\ldots,n_{i-1},n_i,n_{i+1},\ldots,n_k)
		\wwrel \ge
			p(n_1,\ldots,n_{i-1},n_i+1,n_{i+1},\ldots,n_k),
	\]
	\item 
	weakly decreasing upon adding new subtrees,
	\[
			p(n_1,\ldots,n_i,n_{i+1},\ldots,n_k)
		\wwrel \ge
			p(n_1,\ldots,n_i,1,n_{i+1},\ldots,n_k),
	\]
	\item 
	and sub-multiplicative
	\[
			p(n_1,\ldots,n_i,n_{i+1},\ldots,n_k)
		\wwrel \le
			p(n_1,\ldots,n_i) \cdot p(n_{i+1},\ldots,n_k).
	\]
	\end{thmenumerate}
\end{definition}

The sub-multiplicativity allows us to handle shared roots in micro
trees.

\begin{example}[Uniform composition trees]\label{exm:uniformcomposition}
	A simple example of a monotonic fixed-size ordinal-tree source is obtained
	by setting
	\[p(n_1,\ldots,n_k) \wwrel= \frac1{2^{n_1+\cdots+n_k-2}} \wwrel= 2^{-(n-2)}.\]
	In a sense, this is the analog of random BSTs (\wref{exm:bst}) in the world of ordinal trees.
	The distribution is very skewed to wide and short trees.
\end{example}

\begin{example}[Random LRM-trees / Uniform random recursive trees]
\label{exm:lrm}
~\\
	Let $p(n_1,\ldots,n_k) = \prod_{j=1}^k \frac1{n_1+\cdots+n_j}$.
	It is easy to check that 
	$\mathfrak S_{\mathit{fs}}(p)$
	is a monotonic ordinal-tree source.
	Trees with this distribution arise in several interesing ways.
	\begin{itemize}
	\item
		They are the shape of \textbf{LRM-trees}~\cite{BarbayFischerNavarro2012} 
		built on a random permutation; here, the children of the root
		are the indices of left-to-right minima (records) in the permutation,
		and the subtree is constructed recursively from the subpermutation
		following a left-to-right-minimum up to (excluding) the next one.
	\item
		They are also the shapes of (plane/ordered) \textbf{random recursive trees}
		which are grown inductively: when the $i$th node is added,
		it selects its parent uniformly among the $i-1$ existing nodes
		and becomes that node's leftmost child.
		This process is also called uniform attachment.
	\item
		The distribution is also obtained by applying the FCNS mapping
		to random BSTs; hence \wref{lem:fcns-sources-monotonic} below
		provides another proof of monotonicity.
	\end{itemize}
\end{example}

Let $\mathcal S_{\mathit{fs}}(p)$ be a fixed-size binary-tree source. The first-child next-sibling encoding $\FCNS: \mathfrak{T} \to \mathcal{T}$, defined in \wref{def:fcns}, transforms an ordinal tree $t \in \mathfrak{T}_n$ into a binary tree $\FCNS(t) \in \mathcal{T}_n$. However, this mapping is not surjective onto $\mathcal{T}_n$: As the root node of an ordinal tree $t \in \mathfrak{T}_n$ does not have a next sibling, we find that the left subtree of $\FCNS(t) \in \mathcal{T}_n$ is always of size $n-1$, whereas the right subtree is empty. In particular, $\mathcal S_{\mathit{fs}}(p)$ is not a probability distribution on $\FCNS(\mathfrak{T}_n)$. Thus, for a given fixed-size binary-tree source $\mathcal S_{\mathit{fs}}(p)$, we define
\begin{align*}
\tilde{\mathbb{P}}_{\mathcal{S}}[t] 
\wwrel=
\prod_{v \in t \atop v \neq \rho}p(|t_{\ell}[v]|,|t_r[v]|),
\end{align*}
for binary trees $t \in \FCNS(\mathcal{T}_n)$, where the product ranges over all nodes $v$ of $t$ except for the root node $\rho$. We then find that $\tilde{\mathbb{P}}_{\mathcal{S}}: \FCNS(\mathfrak{T}_n) \to [0,1]$ is a probability distribution.
Moreover, we define:

\begin{definition}[FCNS source]\label{def:fcnssource}
	Let $\mathcal S$ be a fixed-size binary-tree source.
	By $\mathfrak S_{\mathit{fcns}}(\mathcal S)$ we denote the ordinal tree source
	that yields $\ProbIn{\mathfrak S}{t} = \tilde{\mathbb{P}}_{\mathcal S}[\FCNS(t)]$ 
	for every $t\in\mathfrak T_n$.
\end{definition}

That is, in order to generate a random tree in $\mathfrak T_n$, 
we can let $\mathcal S$ generate a binary tree $t'\in\mathcal T_{n-1}$ with probability $\mathbb{P}_{\mathcal{S}}[t']$, then add a new root node to $t'$ in order to obtain a tree $t''$, such that $t'$ is the left subtree of $t''$, and compute 
$t = \FCNS^{-1}(t')\in\mathfrak T_n$. We find that $\mathbb{P}_{\mathcal{S}}[t']=\tilde{\mathbb{P}}_{\mathcal{S}}[t'']$.

\begin{lemma}[FCNS preserves monotonicity]
\label{lem:fcns-sources-monotonic}
	Let $\mathcal S_{\mathit{fs}}(p)$ be a monotonic fixed-size \thmemph{binary} 
	tree source.
	Then, $\mathfrak S_{\mathit{fcns}}(\mathcal S_{\mathit{fs}}(p))$ 
	is a \thmemph{monotonic} fixed-size ordinal-tree source.
\end{lemma}
\begin{proof}
We show that $\mathfrak S_{\mathit{fcsn}}(\mathcal S_{\mathit{fs}}(p))$ can be written as $\mathfrak S_{\mathit{fs}}(p')$ for
a $p'$ that fulfills the conditons of \wref{def:monotonic-ordinal-tree-source}.
By definition of $\FCNS$, we have
\begin{align*}
		p'(n_1,\ldots,n_k)
	&\wwrel=
		p(n_1-1,n_2+\cdots+n_k)\cdot p(n_2-1,n_3+\cdots+n_k)
		\cdot \cdots 
	\\*	&\wwrel\ppe{}\cdot{}
	p(n_{k-1}-1,n_k) \cdot p(n_k-1,0).
\end{align*}
The monotonicity conditions follow by directly from monotonicity of $p$.
\end{proof}

\begin{lemma}[monotonicity implies submultiplicativity]
\label{lem:bi-monotonic-ordinal}
	Let $\mathfrak S_{\mathit{fs}}(p)$ be monotonic and
	$t \in \mathfrak{T}$ be decomposed into micro trees $\mu_1, \dots, \mu_m$. 
	Then
	\(
			\Prob t
		\wwrel\le 
			\prod_{i=1}^m \Prob{\mu_i}.
	\)
\end{lemma}
\begin{proof}
Let $v$ be a node of $t$ with children $u_1,\ldots,u_k$ 
and let $\mu_i$ be a micro tree that $v$ belongs to.
As $\mu_i$ is a subtree of $t$, we find $|{\mu_i}[u_j]| \leq |t[u_j]|$.
Note that $\mu_i$ might contain only some of the nodes $u_j$; if a node $u_j$ does not belong to $\mu_i$, we define
$\mu_i[u_j] = \Lambda$ and hence $|\mu_i[u_j]| = 0$.
There are 3 cases for $v$:
\begin{enumerate}
\item $v$ occurs in only one micro tree $\mu_i$.\\
	Then, its contribution to $\Prob t$ satisfies 
	$p(|t[u_1]|,\ldots,|t[u_k]|) \le p(|\mu_i[u_1],\ldots,|\mu_i[u_k]|)$
	by monotonicity of the source.
\item $v$ is a branching node.\\
	Assume $u_1,\ldots,u_k$ are spread over $s$ micro trees $\mu_{i_1},\ldots,\mu_{i_s}$
	that also contain $v$.
	Then, these micro trees each contain an interval of children (\wref{fact:farzan-munro-branching-node}), \ie,
	there are indices $1\le l_1\le r_1\le l_2\le r_2\le \cdots\le l_s\le r_s \le k$ 
	so that 
	$\mu_{i_j}$ contains $u_{l_j},\ldots,u_{r_j}$.
	By monotonicity and since $p(\cdot) \le 1$, we have
	\begin{align*}
			p(|t[u_1]|,\ldots,|t[u_k]|)
		&\wwrel\le
			\prod_{j=1}^s p(|t[u_{l_j}]|,\ldots,|t[u_{r_j}]|)
	\\	&\wwrel\le
			\prod_{j=1}^s p(|\mu_{i_j}[u_{l_j}]|,\ldots,|\mu_{i_j}[u_{r_j}]|).
	\end{align*}
\item $v$ is a path node.\\
	As above, $u_1,\ldots,u_k$ will be spread over $s$ 
	micro trees $\mu_{i_1},\ldots,\mu_{i_s}$ that also contain $v$,
	but one of them, $\mu_{i_h}$ can be missing a child from its interval
	(\wref{fact:farzan-munro-path-node}).
	With indices as above, $\mu_{i_j}$, $j\ne h$, contains $u_{l_j},\ldots,u_{r_j}$,
	and $\mu_{i_h}$ contains $u_{l_h},\ldots,u_{q-1},u_{q+1},\ldots,u_{r_h}$ 
	for a $q\in[k]$.
	We obtain by monotonicity
	\begin{align*}
	&\wwrel\ppe
			p(|t[u_{l_h}]|,\ldots,|t[u_{q-1}]|,|t[u_{q}]|,|t[u_{q+1}]|,\ldots, |t[u_{r_h}]|) 
	\\	&\wwrel\le
			p(|t[u_{l_h}]|,\ldots,|t[u_{q-1}]|,\like{|t[u_{q}]|}{1},|t[u_{q+1}]|,\ldots ,|t[u_{r_h}]|)
	\\	&\wwrel\le
			p(|t[u_{l_h}]|,\ldots,|t[u_{q-1}]|,\like{|t[u_{q}]|,{}}{}|t[u_{q+1}]|,\ldots, |t[u_{r_h}]|);
	\end{align*}
	and hence
	\begin{align*}
			p(|t[u_1]|,\ldots,|t[u_k]|)
		&\wwrel\le
			\prod_{j=1}^s 
				p(|t[u_{l_j}]|,\ldots,|t[u_{r_j}]|)
	\\	&\wwrel\le
			p(|t[u_{l_h}]|,\ldots,|t[u_{q-1}]|,|t[u_{q+1}]|,\ldots, |t[u_{r_h}]|)
		\cdot\\ &\wwrel\ppe
			\prod_{\substack{j=1,\ldots,s\\j\ne h}} 
				p(|t[u_{l_j}]|,\ldots,|t[u_{r_j}]|)
	\\	&\wwrel\le
	p(|\mu_{i_h}[u_{l_h}]|,\ldots,|\mu_{i_h}[u_{q-1}]|,|\mu_{i_h}[u_{q+1}]|,\ldots, |\mu_{i_h}[u_{r_h}]|)
			\cdot\\ &\wwrel\ppe
			\prod_{\substack{j=1,\ldots,s\\j\ne h}} 
				p(|\mu_{i_j}[u_{l_j}]|,\ldots,|\mu_{i_j}[u_{r_j}]|).
	\end{align*}
\end{enumerate}
In all three cases we could bound the contribution of $v$ to 
$\Prob t$ by the product of its contributions to the micro trees
it belongs to.
Therefore we find
\begin{align*}
		\Prob{t} 
	&\wwrel= 
		\prod_{v \in t} p(|t_1[v]|,\ldots,|t_{\deg_t(v)}[v]|) 
\\	&\wwrel\le 
		\prod_{i=1}^m \prod_{v \in \mu_i}p(|(\mu_i)_{1}[v]|,\ldots, |(\mu_i)_{\deg_{\mu_i}(v)}[v]|) 
\\	&\wwrel= 
		\prod_{i=1}^m \mathbb{P}[\mu_i].
\end{align*}
\end{proof}

\subsubsection{Universality of Monotonic Fixed-Size Ordinal Tree Sources}
In order to show universality of our hypersuccinct code for ordinal trees from \wref{sec:hypersuccinct-code-ordinal} with respect to fixed-size ordinal tree sources, we start again with a source-specific encoding for ordinal trees: As for binary trees, we define a depth-first order arithmetic code $D_p$ for ordinal trees, dependent on a given ordinal tree source $\mathfrak{S}_{\mathit{fs}}(p)$. 
Let $t \in \mathfrak{T}$ denote an ordinal tree with $\mathbb{P}[t]>0$. 
Assuming that the fixed-size source $p$ need not be stored as part of the encoding, we again make use of arithmetic coding in order to store $t$'s subtree sizes: Recall that the function $p$ assigns a probability to each possible grouping of the $n-1$ descendants of a tree of size $n$ into subtrees, and that there are $2^{n-2}$ many choices for these groupings: the compositions of $n-1$ into positive integers. Fix an enumeration of these compositions for every $n$, such that if we know $n$, every number $\ell \in \{1, \dots ,2^{n-2}\}$ represents one of these possible groupings.

The depth-first arithmetic code $D_p$ now stores an ordinal tree $t$ as follows:
We initially encode the size of the tree in Elias gamma code: If the tree consists of $n$ nodes, we store the Elias gamma code of $n+1$, $\gamma(n+1)$, in order to take the case into account that $t$ is the empty binary tree. Additionally, while traversing the tree in depth-first order, we encode the grouping of the $|t[v]|-1$ many descendants of $v$ into subtrees for every node $v$ using arithmetic coding: To encode these subtree sizes, we feed the arithmetic coder with the model that the next symbol is a number $\ell \in \{0, \dots, 2^{|t[v]|-1}\}$, representing a composition $(|t[v_1]|, \dots, |t[v_k]|)$ of $|t[v]|-1$ by our fixed enumeration of all compositions of $|t[v]|-1$, with probability $p(|t[v_1]|, \dots, |t[v_k]|)$.
We can reconstruct the tree $t$ recursively from its code $D_p(t)$, as we always know the subtree size of the current node.
This yields an encoding $D_p$ which stores an ordinal tree $t$ with $\mathbb{P}[t]>0$ in
\begin{align} \label{eq:depth-first-ordinal}
|D_p(t)| 
\wwrel\leq 
\lg \left(\frac{1}{\mathbb{P}[t]}\right) + 2 \lfloor\lg(|t|+1) \rfloor+3
\end{align}
many bits.

\begin{lemma}[micro tree code]
\label{lem:fixed-sizemonotonicordinal}
Let $\mathfrak{S}_{\mathit{fs}}(p)$ be a fixed-size tree source and let $t \in \mathfrak{T}_n$ with $\mathbb{P}[t]>0$. If $\mathfrak{S}_{\mathit{fs}}(p)$ is monotonic, then
\begin{align*}
\sum_{i=1}^m |C(\mu_i)| \wwrel\leq \lg\left(\frac{1}{\mathbb{P}[t]}\right) + O\left(\frac{n \log \log n}{\log n}\right),
\end{align*}
where $C$ is a Huffman code for the sequence of micro trees $\mu_1, \dots, \mu_m$ from our tree covering scheme (see \wref{sec:hypersuccinct-code-ordinal}).
\end{lemma}
\begin{proof}
As $\mathfrak{S}_{\mathit{fs}}(p)$ is monotonic, we have $0<\mathbb{P}[t]\leq \mathbb{P}[\mu_i]$ by \wref{lem:bi-monotonic-ordinal} for every $i \in [m]$: Thus, $|D_p(\mu_i)|$ is well-defined for every micro tree $\mu_i$.
By optimality of Huffman codes, we find that
\begin{align*}
\sum_{i=1}^m |C(\mu_i)| \leq \sum_{i=1}^m |D_p(\mu_i)|,
\end{align*}
where $D_p$ is the depth-first arithmetic code for ordinal tree sources.
By our estimate \eqref{eq:depth-first-ordinal} for $|D_p|$, we find that
\begin{align*}
		\sum_{i=1}^m  |D_p(\mu_i)| &\leq \sum_{i=1}^m \left(\lg\left(\frac{1}{\mathbb{P}[\mu_i]}\right)+3+2\lfloor \lg (|\mu_i| +1)\rfloor\right)\\
		& \leq \sum_{i=1}^m \lg\left(\frac{1}{\mathbb{P}[\mu_i]}\right) + O(m\log \mu).
\end{align*}
As $\mathfrak{S}_{\mathit{fs}}(p)$ is monotonic, we find by \wref{lem:bi-monotonic-ordinal}:
\begin{align*}
\sum_{i=1}^m \lg\left(\frac{1}{\mathbb{P}[\mu_i]}\right) + O(m\log \mu) &\leq 
\lg\left(\frac{1}{\mathbb{P}[t]}\right) + O(m \log \mu).
\end{align*}
Altogether, with $m = \Theta(n/\log n) $ and $\mu = \Theta(\log n)$ (see \wref{sec:hypersuccinct-code-ordinal}), we thus obtain
\begin{align*}
\sum_{i=1}^m |C(\mu_i)| \leq \lg\left(\frac{1}{\mathbb{P}[t]}\right) + O\left(\frac{n \log \log n}{\log n}\right).
\end{align*}
\end{proof}

From \wref{lem:fixed-sizemonotonicordinal} and \wref{lem:hypersuccinct-code-ordinal}, we find the following:
\begin{theorem}[Universality for monotonic sources]
\label{thm:monotonicordinal}
Let $\mathfrak{S}_{\mathit{fs}}(p)$ be a monotonic fixed-size tree source. The hypersuccinct code $\mathsf{H}: \mathfrak{T} \to [0,1]$ satisfies
\begin{align*}
|\mathsf{H}(t)| \leq \lg \left(\frac{1}{\mathbb{P}[t]}\right) + O\left(\frac{n \log \log n}{\log n}\right)
\end{align*}
for every $t \in \mathfrak{T}_n$ with $\mathbb{P}[t]>0$.
\end{theorem}

As the ordinal tree sources from \wref{exm:uniformcomposition} and \wref{exm:lrm} are both monotonic, we obtain the following corollary from \wref{thm:monotonicordinal}:

\begin{examplebox}
\begin{corollary}\label{cor:monotonicordinal}
The hypersuccinct code $\mathsf{H}: \mathfrak{T} \to \{0,1\}^\star$ encodes
\begin{itemize}
\item[(i)] \textbf{Uniform composition trees} of size $n$ (see \wref{exm:uniformcomposition}) using
\begin{align*}
|\mathsf{H}(t)| \leq \lg \left(1/\mathbb{P}[t]\right) + O(n \log \log n/\log n)
\end{align*}
many bits,
\item[(ii)] \textbf{Random LRM trees} of size $n$ (see \wref{exm:lrm}) using
\begin{align*}
|\mathsf{H}(t)| \leq \lg \left(1/\mathbb{P}[t]\right) + O(n \log \log n/\log n)
\end{align*}
many bits.
\end{itemize}
\end{corollary}
\end{examplebox}

\subsection{Fringe-Dominated Fixed-Size Ordinal Tree Sources}

As for binary trees, we consider a second class of fixed-size sources, fringe-dominated ordinal tree sources, for which we will be able to prove universality of the hypersuccinct code: Recall that a node $v$ is called \emph{heavy}, if $|t[v]|\geq B$ for the fixed parameter $B$, and \emph{light}, otherwise. With $n_{\geq B}(t)$ we again denote the number of heavy nodes of $t$. Moreover, we call a fringe subtree \emph{heavy}, if its root is heavy, and  \emph{light} otherwise.
With $\ell_{B}(t)$, we denote the total number of \emph{maximal} (non-empty) light fringe subtrees of $t$, \ie, of light nodes $v$ of $t$, such that $\operatorname{parent}(v)$ is heavy. Note that for binary trees, we have $\ell_B(t) \leq n_{\geq B}(t) +1 $, as the set of heavy nodes of a binary tree $t$ induces a (binary, non-fringe) subtree $t'$ of $t$, and every leaf of this subtree $t'$ of $t$ can have at most two children. For ordinal trees, this relation does not hold (consider, for example, an ordinal tree of size $n$ consisting of a root node with $n-1$ children).

\begin{definition}[Average-case fringe-dominated]
We call a fixed-size ordinal tree source \emph{average-case $B$-fringe-dominated}, 
for a function $B$ with $B(n)=\Theta(\log n)$, 
if
\begin{align*}
\sum_{t \in \mathfrak{T}_n}\mathbb{P}[t]\cdot \ell_{B}(t) = o\left(\frac{n}{\log B}\right)
\quad \text{ and } \quad \sum_{t \in \mathfrak{T}_n}\mathbb{P}[t]\cdot n_{\geq B}(t) = o\left(\frac{n}{\log B}\right).
\end{align*}
\end{definition}
\begin{definition}[Worst-case fringe-dominated]
We call a fixed-size ordinal tree source \emph{worst-case $B$-fringe-dominated}, 
for a function $B$ with $B(n)=\Theta(\log n)$, 
if $$\ell_{ B}(t) = o(n/\log B) \quad \text{and} \quad n_{\geq B}(t) = o(n/\log B)$$ for every $t \in \mathfrak{T}_n$ with $\mathbb{P}[t]>0$.
\end{definition}
Note that for binary trees, these definitions accord with \wref{def:avfringe-dominated} and \wref{def:wfringe-dominated} of fringe-dominated binary tree sources, as in this case $\ell_B(t) \leq n_{\geq B}(t) +1 $, by the above considerations.
The parameter $B$ will again be chosen as $B=\Theta(\log n)$.

Fringe-dominated sources can be handled similarly as binary trees
using a great-branching code.
We start with the following lemma:
\begin{lemma}[micro tree code]
\label{lem:fringe-dom-ordinal}
Let $\mathfrak{S}_{\mathit{fs}}(p)$ be a fixed-size tree source and let $t \in \mathfrak{T}_n $ with $\mathbb{P}[t]>0$. Then
\begin{align*}
		\sum_{i=1}^m |C(\mu_i)|
	&\wwrel\leq \lg \left(\frac{1}{\mathbb{P}[t]}\right) 
		\bin+ O(n \log \log n/\log n)
		\bin+O(\ell_{B(n)}(t) \log \log n)
\\*	&\wwrel\ppe
		\bin+O(n_{\geq B}(t) \log \log n),
\end{align*}
where $C$ is a Huffman code for the sequence of micro trees $\mu_1, \dots, \mu_m$ and $B(n) \in \Theta(\log n)$ is the parameter from our tree covering scheme (see \wref{sec:hypersuccinct-code-ordinal}).
\end{lemma}
\begin{proof}
As in the case of binary trees, we first observe that some of the micro trees $\mu_1, \dots, \mu_m$ from the tree covering scheme might be fringe, but many will be internal micro trees, \ie, have child micro trees in the top tier tree $\Upsilon$. 
Let $\mathcal{I}_0 = \{i \in [m] \mid \mu_i \text{ is fringe}\}$ and let $\mathcal{I}_1 = [m]\setminus \mathcal{I}_0$. If $\mu_i$ is a fringe micro tree, then all micro-tree local subtree sizes and node degrees coincide with the corresponding global subtree sizes and node degrees, except for (possibly) the root node's degree: 
The root node of $\mu_i$ might be contained in several micro trees, in that case its global degree and its micro-tree local node degree do not coincide (however, the respective subtree sizes do). 
Let $\rho_i$ denote the root node of micro tree $\mu_i$ and let $f_{i,1}, \dots, f_{i,\deg(\rho_i)}$ denote the fringe subtrees of $\mu_i$ rooted in $\rho_i$'s children, listed in preorder. By definition of the tree covering scheme (\wref{sec:farzan-munro}), we find that all the subtrees $f_{i,1}, \dots, f_{i,\deg(\rho_i)}$ are maximal light subtrees of $t$, and $\rho_i$ corresponds to a heavy node of $t$.

If $\mu_i$ is an internal micro tree, then its root node might be contained in several micro trees as well, resulting in different global and micro-tree local node degrees. 
Furthermore, the subtree sizes of the ancestors of portal nodes change. By \wref{lem:tree-decomposition-ordinal}, there is at most one other edge leading to a node outside of the micro tree $\mu_i$ apart from edges leaving the subtree root: 
Thus, the ancestors of portals in an internal micro tree $\mu_i$ form a \emph{unary path} from the root node to the (non-root-node) portal, if it exists. Let $\mathit{bough}(\mu_i)$ denote the subtree of $\mu_i$ induced by the set of nodes that are ancestors of $\mu_i$'s child micro trees (ancestors of the portals), including the root node. As observed above, $\mathit{bough}(\mu_i)$ is always a unary path -- thus, if we know the length of $\mathit{bough}(\mu_i)$, we also know its shape.
With $g_{i,1,1}, \dots, g_{i,1,k_{i,1}} \dots g_{i,|\mathit{bough}(\mu_i)|,1}, \dots, g_{i,|\mathit{bough}(\mu_i)|,k_{i,|\mathit{bough}(\mu_i)|}}$ we denote the non-empty fringe subtrees of $\mu_i$ hanging off the boughs of $\mu_i$, where $g_{i,j,1}, \dots, g_{i,j,k_{i,j}}$ denote the fringe subtrees attached to the $j$th node of $\mathit{bough}(\mu_i)$ (listed in preorder), and $k_{i,j}$ denotes their respective number. 
Moreover, with $r_{i,j}$ we denote how many of them are right siblings of the $(j+1)$st node of $\mathit{bough}(\mu_i)$ (if $j =|\mathit{bough}(\mu_i)|$, we set $r_{i,|\mathit{bough}(\mu_i)|}=0$). 
As those fringe subtrees $g_{i,j,k}$ of $\mu_i$ are fringe subtrees of $t$ as well and pairwise-disjoint, we find that their micro-tree local subtree sizes and micro-tree local node degrees coincide with the corresponding global subtree sizes and global node degrees. Altogether, we thus have
\begin{align}\label{eq:ordinal-fringe-dominated-prob}
\sum_{i \in \mathcal{I}_0}\sum_{k=1}^{\deg(\rho_i)}\lg\left(\frac{1}{\mathbb{P}[f_{i,k}]}\right) + \sum_{i \in \mathcal{I}_1}\sum_{j=1}^{|\mathit{bough}(\mu_i)|}\sum_{k=1}^{k_{i,j}}\lg\left(\frac{1}{\mathbb{P}[g_{i,j,k}]}\right) \wwrel\leq \lg\left(\frac{1}{\mathbb{P}[t]}\right).
\end{align}
Moreover, the fringe subtrees $g_{i,j,k}$ are maximal light subtrees of $t$, as by definition of the tree covering scheme, $\mathit{bough}$-nodes are heavy.

As in the proof of \wref{lem:great-branching}, we now construct a new encoding, similar to the ``great-branching'' code, for ordinal trees: Let 
\[
	E_{i,j}
	\wwrel=
	\gamma(k_{i,j}+1)\cdot \gamma(r_{i,j}+1)\cdot D_p(g_{i,j,1})\dots D_p(g_{i,j,k_{i,j}}) 
	\wrel\in \{0,1\}^\star,
\]
where $D_p(g_{i,k})$ denotes the depth-first order arithmetic code for ordinal fixed-size tree sources from \wref{sec:monotonic-ordinal}. That is, $E_{i,j}$ stores the number of fringe subtrees attached to the $j$th node of $\mathit{bough}(\mu_i)$, followed by the number $r_{i,j}$ which states how many of them are right siblings of the $j+1$st node of $\mathit{bough}(\mu_i)$, followed by their depth-first order arithmetic codes, listed in preorder.
We set
\begin{align*}
\hat{G}_B(\mu_i)=\begin{cases} \texttt{0}\cdot \gamma(\deg(\rho_i))\cdot D_p(f_{i,1})\cdots D_p(f_{i,\deg(\rho_i)}), & \text{if }\mu_i \text{ is a fringe micro tree;} \\
			\texttt{1}\cdot 
			\gamma(|\mathit{bough}(\mu_i)|) \cdot 
			E_{i,1} \cdots E_{i,|\mathit{bough}(\mu_i)|},
			&\text{otherwise,}
\end{cases}
\end{align*}
Note that this is well-defined, as the encoding $D_p$ is only applied to fringe subtrees $f_{i,k}$ and $g_{i,j,k}$ of $t$, for which $\mathbb{P}[f_{i,k}], \mathbb{P}[g_{i,j,k}]>0$ follows from $\mathbb{P}[t]>0$. 
We can reconstruct $\mu_i$ from $\hat{G}_B(\mu_i)$ as follows: If $\mu_i$ is a fringe subtree, we know the degree of the root of $\mu_i$, followed by the (uniquely decodable) encodings of the root node's subtrees, $D_p(f_{i,1}), \dots, ,D_p(f_{i,\deg(\rho_i)})$. 
If $\mu_i$ is an internal micro tree, we first decode the size (and thus, the shape) of $\mathit{bough}(\mu_i)$. 
Then, for each node of $\mathit{bough}(\mu_i)$, we decode the number of fringe subtrees (which can be zero) attached to that node, followed by how many of them are right siblings of the next bough-node, followed by their depth-first order arithmetic code, which tells us their sizes and shapes, listed in preorder.
The code $\hat{G}_B$ is \emph{not} a prefix-free code over $\Sigma_{\mu}$: there can be micro tree shapes that are assigned \emph{several} codewords by $\hat{G}_B$, depending on which nodes are portals to other micro trees (if any). 
But $\hat{G}_B$ is uniquely decodable to local shapes of micro trees, and can thus be seen as a \emph{generalized prefix-free code}, where more than one codeword per symbol is allowed: 
Thus, the Huffman code $C$ for micro trees used in the hypersuccinct code achieves no worse encoding length than the great-branching code $\hat{G}_B$:
\begin{align*}
\sum_{i=1}^m|C(\mu_i)| \leq \sum_{i=1}^m|\hat{G}_B(\mu_i)|= \sum_{i \in \mathcal{I}_0}|\hat{G}_B(\mu_i)|+\sum_{i \in \mathcal{I}_1}|\hat{G}_B(\mu_i)|.
\end{align*}
By definition of $\hat{G}_B$ and $E_{i,j}$, we find
\begin{align*}
&\sum_{i \in \mathcal{I}_0}|\hat{G}_B(\mu_i)|+\sum_{i \in \mathcal{I}_1}|\hat{G}_B(\mu_i)| =\sum_{i \in \mathcal{I}_0}\left(1+|\gamma(\deg(\rho_i))|+\sum_{k=1}^{\deg(\rho_i)}|D_p(f_{i,k})|\right)\\
&+\sum_{i \in \mathcal{I}_1}\left(1+|\gamma(|\mathit{bough}(\mu_i)|)|+\sum_{j=1}^{|\mathit{bough}(\mu_i)|}\left( |\gamma(k_{i,j}+1)|+ |\gamma(r_{i,j}+1)|+\sum_{k=1}^{k_{i,j}}|D_p(g_{i,j,k})|\right)\right).
\end{align*}
With $\deg(\rho_i),k_{i,j},r_{i,j} \leq \mu-1$, this is upper-bounded by
\begin{align*}
&\sum_{i \in \mathcal{I}_0}\left(2+2\lg(\mu)+\sum_{k=1}^{\deg(\rho_i)}|D_p(f_{i,k})|\right)\\
&+\sum_{i \in \mathcal{I}_1}\left(2+2\lg(|\mathit{bough}(\mu_i)|)+|\mathit{bough}(\mu_i)|(4\lg(\mu)+2)+\sum_{j=1}^{|\mathit{bough}(\mu_i)|}\sum_{k=1}^{k_{i,j}}|D_{p}(g_{i,j,k})|\right).
\end{align*}
Using the estimate \eqref{eq:depth-first-ordinal} for  $|D_p|$, we can upper-bound this by
\begin{align*}
&\sum_{i \in \mathcal{I}_0}\left(2+2\lg(\mu)+\sum_{k=1}^{\deg(\rho_i)}\left(\lg\left(\frac{1}{\mathbb{P}[f_{i,k}]}\right)+2\lg(|f_{i,k}|+1)+3\right)\right)\\
+&\sum_{i \in \mathcal{I}_1}\left(2+2\lg(|\mathit{bough}(\mu_i)|)+|\mathit{bough}(\mu_i)|(4\lg(\mu)+2)\right)\\
+&\sum_{i \in \mathcal{I}_1}\sum_{j=1}^{|\mathit{bough}(\mu_i)|}\sum_{k=1}^{k_{i,j}}\left(\lg\left(\frac{1}{\mathbb{P}[g_{i,j,k}]}\right)+2\lg(|g_{i,j,k}|+1)+3\right).
\end{align*}
With $|\mathcal{I}_0| \leq m$ and by inequality \eqref{eq:ordinal-fringe-dominated-prob}, this is smaller than
\begin{align*}
&\lg\left(\frac{1}{\mathbb{P}[t]}\right)+4m\lg(\mu)+5\sum_{i \in \mathcal{I}_0}\sum_{k=1}^{\deg(\rho_i)}\lg(|f_{i,k}|+1)
+10\sum_{i \in \mathcal{I}_1}|\mathit{bough}(\mu_i)|\lg(\mu)\\
&+5\sum_{i \in \mathcal{I}_1}\sum_{j=1}^{|\mathit{bough}(\mu_i)|}\sum_{k=1}^{k_{i,j}}\lg(|g_{i,j,k}|+1).
\end{align*}
Recall that all fringe subtrees $f_{i,k}$ and $g_{i,j,k}$ are distinct maximal light subtrees of $t$. Thus, their total number is upper-bounded by the number $\ell_B(t)$ of maximal light fringe subtrees:
\begin{align}\label{eq:light-subtrees-ordinal}
\sum_{i \in \mathcal{I}_0}\deg(\rho_i) + \sum_{i \in \mathcal{I}_1}\sum_{j=1}^{|\mathit{bough}(\mu_i)|}k_{i,j} \leq \ell_{B}(t).
\end{align}
Furthermore, every $\mathit{bough}$-node is heavy. However, the paths $\mathit{bough}(\mu_i)$ are not necessarily disjoint subtrees of $t$, as possibly many micro tree root nodes correspond to the same node of $t$: Thus, at most one node per micro tree is counted multiple times if we add up the sizes of the boughs. We thus have
\begin{align}\label{eq:bough-nodes-ordinal}
\sum_{i \in \mathcal{I}_1}|\mathit{bough}(\mu_i)|\leq n_{\geq B}(t) + m.
\end{align}
With the bounds \eqref{eq:light-subtrees-ordinal} and \eqref{eq:bough-nodes-ordinal}, and as $|f_{i,j}|, |g_{i,j,k}| \leq \mu$, we find altogether:
\begin{align*}
\sum_{i=1}	^m|C(\mu_i)|\leq \lg \left(\frac{1}{\mathbb{P}[t]}\right) + O(m \log(\mu)) + O(\ell_B(t) \log(\mu)) + O(n_{\geq B}(t)\log(\mu)).
\end{align*}
As $m = \Theta(n/\log n)$ and $\mu = \Theta(\log n)$ (see \wref{sec:hypersuccinct-code-ordinal}), we have
\begin{align*}
\sum_{i=1}	^m|C(\mu_i)|\leq \lg \left(\frac{1}{\mathbb{P}[t]}\right) + O\left(\frac{n \log \log n}{\log n}\right) + O(\ell_B(t) \log \log n) + O(n_{\geq B}(t)\log\log n).
\end{align*}
\end{proof}

From \wref{lem:fringe-dom-ordinal} and \wref{lem:hypersuccinct-code-ordinal}, we find the following:
\begin{theorem}[Universality from fringe dominance]
	Let $\mathfrak{S}_{\mathit{fs}}(p)$ be an average-case fringe-dominated fixed-size ordinal tree source. Then the hypersuccinct code $\mathsf{H}: \mathfrak{T} \to \{0,1\}^\star$ satisfies
	\begin{align*}
	\sum_{t \in \mathfrak{T}_n}\mathbb{P}[t]|\mathsf{H}(t)|
	\wwrel\leq \sum_{t \in \mathfrak{T}_n}\mathbb{P}[t]\lg\left(\frac{1}{\mathbb{P}[t]}\right)+o(n).
	\end{align*}
	Let $\mathfrak{S}_{\mathit{fs}}(p)$ be a worst-case fringe-dominated fixed-size ordinal tree source. Then the hypersuccinct code $\mathsf{H}: \mathfrak{T} \to \{0,1\}^\star$ satisfies
	\begin{align*}
	|\mathsf{H}(t)| \wwrel\leq \lg\left(\frac{1}{\mathbb{P}[t]}\right)+o(n)
	\end{align*}
	for every ordinal tree $t \in \mathfrak{T}_n$ with $\mathbb{P}[t]>0$.
\end{theorem}

\begin{remark}[Fixed-height sources?]
Fixed-height sources for ordinal trees could in principle be handled similar 
to the fixed-size ones below; but unless node degrees are bounded, 
there are infinitely many ordinal trees of a given height, which 
makes the utility of such sources questionable 
(and would not satisfy to the filter definitions in \cite{ZhangYangKieffer2014}).
We will therefore not explore this route.
\end{remark}

\section{Label-Shape Entropy}
\label{sec:label-shape}

In \cite{HuckeLohreySeelbachBenkner2019} (see also \cite{HuckeLohreySeelbachBenkner2020}), another measure of empirical entropy for (node-labeled) ordinal trees was introduced that we denote with $\mathcal{H}_k^s$: 
In \cite{HuckeLohreySeelbachBenkner2020}, this measure is referred to as \emph{label-shape-entropy}, as it considers both labels and structure of the tree, but this notion of empirical entropy is also a suitable entropy measure for \emph{unlabeled} ordinal trees. Since we do not consider labeled trees in this work, we refer to this notion of empirical entropy for trees as \emph{shape-entropy} for short.
In this section, we show that the length of our hypersuccinct code $\mathsf{H}$ for binary trees (see \wref{sec:hypersuccinct-code}) can be upper-bounded in terms of the $k$th-order shape entropy $\mathcal{H}_k^s$ of an \emph{ordinal} tree (for suitable~$k$), plus lower-order terms.

\begin{remark}[Relation to degree entropy]
In \cite{HuckeLohreySeelbachBenkner2020}, it is shown that $k$th order shape entropy $\mathcal{H}_k^s$ can be exponentially smaller than the degree entropy $H^{\deg}$ (see \wref{def:degree-entropy}), 
but that a reverse statement cannot hold, that is, the following two statements are shown:

\begin{lemma}[Lemmas 4 and 5, \cite{HuckeLohreySeelbachBenkner2020}]
There exists a family of trees $(t_n)_{n \in \mathbb{N}}$, such that $|t_n| = \Theta(n)$, $H^{\deg}(t_n)=(2-o(1))n$ and $\mathcal{H}_k^s(t_n) \leq \lg(en)$.
\end{lemma}

\begin{theorem}[Theorem 4, \cite{HuckeLohreySeelbachBenkner2020}]
For every ordinal tree $t \in \mathfrak{T}$ of size $|t|\geq 2$ and integer $k \geq 1$, we have $\mathcal{H}_k^s(t) \leq 2H^{\deg}(t) + 2\lg(|t|)+4$.
\end{theorem}
\end{remark}

We need some additional notation. 
We introduce two additional types of tree processes to apply our proof template for universality. 
We call them \emph{shape-processes} (as considered before in \cite{HuckeLohreySeelbachBenkner2019}) and \emph{childtype-processes}.
The childtype-processes will allow us to write $\mathcal{H}_k^s(t)$ as $\lg\left(1/\mathbb{P}[t]\right)$, where $\mathbb{P}[t]$ is the probability that a certain tree process (a childtype process) generates $t$, 
which then can be written as a product of contributions of the nodes of the tree.

Let $\mathcal{T}^{\diamond}$ denote the set of full binary trees, and let $\mathcal{T}^{\diamond}_n$ likewise denote the set of full binary trees of size $n$.
Let $v$ be a node of a full binary tree $t \in \mathcal{T}^{\diamond}$. We define the \emph{shape-history} $h^s(v)$ of $v$ inductively as follows: If $v$ is the root node of $t$, we set $h^s(v) = \varepsilon$ (the empty string). If $v$ is the left child of a node $w$ of $t$, we set $h^s(v)=h^s(w)0$ and if $v$ is a right child of a node $w$ of $t$, we set $h^s(v)=h^s(w)1$. In other words, in order to obtain $h^s(v)$, we walk downwards in the tree from the root node to node $v$, and concatenate bits $0$ and $1$ for each edge we traverse, where a number $0$ (resp., $1$) states that we move on to a left (resp. right) child node. 
Morever, we define the \emph{$k$-th order shape history} $h_k^s(v) \in \{0,1\}^k$ of a node $v$ of a full binary tree $t \in \mathcal{T}^{\diamond}$ as the length-$k$-suffix of the string $0^kh^s(v)$, that is, if $|h^s(v)|\geq k$, we take the last $k$ directions $0$ and $1$ on the path from the root to the node $v$, and if $|h^s(v)|<k$, we pad this too short history with $0$'s, in order to obtain a string of length $k$. (This accords with the definition in \cite{HuckeLohreySeelbachBenkner2019}: Several alternatives of how to define $k$-shape histories of nodes $v$ for which $|h^s(v)|< k$ are discussed in the long version of \cite{HuckeLohreySeelbachBenkner2019}). 
Recall the definition of $\type(v)$ for a node $v$ of a binary tree $t$ from \wref{sec:memoryless-binary}: In particular, we find that $\type(v) \in \{0,2\}$ if $v$ is a node of a full binary tree. 
For a string $z \in \{0,1\}^k$ and an integer $i \in \{0,2\}$, we define $m_z^t$ as the number of nodes $v$ of $t$, for which $h_k^s(v)=z$, and $m_{z,i}^t$ as the number of nodes $v$ of $t$, for which $h_k^s(v)=z$ and $\type(v)=i$. 
A \emph{$k$th order shape process} $\vartheta= (\vartheta_z)_{z \in \{0,1\}^k}$ is a tuple of probability distributions $\vartheta_z: \{0,2\} \to [0,1]$ (see \cite{HuckeLohreySeelbachBenkner2019}). A $k$th order shape process $\vartheta$ assigns a probability $\mathbb{P}_{\vartheta}(t)$ to a full binary tree $t \in \mathcal{T}^{\diamond}$ by
\begin{align}\label{eq:shapeprocess}
\mathbb{P}_{\vartheta}[t]=\prod_{v \in t}\vartheta_{h_k^s(v)}(\type(v))=\prod_{z \in \{0,1\}^k}\prod_{i \in \{0,2\}}\left(\vartheta_{z}(i)\right)^{m_{z,i}^t}.
\end{align}
A $k$th order shape process randomly generates a full binary tree as follows: In a top-down way, starting at the root node, we determine for each node $v$ its type $\type(v) \in \{0,2\}$, where this decision depends on the $k$-shape-history $h_k^s(v)$: The probability that a node $v$ is of type $i$ is given by $\vartheta_{h_k^s(v)}(i)$.
If $i=0$, this node becomes a leaf and the process stops at this node. Otherwise, i.e., if $i=2$, we attach a left and a right child to the node and continue the process at these child nodes. Note that this process might generate infinite trees with non-zero probability.
In \cite{HuckeLohreySeelbachBenkner2019}, the \emph{$k$th order empirical shape entropy} of a full binary tree $t$ is defined as follows:

\begin{definition}[Shape entropy for full binary trees, \cite{HuckeLohreySeelbachBenkner2019}]\label{def:shapeentropybinary}
Let $k \geq 0$ be an integer and let $t \in  \mathcal{T}^{\diamond}$ be a full binary tree. The (unnormalized) $k$th-order shape entropy of $t$ is defined as
\begin{align*}
\mathcal{H}_k^s(t) \wwrel= \sum_{z \in \{0,1\}^k}\sum_{i \in \{0,2\}}m_{z,i}^t\lg\left(\frac{m_z^t}{m_{z,i}^t}\right).
\end{align*}
\end{definition}
The corresponding normalized tree entropy is obtained by dividing by the tree size. Note that shape entropy for full binary trees was already considered in \wref{rem:shape-entropy}.
For a full binary tree $t \in \mathcal{T}^{\diamond}$, we define the corresponding \emph{empirical $k$th order shape process} as the shape process $(\vartheta_z^t)_{z \in \{0,1\}^k}$ with $\vartheta_z^t(i)=m_{z,i}^t/m_z^t$ for every $z \in \{0,1\}^k$ and $i \in \{0,2\}$ (if $m_z^t=0$, we simply set $\vartheta_z^t(0)=1$). In particular, for the $k$th order empirical shape process $(\vartheta_z^t)_{z \in \{0,1\}^k}$ of a full binary tree $t \in \mathcal{T}^{\diamond}$, we find 
\begin{align}\label{eq:empiricalshape}
\lg\left(\frac{1}{\mathbb{P}_{\vartheta^t}[t]}\right)=\sum_{z \in \{0,1\}^k}\sum_{i \in \{0,2\}}m_{z,i}^t\lg\left(\frac{1}{\vartheta_z^t(i)}\right)=\sum_{z \in \{0,1\}^k}\sum_{i \in \{0,2\}}m_{z,i}^t\lg\left(\frac{m_z^t}{m_{z,i}^t}\right)=\mathcal{H}_k^s(t).
\end{align}

Next, we define a \emph{modified first-child next-sibling encoding} $\FCNS^\diamond: \mathfrak{F} \to \mathcal{T}^{\diamond}$, which maps a forest to a \emph{full} binary tree, as follows:
\begin{definition}[Modified $\FCNS$]\label{def:modifiedfcns}
The \emph{modified first-child next-sibling encoding} $\FCNS^\diamond: \mathfrak{F} \to \mathcal{T}^{\diamond}$ is recursively defined by $\FCNS^{\diamond}(\varepsilon) = \treenode $ for the empty forest $\varepsilon$, and 
\begin{align*}
\FCNS^{\diamond}(\treenode(f)g)=\treenode (\FCNS^{\diamond}(f), \FCNS^{\diamond}(g))
\end{align*}
for forests $f,g \in \mathfrak{F}$.
\end{definition}
That is, the left child (resp. right child) of a node in $\FCNS^{\diamond}(f)$ is its first child (resp. next sibling) in $f$ or a newly-added leaf, if it does not exist.
In particular, we find that $\FCNS^{\diamond}(f)$ is always a full binary tree, and that $\FCNS^\diamond: \mathfrak{F} \to \mathcal{T}^{\diamond}$ is a bijection. Moreover, we find that we obtain the modified first-child next-sibling encoding $\FCNS^{\diamond}(f)$ from $\FCNS(f)$ (as defined in \wref{def:fcns}) by adding a leaf to each null-pointer of $\FCNS(f)$.
Furthermore, we find that each node $v$ of a forest $f \in \mathfrak{F}$ uniquely corresponds  to an inner node of $\FCNS^{\diamond}(f)$, which we denote with $\operatorname{id}_{\FCNS}^{\diamond}(v)$. 
The shape entropy of an ordinal tree is defined as the shape entropy of its corresponding modified first-child next-sibling encoding in \cite{HuckeLohreySeelbachBenkner2019}:
\begin{definition}[Shape entropy for ordinal trees, \cite{HuckeLohreySeelbachBenkner2019}]\label{def:shapeentropyordinal}
Let $k \geq 0$ be an integer and let $t \in \mathfrak{T}$ be an ordinal tree. The (unnormalized) $k$th order shape entropy of $t$ is defined as
\begin{align*}
\mathcal{H}_k^s(t) = \mathcal{H}_k^s(\FCNS^{\diamond}(t)).
\end{align*}
\end{definition}
For an inner node $v$ of a full binary tree $t \in \mathcal{T}^{\diamond}$, we define its \emph{childtype} as follows:
\begin{align*}
\childtype(v)=\begin{cases} 0 \quad \text{if $v$'s children are both leaves,}\\
1 \quad \text{if only $v$'s left child is a leaf,}\\
2 \quad \text{if only $v$'s right child is a leaf,}\\
3 \quad \text{if $v$'s children are both inner nodes.}
\end{cases}
\end{align*}
Moreover, for a node $v$ of a forest $f \in \mathfrak{F}$, we set $\childtype(v)=\childtype(\idd_{\FCNS}^{\diamond}(v))$. In particular, we find:
\begin{lemma}\label{lem:childtype}
Let $v$ be a node of a forest $f \in \mathfrak{F}$, then
\begin{align*}
\childtype(v)=\begin{cases} 0 \quad \text{if $v$ is a leaf and does not have a next sibling, }\\
1 \quad \text{if $v$ is a leaf and has a next sibling, }\\
2 \quad \text{if $v$ is not a leaf and does not have a next sibling, }\\
3 \quad \text{if $v$ is not a leaf and has a next sibling. }\\
\end{cases}
\end{align*}
\end{lemma}
The proof of \wref{lem:childtype} follows immediately from \wref{def:modifiedfcns} and the definition of the $\childtype$-mapping.
Furthermore, for a node $v$ of a forest $f \in \mathfrak{F}$, we define the shape-history $h^s(v)$ as $h^s(\operatorname{id}_{\FCNS}^{\diamond}(v))$, i.e., as the shape-history of its corresponding node in $\FCNS^{\diamond}(f)$. We find that if $v$ is the root node of the first tree in (the sequence of trees) $f$, then $h^s(v)=\varepsilon$ (the empty string). Otherwise, if $v$ is the first child of a node $w$ of $f$, then $h^s(v)=h^s(w)0$ and if $v$ is the next sibling of a node $w$ of $f$, then $h^s(v)=h^s(w)1$.
Note that basically, for a node $v$ of a forest $f$, $h^s(v)$ represents the numbers of $v$'s left siblings and of $v$'s ancestors' left siblings in unary. 
Similarly, we define $h_k^s(v)$ as $h_k^s(\operatorname{id}_{\FCNS}^{\diamond}(v))$.

A \emph{$k$th order childtype process} $\zeta=(n_{\zeta},(\zeta_z)_{z \in \{0,1\}^k})$ is a tuple of probability distributions $\zeta_z: \{0,1,2,3\} \to [0,1]$ together with a number $n_{\zeta} \in [0,1]$. A $k$th order childtype process $\zeta$ assigns a probability $\mathbb{P}_{\zeta}$ to a full binary tree $t \in \mathcal{T}^{\diamond}$ by
\begin{align}\label{eq:childtypeprocess}
\mathbb{P}_{\zeta}[t]=\begin{dcases} 1-n_{\zeta} \quad &\text{if } |t| = 1,\\
n_{\zeta} \cdot \!\!\!\!\!\!\prod_{v \in t \atop v \text{ inner node of } t}\!\!\!\!\!\!\zeta_{h_k^s(v)}(\childtype(v))
&\text{otherwise.}
\end{dcases}
\end{align}
A $k$th order childtype process randomly generates a full binary tree $t$ as follows: 
With probability $1-n_{\zeta}$, $t$ consists of just one node. Otherwise, in a top-down way, starting at the root node, we determine for each node $v$ its $\childtype(v) \in \{0,1,2,3\}$, where this decision depends on the $k$-shape-history $h_k^s(v)$: The probability that a node $v$ is of childtype $i$ is given by $\zeta_{h_k^s(v)}(i)$. We add a left child and a right child to the node and if $i=0$, we (implicitly) mark both of them as leaves, if $i=1$, we mark the left child as a leaf, if $i=2$, we mark the right child as a leaf and if $i=3$, we do not mark the children as leaves. The process then continues at child nodes which are not marked as leaves. 
For a forest $f \in \mathfrak{F}$, we set 
\begin{align}\label{def:childtypeordinal}
\mathbb{P}_{\zeta}[f]=\mathbb{P}_{\zeta}[\FCNS^{\diamond}(f)].
\end{align}
Thus, via the $\FCNS^{\diamond}$-encoding, a $k$th order childtype process can be seen as a process randomly generating a forest $f$ as follows: With probability $1-n_{\zeta}$, the forest is empty. Otherwise, in a top-down left-to-right way, starting at the root node of the first tree in the forest, we determine for each node $v$ its $\childtype(v) \in \{0,1,2,3\}$ (i.e., whether this node has a first child and whether this node has a next sibling), where this decision depends on the $k$-shape-history $h_k^s(v)$: Note that as we generate $f$ in a top-down left-to-right way, we always know $h_k^s(v)$ at every node we visit. If $\childtype(v) = 0$, the process stops at this node. If $\childtype(v)=1$, then we add a new child node to $v$'s parent (respectively, if $v$ is a root node itself, we add a new tree of size one to the forest), if $\childtype(v)=2$, we add a new child to $v$, and if $\childtype(v)=3$, we add a new child to $v$ and a new child to $v$'s parent node. The process then continues at these newly added nodes. 
In particular, we find
\begin{lemma}\label{lem:defprobchildtype}
Let $t \in \mathfrak{T}$ be a non-empty ordinal tree, then 
\begin{align*}
\mathbb{P}_{\zeta}[t]=n_{\zeta}\cdot \prod_{v \in t}\zeta_{h_k^s(v)}(\childtype(v)). 
\end{align*}
\end{lemma}
\begin{proof}
We find by the definition of $\mathbb{P}_{\zeta}$ (see \eqref{eq:childtypeprocess}), the definition of the $k$-shape-history and the definition of the mapping $\childtype$:
\begin{align*}
\mathbb{P}_{\zeta}[t]&=\mathbb{P}_{\zeta}[\FCNS^{\diamond}(t)]=n_{\zeta} \cdot \!\!\!\!\!\!\prod_{v \in \FCNS^{\diamond}(t) \atop v \text{ inner node } }\!\!\!\!\!\! \zeta_{h_k^s(v)}(\childtype(v)) \\
&=n_{\zeta} \cdot\prod_{v \in t} \zeta_{h_k^s(\idd_{\FCNS}^{\diamond}(v))}(\childtype(\idd_{\FCNS}^{\diamond}(v)))=n_{\zeta}\cdot \prod_{v \in t}\zeta_{h_k^s(v)}(\childtype(v)).
\end{align*}
\end{proof}
Finally, we make the following definition:

\begin{definition}\label{def:shapechildtypeprocess}
Let $\vartheta=(\vartheta_z)_{z \in \{0,1\}^k}$ be a $k$th order shape process. We define the corresponding $k-1$st-order childtype process $\zeta^{\vartheta}=(n_{\zeta}^{\vartheta}, (\zeta_z^{\vartheta})_{z \in \{0,1\}^{k-1}})$ by setting $n_{\zeta^{\vartheta}}=\vartheta_{0^k}(2)$ and 
\begin{align*}
\zeta_z^{\vartheta}(0)= \vartheta_{z0}(0)\cdot \vartheta_{z1}(0), &\quad
\zeta_z^{\vartheta}(1)= \vartheta_{z0}(0)\cdot \vartheta_{z1}(2),\\
\zeta_z^{\vartheta}(2)= \vartheta_{z0}(2)\cdot \vartheta_{z1}(0), &\quad
\zeta_z^{\vartheta}(3)= \vartheta_{z0}(2)\cdot \vartheta_{z1}(2),
\end{align*}
for every $z \in \{0,1\}^{k-1}$.
\end{definition}
It is easy to see that $\zeta_z^{\vartheta}$ is well-defined.
In particular, we find
\begin{lemma}\label{lem:shapechildtype}
Let $t \in \mathcal{T}^{\diamond}$ be a full binary tree. Then
\begin{align*}
\mathbb{P}_{\vartheta}[t]=\mathbb{P}_{\zeta^{\vartheta}}[t].
\end{align*}
\end{lemma}
\begin{proof}
First, let $|t|=1$: Then $t$ consists of only one leaf node $v$ of $k$-history $0^k$, and thus, we have
\begin{align*}
\mathbb{P}_{\vartheta}[t]=\vartheta_{0^k}(0) = 1-n_{\zeta^{\vartheta}}=\mathbb{P}_{\zeta^{\vartheta}}[t].
\end{align*}
In the next part of the proof, assume that $|t|>1$.
Let $\tilde{m}_{z,i}^t$ denote the number of inner nodes of $t$ with $k$-shape-history $z \in \{0,1\}^\star$ and of childtype $i \in \{0,1,2,3\}$, and recall that $m_{z,i}^t$ denotes the number of nodes of $t$ with $k$-shape-history $z$ and with $\type(v)=i \in \{0,2\}$.
Let $v$ be a node of $t$. 
First, we assume that $h_k^s(v)=z0$ for some $z \in \{0,1\}^{k-1}$ with $z \neq 0^{k-1}$ (thus, $v$ is not the root node of $t$), and that $v$ is a leaf: Then $v$'s parent $w$ is of $k-1$-shape-history $z$, and $w$'s childtype is either $0$ or $1$. In particular, the correspondence between  leaves $v$ of $t$ with $k$-shape-history $z0$ and inner nodes $w=\operatorname{parent}(v)$ of $t$ with $k-1$-shape-history $z$ and childtype $0$ or $1$ is bijective, as every node $v$ with $k$-shape-history $z0$ is a left child of its parent node.
We thus have
\begin{align*}
m_{z0,0}^t=\tilde{m}_{z,0}^t + \tilde{m}_{z,1}^t.
\end{align*}
In a similar way, we find that inner nodes $v$ of $t$ with $k$-shape-history $z0$ for $z \neq 0^{k-1}$ correspond to inner nodes $w=\operatorname{parent}(v)$ of $t$ with $k-1$-shape-history $z$ and childtype $i \in \{2,3\}$: We find
\begin{align*}
m_{z0,2}^t=\tilde{m}_{z,2}^t + \tilde{m}_{z,3}^t.
\end{align*}
Furthermore, we obtain the following relations in the same way:
\begin{align*}
m_{z1,0}^t=\tilde{m}_{z,0}^t + \tilde{m}_{z,2}^t,\\
m_{z1,2}^t=\tilde{m}_{z,1}^t + \tilde{m}_{z,3}^t,
\end{align*}
for every $z \in \{0,1\}^k$. 
It remains to deal with nodes of $k$-shape-history $z=0^{k}$: We find that every inner node $v$ of $t$ of $k$-shape-history $0^k$ uniquely corresponds to an inner node $w=\operatorname{parent}(v)$ of $t$ of $k-1$-shape-history $0^{k-1}$ and childtype $i \in \{2,3\}$, except for the root node: We thus have
\begin{align*}
m_{0^k,2}^t-1=\tilde{m}_{0^{k-1},2}^t + \tilde{m}_{0^{k-1},3}^t.
\end{align*}
Finally, every leaf $v$ of $t$ of $k$-shape-history $0^k$ uniquely corresponds to an inner node $w=\operatorname{parent}(v)$ of $t$ of $k-1$-shape-history $0^{k-1}$ and childtype $i \in \{1,2\}$, as the root node is an inner node by assumption:
\begin{align*}
m_{0^k,0}^t=\tilde{m}_{0^{k-1},0}^t + \tilde{m}_{0^{k-1},1}^t.
\end{align*}
Altogether, we thus have for trees $t$ with $|t|>1$:
\begin{align*}
\!\!\!\mathbb{P}_{\vartheta}[t]&=\!\!\!\!\prod_{z \in \{0,1\}^k} \prod_{i \in \{0,2\}}\left(\vartheta_z(i)\right)^{m_{z,i}^t} 
= \left(\vartheta_{0^k}(0)\right)^{\tilde{m}_{0^{k-1},0}^t+\tilde{m}_{0^{k-1},1}^t}
\cdot \left(\vartheta_{0^{k}}(2)\right)^{\tilde{m}_{0^{k-1},2}^t+\tilde{m}_{0^{k-1},3}^t+1} \\
&\cdot 
\prod_{z \in \{0,1\}^{k-1} \atop z \neq 0^{k-1}}\left(\vartheta_{z0}(0)\right)^{\tilde{m}_{z,0}^t+\tilde{m}_{z,1}^t} \cdot \left(\vartheta_{z0}(2)\right)^{\tilde{m}_{z,2}^t+\tilde{m}_{z,3}^t}\\
&\cdot \prod_{z \in \{0,1\}^{k-1}}\left(\vartheta_{z1}(0)\right)^{\tilde{m}_{z,0}^t+\tilde{m}_{z,2}^t} \cdot \left(\vartheta_{z1}(2)\right)^{\tilde{m}_{z,1}^t+\tilde{m}_{z,3}^t}\\
&=\vartheta_{0^k}(2)\cdot 
\prod_{z \in \{0,1\}^{k-1} }\left(\vartheta_{z0}(0)\cdot \vartheta_{z1}(0)\right)^{\tilde{m}_{z,0}^t}\cdot
\prod_{z \in \{0,1\}^{k-1} }\left(\vartheta_{z0}(0)\cdot \vartheta_{z1}(2)\right)^{\tilde{m}_{z,1}^t}\\ &\cdot
\prod_{z \in \{0,1\}^{k-1} }\left(\vartheta_{z0}(2)\cdot \vartheta_{z1}(0)\right)^{\tilde{m}_{z,2}^t}\cdot
\prod_{z \in \{0,1\}^{k-1}}\left(\vartheta_{z0}(2)\cdot \vartheta_{z1}(2)\right)^{\tilde{m}_{z,3}^t}\\
&=n_{\zeta}\cdot\prod_{z \in \{0,1\}^{k-1} }\left(\zeta_z^{\vartheta}(0)\right)^{\tilde{m}_{z,0}^t}\cdot
\prod_{z \in \{0,1\}^{k-1} }\left(\zeta_z^{\vartheta}(1)\right)^{\tilde{m}_{z,1}^t}\\ &\cdot
\prod_{z \in \{0,1\}^{k-1} }\left(\zeta_z^{\vartheta}(2)\right)^{\tilde{m}_{z,2}^t}\cdot
\prod_{z \in \{0,1\}^{k-1}}\left(\zeta_z^{\vartheta}(3)\right)^{\tilde{m}_{z,3}^t}=\mathbb{P}_{\zeta^{\vartheta}}[t].
\end{align*}
This finishes the proof.
\end{proof}

\begin{corollary}\label{cor:ordinalshapeentropy}
Let $t \in \mathfrak{T}$ be an ordinal tree, and let $\vartheta^t := \vartheta^{\FCNS^{\diamond}(t)}$ denote the empirical shape process of its corresponding first-child next-sibling encoding.
Then
\begin{align*}
\mathcal{H}_k^s(t) = \lg\left(\frac{1}{\mathbb{P}_{\zeta^{\vartheta^t}}[t]}\right).
\end{align*}
\end{corollary}
\begin{proof}
We have
\begin{align*}
\mathcal{H}_k^s(t)=\mathcal{H}_k^s(\FCNS^{\diamond}(t))=\lg\left(\frac{1}{\mathbb{P}_{\vartheta^t}[\FCNS^{\diamond}(t)]}\right)=\lg\left(\frac{1}{\mathbb{P}_{\zeta^{\vartheta^t}}[\FCNS^{\diamond}(t)]}\right)=\lg\left(\frac{1}{\mathbb{P}_{\zeta^{\vartheta^t}}[t]}\right),
\end{align*}
where the first equality follows from \wref{def:shapeentropyordinal}, the second equality follows from the fact that $\vartheta^t$ is the empirical $k$th order shape-process of $\FCNS^{\diamond}(t)$ (see \eqref{eq:empiricalshape}), the third equality follows from \wref{lem:shapechildtype} and the last equality follows from \eqref{def:childtypeordinal}.
\end{proof}

In order to show that our hypersuccinct encoding from \wref{sec:hypersuccinct-code} achieves the shape-entropy defined in \cite{HuckeLohreySeelbachBenkner2019} for ordinal trees, we start with defining a \emph{source-specific}  encoding (called depth-first order arithmetic code) with respect to a given $k$th order childtype process $\zeta$, against which we will compare the hypersuccinct code: The formula for $\mathbb{P}_{\zeta}[t]$ from \wref{lem:defprobchildtype} suggests a route for an (essentially) optimal source-specific encoding of any ordinal tree $t \in \mathfrak{T}$ with $\mathbb{P}_{\zeta}[t]>0$, that, given a $k$th order childtype process $\zeta$, spends $\lg(1/\mathbb{P}[t])$ (plus lower-order terms) many bits in order to encode an ordinal tree $t \in \mathfrak{T}$ with $\mathbb{P}_{\zeta}[t]>0$: Such an encoding may spend $\lg\left(1/\zeta_{h_k^s(v)}(\childtype(v))\right)$ many bits per node $v$ of $t$, plus $\lg\left(1/n_{\zeta}\right)$ many bits, if $t$ is non-empty, respectively, $\lg\left(1/(1-n_{\zeta})\right)$ many bits, if $t$ is the empty tree. (Note that as $\mathbb{P}_{\zeta}[t]>0$ by assumption, we have $\zeta_{h_k^s(v)}(\childtype(v))>0$ for every node $v$ of $t$.)
Assuming that we know the childtype process $\zeta=(\zeta_z)_{z \in \{0,1\}^k}$, i.e., that we need not store it as part of the encoding, we can make use of \emph{arithmetic coding} in order to devise a simple (source-dependent) encoding $\mathcal{D}_{\zeta}$, dependent on $\zeta$, that stores an ordinal tree $t$ as follows: 
First, we store a number $i \in \{1,2\}$ which tells us whether $t$ is empty ($i=1$) or non-empty ($i=2$) using arithmetic encoding, i.e., we feed the arithmetic coder with the model that the next symbol is a number $i \in \{1,2\}$ with probability $1-n_{\zeta}$, respectively, $n_{\zeta}$.
Next, while traversing the tree in depth-first order, we encode $\childtype(v) \in \{0,1,2,3\}$ for each node $v$ of $t$ that we pass, using arithmetic coding: To encode $\childtype(v)$ (i.e., whether $v$ is a leaf or not and whether $v$ has a next sibling or not, see \wref{lem:childtype}), we feed the arithmetic coder with the model that the next symbol is a number $i \in \{0,1,2,3\}$ with probability $\zeta_{h_k^s(v)}(i)$. Note that we always know $h_k^s(v)$ at each node $v$ we traverse: By definition, we have $h_k^s(v)=h_k^s(\idd_{\FCNS}^{\diamond}(v))$. If $v$ is the root node of $t$, then $\idd_{\FCNS}^{\diamond}(v)$ is the root node of $\FCNS^{\diamond}(t)$ and thus $h_k^s(v)=h_k^s(\idd_{\FCNS}^{\diamond}(v))=0^k$.
Otherwise, $h^s(\idd_{\FCNS}^{\diamond}(v))=h^s(\idd_{\FCNS}^{\diamond}(w))0$ or $h^s(\idd_{\FCNS}^{\diamond}(v))=h^s(\idd_{\FCNS}^{\diamond}(w))1$ for a node $w$ of $t$, which is either $v$'s left sibling or, if $v$ is the first child of its parent node in $t$, $v$'s parent in $t$, as $\idd_{\FCNS}^{\diamond}(w)$ is $\idd_{\FCNS}^{\diamond}(v)$'s parent. Thus, as we visit the nodes of $t$ in depth-first order, we have already visited $w$ and know $h_k^s(w)$, from which we can compute $h_k^s(v)$.
Altogether, this yields a source dependent code $\mathcal{D}_{\zeta}(t)$, which we refer to as the \emph{depth-first arithmetic code} with respect to the childtype-process $\zeta$. Note that an ordinal tree $t$ is always uniquely decodable from $\mathcal{D}_{\zeta}(t)$. 
As arithmetic coding uses at most $\lg\left(1/\zeta_{h_k^s(v)}(\childtype(v))\right)$ many bits per node $v$, plus $\lg\left(1/n_{\zeta}\right)$ many bits if $t$ is non-empty, plus at most $2$ bits of overhead, we find
\begin{align*}
|\mathcal{D}_{\zeta}(t)| \leq \begin{dcases}\sum_{v \in t}\lg\left(1/\zeta_{h_k^s(v)}(\childtype(v))\right) + \lg\left(1/n_{\zeta}\right) +2 \quad &\text{if } t \text{ is non-empty,}\\
\lg\left(1/(1-n_{\zeta})\right)+2 &\text{otherwise.}
\end{dcases}
\end{align*}
We now start with the following lemma:
\begin{lemma}\label{lem:shape-entropyhuffman}
Let $(\zeta_z)_{z \in \{0,1\}^k}$ be a $k$th order childtype process and let $t \in \mathfrak{T}$ be an ordinal tree of size $n$ with $\mathbb{P}_{\zeta}[t]>0$. Then
\begin{align*}
\sum_{i=1}^m|C(\mu_i)|\leq \lg\left(\frac{1}{\mathbb{P}_{\zeta}[t]}\right)+ O\left(\frac{n \log \log n+ kn}{\log n}\right),
\end{align*}
where $C$ is a Huffman code for the sequence of micro trees $\mu_1, \dots, \mu_m$ obtained from the tree-covering scheme.
\end{lemma}
\begin{proof}
Recall that the micro trees $\mu_1, \dots, \mu_m$ from our tree partitioning scheme for ordinal trees
are pairwise disjoint except for (potentially) sharing a common subtree root and that apart from
edges leaving the subtree root, at most one other edge leads to a node outside of the subtree (see \wref{fact:farzan-munro}).
The probability $\mathbb{P}_{\zeta}[t]$ consists of the contributions $\zeta_{h_k^s(v)}(\childtype(v))$ for every node $v$ of $t$.
However, $\zeta_{h_k^s(v)}(\childtype(v))$ depends on the childtype and $k$-shape-history of each node $v$, and there might be nodes, for which childtype and $k$-shape-history differ in $t$ and $\mu_i$. 
For the sake of clarity, let $h_k^s(v,t)$ denote the $k$-shape history of a node $v$ in $t$ (and likewise $h_k^s(v,\mu_i)$ the $k$-shape history of a node $v$ in a micro tree $\mu_i$), and let $ \childtype(v,t)$ (resp. $ \childtype(v,\mu_i)$) denote the childtype of a node $v$ in $t$ (resp. $\mu_i$).
First, we investigate under which conditions it might occur that a node $v$ of micro tree $\mu_i$ satisfies $h_k^s(v,t) \neq h_k^s(v,\mu_i)$ or $ \childtype(v,t)\neq  \childtype(v,\mu_i)$. We find:
\begin{itemize}
\item[(i)] If $v$ is the root node of a micro tree $\mu_i$, then it might have left, respectively, right siblings in $t$, which it does not have in $\mu_i$: Thus, its childtype and its $k$-shape-history might change.
\item[(ii)] If $v$ is the first child of the root of $\mu_i$, then it might have left siblings in $t$, which it does not have in $\mu_i$. Thus, its $k$-shape-history changes. Furthermore, the $k$-shape-history of its close descendants and right siblings thus changes as well, i.e., the $k$-shape-history of the descendants of order less than $k$ of $\idd_{\FCNS}^{\diamond}(v)$: 
However, if we know $h_k^s(v,t)$, we are able to recover $h_k^s(w,t)$ for all nodes $w$ which are descendants, right siblings, or right siblings of descendants of $v$.
\item[(iii)] If $v$ is the last child of the root of $\mu_i$, then it might have  right siblings in $t$, which it does not have in $\mu_i$: Thus, its childtype might change.
\item[(iv)] The root node's children in $\mu_i$ are consecutive children of this node in $t$,  except for possibly one child node $x$, which might be missing in $\mu_i$ (see \wref{fact:farzan-munro}). Thus, if $v$ is the right sibling of $x$ in $t$, its $k$-shape-history in $\mu_i$ might differ from its $k$-shape-history in $t$. Furthermore, the $k$-shape-histories of nodes corresponding to the descendants of order at most $k$ of $\idd_{\FCNS}^{\diamond}(v)$ in $\FCNS^{\diamond}(t)$ might change as well. Again, if we know $h_k^s(v,t)$, we are able to recover $h_k^s(w,t)$ of nodes $w$ which correspond to  descendants of $\idd_{\FCNS}^{\diamond}(v)$ in $\FCNS^{\diamond}(t)$. 
\item[(v)] There is at most one other edge which leads to a node outside of the micro tree $\mu_i$, besides edges emanating from the root of $\mu_i$ (see \wref{fact:farzan-munro}). Let $v$ be the node in $\mu_i$, from which this other edge emanates: If $v$ has only one child in $t$, then it does not have a child node in $\mu_i$, and thus, its childtypes in $t$ and $\mu_i$ do not coincide. Otherwise, the degree of $v$ in $t$ is greater than one and in particular, there might be a child node $w$ of $v$, whose left sibling in $t$ does not belong to $\mu_i$. Thus, $w$'s $k$-shape-history might change, as well as the $k$-shape-history of the nodes corresponding to the descendants of order less than $k$ of $\idd_{\FCNS}^{\diamond}(w)$: Again, if we know $h_k^s(v,t)$, we are able to recover $h_k^s(w,t)$ of nodes $w$ which correspond to  descendants of $\idd_{\FCNS}^{\diamond}(v)$ in $\FCNS^{\diamond}(t)$. 
Finally, there might be a child node $u$ of $v$, which has a right sibling in $t$ and which does not have a right sibling in $\mu_i$; thus, its childtype changes.

\end{itemize}
By the above considerations, there can be several nodes $v$ in $\mu_i$ for which $h_k^s(v,t) \neq h_k^s(v, \mu_i)$, however, we only need to know $h_k^s(v,t)$ for at most four of these nodes (see items (i), (ii), (iv) and (v)) in order to be able to determine the $k$-shape-history in $t$ of all nodes of $\mu_i$. Let $\ell_i$ denote the number of $k$-shape-histories we need to know in order to be able to determine $h_k^s(v,t)$ for all nodes $v$ of $\mu_i$. 
Furthermore, let $j_i$ denote the number of nodes $v$ of $\mu_i$, for which $\childtype(v,t) \neq \childtype(v,\mu_i)$, where we always include the root node $\pi_i$ of $\mu_i$ in this $j_i$ many nodes (even if its childtypes in $t$ and $\mu_i$ are identical).
By the above considerations, we find that $j_i$ is upper-bounded by four (see items (i), (iii) and (v)). 
Let $S_i \in \{0,1\}^\star$ denote the following binary string, obtained as the concatenation of
\begin{itemize}
\item an encoding of the number $j_i$ using two bits,
\item the preorder positions in $\mu_i$ of the $j_i$ many nodes for which $\childtype(v,t) \neq \childtype(v,\mu_i)$ (plus the root node $\pi_i$ of $\mu_i$), encoded in Elias gamma code and listed in preorder,
\item the encodings of the childtypes in $t$ of these $j_i$ nodes using two bits each, listed in preorder,
\item the encodings of the childtypes in $\mu_i$ of these $j_i$ nodes using two bits each, listed in preorder,
\item an encoding of the number $\ell_i$ using two bits,
\item the Elias gamma encodings of the preorder positions in $\mu_i$ of the $\ell_i$ many nodes from whose $k$-shape histories in $t$ we are able to determine the $k$-shape history in $t$ of all nodes of $\mu_i$, listed in preorder,
\item the $k$-shape-histories of these $\ell_i$ nodes, listed in preorder, using $k$ bits each.
\end{itemize}
We find that $|S_i| \leq O(\log(\mu) + k)$. We define the following modification of the depth-first order arithmetic code $\mathcal{D}_{\zeta}$, which we denote with $\bar{\mathcal{D}}_{\zeta}$: The encoding $\bar{\mathcal{D}}_{\zeta}(\mu_i)$ consists of the string $S_i$ followed by an encoding of $\childtype(v,t)$ for every node $v$ of $\mu_i$ in depth-first order (preorder) of $\mu_i$ except for the root node $\pi_i$ of $\mu_i$, using arithmetic coding:
The childtype of the root node $\pi_i$ is already stored in $S_i$. We traverse the tree $\mu_i$ in depth-first order; to encode $\childtype(v,t)$, we feed the arithmetic coder with the model that the next symbol is a number $i \in \{0,1,2,3\}$ with probability $\zeta_{h_k^s(v,t)}(i)$. Note that at each node $v$ that we pass, we know $h_k^s(v,t)$ (either from $S_i$ or as we are able to determine $h_k^s(v,t)$ from the $k$-shape-history of the node $v$'s left sibling or parent) and we know both $\childtype(v,t)$ and $\childtype(v,\mu_i)$ (either because $\childtype(v,t)=\childtype(v,\mu_i)$ or because we have stored both $\childtype(v,t)$ and $\childtype(v,\mu_i)$ explicitly in $S_i$). Altogether, this yields the encoding $\bar{\mathcal{D}}(\mu_i)$.
Note that we leave out the $\lg\left(1/n_{\zeta}\right)$ many bits (used in the encoding $\mathcal{D}(\mu_i)$) which encode the number $i \in \{1,2\}$ which tells us whether $\mu_i$ is empty or not (by definition, every micro tree $\mu_i$ of a non-empty tree $t$ is non-empty).
As we have $\zeta_{h_k^s(v,t)}(\childtype(v,t))>0$ for every node $v$ the encoding $\bar{\mathcal{D}}_{\zeta}(\mu_i)$ is well-defined. We find that
\begin{align*}
|\bar{\mathcal{D}}_{\zeta}(\mu_i)| \leq |S_i|+\sum_{v \in \mu_i \atop v \neq \pi_i}\lg(1/\zeta_{h_k^s(v,t)}(\childtype(v,t)))+2.
\end{align*}
Furthermore, note that we can uniquely recover a micro tree shape $\mu_i$ from the encoding $\bar{\mathcal{D}}_{\zeta}(\mu_i)$ and that formally, $\bar{\mathcal{D}}_{\zeta}$ is not a \emph{prefix-free} code over $\Sigma_{\mu}$, as as there can be micro tree shapes that
are assigned several codewords by $\bar{\mathcal{D}}_{\zeta}$. But $\bar{\mathcal{D}}_{\zeta}$ can again be seen as a generalized prefix-free code, where more than one codeword per symbol is allowed, as $\bar{\mathcal{D}}_{\zeta}$ is uniquely decodable to local shapes of micro trees. 
Thus, as a Huffman code minimizes the encoding length over the class of generalized prefix-free codes, we find:
\begin{align*}
\sum_{i=1}^m|C(\mu_i)|\leq \sum_{i=1}^m|\bar{\mathcal{D}}_{\zeta}(\mu_i)|\leq \sum_{i=1}^m\left(|\mathcal{S}_i| + \sum_{v \in \mu_i \atop v \neq \pi_i}\lg(1/\zeta_{h_k^s(v,t)}(\childtype(v,t)))+2 \right).
\end{align*}
Recall that the micro trees $\mu_i$ are disjoint except for possibly sharing a common root node and that $|S_i| \leq O(\log \mu + k)$. Thus, we  have
\begin{align*}
\sum_{i=1}^m|C(\mu_i)|\leq  \sum_{v \in t }\lg(1/\zeta_{h_k^s(v,t)}(\childtype(v,t)))+O(m\log\mu + mk).
\end{align*}
With $m=\Theta(n/\log n)$ and $\mu=\Theta(\log n)$ (see \wref{sec:hypersuccinct-code-ordinal}), we have
\begin{align*}
\sum_{i=1}^m|C(\mu_i)|\leq \log\left(\frac{1}{\mathbb{P}_{\zeta}[t]}\right) + O\left(\frac{n \log \log  n + kn}{\log n}\right).
\end{align*}
This finishes the proof.
\end{proof}

From \wref{lem:hypersuccinct-code-ordinal}, \wref{lem:shape-entropyhuffman} and \wref{cor:ordinalshapeentropy}, we now find the following:

\begin{corollary}\label{cor:shape-entropy}
The hypersuccinct code $\mathsf{H}: \mathfrak{T} \to \{0,1\}^{\star}$ satisfies
\begin{align*}
		|\mathsf{H}(t)|
	\wwrel\leq 
		\mathcal{H}_k^s(t) +  O\left(\frac{n \log \log n + kn}{\log n}\right)
\end{align*}
for every ordinal tree $t \in \mathfrak{T}$ of size $n$.
\end{corollary}
It remains to remark that the above result from \wref{cor:shape-entropy} requires $k \in o(\log n)$ in order to be non-trivial: This bound on $k$ also occurs in \cite{HuckeLohreySeelbachBenkner2019}.

	\clearpage
	\section{Notation Index}
	\label{sec:notation}
	
	We collect used notation here for reference.
	
	\subsection{Elementary Notation}
	\begin{notations}
	\notation{$\N$, $\N_0$}{natural numbers without 0 (resp., with $0$), $\N=\{1,2,\ldots\}$, $\N_0=\{0,1,2,\ldots\}$}
	\notation{$\ln(n)$, $\lg(n)$}
		natural and binary logarithm; $\ln(n) = \log_e(n)$, $\lg(n) = \log_2(n)$.
	\notation{$[m..n]$, $[n]$}
		integer intervals, $[k..n] = \{k,k+1,\ldots,n\}$;
		$[n] = [1..n]$.
	\notation{$\Oh(f(n))$, $\Omega$, $\Theta$, $\sim$}
		asymptotic notation as defined, \eg, in \cite[\S\,A.2]{FlajoletSedgewick2009};
		in particular, $f\sim g$ means $f = g(1 + o(1))$;
		$f=g\pm\Oh(h)$ is equivalent to $|f-g| \in \Oh(|h|)$.
	\notation{$x \pm y$}
		$x$ with absolute error $|y|$; formally the interval $x \pm y = [x-|y|,x+|y|]$;
		as with $\Oh$-terms, we use ``one-way equalities'': $z=x\pm y$ instead of $z \in x \pm y$.
	\end{notations}
	
	\subsection{Tree Notation}
	\begin{notations}
	\notation{$\mathcal T_n$, $\mathcal T$}
		set of binary tree over $n$ nodes, 
		$\mathcal T = \bigcup_{n\ge0} \mathcal T_n$
	\notation{$\mathcal T^h$}
		set of binary tree of height $h$
	\notation{$\mathfrak T_n$, $\mathfrak T$}
		set of ordinal tree over $n$ nodes,
		$\mathfrak T = \bigcup_{n\ge0} \mathfrak T_n$
	\notation{$\mathfrak T^h$}
		set of ordinal trees of height $h$
	\notation{$\mathfrak{F}$} the set of all forests, \ie, (possibly empty) sequences of trees from $\mathfrak{T}$	
	\notation{$\Lambda$}
		the empty tree ``null''
	\notation{$v\in t$}
		$v$ is a node in tree $t$; 
		unless indicated otherwise, we identify nodes with their preorder rank
	\notation{$|t|$}
		number of nodes in $t$, \ie, 
		$t\in\mathcal T_n$ or $t\in\mathfrak T_n$ implies $|t|=n$
	\notation{$h(t)$} height of the tree $t$
	\notation{$\operatorname{type}(v)$} type of a node of a binary tree (leaf, left-unary, right-unary or binary)	
	\notation{$\deg(v)$}
		degree of $v$, \ie, the number of children of $v$
	\notation{$t[v]$}
		subtree of $t$ rooted at $v$; 
		if $v$ does not occur in $t$, $t[v]=\Lambda$
	\notation{$t_\ell[v]$, $t_r[v]$}
		left resp.\ right subtree of $v\in t \in\mathcal T$
	\notation{$t_{\ell}$, $t_r$} left resp.\ right subtree of the root of tree $t \in \mathcal{T}$	
	\notation{$t_k[v]$}
		$k$th subtree of $v\in t\in \mathfrak T$, for $k\in[\deg(v)]$
	\notation{$\mathit{BP}(t)$} balanced parenthesis encoding of the binary tree $t \in \mathcal{T}$, see \wref{def:bp-sequence}
	\notation{$\mathit{BP}_o (t)$} balanced parenthesis encoding of the ordinal tree $t \in \mathfrak{T}$, see \wref{def:bp-sequence}	
	\notation{$\FCNS(t)$} the first-child next-sibling encoding of the binary tree $t$, see \wref{def:fcns}
	\notation{$h(v)$, $h_k(v)$} ($k$-) history of a node $v$ of a binary tree: string consisting of the node types of $v$'s ($k$ closest) ancestors
	\notation{$n_z^t$} number of nodes of $t$ with $k$-history $z$
	\notation{$n_{z,i}^t$} number of nodes of $t$ with $k$-history $z$ and type $i$
	\notation{$\nu_i^t$} number of nodes of degree $i$ of $t$
	\notation{$n_{b}(t), n_{\geq b}(t)$} number of nodes of $t$ with $|t[v]|=b$, resp. $|t[v]|\geq b$
	\notation{\TrHeight etc}
		operations on trees; see \wref{tab:binary-operations} and~\wref{tab:operations}
	\end{notations}
	
	\subsection{Tree Covering}
	\begin{notations}
	\notation{$B$}
		parameter of micro tree size, $B = \lceil\frac18\lg n\rceil$
	\notation{$\mu$}
		$\mu=\frac14 \lg n$ maximal micro tree size
	\notation{$m$}
		number of micro trees, $m=\Theta(n/B)$
	\notation{$\mu_1,\ldots,\mu_m$}
		micro trees in preorder of their roots, 
		with ties broken by next node in micro tree
	\notation{$\Upsilon$} top tier tree, obtained by contracting each micro tree into a single node
	\notation{$\Sigma_{\mu}$} set consisting of (the different shapes of) micro trees $\mu_1, \dots, \mu_m$	
	\end{notations}
	
	\subsection{Tree Sources}
	\begin{notations}
	\notation{$T_n$} a random tree of size $n$, \ie, a random variable taking values in $\mathcal{T}_n$ or $\mathfrak{T}_n$ with respect to some probability distribution
	\notation{$\tau=(\tau_z)_{z \in \{1,2,3\}^k}$} a $k$th-order type process, see \wref{sec:memoryless-binary}
	\notation{$H_k^{\type}(t)$} $k$th-order empirical type entropy of a binary tree $t$, see \wref{def:empiricaltype}
	\notation{$d=(d_i)_{i \in \mathbb{N}_0}$} a degree distribution, see \wref{sec:degreeentropy}
	\notation{$H^{\deg}(t)$} the degree entropy of an ordinal tree $t$, see \wref{def:degree-entropy}
	\notation{$\mathcal S_{\mathit{fs}}(p)$}
		fixed-size binary tree source induced by $p$, see \wref{sec:fixedsize}
	\notation{$\mathcal S_{\mathit{fh}}(p)$}
		fixed-height binary tree source induced by $p$, see \wref{sec:fixedheight}
		\notation{$\mathfrak S_{\mathit{fs}}(p)$}
		fixed-size ordinal tree source induced by $p$, see \wref{sec:fixedsizeordinal}
		\notation{$\mathfrak{S}_{\mathit{fcns}}(\mathcal{S})$} $\operatorname{FCNS}$-source of the fixed-size binary tree source $\mathcal{S}$, see \wref{def:fcnssource}
	\notation{$H_n(\mathcal S_{\mathit{fs}}(p))$} entropy induced by the fixed-size source $\mathcal S_{\mathit{fs}}(p)$ over the set $\mathcal{T}_n$, \wref{sec:lower-bound}
	\notation{$H_h(\mathcal S_{\mathit{fh}}(p))$} entropy induced by the fixed-height source $\mathcal S_{\mathit{fh}}(p)$ over the set $\mathcal{T}^h$, see \wref{sec:lower-bound}
	\notation{$\mathcal{T}_n(\mathcal{P})$, $\mathcal{T}(\mathcal{P})$} set of binary trees of size $n$ which satisfy property $\mathcal{P}$, $\mathcal{T}(\mathcal{P})=\bigcup_{n \geq 0}\mathcal{T}_n(\mathcal{P})$
	\notation{$\mathcal{U}_{\mathcal{P}}$} uniform subclass source with respect to property $\mathcal{P}$, see \wref{sec:uniform-binary}
	\notation{$\mathcal{T}(\mathcal{A})$} set of AVL trees, see \wref{exm:avl-uniform-height} and \wref{exm:AVLsize}
	\notation{$\mathcal{T}(\mathcal{R})$} set of red-black trees, see \wref{exm:redblack}
	\notation{$\mathcal{T}(\mathcal{W}_{\alpha})$} set of $\alpha$-weight-balanced trees, see \wref{exm:weightbalanced}
	\end{notations}

\clearpage
\bibliography{range-min-max}

\begin{thebibliography}{10}

\bibitem{succinct}
\url{http://github.com/ot/succinct}.

\bibitem{AlstrupHolmDeLichtenbergThorup2005}
Stephen Alstrup, Jacob Holm, Kristian De~Lichtenberg, and Mikkel Thorup.
\newblock Maintaining information in fully dynamic trees with top trees.
\newblock {\em {ACM} Transactions on Algorithms}, 1(2):243--264, October 2005.
\newblock \href {https://doi.org/10.1145/1103963.1103966}
  {\path{doi:10.1145/1103963.1103966}}.

\bibitem{ArroyueloRaman2019}
Diego Arroyuelo and Rajeev Raman.
\newblock Adaptive succinctness.
\newblock In {\em String Processing and Information Retrieval (SPIRE)}, pages
  467--481. Springer International Publishing, 2019.
\newblock \href {https://doi.org/10.1007/978-3-030-32686-9_33}
  {\path{doi:10.1007/978-3-030-32686-9_33}}.

\bibitem{BarbayFischerNavarro2012}
J{\'{e}}r{\'{e}}my Barbay, Johannes Fischer, and Gonzalo Navarro.
\newblock {LRM}-trees: Compressed indices, adaptive sorting, and compressed
  permutations.
\newblock {\em Theoretical Computer Science}, 459:26--41, November 2012.
\newblock \href {https://doi.org/10.1016/j.tcs.2012.08.010}
  {\path{doi:10.1016/j.tcs.2012.08.010}}.

\bibitem{BaumstarkGogHeuerLabeit2017}
Niklas Baumstark, Simon Gog, Tobias Heuer, and Julian Labeit.
\newblock Practical range minimum queries revisited.
\newblock In Costas~S. Iliopoulos, Solon~P. Pissis, Simon~J. Puglisi, and
  Rajeev Raman, editors, {\em International Symposium on Experimental
  Algorithms (SEA)}, volume~75 of {\em LIPIcs}, pages 12:1--12:16. Schloss
  Dagstuhl--Leibniz-Zentrum fuer Informatik, 2017.
\newblock \href {https://doi.org/10.4230/LIPIcs.SEA.2017.12}
  {\path{doi:10.4230/LIPIcs.SEA.2017.12}}.

\bibitem{BelazzouguiMakinenValenzuela2014}
Djamal Belazzougui, Veli M\"{a}kinen, and Daniel Valenzuela.
\newblock Compressed suffix array.
\newblock In {\em Encyclopedia of Algorithms}, pages 1--6. Springer {US}, 2014.
\newblock \href {https://doi.org/10.1007/978-3-642-27848-8_82-2}
  {\path{doi:10.1007/978-3-642-27848-8_82-2}}.

\bibitem{BilleGortzLandauWeimann2015}
Philip Bille, Inge~Li G{\o}rtz, Gad~M. Landau, and Oren Weimann.
\newblock Tree compression with top trees.
\newblock {\em Information and Computation}, 243:166--177, August 2015.
\newblock \href {https://doi.org/10.1016/j.ic.2014.12.012}
  {\path{doi:10.1016/j.ic.2014.12.012}}.

\bibitem{BilleLandauRamanSadakaneRaoWeimann2015}
Philip Bille, Gad~M. Landau, Rajeev Raman, Kunihiko Sadakane, Srinivasa
  Rao~Satti, and Oren Weimann.
\newblock Random access to grammar-compressed strings and trees.
\newblock {\em {SIAM} Journal on Computing}, 44(3):513--539, January 2015.
\newblock \href {https://doi.org/10.1137/130936889}
  {\path{doi:10.1137/130936889}}.

\bibitem{ChoiSzpankowski2012}
Yongwook Choi and Wojciech Szpankowski.
\newblock Compression of graphical structures: Fundamental limits, algorithms,
  and experiments.
\newblock {\em {IEEE} Transactions on Information Theory}, 58(2):620--638,
  February 2012.
\newblock \href {https://doi.org/10.1109/tit.2011.2173710}
  {\path{doi:10.1109/tit.2011.2173710}}.

\bibitem{CoverThomas2006}
Thomas~M. Cover and Joy~A. Thomas.
\newblock {\em Elements of Information Theory}.
\newblock Wiley Interscience, 2nd edition, 2006.

\bibitem{DavoodiNavarroRamanRao2014}
Pooya Davoodi, Gonzalo Navarro, Rajeev Raman, and Srinivasa Rao~Satti.
\newblock Encoding range minima and range top-2 queries.
\newblock {\em Philosophical Transactions of the Royal Society A: Mathematical,
  Physical and Engineering Sciences}, 372(2016):20130131--20130131, apr 2014.
\newblock \href {https://doi.org/10.1098/rsta.2013.0131}
  {\path{doi:10.1098/rsta.2013.0131}}.

\bibitem{DavoodiRamanRao2017}
Pooya Davoodi, Rajeev Raman, and Srinivasa~Rao Satti.
\newblock On succinct representations of binary trees.
\newblock {\em Mathematics in Computer Science}, 11(2):177--189, March 2017.
\newblock \href {https://doi.org/10.1007/s11786-017-0294-4}
  {\path{doi:10.1007/s11786-017-0294-4}}.

\bibitem{DowneySethiTarjan1980}
Peter~J. Downey, Ravi Sethi, and Robert~Endre Tarjan.
\newblock Variations on the common subexpression problem.
\newblock {\em Journal of the {ACM} ({JACM})}, 27(4):758--771, October 1980.
\newblock \href {https://doi.org/10.1145/322217.322228}
  {\path{doi:10.1145/322217.322228}}.

\bibitem{DudekG18}
Bartlomiej Dudek and Pawel Gawrychowski.
\newblock Slowing down top trees for better worst-case compression.
\newblock In Gonzalo Navarro, David Sankoff, and Binhai Zhu, editors, {\em
  Annual Symposium on Combinatorial Pattern Matching, {CPM} 2018, July 2-4,
  2018 - Qingdao, China}, volume 105 of {\em LIPIcs}, pages 16:1--16:8. Schloss
  Dagstuhl - Leibniz-Zentrum f{\"{u}}r Informatik, 2018.
\newblock \href {https://doi.org/10.4230/LIPIcs.CPM.2018.16}
  {\path{doi:10.4230/LIPIcs.CPM.2018.16}}.

\bibitem{EffrosVisweswariahKulkarniVerdu2002}
M.~Effros, K.~Visweswariah, S.~R. Kulkarni, and S.~Verdu.
\newblock Universal lossless source coding with the burrows wheeler transform.
\newblock {\em {IEEE} Transactions on Information Theory}, 48(5):1061--1081,
  May 2002.
\newblock \href {https://doi.org/10.1109/18.995542}
  {\path{doi:10.1109/18.995542}}.

\bibitem{FarzanMunro2014}
Arash Farzan and J.~Ian Munro.
\newblock A uniform paradigm to succinctly encode various families of trees.
\newblock {\em Algorithmica}, 68(1):16--40, June 2014.
\newblock \href {https://doi.org/10.1007/s00453-012-9664-0}
  {\path{doi:10.1007/s00453-012-9664-0}}.

\bibitem{FarzanRamanRao2009}
Arash Farzan, Rajeev Raman, and S.~Srinivasa Rao.
\newblock Universal succinct representations of trees?
\newblock In {\em International Colloquium on Automata, Languages and
  Programming (ICALP)}, pages 451--462. Springer, 2009.
\newblock \href {https://doi.org/10.1007/978-3-642-02927-1_38}
  {\path{doi:10.1007/978-3-642-02927-1_38}}.

\bibitem{FerradaNavarro2017}
H{\'{e}}ctor Ferrada and Gonzalo Navarro.
\newblock Improved range minimum queries.
\newblock {\em Journal of Discrete Algorithms}, 43:72--80, mar 2017.
\newblock \href {https://doi.org/10.1016/j.jda.2016.09.002}
  {\path{doi:10.1016/j.jda.2016.09.002}}.

\bibitem{FerraginaManzini2000}
P.~Ferragina and G.~Manzini.
\newblock Opportunistic data structures with applications.
\newblock In {\em Annual Symposium on Foundations of Computer Science
  ({FOCS})}. {IEEE} Comput. Soc, 2000.
\newblock \href {https://doi.org/10.1109/sfcs.2000.892127}
  {\path{doi:10.1109/sfcs.2000.892127}}.

\bibitem{FerraginaVenturini2007}
Paolo Ferragina and Rossano Venturini.
\newblock A simple storage scheme for strings achieving entropy bounds.
\newblock {\em Theoretical Computer Science}, 372(1):115--121, March 2007.
\newblock \href {https://doi.org/10.1016/j.tcs.2006.12.012}
  {\path{doi:10.1016/j.tcs.2006.12.012}}.

\bibitem{FischerHeun2011}
Johannes Fischer and Volker Heun.
\newblock Space-efficient preprocessing schemes for range minimum queries on
  static arrays.
\newblock {\em {SIAM} Journal on Computing}, 40(2):465--492, January 2011.
\newblock \href {https://doi.org/10.1137/090779759}
  {\path{doi:10.1137/090779759}}.

\bibitem{FlajoletSedgewick2009}
Philippe Flajolet and Robert Sedgewick.
\newblock {\em Analytic Combinatorics}.
\newblock Cambridge University Press, 2009.
\newblock (available on author's website:
  \url{http://algo.inria.fr/flajolet/Publications/book.pdf}).

\bibitem{GabowBentleyTarjan1984}
Harold~N. Gabow, Jon~Louis Bentley, and Robert~E. Tarjan.
\newblock Scaling and related techniques for geometry problems.
\newblock In {\em STOC 1984}. {ACM} Press, 1984.
\newblock \href {https://doi.org/10.1145/800057.808675}
  {\path{doi:10.1145/800057.808675}}.

\bibitem{GanardiHuckeJezLohreyNoeth2017}
Moses Ganardi, Danny Hucke, Artur Jez, Markus Lohrey, and Eric Noeth.
\newblock Constructing small tree grammars and small circuits for formulas.
\newblock {\em Journal of Computer and System Sciences}, 86:136--158, June
  2017.
\newblock \href {https://doi.org/10.1016/j.jcss.2016.12.007}
  {\path{doi:10.1016/j.jcss.2016.12.007}}.

\bibitem{GanardiHuckeLohreySeelbachBenkner2019}
Moses Ganardi, Danny Hucke, Markus Lohrey, and Louisa~Seelbach Benkner.
\newblock Universal tree source coding using grammar-based compression.
\newblock {\em {IEEE} Transactions on Information Theory}, 65(10):6399--6413,
  October 2019.
\newblock \href {https://doi.org/10.1109/tit.2019.2919829}
  {\path{doi:10.1109/tit.2019.2919829}}.

\bibitem{GanardiHuckeLohreyNoeth2017}
Moses Ganardi, Danny Hucke, Markus Lohrey, and Eric Noeth.
\newblock Tree compression using string grammars.
\newblock {\em Algorithmica}, 80(3):885--917, February 2017.
\newblock \href {https://doi.org/10.1007/s00453-017-0279-3}
  {\path{doi:10.1007/s00453-017-0279-3}}.

\bibitem{GanardiJL19}
Moses Ganardi, Artur Jez, and Markus Lohrey.
\newblock Balancing straight-line programs.
\newblock In David Zuckerman, editor, {\em 60th {IEEE} Annual Symposium on
  Foundations of Computer Science, {FOCS} 2019, Baltimore, Maryland, USA,
  November 9-12, 2019}, pages 1169--1183. {IEEE} Computer Society, 2019.
\newblock \href {https://doi.org/10.1109/FOCS.2019.00073}
  {\path{doi:10.1109/FOCS.2019.00073}}.

\bibitem{Ganczorz19}
Michal Ganczorz.
\newblock Entropy lower bounds for dictionary compression.
\newblock In Nadia Pisanti and Solon~P. Pissis, editors, {\em 30th Annual
  Symposium on Combinatorial Pattern Matching, {CPM} 2019, June 18-20, 2019,
  Pisa, Italy}, volume 128 of {\em LIPIcs}, pages 11:1--11:18. Schloss Dagstuhl
  - Leibniz-Zentrum f{\"{u}}r Informatik, 2019.
\newblock \href {https://doi.org/10.4230/LIPIcs.CPM.2019.11}
  {\path{doi:10.4230/LIPIcs.CPM.2019.11}}.

\bibitem{Ganczorz2020}
Micha{\l} Ga{\'{n}}czorz.
\newblock Using statistical encoding to achieve tree succinctness never seen
  before.
\newblock In {\em Symposium on Theoretical Aspects of Computer Science
  (STACS)}, LIPIcs, pages 22:1--22:29. Schloss Dagstuhl - Leibniz-Zentrum
  f\"{u}r Informatik, 2020.
\newblock \href {https://doi.org/10.4230/LIPICS.STACS.2020.22}
  {\path{doi:10.4230/LIPICS.STACS.2020.22}}.

\bibitem{GasconLohreyManethRehSieber2020}
Adri{\`{a}} Gasc{\'{o}}n, Markus Lohrey, Sebastian Maneth, Carl~Philipp Reh,
  and Kurt Sieber.
\newblock Grammar-based compression of unranked trees.
\newblock {\em Theory of Computing Systems}, 64(1):141--176, 2020.
\newblock \href {https://doi.org/10.1007/s00224-019-09942-y}
  {\path{doi:10.1007/s00224-019-09942-y}}.

\bibitem{GawrychowskiJez2016}
Pawel Gawrychowski and Artur Jez.
\newblock {LZ77 Factorisation of Trees}.
\newblock In Akash Lal, S.~Akshay, Saket Saurabh, and Sandeep Sen, editors,
  {\em Annual Conference on Foundations of Software Technology and Theoretical
  Computer Science (FSTTCS 2016)}, volume~65 of {\em LIPIcs}, pages
  35:1--35:15, Dagstuhl, Germany, 2016. Schloss Dagstuhl--Leibniz-Zentrum fuer
  Informatik.
\newblock \href {https://doi.org/10.4230/LIPIcs.FSTTCS.2016.35}
  {\path{doi:10.4230/LIPIcs.FSTTCS.2016.35}}.

\bibitem{GawrychowskiJoMozesWeimann2020}
Pawe{\l} Gawrychowski, Seungbum Jo, Shay Mozes, and Oren Weimann.
\newblock Compressed range minimum queries.
\newblock {\em Theoretical Computer Science}, 812:39--48, April 2020.
\newblock \href {https://doi.org/10.1016/j.tcs.2019.07.002}
  {\path{doi:10.1016/j.tcs.2019.07.002}}.

\bibitem{GawrychowskiNicholson2015}
Pawe{\l} Gawrychowski and Patrick~K. Nicholson.
\newblock Optimal encodings for range top-$k$, selection, and min-max.
\newblock In {\em International Colloquium on Automata, Languages, and
  Programming (ICALP)}, pages 593--604, 2015.
\newblock \href {https://doi.org/10.1007/978-3-662-47672-7_48}
  {\path{doi:10.1007/978-3-662-47672-7_48}}.

\bibitem{GearyRamanRaman2006}
Richard~F. Geary, Rajeev Raman, and Venkatesh Raman.
\newblock Succinct ordinal trees with level-ancestor queries.
\newblock {\em {ACM} Transactions on Algorithms}, 2(4):510--534, October 2006.
\newblock \href {https://doi.org/10.1145/1198513.1198516}
  {\path{doi:10.1145/1198513.1198516}}.

\bibitem{GogBellerMoffatPetri2014}
Simon Gog, Timo Beller, Alistair Moffat, and Matthias Petri.
\newblock From theory to practice: Plug and play with succinct data structures.
\newblock In {\em International Symposium on Experimental Algorithms (SEA)},
  pages 326--337, 2014.
\newblock \href {https://doi.org/10.1007/978-3-319-07959-2_28}
  {\path{doi:10.1007/978-3-319-07959-2_28}}.

\bibitem{GolebiewskiMagnerSzpankowski2019}
Zbigniew Go{\l}{\k{e}}biewski, Abram Magner, and Wojciech Szpankowski.
\newblock Entropy and optimal compression of some general plane trees.
\newblock {\em {ACM} Transactions on Algorithms}, 15(1):1--23, January 2019.
\newblock \href {https://doi.org/10.1145/3275444} {\path{doi:10.1145/3275444}}.

\bibitem{GolinIaconoKrizancRamanRaoShende2016}
Mordecai Golin, John Iacono, Danny Krizanc, Rajeev Raman, Srinivasa~Rao Satti,
  and Sunil Shende.
\newblock Encoding 2d range maximum queries.
\newblock {\em Theoretical Computer Science}, 609:316--327, January 2016.
\newblock \href {https://doi.org/10.1016/j.tcs.2015.10.012}
  {\path{doi:10.1016/j.tcs.2015.10.012}}.

\bibitem{GonzalezNavarro2006}
Rodrigo Gonz{\'{a}}lez and Gonzalo Navarro.
\newblock Statistical encoding of succinct data structures.
\newblock In {\em Combinatorial Pattern Matching}, pages 294--305. Springer
  Berlin Heidelberg, 2006.
\newblock \href {https://doi.org/10.1007/11780441_27}
  {\path{doi:10.1007/11780441_27}}.

\bibitem{GrahamKnuthPatashnik1994}
Ronald~L. Graham, Donald~E. Knuth, and Oren Patashnik.
\newblock {\em Concrete Mathematics: A Foundation For Computer Science}.
\newblock Addison-Wesley, 1994.

\bibitem{Grossi2013}
Roberto Grossi.
\newblock Random access to high-order entropy compressed text.
\newblock In {\em Lecture Notes in Computer Science}, pages 199--215. Springer
  Berlin Heidelberg, 2013.
\newblock \href {https://doi.org/10.1007/978-3-642-40273-9_14}
  {\path{doi:10.1007/978-3-642-40273-9_14}}.

\bibitem{GrossiVitter2000}
Roberto Grossi and Jeffrey~Scott Vitter.
\newblock Compressed suffix arrays and suffix trees with applications to text
  indexing and string matching (extended abstract).
\newblock In {\em {ACM} Symposium on Theory of Computing ({STOC})}. {ACM}
  Press, 2000.
\newblock \href {https://doi.org/10.1145/335305.335351}
  {\path{doi:10.1145/335305.335351}}.

\bibitem{GrossiVitter2005}
Roberto Grossi and Jeffrey~Scott Vitter.
\newblock Compressed suffix arrays and suffix trees with applications to text
  indexing and string matching.
\newblock {\em {SIAM} Journal on Computing}, 35(2):378--407, January 2005.
\newblock \href {https://doi.org/10.1137/s0097539702402354}
  {\path{doi:10.1137/s0097539702402354}}.

\bibitem{Gusfield1997}
Dan Gusfield.
\newblock {\em Algorithms on Strings, Trees and Sequences}.
\newblock Cambridge University Press, 1997.

\bibitem{HeMunroRao2012}
Meng He, J.~Ian Munro, and Srinivasa~Satti Rao.
\newblock Succinct ordinal trees based on tree covering.
\newblock {\em {ACM} Transactions on Algorithms}, 8(4):1--32, September 2012.
\newblock \href {https://doi.org/10.1145/2344422.2344432}
  {\path{doi:10.1145/2344422.2344432}}.

\bibitem{HubschleSchneiderRaman15}
Lorenz H{\"{u}}bschle{-}Schneider and Rajeev Raman.
\newblock Tree compression with top trees revisited.
\newblock In Evripidis Bampis, editor, {\em Experimental Algorithms - 14th
  International Symposium, {SEA} 2015, Paris, France, June 29 - July 1, 2015,
  Proceedings}, volume 9125 of {\em Lecture Notes in Computer Science}, pages
  15--27. Springer, 2015.
\newblock \href {https://doi.org/10.1007/978-3-319-20086-6\_2}
  {\path{doi:10.1007/978-3-319-20086-6\_2}}.

\bibitem{HuckeLohreySeelbachBenkner2019}
Danny Hucke, Markus Lohrey, and Louisa~Seelbach Benkner.
\newblock Entropy bounds for grammar-based tree compressors.
\newblock In {\em {IEEE} International Symposium on Information Theory
  ({ISIT})}. {IEEE}, July 2019.
\newblock \href {https://doi.org/10.1109/isit.2019.8849372}
  {\path{doi:10.1109/isit.2019.8849372}}.

\bibitem{HuckeLohreySeelbachBenkner2020}
Danny Hucke, Markus Lohrey, and Louisa~Seelbach Benkner.
\newblock A comparison of empirical tree entropies.
\newblock In Christina Boucher and Sharma~V. Thankachan, editors, {\em String
  Processing and Information Retrieval - 27th International Symposium, {SPIRE}
  2020, Orlando, FL, USA, October 13-15, 2020, Proceedings}, volume 12303 of
  {\em Lecture Notes in Computer Science}, pages 232--246. Springer, 2020.
\newblock \href {https://doi.org/10.1007/978-3-030-59212-7\_17}
  {\path{doi:10.1007/978-3-030-59212-7\_17}}.

\bibitem{HwangNeininger2002}
Hsien-Kuei Hwang and Ralph Neininger.
\newblock Phase change of limit laws in the quicksort recurrence under varying
  toll functions.
\newblock {\em {SIAM} Journal on Computing}, 31(6):1687--1722, jan 2002.
\newblock \href {https://doi.org/10.1137/s009753970138390x}
  {\path{doi:10.1137/s009753970138390x}}.

\bibitem{OEIS-Narayana-numbers}
OEIS~Foundation Inc.
\newblock {The On-Line Encyclopedia of Integer Sequences, A001263}, 2021.
\newblock URL: \url{https://oeis.org/A001263}.

\bibitem{Jacobson1989}
G.~Jacobson.
\newblock Space-efficient static trees and graphs.
\newblock In {\em Symposium on Foundations of Computer Science (FOCS)}. {IEEE},
  1989.
\newblock \href {https://doi.org/10.1109/sfcs.1989.63533}
  {\path{doi:10.1109/sfcs.1989.63533}}.

\bibitem{JacquetSzpankowski2015}
Philippe Jacquet and Wojciech Szpankowski.
\newblock {\em Analytic Pattern Matching}.
\newblock Cambridge University Press, 2015.

\bibitem{JanssonSadakaneSung2012}
Jesper Jansson, Kunihiko Sadakane, and Wing-Kin Sung.
\newblock Ultra-succinct representation of ordered trees with applications.
\newblock {\em Journal of Computer and System Sciences}, 78(2):619--631, March
  2012.
\newblock \href {https://doi.org/10.1016/j.jcss.2011.09.002}
  {\path{doi:10.1016/j.jcss.2011.09.002}}.

\bibitem{KempaPrezza2018}
Dominik Kempa and Nicola Prezza.
\newblock At the roots of dictionary compression: string attractors.
\newblock In {\em Annual {ACM} {SIGACT} Symposium on Theory of Computing
  ({STOC})}. {ACM} Press, 2018.
\newblock \href {https://doi.org/10.1145/3188745.3188814}
  {\path{doi:10.1145/3188745.3188814}}.

\bibitem{KiefferYang2000}
J.C. Kieffer and En-Hui Yang.
\newblock Grammar-based codes: a new class of universal lossless source codes.
\newblock {\em {IEEE} Transactions on Information Theory}, 46(3):737--754, May
  2000.
\newblock \href {https://doi.org/10.1109/18.841160}
  {\path{doi:10.1109/18.841160}}.

\bibitem{KiefferYangSzpankowski2009}
John~C. Kieffer, En-Hui Yang, and Wojciech Szpankowski.
\newblock Structural complexity of random binary trees.
\newblock In {\em 2009 {IEEE} International Symposium on Information Theory}.
  {IEEE}, jun 2009.
\newblock \href {https://doi.org/10.1109/isit.2009.5205704}
  {\path{doi:10.1109/isit.2009.5205704}}.

\bibitem{Lohrey2015}
Markus Lohrey.
\newblock Grammar-based tree compression.
\newblock In {\em International Conference on Developments in Language Theory},
  pages 46--57. Springer, 2015.

\bibitem{LohreyManethSchmidtSchauss2012}
Markus Lohrey, Sebastian Maneth, and Manfred Schmidt-Schau{\ss}.
\newblock Parameter reduction and automata evaluation for grammar-compressed
  trees.
\newblock {\em Journal of Computer and System Sciences}, 78(5):1651--1669,
  September 2012.
\newblock \href {https://doi.org/10.1016/j.jcss.2012.03.003}
  {\path{doi:10.1016/j.jcss.2012.03.003}}.

\bibitem{LohreyRehSieber2017}
Markus Lohrey, Carl~Philipp Reh, and Kurt Sieber.
\newblock Optimal top dag compression, 2017.
\newblock \href {http://arxiv.org/abs/1712.05822} {\path{arXiv:1712.05822}}.

\bibitem{LuczakMagnerSzpankowski2019}
Tomasz Luczak, Abram Magner, and Wojciech Szpankowski.
\newblock Compression of preferential attachment graphs.
\newblock In {\em 2019 {IEEE} International Symposium on Information Theory
  ({ISIT})}. {IEEE}, July 2019.
\newblock \href {https://doi.org/10.1109/isit.2019.8849739}
  {\path{doi:10.1109/isit.2019.8849739}}.

\bibitem{MagnerTurowskiSzpankowski2018}
Abram Magner, Krzysztof Turowski, and Wojciech Szpankowski.
\newblock Lossless compression of binary trees with correlated vertex names.
\newblock {\em {IEEE} Transactions on Information Theory}, 64(9):6070--6080,
  sep 2018.
\newblock \href {https://doi.org/10.1109/tit.2018.2851224}
  {\path{doi:10.1109/tit.2018.2851224}}.

\bibitem{Martinez1992}
Conrado Mart{\'{\i}}nez.
\newblock {\em Statistics under the {BST} model}.
\newblock PhD thesis, University Barcelona, 1992.

\bibitem{MillerPippengerRosenbergSnyder1979}
Raymond~E. Miller, Nicholas Pippenger, Arnold~L. Rosenberg, and Lawrence
  Snyder.
\newblock Optimal 2,3-trees.
\newblock {\em {SIAM} Journal on Computing}, 8(1):42--59, February 1979.
\newblock \href {https://doi.org/10.1137/0208004} {\path{doi:10.1137/0208004}}.

\bibitem{MunroWild2019}
J.~Ian Munro and Sebastian Wild.
\newblock Entropy trees and range-minimum queries in optimal average-case
  space, 2019.
\newblock \href {http://arxiv.org/abs/1903.02533} {\path{arXiv:1903.02533}}.

\bibitem{Navarro2016}
Gonzalo Navarro.
\newblock {\em Compact Data Structures -- A practical approach}.
\newblock Cambridge University Press, 2016.

\bibitem{Navarro2021a}
Gonzalo Navarro.
\newblock Indexing highly repetitive string collections, part {I}:
  Repetitiveness measures.
\newblock {\em {ACM} Computing Surveys}, 54(2):29:1--29:36, February 2021.

\bibitem{Navarro2021b}
Gonzalo Navarro.
\newblock Indexing highly repetitive string collections, part {II}: Compressed
  indexes.
\newblock {\em {ACM} Computing Surveys}, 54(2):26:1--26:38, February 2021.
\newblock \href {https://doi.org/10.1145/3432999} {\path{doi:10.1145/3432999}}.

\bibitem{NavarroMakinen2007}
Gonzalo Navarro and Veli M\"{a}kinen.
\newblock Compressed full-text indexes.
\newblock {\em {ACM} Computing Surveys}, 39(1):2, April 2007.
\newblock \href {https://doi.org/10.1145/1216370.1216372}
  {\path{doi:10.1145/1216370.1216372}}.

\bibitem{NavarroSadakane2014}
Gonzalo Navarro and Kunihiko Sadakane.
\newblock Fully functional static and dynamic succinct trees.
\newblock {\em {ACM} Transactions on Algorithms}, 10(3):1--39, may 2014.
\newblock \href {https://doi.org/10.1145/2601073} {\path{doi:10.1145/2601073}}.

\bibitem{NievergeltReingold1973}
J{\"{u}}rg Nievergelt and Edward~M. Reingold.
\newblock Binary search trees of bounded balance.
\newblock {\em {SIAM} J. Comput.}, 2(1):33--43, 1973.
\newblock \href {https://doi.org/10.1137/0202005} {\path{doi:10.1137/0202005}}.

\bibitem{Odlyzko1984}
Andrew~M. Odlyzko.
\newblock Some new methods and results in tree enumeration.
\newblock {\em Congressus Numerantium}, 42:27--52, 1984.
\newblock URL:
  \url{http://www.dtc.umn.edu/~odlyzko/doc/arch/enumer.methods.pdf}.

\bibitem{Prezza2019}
Nicola Prezza.
\newblock {Optimal Rank and Select Queries on Dictionary-Compressed Text}.
\newblock In Nadia Pisanti and Solon~P. Pissis, editors, {\em Symposium on
  Combinatorial Pattern Matching (CPM)}, volume 128 of {\em LIPIcs}, pages
  4:1--4:12, Dagstuhl, Germany, 2019. Schloss Dagstuhl--Leibniz-Zentrum fuer
  Informatik.
\newblock \href {https://doi.org/10.4230/LIPIcs.CPM.2019.4}
  {\path{doi:10.4230/LIPIcs.CPM.2019.4}}.

\bibitem{RamanRamanRao2007}
Rajeev Raman, Venkatesh Raman, and Srinivasa~Rao Satti.
\newblock Succinct indexable dictionaries with applications to encoding k-ary
  trees, prefix sums and multisets.
\newblock {\em {ACM} Transactions on Algorithms}, 3(4):43--es, nov 2007.
\newblock \href {https://doi.org/10.1145/1290672.1290680}
  {\path{doi:10.1145/1290672.1290680}}.

\bibitem{RamanRao2013}
Rajeev Raman and S.~Srinivasa Rao.
\newblock Succinct representations of ordinal trees.
\newblock In Brodnik A., López-Ortiz A., Raman V., and Viola A., editors, {\em
  Space-Efficient Data Structures, Streams, and Algorithms}, volume 8066 of
  {\em LNCS}, pages 319--332. Springer, 2013.
\newblock \href {https://doi.org/10.1007/978-3-642-40273-9_20}
  {\path{doi:10.1007/978-3-642-40273-9_20}}.

\bibitem{Sedgewick2008}
Robert Sedgewick.
\newblock Left-leaning red-black trees, 2008.
\newblock URL: \url{http://www.cs.princeton.edu/~rs/talks/LLRB/LLRB.pdf}.

\bibitem{SedgewickFlajolet1996}
Robert Sedgewick and Philippe Flajolet.
\newblock {\em An introduction to the analysis of algorithms}.
\newblock Addison-Wesley-Longman, 1996.

\bibitem{SeelbachBenknerLohrey2018}
Louisa Seelbach~Benkner and Markus Lohrey.
\newblock {Average Case Analysis of Leaf-Centric Binary Tree Sources}.
\newblock In Igor Potapov, Paul Spirakis, and James Worrell, editors, {\em
  Symposium on Mathematical Foundations of Computer Science (MFCS)}, volume 117
  of {\em LIPIcs}, pages 16:1--16:15, Dagstuhl, Germany, 2018. Schloss
  Dagstuhl.
\newblock \href {https://doi.org/10.4230/LIPIcs.MFCS.2018.16}
  {\path{doi:10.4230/LIPIcs.MFCS.2018.16}}.

\bibitem{Tsur2018}
Dekel Tsur.
\newblock Representation of ordered trees with a given degree distribution,
  2018.
\newblock \href {http://arxiv.org/abs/1807.00371} {\path{arXiv:1807.00371}}.

\bibitem{VerbinYu2013}
Elad Verbin and Wei Yu.
\newblock Data structure lower bounds on random access to grammar-compressed
  strings.
\newblock In {\em Combinatorial Pattern Matching}, pages 247--258. Springer
  Berlin Heidelberg, 2013.
\newblock \href {https://doi.org/10.1007/978-3-642-38905-4_24}
  {\path{doi:10.1007/978-3-642-38905-4_24}}.

\bibitem{Wild2010}
Sebastian Wild.
\newblock {\em An {Earley}-style Parser for Solving the {RNA-RNA} Interaction
  Problem}.
\newblock Bachelor's thesis, TU Kaiserslautern, 2010.
\newblock URL: \url{https://nbn-resolving.org/urn:nbn:de:hbz:386-kluedo-22827}.

\bibitem{Wild2016}
Sebastian Wild.
\newblock {\em Dual-Pivot Quicksort and Beyond: Analysis of Multiway
  Partitioning and Its Practical Potential}.
\newblock Dissertation ({Ph.\,D.\ }thesis), 2016.
\newblock URL: \url{https://www.wild-inter.net/publications/wild-2016}.

\bibitem{Wild2018}
Sebastian Wild.
\newblock Quicksort is optimal for many equal keys.
\newblock In {\em Workshop on Analytic Algorithmics and Combinatorics
  ({ANALCO})}, pages 8--22. SIAM, 2018.
\newblock \href {http://arxiv.org/abs/1608.04906} {\path{arXiv:1608.04906}},
  \href {https://doi.org/10.1137/1.9781611975062.2}
  {\path{doi:10.1137/1.9781611975062.2}}.

\bibitem{WittenNealCleary1987}
Ian~H. Witten, Radford~M. Neal, and John~G. Cleary.
\newblock Arithmetic coding for data compression.
\newblock {\em Communications of the {ACM}}, 30(6):520--540, jun 1987.
\newblock \href {https://doi.org/10.1145/214762.214771}
  {\path{doi:10.1145/214762.214771}}.

\bibitem{ZhangYangKieffer2014}
Jie Zhang, En-Hui Yang, and John~C. Kieffer.
\newblock A universal grammar-based code for lossless compression of binary
  trees.
\newblock {\em {IEEE} Transactions on Information Theory}, 60(3):1373--1386,
  March 2014.
\newblock \href {https://doi.org/10.1109/tit.2013.2295392}
  {\path{doi:10.1109/tit.2013.2295392}}.

\end{thebibliography}

\end{document}